\documentclass[12pt]{article}
\usepackage[utf8]{inputenc}
\usepackage[margin=1in]{geometry}
\usepackage{multirow,multicol}
\usepackage{amsfonts}
\usepackage{amsmath,amssymb,amsthm}
\usepackage[mathscr]{euscript}
\usepackage{comment}
\usepackage{graphicx,color}
\usepackage{mathtools}
\usepackage{mathabx} %widecheck
\usepackage[linesnumbered,ruled,vlined]{algorithm2e}
\usepackage{setspace}
\usepackage{sectsty}
\usepackage{makecell}
\usepackage{chngcntr}
\usepackage{apptools}
\usepackage{caption}
\usepackage{subcaption}
\usepackage{booktabs}
\usepackage{etoolbox}

\usepackage{tabularx}
\newcolumntype{Y}{>{\raggedleft\arraybackslash}X}
\newcolumntype{Z}{>{\centering\arraybackslash}X}
\usepackage{enumitem}
\usepackage{underscore} 
\usepackage{array} 
\usepackage{arydshln}
\usepackage{soul}
\usepackage{tikz}
\usetikzlibrary{matrix}
\usepackage[english]{babel}
\usepackage[colorlinks=true,linkcolor=blue,urlcolor=black,citecolor=blue,bookmarksopen=true]{hyperref}
\usepackage{bookmark}
\usepackage{multirow} % use multirow

\usepackage{lipsum}
\usepackage{calrsfs} % for calligraphic letters in the OMS encoding 
\usepackage{pdflscape}
\usepackage{float}
\usepackage{natbib}
\usepackage{bibunits}

\DeclareMathAlphabet{\pazocal}{OMS}{zplm}{m}{n} % usepackage calrsfs

\DeclareFontFamily{OT1}{pzc}{}
\DeclareFontShape{OT1}{pzc}{m}{it}{<-> s * [1.10] pzcmi7t}{}
\DeclareMathAlphabet{\mathpzc}{OT1}{pzc}{m}{it}

\newtheorem{assumption}{Assumption}

\newtheorem{definition}{Definition}
\newtheorem{lemma}{Lemma}
\newtheorem{proposition}{Proposition}
\newtheorem{theorem}{Theorem}

\newtheorem{remark}{Remark}

\newcolumntype{P}[1]{>{\centering\arraybackslash}p{#1}}

\newcommand{\bm}{\boldsymbol}
\newcommand{\cm}[1]{\mbox{\boldmath$\mathscr{#1}$}}

\newcommand{\Fr}{{\mathrm{F}}}
\newcommand{\op}{{\mathrm{op}}}

\newcommand{\ma}{{\mathrm{MA}}}
\newcommand{\ar}{{\mathrm{AR}}}

\newcommand{\colsp}{\mathrm{colsp}}
\newcommand{\rsc}{{\mathrm{rsc}}}
\newcommand{\dev}{{\mathrm{dev}}}
\newcommand{\init}{{\mathrm{init}}}

\newcommand{\HW}{{\mathrm{HW}}}
\newcommand{\stcomp}[1]{{#1}^\complement} % set complement
\def\HH{{\mathrm{\scriptscriptstyle\mathsf{H}}}}

\DeclareMathOperator*{\vect}{vec}
\DeclareMathOperator*{\rank}{rank}
\DeclareMathOperator*{\trace}{tr}
\DeclareMathOperator*{\argmin}{arg\,min}

\DeclareMathOperator*{\diag}{diag}

\DeclareMathOperator*{\var}{var}

\renewcommand{\arraystretch}{1}

\numberwithin{equation}{section}

\title{\vspace{-30pt} An Interpretable and Efficient Infinite-Order Vector Autoregressive Model for High-Dimensional Time Series}

\author{Yao Zheng\\\textit{University of Connecticut}}
\date{}

\begin{document}
\begin{bibunit}[apalike]	
\setlength{\parindent}{16pt}
	
\maketitle
	
\begin{abstract}
As a special infinite-order vector autoregressive (VAR) model, the vector autoregressive moving average (VARMA) model can capture much richer temporal patterns than the widely used finite-order VAR model. However, its practicality has long been hindered by its non-identifiability, computational intractability, and difficulty of interpretation, especially for high-dimensional time series. This paper proposes a novel sparse infinite-order VAR model for high-dimensional time series, which avoids all above drawbacks while inheriting essential temporal patterns of the VARMA model. As another attractive feature, the temporal and cross-sectional structures of the VARMA-type dynamics captured by this model can be interpreted separately, since they are characterized by different sets of parameters. This separation naturally motivates the sparsity assumption on the parameters determining the cross-sectional dependence. As a result, greater statistical efficiency and interpretability can be achieved with little loss of temporal information. We introduce two $\ell_1$-regularized estimation methods for the proposed model, which can be efficiently implemented via block coordinate descent algorithms, and derive the corresponding nonasymptotic error bounds. A consistent model order selection method based on the Bayesian information criteria is also developed. The merit of the proposed approach is supported by simulation studies and a real-world macroeconomic data analysis.
\end{abstract}


%\textit{MSC2020 subject classifications}: Primary 62M10; secondary 62H12, 60G10

\newpage
\section{Introduction}\label{section:intro}

Let $\bm{y}_t\in\mathbb{R}^N$ be the observation of an $N$-dimensional time series at time $t$. The need for modeling $\bm{y}_t$ with a large dimension $N$ is ubiquitous, ranging from economics and finance \citep{NWBM20, WBBM21} to biology and neuroscience \citep{LZLR09, Gorrostieta2012}, and to environmental and health sciences \citep{DP16, DZZ16}. For modeling $\bm{y}_t$, three issues are of particular importance: 
\begin{itemize}[itemsep=2pt,parsep=2pt,topsep=4pt,partopsep=2pt]
	\item [(I1)] Flexibility of temporal dynamics: As $N$ increases, it is more likely that $\bm{y}_t$ contains component series with complex temporal dependence structures. Then information further in the past may be needed to generate more flexible temporal dynamics. 
	\item [(I2)] Efficiency: It is important that the estimation is efficient both statistically and computationally under large $N$, so that accurate forecasts can be obtained.
	\item [(I3)] Interpretability: Ideally, the model should have easy interpretations, such as direct  implications of Granger causality  \citep{Granger69} among the $N$ component series.
\end{itemize}

The  finite-order vector autoregressive (VAR) model, coupled with dimension reduction techniques such as sparse \citep{BM21} and low-rank \citep{Wang2021High}  methods, has been widely studied for high-dimensional time series. This model is highly popular due to its theoretical and computational tractability, and the coefficient matrices have intuitive interpretations analogous to those in the multivariate linear regression. However, in practice, a large lag order is often required   for the VAR model to adequately fit the data \citep{CEK16, NWBM20}. Thus, it is more realistic to assume that the data  follow the more general, infinite-order VAR (VAR($\infty$)) process:
\begin{equation} \label{eq:VARinf}
	\bm{y}_t =\sum_{h=1}^{\infty} \bm{A}_h\bm{y}_{t-h}+\bm{\varepsilon}_t,
\end{equation}
where $\bm{\varepsilon}_t$ are the innovations, and $\bm{A}_h\in\mathbb{R}^{N\times N}$ are the AR coefficient matrices; in particular, it reduces to the VAR($P$) model when $\bm{A}_h=\bm{0}$ for $h>P$. In fact, if a sample $\{\bm{y}_t\}_{t=1}^T$ is generated from \eqref{eq:VARinf}, we can approximate it by a VAR($P$) model provided that $P\rightarrow\infty$ at an appropriate rate as the sample size $T\rightarrow\infty$ \citep{Lutkepohl2005}, which in turn explains the practical need for a large $P$. Nonetheless, for $\bm{y}_t$ in \eqref{eq:VARinf} to be stationary, $\bm{A}_h$ must diminish quickly as $h\rightarrow\infty$; otherwise, the infinite sum will be ill-defined.
The decay property of $\bm{A}_h$, coupled with a large $P$,  will not only pose difficulties in high-dimensional estimation, but make the fitted VAR($P$) model hard to interpret. Take the Lasso estimator of the VAR($P$) model with sparse  $\bm{A}_h$'s. Since all entries of $\bm{A}_h$ must be small at even moderately large $h$,  the Lasso may fail to capture the significant yet small entries. Moreover, the sparsity pattern of  $\bm{A}_h$ for the fitted model generally  varies substantially across $h$, making it even more difficult to interpret $\bm{A}_h$'s simultaneously \citep{SBM12, NWBM20}.

In the literature on multivariate time series, an alternative approach to infinite-order VAR modeling is to consider  the  vector autoregressive moving average (VARMA) model. For example, the VARMA($1,1$) model is 
\begin{equation} \label{eq:VARMA11}
	\bm{y}_t = \bm{\Phi}\bm{y}_{t-1}+\bm{\varepsilon}_t -\bm{\Theta}\bm{\varepsilon}_{t-1},
\end{equation} 
where $\bm{\Phi}, \bm{\Theta}\in\mathbb{R}^{N\times N}$ are the AR and MA coefficient matrices.
Assuming that \eqref{eq:VARMA11} is invertible, that is, all eigenvalues of $\bm{\Theta}$ are less than one in absolute value, \eqref{eq:VARMA11} can be written as the VAR($\infty$) process in \eqref{eq:VARinf} with $\bm{A}_h=\bm{A}_h(\bm{\Phi},\bm{\Theta})=\bm{\Theta}^{h-1}(\bm{\Phi}-\bm{\Theta})$ for $h\geq 1$. Note that $\bm{A}_h$
diminishes quickly as $h\to \infty$ due to the exponential factor $\bm{\Theta}^{h-1}$, so the VAR($\infty$) process is well defined. Hence, the MA part of the model is the key to  parsimoniously generating VAR($\infty$)-type temporal dynamics. For the general VARMA($p,q$) model, $\bm{y}_t =\sum_{i=1}^{p} \bm{\Phi}_i\bm{y}_{t-i}+\bm{\varepsilon}_t - \sum_{j=1}^{q}\bm{\Theta}_j\bm{\varepsilon}_{t-j}$, the richness of temporal patterns will increase with $p$ and $q$, but with only small orders $p$ and $q$, the VARMA model   can usually provide more accurate forecasts than large-order VAR models in practice \citep{AV08, CEK16}. 
Compared with finite-order VAR models, the VARMA model is more favorable in terms of (I1) but suffers from severe drawbacks regarding (I2), as its computation is generally complicated due to the following two problems:
\begin{itemize}[itemsep=2pt,parsep=2pt,topsep=4pt,partopsep=2pt]
	\item [(P1)] Non-identifiability: 
	For example, in the VARMA($1,1$) case, there are multiple pairs of $(\bm{\Theta}, \bm{\Phi})$ corresponding to the same process. The root cause of this problem is the matrix multiplications in the parametric form of $\bm{A}_h(\bm{\Phi},\bm{\Theta})=\bm{\Theta}^{h-1}(\bm{\Phi}-\bm{\Theta})$.
	\item [(P2)] High-order matrix polynomials: Consider as an example the ordinary least squares (OLS)  estimation of the VARMA($1,1$) model. For a sample $\{\bm{y}_t\}_{t=1}^T$, since $\bm{A}_h(\bm{\Phi},\bm{\Theta})$ is an $h$th-order matrix polynomial for $1\leq h\leq T$, the loss function  will have a computational complexity of $O(T^2 N^3)$\footnote{The computational complexity in this paper is calculated in a model of computation where field operations (addition and multiplication) take constant time.}, hence unscalable under large $N$.
\end{itemize}

While recent attempts have been made to improve the feasibility of VARMA models \citep{MS07,CEK16,Dias18,WBBM21}, they do not tackle  (P1) and (P2) directly, but rather resort to sophisticated  identification constraints and optimization methods. Moreover,  high-dimensional VARMA models can be difficult  to interpret due to their latent MA structures. Particularly, while it may be natural to assume that $\bm{\Theta}$ and $\bm{\Phi}$ in \eqref{eq:VARMA11} are sparse under large $N$  \citep{WBBM21}, this  does not necessarily result in a sparse VAR($\infty$) model; i.e., $\bm{A}_h(\bm{\Phi},\bm{\Theta})$'s may not be sparse. Thus, the sparse VARMA model is not particularly attractive in terms of (I3).

For high-dimensional time series, we aim to develop a sparse VAR($\infty$) model  that is favorable in all of (I1)--(I3). The proposed approach is motivated by reparametrizing the  VAR($\infty$) form of the VARMA($p,q$) model into formulation \eqref{eq:VARinf} with
\begin{equation}\label{eq:parA}
	\bm{A}_h= \sum_{k=1}^d \ell_{h,k}(\bm{\omega})  \bm{G}_k \quad\text{for}\quad h\geq1,
\end{equation}
where  $\bm{G}_1, \dots, \bm{G}_d\in\mathbb{R}^{N\times N}$ are unknown coefficient matrices, $\{\ell_{h,k}(\cdot)\}_{h=1}^\infty$ for $1\leq k\leq d$ are different  sequences of real-valued functions characterizing the exponential decay pattern of $\bm{A}_h$, with $\ell_{h,k}(\bm{\omega})\rightarrow0$ as $h\rightarrow\infty$ for each $k$,  and  $\bm{\omega}$ is an unknown low-dimensional parameter vector;  see also \cite{HZLL22} for a high-dimensional Tucker-low-rank time series model concurrently developed  from (1.3) with different techniques and interpretations. 
Similar to the orders $(p,q)$ of the VARMA model, $d$  can be viewed as the overall order that controls the complexity of temporal patterns of the VAR($\infty$) model; see Section \ref{section:model} for the detailed model formulation.
Note that  \eqref{eq:parA}  preserves the essential temporal patterns of the VARMA process, since it is derived directly from the former with little loss of generality. Thus, it is fundamentally more flexible than finite-order VAR models, i.e.,  more desirable regarding (I1). 
Moreover, each $\bm{A}_h=\bm{A}_h(\bm{\omega}, \bm{G}_1, \dots, \bm{G}_d)$ in \eqref{eq:parA} is a linear combination of matrices. Hence, unlike $\bm{A}_h(\bm{\Phi},\bm{\Theta})$ mentioned above, this form of $\bm{A}_h$ gets rid of all matrix multiplications. As a result, both problems (P1) and (P2) are eliminated, and then (I2) can be achieved.
To tackle the high dimensionality, we  assume that $\bm{G}_k$'s are sparse, leading to the proposed sparse parametric VAR($\infty$) (SPVAR($\infty$)) model.
In addition to improving the estimation efficiency  as required by (I2), the sparsity assumption enables greater interpretability, i.e., (I3), thanks to the novel separation of temporal and cross-sectional dependence in parameterizing the VARMA-type dynamic structure:
\begin{itemize}[itemsep=2pt,parsep=2pt,topsep=4pt,partopsep=2pt]
	\item[(D1)] Temporal dependence: In  \eqref{eq:parA}, the decay pattern of $\bm{A}_h$ as $h\rightarrow\infty$ is fully characterized by the scalar weights $\ell_{h,k}(\bm{\omega})$'s. 
	\item[(D2)] Cross-sectional dependence: The $\bm{G}_k$'s,  independent of the above decay pattern as $h\rightarrow\infty$, fully capture the cross-sectional dependence.
\end{itemize}
As a result of (D2), the Granger causal network of the $N$ component series of $\bm{y}_t$ is directly linked to the aggregate sparsity pattern of $\bm{G}_k$'s. Moreover, as  detailed in Section \ref{subsec:VARMA}, $\{\ell_{h,k}(\bm{\omega})\}_{h=1}^\infty$'s in \eqref{eq:parA} are specifically  defined such 
that $\bm{A}_k=\bm{G}_k$ for $1\leq k\leq p$, whereas $\bm{A}_{p+j}$ for $j\geq 1$ are expressed as linear combinations of $\bm{G}_{p+1},\dots, \bm{G}_d$, where $p$ is the AR order of the VARMA($p,q$)  model from which  \eqref{eq:parA} originates. Consequently, there is  an interesting  dichotomy  in the interpretations of different $\bm{G}_k$'s: On the one hand, each $\bm{G}_k$ with $1\leq k\leq p$ has the same interpretation as the lag-$k$  AR coefficient matrix of the VAR($p$) model, capturing the short-term cross-sectional dependence. On the other hand, the ``MA'' coefficient matrices $\bm{G}_{p+1}, \dots, \bm{G}_d$ encapsulate the cross-sectional dependence associated with the VARMA-type temporal structure, i.e., the long-term influence among the component series that extends into high lags. It is worth noting that the Granger causal network each $\bm{G}_k$ individually captures is specific to a particular temporal pattern characterized by $\{\ell_{h,k}(\bm{\omega})\}_{h=1}^\infty$.
This granularity provides a more detailed perspective on Granger causality from a temporal standpoint; see Section \ref{subsec:model} for details.
Additionally, in view of (D1), the sparsity of $\bm{G}_k$'s  incurs little loss of temporal information, so the essential VARMA-type temporal pattern is well preserved. This is a distinct advantage over regularized VARMA models \citep{CEK16, WBBM21}.

In fact, even compared to sparse finite-order VAR models, the proposed model can  be more interpretable for the  following two reasons. Firstly,  while the AR coefficient matrices $\bm{A}_h$ must diminish quickly as $h\rightarrow\infty$ to ensure stationarity of $\bm{y}_t$, $\bm{G}_k$'s do not need to decay thanks to the diminishing  $\ell_{h,k}(\bm{\omega})$'s. Consequently, $\bm{G}_k$'s, which have relatively strong signals, can be easier to interpret than the diminishing  $\bm{A}_h$'s.
Second, similar to the orders $(p,q)$ of VARMA models, the required $d$ is generally small in practice. For example, $d=2$ works well for the macroeconomic data in Section \ref{sec:empirical}, so we only need to interpret two adjacency matrices $\bm{G}_1$ and $\bm{G}_2$. However, if the VAR($P$) model were fitted,  we would have to interpret $P$ adjacency matrices, where the required $P$ would be much larger.  

We summarize the  main contributions of this paper as follows:
\begin{itemize}[itemsep=2pt,parsep=2pt,topsep=4pt,partopsep=2pt]
	\item [(i)] A sparse parametric VAR($\infty $) model is introduced for high-dimensional time series,
	which is favorable regarding (I1)--(I3), while avoiding problems (P1) and (P2).
	\item [(ii)] We develop two $\ell_1$-regularized estimators, which can be implemented via efficient block coordinate descent algorithms, and derive their nonasymptotic error bounds under weak sparsity; particularly, our theory takes into account the effect of initializing $\bm{y}_t=\bm{0}$ for $t\leq 0$, which is needed for feasible estimation of VAR($\infty$) models. 
	\item [(iii)]  A high-dimensional Bayesian information criterion (BIC) is proposed for model order selection, and its consistency is established. 
\end{itemize}

The remainder of this paper is organized as follows. Section \ref{section:model} introduces the proposed model and its interpretation. Section \ref{sec:HDmethod} presents two $\ell_1$-regularized estimators and their nonasymptotic theory. Section \ref{sec:BIC} introduces the proposed BIC. Sections \ref{sec:sim} and \ref{sec:empirical} provide simulation  and empirical studies. Section \ref{sec:conclusion} concludes with a brief discussion. The block coordinate descent  algorithms for implementing the estimation, additional simulation and empirical results, and all technical proofs are provided in a separate supplementary file. 

Unless otherwise specified, we denote scalars, vectors and matrices by lowercase letters (e.g., $x$), boldface lowercase letters (e.g., $\bm{x}$), and boldface capital letters (e.g., $\bm{X}$), respectively. Let $\mathbb{I}_{\{\cdot\}}$ be the indicator function taking value one when the condition is true and zero otherwise. 
For any $a,b\in\mathbb{R}$, let $a\vee b=\max\{a,b\}$ and $a\wedge b=\min\{a,b\}$. 
The $\ell_q$-norm of any $\bm{x}\in\mathbb{R}^p$ is denoted by $\|\bm{x}\|_q=(\sum_{j=1}^p|x_j|^q)^{1/q}$ for $q>0$. 
For any $\bm{X}\in\mathbb{R}^{d_1\times d_2}$, let $\bm{X}^\top$, $\sigma_{\max}(\bm{X})$ (or $\sigma_{\min}(\bm{X})$), $\lambda_{\max}(\bm{X})$ (or $\lambda_{\min}(\bm{X})$),  $\vect(\bm{X})$, $ \|\bm{X}\|_\op$, and $\|\bm{X}\|_\Fr$ be its transpose, largest (or smallest) singular value, largest (or smallest)  eigenvalue, vectorization, operator norm $\|\bm{X}\|_\op=\sigma_{\max}(\bm{X})$, and Frobenius norm $\|\bm{X}\|_\Fr=\sqrt{\trace(\bm{X}^\top \bm{X})}$, respectively. We use $C>0$ (or $c>0$) to denote generic large (or small) absolute constants. 
For any sequences $x_n$ and $y_n$, denote $x_n\lesssim y_n$ (or $x_n\gtrsim y_n$) if there is  $C>0$ such that $x_n\leq C y_n$ (or $x_n\geq C y_n$). We write $x_n\asymp y_n$ if $x_n\lesssim y_n$ and $x_n\gtrsim y_n$. In addition, $x_n\gg y_n$ if $y_n/x_n\rightarrow 0$ as $n\rightarrow\infty$.

\section{Proposed model} \label{section:model}
\subsection{Motivation: Reparameterization of VARMA models}\label{subsec:VARMA}

This section introduces the motivation behind the proposed model. Recall that the shared root cause of problems (P1) and (P2) of the VARMA($1,1$) model, as discussed in Section \ref{section:intro}, lies in the matrix multiplications involved in computing the AR coefficient matrices $\bm{A}_h(\bm{\Phi},\bm{\Theta})=\bm{\Theta}^{h-1}(\bm{\Phi}-\bm{\Theta})$ in the VAR($\infty$) form of the  model. Thus, the key to overcoming both problems is to eliminate the  matrix multiplications in the parameterization of $\bm{A}_h$. 

To this end, we show that a reparameterization of $\bm{A}_h(\bm{\Phi},\bm{\Theta})$ free of matrix multiplications can be derived via the following two main steps: (1) Block-diagonalize $\bm{\Theta}$ via the Jordan decomposition,
$\bm{\Theta}= \bm{B} \bm{J}\bm{B}^{-1}$, where  $\bm{B}\in\mathbb{R}^{N \times N}$ is an invertible matrix, and $\bm{J}\in\mathbb{R}^{N \times N}$ is the real Jordan form   containing eigenvalues of $\bm{\Theta}$; see \eqref{eq:realJ} below for details. (2) Then, merge $\bm{B}$ with all remaining components in the expression of $\bm{A}_h(\bm{\Phi},\bm{\Theta})$. 

Specifically, by Theorem 1 in \cite{hartfiel1995dense}, for any $0<n\leq N$, real matrices with $n$ distinct nonzero eigenvalues are dense in the set of all $N\times N$ real matrices with rank at most $n$. Thus, with only a little loss of generality, we can assume that $\bm{\Theta}$ is a real matrix with $n$ distinct nonzero eigenvalues, where $n=\rank(\bm{\Theta})$;  a  more general result allowing repeated eigenvalues is derived in  the technical appendix of \cite{HZLL22}.
Then suppose that $\bm{\Theta}$ has $r$ nonzero real eigenvalues, $\lambda_1, \dots, \lambda_r$, and $s$ conjugate pairs of nonzero complex eigenvalues, $(\lambda_{r+2m-1}, \lambda_{r+2m})=(\gamma_m e^{i\theta_m},\gamma_m e^{-i\theta_m})$ for $1\leq m\leq s$, where   $|\lambda_j|\in(0,1)$ for $1\leq j\leq r$,  $\gamma_m\in(0,1)$ and $\theta_m \in (0, \pi)$ for $1\leq m\leq s$, and $i$ represents the imaginary unit. 
Therefore, $n=r+2s$, and the real Jordan form of $\bm{\Theta}$ is a real block diagonal matrix:
\begin{equation}\label{eq:realJ}
	\bm{J}=\diag\left \{\lambda_1, \dots, \lambda_r, \bm{C}_1, \dots, \bm{C}_s, \bm{0}\right \},\quad	
	\bm{C}_m=
	\gamma_m\cdot \left(\begin{matrix}
		\cos \theta_m& \sin \theta_m\\
		- \sin \theta_m& \cos \theta_m
	\end{matrix}\right) \in\mathbb{R}^{2\times 2}, %\hspace{5mm} 1\leq m\leq s;
\end{equation}
where $1\leq m\leq s$; see Chapter 3 in \cite{HJ12}.  

Let $\bm{A}_1=\bm{\Phi}-\bm{\Theta}:=\bm{G}_1$. Substituting the Jordan decomposition $\bm{\Theta}= \bm{B} \bm{J}\bm{B}^{-1}$ into the expression of $\bm{A}_h$, we can show that for all $h\geq 2$, $\bm{A}_{h}=\bm{B} \bm{J}^{h-1}\bm{B}^{-1}(\bm{\Phi}-\bm{\Theta})=\sum_{j=1}^{r} \lambda_j^{h-1} \bm{G}_{1+j}
+\sum_{m=1}^{s} \gamma_{m}^{h-1} \left [ \cos\{(h-1) \theta_{m}\} \bm{G}_{1+r+2m-1}+\sin\{(h-1) \theta_{m}\} \bm{G}_{1+r+2m}\right ]$,
where $\bm{G}_{2}, \dots, \bm{G}_{1+r+2s}\in\mathbb{R}^{N\times N}$ are determined jointly by $\bm{B}$ and $\bm{B}^{-1}(\bm{\Phi}-\bm{\Theta})$;  see the proof of Proposition 1 in the supplementary file for  details.
This result is a reparameterization of $\bm{A}_h$'s in terms of the scalars $\lambda_j$'s, $\gamma_{m}$'s, $\theta_{m}$'s, and matrices $\bm{G}_1, \dots, \bm{G}_{1+r+2s}$. 
As each $\bm{A}_h$ is a linear combination of $\bm{G}_1, \dots, \bm{G}_{1+r+2s}$, problems (P1) and (P2) are tackled at their root: It not only ensures the identifiability of the parameters $\lambda_j$'s, $\gamma_{m}$'s, $\theta_{m}$'s, and the $\bm{G}$-matrices, up to a permutation in the indices $j$ and $m$, but also leads to a significantly reduced computational complexity, such as $O(TN^2+T^2N)$ for the squared loss function. 

In general, the VARMA($p,q$) model is given by
$\bm{y}_t =\sum_{i=1}^{p} \bm{\Phi}_i\bm{y}_{t-i}+\bm{\varepsilon}_t - \sum_{j=1}^{q}\bm{\Theta}_j\bm{\varepsilon}_{t-j}$,
where $\bm{\Phi}_i, \bm{\Theta}_j\in\mathbb{R}^{N\times N}$ for $1\leq i\leq p$ and $1\leq j\leq q$.
Assuming invertibility, it has the following VAR($\infty$) representation:
\begin{equation}\label{eq:VARMA_inf}
	\bm{y}_t =  \sum_{h=1}^{\infty} \underbrace{ \left ( \sum_{i=0}^{p\wedge h} \bm{P} \underline{\bm{\Theta}}^{h-i} \bm{P}^\top \bm{\Phi}_i \right )}_{\bm{A}_h}\bm{y}_{t-h}+\bm{\varepsilon}_t, \hspace{5mm} \underline{\bm{\Theta}}=
	\left(\begin{matrix}
		\bm{\Theta}_1&\bm{\Theta}_2&\cdots&\bm{\Theta}_{q-1}&\bm{\Theta}_q\\
		\bm{I}&\bm{0}&\cdots&\bm{0}&\bm{0}\\
		\bm{0}&\bm{I}&\cdots&\bm{0}&\bm{0}\\
		\vdots&\vdots&\ddots&\vdots&\vdots\\
		\bm{0}&\bm{0}&\cdots&\bm{I}&\bm{0}
	\end{matrix} \right ), 
\end{equation}
where  $\bm{\Phi}_0=-\bm{I}$ and $\bm{P} = (\bm{I}_{N}, \bm{0}_{N\times N(q-1)})$ are constant matrices,  $\underline{\bm{\Theta}}$ is called the MA companion matrix, and all eigenvalues of $\underline{\bm{\Theta}}$ are less than one in absolute value; see \cite{Lutkepohl2005}. Similar to the VARMA($1,1$) case,  the following reparameterization can be derived. 

\begin{proposition}\label{prop:VARMA}
	Suppose that all nonzero eigenvalues of $\underline{\bm{\Theta}}$ are distinct, and  there are $r$ distinct nonzero real eigenvalues of  $\underline{\bm{\Theta}}$, $\lambda_j\in(-1,0)\cup(0,1)$ for $1\leq j\leq r$, and $s$ distinct conjugate pairs of nonzero complex eigenvalues of  $\underline{\bm{\Theta}}$, $(\lambda_{r+2m-1}, \lambda_{r+2m})=(\gamma_m e^{i\theta_m},\gamma_m e^{-i\theta_m})$ with $\gamma_m\in(0,1)$ and $\theta_m \in (0, \pi)$ for $1\leq m\leq s$. Then for all $h\geq1$, we have 
	\begin{align}\label{eq:linearcomb2}
		\begin{split}
			\bm{A}_{h} 
			&=\sum_{k=1}^{p}\mathbb{I}_{\{h=k\}}\bm{G}_{k}
			+ 
			\sum_{j=1}^{r} \mathbb{I}_{\{h\geq p+1\}} \lambda_j^{h-p} \bm{G}_{p+j}\\
			&\hspace{5mm} +\sum_{m=1}^{s} \mathbb{I}_{\{h\geq p+1\}} \gamma_{m}^{h-p} \left [ \cos\{(h-p) \theta_{m}\} \bm{G}_{p+r+2m-1} + \sin\{(h-p)\theta_{m}\} \bm{G}_{p+r+2m} \right],
		\end{split}
	\end{align}
	where  $\bm{G}_k=\bm{A}_k$ for $1\leq k\leq p$, and  $\{\bm{G}_{k}\}_{k=p+1}^{p+r+2s}$  are determined jointly by $\bm{\widetilde{B}}$ and $\bm{\widetilde{B}}_{-}$, with $\bm{\widetilde{B}}=\bm{P} \bm{B}$
	and $\bm{\widetilde{B}}_-=\bm{B}^{-1}\left (\sum_{i=0}^{p} \underline{\bm{\Theta}}^{p-i} \bm{P}^\top \bm{\Phi}_i \right )$. In addition, the corresponding term in \eqref{eq:linearcomb2} is suppressed if $p, r$ or $s$ is zero.
\end{proposition}

Throughout this paper, we denote
$d = p+r+2s$. Let $\bm{\omega}=(\lambda_1, \dots, \lambda_r, \bm{\eta}_1^\top, \dots \bm{\eta}_s^\top)^\top \in\mathbb{R}^{r+2s}$, where $\bm{\eta}_m=(\gamma_m,\theta_m)^\top $ for $1\leq m\leq s$, and 
$\bm{g}=\vect(\bm{G})\in\mathbb{R}^{N^2d}$, where $\bm{G}=(\bm{G}_1, \dots, \bm{G}_d)\in\mathbb{R}^{N\times Nd}$.
Then, we can succinctly write \eqref{eq:linearcomb2} in the parametric form of $\bm{A}_h=\bm{A}_h(\bm{\omega},\bm{g})=\sum_{k=1}^{d}\ell_{h,k}(\bm{\omega})\bm{G}_k$ for all $h\geq1$. Here $\ell_{h,k}(\cdot)$'s  are  real-valued functions predetermined according to \eqref{eq:linearcomb2}, which can   be defined conveniently through a matrix as follows: for any $h\geq1$ and $1\leq k\leq d$, $\ell_{h,k}(\bm{\omega})$  is the $(h,k)$-th entry of the $\infty\times d$ matrix, 
\begin{equation*}%\label{eq:Lfunc}
	\bm{L}(\bm{\omega})=\left (\ell_{h,k}(\bm{\omega})\right )_{h\geq1, 1\leq k\leq d}
	=\left (\begin{matrix}
		\bm{I}_p &\bm{0}_{p\times1}&\cdots &\bm{0}_{p\times1}&\bm{0}_{p\times2}&\cdots&\bm{0}_{p\times2}\\
		\bm{0}_{\infty\times p}& \bm{\ell}^{I}(\lambda_1) & \cdots & \bm{\ell}^{I}(\lambda_r) & \bm{\ell}^{II}(\bm{\eta}_1) & \cdots & \bm{\ell}^{II}(\bm{\eta}_s)
	\end{matrix}\right ) \in\mathbb{R}^{\infty\times d},
\end{equation*}
where, for  any $\lambda$ and $\bm{\eta}=(\gamma,\theta)^\top$,  the blocks $\bm{\ell}^{I}(\lambda)$ and $\bm{\ell}^{II}(\bm{\eta})$ are defined as
\begin{equation*}%\label{eq:Lcolumn}
	\bm{\ell}^{I}(\lambda)= (\lambda, \lambda^2, \lambda^3, \dots)^\top \in\mathbb{R}^\infty,
	\quad
	\bm{\ell}^{II}(\bm{\eta}) =
	\left (\begin{array}{cccc}
		\gamma \cos(\theta)& \gamma^2\cos(2\theta)& \gamma^3\cos(3\theta)&\cdots\\
		\gamma \sin(\theta)& \gamma^2\sin(2\theta)& \gamma^3\sin(3\theta)&\cdots\\
	\end{array}\right )^\top  \in\mathbb{R}^{\infty\times 2}.
\end{equation*}	

\subsection{Proposed sparse parametric VAR($\infty$) model}\label{subsec:model}

Motivated by the discussion in Section \ref{subsec:VARMA}, we propose the following VAR($\infty$) model for high-dimensional time series:
\begin{equation}\label{eq:model-scalar}
	\bm{y}_t=\sum_{h=1}^\infty  \bm{A}_h(\bm{\omega},\bm{g}) \bm{y}_{t-h}+\bm{\varepsilon}_t
	=\sum_{k=1}^{d}\bm{G}_k \sum_{h=1}^\infty\ell_{h,k}(\bm{\omega}) \bm{y}_{t-h}+\bm{\varepsilon}_t,
\end{equation}%\ell_{h,k}(\bm{\omega}; p,r,s)
where  $\bm{\omega}\in (-1,1)^r \times \bm{\varPi}^s \subset \mathbb{R}^{r+2s}$ is a parameter vector, with $\bm{\varPi}=[0, 1)\times (0, \pi)$, $\ell_{h,k}(\cdot)$'s are known real-valued functions defined as in  Section \ref{subsec:VARMA},  $\bm{G}_k\in\mathbb{R}^{N\times N}$ for $1\leq k\leq d$ are parameter matrices with  $d=p+r+2s$. To handle the high-dimensionality, we assume that  $\bm{G}_k$'s are sparse matrices. In this section, we will focus on the exact sparsity as it is instrumental for model interpretability. However, it will be relaxed to weak sparsity in our theoretical analysis; see Assumptions \ref{assum:sparse} and \ref{assum:rowsparse} in Section \ref{sec:HDmethod}. We call model \eqref{eq:model-scalar} with exactly or weakly sparse $\bm{G}_k$'s the Sparse Parametric VAR($\infty$)  (SPVAR($\infty$)) model. 

Note that if no sparsity assumption is imposed on $\bm{G}_k$'s, then \eqref{eq:model-scalar} provides an alternative low-dimensional time series model comparable to the VARMA model; see Section \ref{subsec:statn} for its stationarity condition. While formulation \eqref{eq:model-scalar} is derived from the VARMA model, it is worth clarifying that it relaxes the restrictions on $\bm{G}_{p+j}$ for $1\leq j\leq r+2s$. Specifically, by Proposition \ref{prop:VARMA}, if $\{\bm{y}_t\}$ is indeed generated from  a VARMA model, then $\bm{G}_{p+j}$'s would fulfill certain restrictions as determined by the Jordan decomposition of the MA companion matrix $\underline{\bm{\Theta}}$. By contrast, \eqref{eq:model-scalar}  treats these matrices as free parameters. 

The resemblance between  \eqref{eq:model-scalar} and the VARMA model is mainly achieved by  $\ell_{h,k}(\cdot)$'s, which yield VARMA-type decay patterns of  $\bm{A}_h$ as $h\rightarrow\infty$. According to \eqref{eq:linearcomb2},  $\ell_{h,k}(\cdot)$'s  implicitly depend on the orders $(p,r,s)$. Note that $p$  and $(r,s)$ are counterparts of the AR and MA orders of the VARMA model, respectively.  In fact, when $r=s=0$, \eqref{eq:model-scalar} reduces to  the VAR($p$) model, $\bm{y}_t=\sum_{h=1}^{p}\bm{G}_h \bm{y}_{t-h}+\bm{\varepsilon}_t$.  
For this reason, we call $\bm{G}_{1},\dots, \bm{G}_p$ and $\bm{G}_{p+1},\dots, \bm{G}_d$ the AR and MA coefficient matrices of the model, respectively.
While larger $(p,r,s)$  allow for more complex temporal patterns, similar to the VARMA model,  usually it suffices to use  small orders  in practice; see Section \ref{sec:empirical} for empirical evidence.

\begin{figure}[t]
	\centering
	\includegraphics[width=\textwidth]{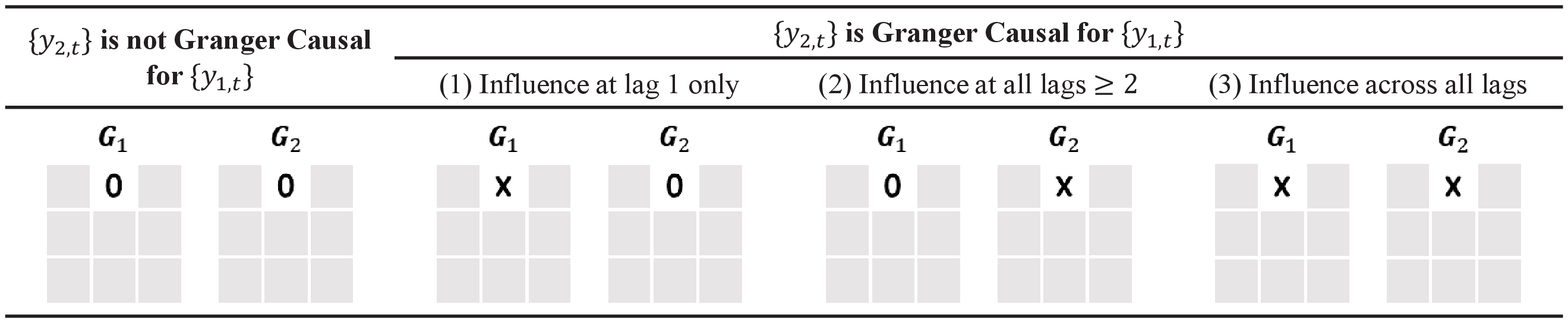}
	\caption{Illustration for different scenarios of  Granger causality of $\{y_{2,t}\}$ for $\{y_{1,t}\}$ when $(p,r,s)=(1,1,0)$ and $N=3$, as determined by the $(1,2)$th entry of $\bm{G}_1$ and $\bm{G}_2$. Cell $(1,2)$ of $\bm{G}_k$ is marked with ``0'' when $g_{1,2,k}= 0$, and ``X'' when $g_{1,2,k}\neq 0$.}
	\label{fig:2G}
\end{figure}

The proposed model can be directly used to infer the multivariate Granger causality (MGC), which concerns Granger causal (GC) relations \citep{Granger69} between any pair of component series in $\bm{y}_t=(y_{1,t}, \dots, y_{N,t})^\top$; see \cite{SF21} for an excellent review. By definition, $\{y_{j,t}\}$ is GC for $\{y_{i,t}\}$ if the past information of  $y_{j,t}$ can improve the forecast of $y_{i,t}$, where $1\leq i\neq j\leq N$. 
Most existing works study the MGC under the finite-order VAR for its convenience: Under the model $\bm{y}_t=\sum_{h=1}^P \bm{A}_h\bm{y}_{t-h}+\bm{\varepsilon}_t$, $\{y_{j,t}\}$ is GC for $\{y_{i,t}\}$ if $a_{i,j,h}\neq 0$ for some $h\in\{1,\dots,P\}$, where $a_{i,j,h}$ is the $(i,j)$-th entry of $\bm{A}_h$, for $1\leq i\neq j\leq N$. Notably, while  working with $\bm{A}_h$'s would be infeasible when $P=\infty$, we can directly infer the MGC through $\bm{G}_k$'s:
By \eqref{eq:model-scalar}, we have that $\{y_{j,t}\}$ is GC for $\{y_{i,t}\}$ if $g_{i,j,k}\neq 0$  for some $k\in\{1,\dots,d\}$, where $g_{i,j,k}$ is the $(i,j)$-th entry of $\bm{G}_k$, for $1\leq i\neq j\leq N$; see Figure \ref{fig:2G} for an illustration with $(i,j)=(1,2)$, $(p,r,s)=(1,1,0)$, and  $N=3$.

More interestingly,  since each  $\bm{G}_k$ captures a piece of cross-sectional information associated with a particular sequence $\{\ell_{h,k}(\bm{\omega})\}_{h=1}^\infty$, we can discern the decay pattern of any GC relations over time, achieving a more granular understanding of the MGC. For simplicity, consider the model for $y_{1,t}$ when $(p,r,s)=(1,1,0)$:
$y_{1,t}= \sum_{j=1}^N  g_{1,j,1} y_{j,t-1} +\sum_{j=1}^{N} g_{1,j,2}  \sum_{h=2}^{\infty} \lambda^{h-1}y_{j,t-h} + \varepsilon_{1,t}$, 
where $g_{i,j,k}$ denotes the $(i,j)$-th entry of $\bm{G}_k$. First, it is clear that $\{y_{j,t}\}$ is GC for $\{y_{1,t}\}$ if $g_{1,j,1}$ and $g_{1,j,2}$ are not both zero. Second, if this GC relation exists, the lagged influence of  $\{y_{j,t}\}$ on $\{y_{1,t}\}$  can be classified into the following three scenarios: (1) \textit{lag-one only}, if $g_{1,j,1}\neq0$ and $g_{1,j,2}=0$; (2)\textit{all lags beyond lag one}, if $g_{1,j,1}=0$ and $g_{1,j,2}\neq 0$; and (3) \textit{all lags}, if $g_{1,j,1}\neq 0$ and $g_{1,j,2}\neq 0$. In scenarios (2) and (3), the exponential decay of the influence over time is determined by $\lambda$; see Figure \ref{fig:2G} for an illustration for $j=2$.

\begin{figure}[t]
	\centering
	\includegraphics[width=\textwidth]{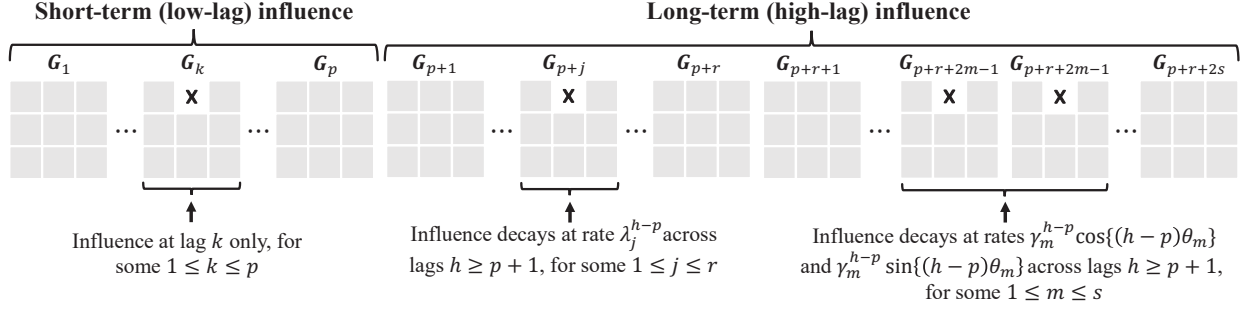}
	\caption{Illustration for different types of lagged influence of  $\{y_{2,t}\}$ on $\{y_{1,t}\}$ under general orders $(p,r,s)$ and $N=3$. Cell $(1,2)$ of $\bm{G}_k$ is marked with ``X'' when $g_{1,2,k}\neq 0$.}
	\label{fig:allG}
\end{figure}

In general, with orders $(p,r,s)$, the model equation for $y_{1,t}$ will consist of two conditional mean terms: The first term involves the sum of $g_{1,j,k} y_{j,t-k}$ for lags $1\leq k\leq p$, whereas the second term captures the influence beyond lag $p$. The latter involves a weighted mixture of  $r$ distinct exponential decay rates and $s$ distinct pairs of damped cosine and sine waves. Then the lagged influence of  $\{y_{j,t}\}$ on $\{y_{1,t}\}$ can be generalized to the following three scenarios, if the GC relation exists:  (1) \textit{short-term only}, if $g_{1,j,k}\neq 0$ for some $1\leq k \leq p$, while $g_{1,j,p+1}=\cdots =g_{1,j,d}=0$; (2)  \textit{long-term only}, if  $g_{1,j,1}=\cdots =g_{1,j,p}=0$, while $g_{1,j,k}\neq 0$ for some $p+1\leq k \leq d$; and (3) \textit{both short-term and long-term influences}, if $g_{1,j,k}\neq 0$ for  some $1\leq k \leq p$ and some $p+1\leq k \leq d$. A more detailed illustration is given in Figure \ref{fig:allG}.

\begin{remark}
	In many applications, the cross-sectional dependence  may not be time-invariant; e.g., \cite{BB17} found that the estimated Granger causal network in a sparse VAR system for stock volatilities may be time-varying. Time-varying cross-sectional dependence is also common in behavioral and neural studies: e.g., different segments of video time series of freely moving animals may correspond to distinct behaviors \citep{Linderman2022}, and discrete shifts in the dynamics of neural activity may reflect changes in underlying brain state \citep{Fiecas2023}. To accommodate  such applications, the proposed model can be extended to allow $\bm{G}_k$'s to be time varying; e.g., a Markov-switching  SPVAR($\infty$) model may be developed along the lines of \cite{LSS2022}.
\end{remark}

\begin{remark}
	In VAR models, the GC relations as captured by the coefficient matrices $\bm{A}_h$'s correspond to lagged cross-sectional dependence, whereas the instantaneous cross-sectional dependence is captured by the variance-covariance matrix $\bm{\Sigma}_\varepsilon$ of $\bm{\varepsilon}_t$. 
	While this section focuses on the former, $\bm{\Sigma}_\varepsilon$ can also be estimated based on residuals from the fitted SPVAR($\infty$) model; see Remark \ref{remark:Sigma} in Section \ref{subsec:JE}.
\end{remark}

\begin{remark}\label{remark:IRA}
	We can also conduct impulse response analysis based on the VMA($\infty$) form of the proposed model; see Theorem \ref{thm:stationary} in Section \ref{subsec:statn} for the  VMA($\infty$) representation. For example, when $(p,r,s)=(1,1,0)$, the corresponding MA coefficient matrices are $\bm{\Psi}_1=\bm{G}_1$, $\bm{\Psi}_2=\bm{G}_1^2+\lambda \bm{G}_2$, $\bm{\Psi}_3=\bm{G}_1^3+\lambda \bm{G}_1\bm{G}_2 + \lambda \bm{G}_2\bm{G}_1 + \lambda^2\bm{G}_2$, etc. When  $\bm{G}_1$ and $\bm{G}_2$ are both sparse with their non-zero entries in sufficiently different positions, all $\bm{\Psi}_j$'s will also tend to be sparse; this is indeed the case for the empirical example in Section \ref{sec:empirical}. Thus, we can alternatively interpret the high-dimensional time series via the impulse response analysis.
\end{remark}

\subsection{Stationarity condition}\label{subsec:statn}
We  provide a sufficient condition on $\bm{\omega}$ and $\bm{G}_k$'s for the existence of a unique strictly stationary solution for \eqref{eq:model-scalar} in the following theorem, which is valid whether $\bm{G}_k$'s are sparse or not. Similar  to the AR companion matrix of a VARMA($p,q$) model, denote
\begin{equation*}
	\underline{\bm{G}}_1=
	\left(\begin{matrix}
		\bm{G}_1&\bm{G}_2&\cdots&\bm{G}_{p-1}&\bm{G}_p\\
		\bm{I}&\bm{0}&\cdots&\bm{0}&\bm{0}\\
		\bm{0}&\bm{I}&\cdots&\bm{0}&\bm{0}\\
		\vdots&\vdots&\ddots&\vdots&\vdots\\
		\bm{0}&\bm{0}&\cdots&\bm{I}&\bm{0}
	\end{matrix} \right ).
\end{equation*}

\begin{theorem}\label{thm:stationary}
	Suppose that there exists $0<\bar{\rho}<1$ such that 
	\[
	\max\{|\lambda_1|,\ldots,|\lambda_r|, \gamma_1,\ldots,\gamma_s\}\leq \bar{\rho} \quad \text{and} \quad \rho(\underline{\bm{G}}_1)+\frac{\bar{\rho}}{1-\bar{\rho}} \sum_{k=1}^{r+2s}\rho(\bm{G}_{p+k}) <1,
	\]
	where $\rho(\cdot)$ denotes the spectral radius of a matrix, and $\rho(\underline{\bm{G}}_1)$ disappears when $p=0$. Moreover, $\{\bm{\varepsilon}_t\}$ is a strictly stationary sequence. 
	Then there exists a unique strictly stationary solution to the model equation in \eqref{eq:model-scalar}, given by $\bm{y}_t = \bm{\varepsilon}_t+ \sum_{j=1}^{\infty}\bm{\Psi}_j\bm{\varepsilon}_{t-j}$, where 
	$\bm{\Psi}_j=\sum_{k=1}^{\infty}\sum_{j_1+\cdots+ j_k=j}\bm{A}_{j_1}\cdots \bm{A}_{j_k}$ for $j\geq1$, with
	$\bm{A}_h=\sum_{k=1}^{d}\ell_{h,k}(\bm{\omega}) \bm{G}_k$ for $h\geq 1$.
\end{theorem}

When $r=s=0$, the condition in Theorem \ref{thm:stationary} reduces to $\rho(\underline{\bm{G}}_1)<1$, which coincides with the necessary and sufficient condition for the strict stationarity of the VAR($p$) model. When $r$ and $s$ are not both zero, the stationarity region for $\bm{G}_{k}$'s in Theorem \ref{thm:stationary} will be larger if $\bar{\rho}$ becomes smaller, i.e., if $\bm{A}_{h}$ diminishes more quickly as $h\rightarrow\infty$.

\begin{remark}\label{remark:stat}
	If $\{\bm{y}_t\}$ is  a VARMA($p,q$) process fulfilling the representation in \eqref{eq:model-scalar}, it is known that the necessary and sufficient condition for its strict stationarity is simply $\rho(\underline{\bm{G}}_1)<1$; see \cite{Lutkepohl2005}. This suggests that the sufficient condition in  Theorem \ref{thm:stationary}  could sometimes be restrictive. Indeed, the condition on $\bm{\omega}$ and $\bm{G}_k$'s in  Theorem \ref{thm:stationary} is derived from the necessary and sufficient condition: 
	%\[
	$\sum_{j=1}^{\infty}\|\bm{\Psi}_j\| < \infty$,
	%\]
	where  $\bm{\Psi}_j$'s are functions of $\bm{A}_h$'s as defined in the VMA($\infty$) form of $\{\bm{y}_t\}$ in Theorem \ref{thm:stationary}, and $\|\cdot\|$ is any submultiplicative matrix norm. This motivates us to  recommend a more general numerical method to check stationarity for practical use: first compute the sequence $\{\bm{\Psi}_j\}$ using the parameters $\bm{\omega}$ and $\bm{G}_k$'s, and then numerically check whether the partial sum $\sum_{j=1}^{J}\|\bm{\Psi}_j\|$ converges as $J\rightarrow \infty$.
	This method is applied in Section \ref{sec:empirical} to check the stationarity of the fitted model.
\end{remark}

\section{High-dimensional estimation}\label{sec:HDmethod}
\subsection{$\ell_1$-regularized joint estimator}\label{subsec:JE}

We first  propose an $\ell_1$-regularized estimator for the SPVAR($\infty$) model via jointly fitting all component series of $\bm{y}_t$. An alternative estimator will be introduced in the next section.

For $\{\bm{y}_t\}_{t=1}^T$ generated from \eqref{eq:model-scalar} with orders $(p, r, s)$, the  squared loss is $\mathbb{L}_T(\bm{\omega},\bm{g})=T^{-1}\sum_{t=1}^{T} \| \bm{y}_t -  \sum_{h=1}^\infty\bm{A}_h(\bm{\omega}, \bm{g}) \bm{y}_{t-h}\|_2^2=T^{-1}\sum_{t=1}^{T} \| \bm{y}_t - \sum_{k=1}^{d}\bm{G}_k \sum_{h=1}^\infty\ell_{h,k}(\bm{\omega}) \bm{y}_{t-h}\|_2^2$. Here  $\bm{g}=\vect(\bm{G})$, where  $\bm{G}=(\bm{G}_1,\dots, \bm{G}_d)\in\mathbb{R}^{N\times Nd}$. Since the loss function depends on observations in the  infinite past, initial values for $\{\bm{y}_t, t\leq 0\}$ will be needed in practice. We  set them to zero as $\mathbb{E}(\bm{y}_t)=\bm{0}$, and then the corresponding loss becomes
\begin{equation}\label{eq:sqrloss}
	\widetilde{\mathbb{L}}_T(\bm{\omega},\bm{g})
	=\frac{1}{T}\sum_{t=1}^{T} \Big \| \bm{y}_t -  \sum_{h=1}^{t-1}\bm{A}_h(\bm{\omega}, \bm{g}) \bm{y}_{t-h}\Big\|_2^2=\frac{1}{T}\sum_{t=1}^{T} \Big\| \bm{y}_t - \sum_{k=1}^{d}\bm{G}_k \sum_{h=1}^{t-1}\ell_{h,k}(\bm{\omega}) \bm{y}_{t-h}\Big\|_2^2.
\end{equation}
The initialization effect will be taken into account in our theoretical analysis, and its negligibility  is confirmed by our simulation study;  see Lemmas S6--S8 and Section S2 in the supplementary file.
We propose the $\ell_1$-regularized joint estimator (JE) as follows:
\begin{equation}\label{eq:lasso}
	(\bm{\widehat{\omega}},\bm{\widehat{g}})=\argmin_{\bm{\omega}\in\bm{\Omega}, \bm{g}\in\mathbb{R}^{N^2d} }\left \{\widetilde{\mathbb{L}}_T(\bm{\omega},\bm{g})+\lambda_g \|\bm{g} \|_1 \right \},
\end{equation}
where $\lambda_g>0$ is the regularization parameter, and
$\bm{\Omega}\subset (-1,1)^r \times \bm{\varPi}^s$ denotes the parameter space of 
$\bm{\omega}$.
Let $\bm{a}=\vect(\bm{A})$, where $\bm{A}=(\bm{A}_1,\bm{A}_2,\dots)$ is the horizontal concatenation of $\{\bm{A}_h\}_{h=1}^\infty$. Note that $\bm{a}=(\bm{L}(\bm{\omega})\otimes \bm{I}_{N^2})\bm{g}$. Based on \eqref{eq:lasso}, the estimator of $\bm{A}_h$ is $\bm{\widehat{A}}_h=\sum_{k=1}^{d}\ell_{h,k}(\bm{\widehat{\omega}})\bm{\widehat{G}}_k$ for $h\geq 1$. Then, $\bm{\widehat{a}}=\vect(\bm{\widehat{A}})=(\bm{L}(\bm{\widehat{\omega}})\otimes \bm{I}_{N^2})\bm{\widehat{g}}$, where $\bm{\widehat{A}}=(\bm{\widehat{A}}_1,\bm{\widehat{A}}_2,\dots)$.   

%%%%%%%%%%%%%%%%%%%%%%%%%%%%%%%%%%%%%%%%%%%%%%%%%%%%%%%%%%%%%%%%%%%%%%%%%%%%%%%%%%%%%%%%%% 

Denote the true value  of any parameter with the superscript ``$*$'', e.g., $\bm{g}^*$, $\bm{\omega}^*$, and $\bm{a}^*$.  For $\bm{\omega}^* \in\bm{\Omega}$, let
$\nu_{\mathrm{lower}}^*= (\min_{1\leq j\leq r}|\lambda_j^*|)\wedge (\min_{1\leq m\leq s}|\gamma_m^*|)$ and $\nu_{\mathrm{gap}}^*=\min_{1\leq j\neq k\leq r+2s} |x_j^*-x_k^*|$, 
where $x_j^* = \lambda_j^*$ for $1\leq j \leq r$ and $(x_{r+2m-1}^*, x_{r+2m}^*)  = (\gamma_m^* e^{i\theta_m^*}, \gamma_m^* e^{-i\theta_m^*})$ for $ 1\leq m\leq s$. The assumptions for our theoretical analysis are presented as follows.

\begin{assumption}[Parameter space and stationarity]\label{assum:statn}
	(i)  There exists an  absolute constant $0<\bar{\rho}<1$ such that $|\lambda_1|,\ldots,|\lambda_r|, \gamma_1,\ldots,\gamma_s\leq \bar{\rho}$ for all $\bm{\omega}\in\bm{\Omega}$; and (ii) the time series $\{\bm{y}_t\}$ is stationary.
\end{assumption}

\begin{assumption}[Separability]\label{assum:gap}
	(i) There exists an absolute constant $c_\nu>0$ such that $\nu_{\mathrm{lower}}^*\geq c_\nu$ and $\nu_{\mathrm{gap}}^*\geq c_\nu$; and (ii) $r$ and $s$ are fixed.
\end{assumption}

\begin{assumption}[Sub-Gaussian errors]\label{assum:error}
	Let $\bm{\varepsilon}_t = \bm{\Sigma}_\varepsilon^{1/2}\bm{\xi}_t$, where  $\bm{\xi}_t$ is a sequence of i.i.d. random vectors with zero mean and $\var(\bm{\xi}_t) = \bm{I}_{N}$, and $\bm{\Sigma}_\varepsilon$ is a positive definite covariance matrix.
	In addition, the coordinates $(\xi_{it})_{1\leq i\leq N}$ within $\bm{\xi}_t$ are mutually independent and $\sigma^2$-sub-Gaussian.
\end{assumption}

Assumption \ref{assum:statn}(i) ensures that $|\lambda_{j}|$'s and $\gamma_m$'s are bounded away from one. A sufficient condition for Assumption \ref{assum:statn}(ii) is given in Theorem \ref{thm:stationary}. Under stationarity, $\{\bm{y}_t\}$ has the  VMA($\infty$) form 
$\bm{y}_t =\bm{\Psi}_*(B)\bm{\varepsilon}_{t}$, where $\bm{\Psi}_*(B) = \bm{I}_N+\sum_{j=1}^{\infty}\bm{\Psi}_j^* B^j$, and $B$ is the backshift operator; see Theorem \ref{thm:stationary}. Let
$\mu_{\min}(\bm{\Psi}_*) = \min_{|z|=1}\lambda_{\min}(\bm{\Psi}_*(z)\bm{\Psi}_*^{\HH}(z))$ and $\mu_{\max}(\bm{\Psi}_*) = \max_{|z|=1}\lambda_{\max}(\bm{\Psi}_*(z)\bm{\Psi}_*^{\HH}(z))$,
where   $\bm{\Psi}_*^{\HH}(z)$ is the conjugate transpose of $\bm{\Psi}_*(z)$ for $z\in\mathbb{C}$. It can be verified that $\mu_{\min}(\bm{\Psi}_*) > 0$; see also \cite{basu2015regularized}. Then we define the positive
constants $\kappa_1=	\lambda_{\min}(\bm{\Sigma}_\varepsilon)\mu_{\min}(\bm{\Psi}_*)$  and  $\kappa_2=\lambda_{\max}(\bm{\Sigma}_\varepsilon)\mu_{\max}(\bm{\Psi}_*)$.
Assumption \ref{assum:gap}(i) requires that different $\lambda_{j}^*$'s or $\bm{\eta}_m^*$'s are bounded away from zero and from each other. Since these parameters lie in bounded parameter spaces, this also entails that $r$ and $s$ must be fixed; see Assumption \ref{assum:gap}(ii). Assumption \ref{assum:error} relaxes the  Gaussian assumption commonly used in the literature on high-dimensional time series models \citep[e.g.,][]{basu2015regularized}  to sub-Gaussianity. 

Let $\bm{g}_\ar=\vect(\bm{G}_\ar)$ and $\bm{g}_\ma=\vect(\bm{G}_\ma)$, where  $\bm{G}_\ar=(\bm{G}_{1},\dots, \bm{G}_p)\in\mathbb{R}^{N\times Np}$ and $\bm{G}_\ma=(\bm{G}_{p+1},\dots, \bm{G}_d)\in\mathbb{R}^{N\times N(r+2s)}$. Let $g_{i,j,k}$ be the $(i,j)$th entry of $\bm{G}_k$. Then, we define the weak sparsity of $\bm{g}_\ar^*$ and $\bm{g}_\ma^*$ by restricting them into the $\ell_q$-``balls'', $\mathbb{B}_q(R_q^\ar):=\{\bm{g}_\ar\in\mathbb{R}^{N^2p}\mid \sum_{k=1}^{p}\sum_{i=1}^{N}\sum_{j=1}^{N} |g_{i,j,k}|^q\leq R_q^\ar\}$ and $\mathbb{B}_q(R_q^\ma):=\{\bm{g}_\ma\in\mathbb{R}^{N^2 (r+2s)}\mid  \sum_{k=p+1}^{d}\sum_{i=1}^{N}\sum_{j=1}^{N} \allowbreak  |g_{i,j,k}|^q\leq R_q^\ma\}$, respectively, which is a  more  general assumption than exact sparsity.

\begin{assumption}[Weak sparsity]\label{assum:sparse}
	There exists $q\in[0,1]$ such that $\bm{g}_\ar^*\in\mathbb{B}_q(R_q^\ar)$ and $\bm{g}_\ma^*\in\mathbb{B}_q(R_q^\ma)$ for some radii  $R_q^\ar, R_q^\ma>0$.
\end{assumption}

Assumption \ref{assum:sparse} implies that $\bm{g}^*\in \mathbb{B}_q(R_q)$, where $R_q:=R_q^\ar + R_q^\ma$ and $\mathbb{B}_q(R_q):=\{\bm{g}\in\mathbb{R}^{N^2d}\mid \sum_{k=1}^{d}\sum_{i=1}^{N}\sum_{j=1}^{N} |g_{i,j,k}|^q\leq R_q\}$. If $q=0$, Assumption \ref{assum:sparse} becomes the exact sparsity constraints---$\bm{g}_\ar^*$ and $\bm{g}_\ma^*$ have at most $R_q^\ar$ and $R_q^\ma$  nonzero entries, respectively. If $q\in(0,1]$, the $\ell_q$-``balls'' enforce a certain decay rate on the absolute values of the entries  in  $\bm{g}^*$ as the dimension $N$ grows. Note that we do not require  $R_q^\ar$ and $R_q^\ma$ to be fixed.

A main theoretical  challenge is that the loss function $\widetilde{\mathbb{L}}_T(\bm{\omega},\bm{g})$ is highly nonconvex with respect to $\bm{\omega}$. 
Consequently, the global statistical consistency commonly established for high-dimensional convex M-estimators is not available. However,  if the nonconvex loss function  exhibits a benign  convex curvature over local regions, then a form of  local statistical consistency  can be established; see, e.g.,  \cite{Loh2017}. For many nonconvex $M$-estimators, certain convexity holds within a constant-radius neighborhood of the true parameter value; for the high-dimensional setup, this is termed as local restricted strong convexity in \cite{Loh2017}. Then it can be shown that all local optima within this region  can  enjoy the same convergence rate as the $\ell_1$-regularized least squared estimator for  linear regression; see also \cite{JG2021} and \cite{WH2022}  for other works on local statistical guarantees for estimators with nonconvex losses or regularizers.
Our method is reminiscent of that for high-dimensional nonconvex M-estimators in the literature. However, our setting is special in that $\widetilde{\mathbb{L}}_T(\bm{\omega},\bm{g})$  is only partially nonconvex, as it is  convex with respect to $\bm{g}$, for any fixed $\bm{\omega}$. Thus, unlike \cite{Loh2017}, we only need to restrict $\bm{\omega}$ within a local region of restricted curvature around $\bm{\omega}^*$, while $\bm{g}$ can be free. 

Let $\underline{\alpha}_\ma = \min_{1\leq j\leq r+2s} \|\bm{G}_{p+j}^*\|_{\Fr}$ and  $\overline{\alpha}_\ma = \max_{1\leq j\leq r+2s} \|\bm{G}_{p+j}^*\|_{\Fr}$, which are both allowed to grow with $N$. Then let $\alpha= \overline{\alpha}_\ma/\underline{\alpha}_\ma$. The local convexity of our loss function around $\bm{\omega}^*$ is an immediate consequence of the following proposition.

\begin{proposition}\label{prop:perturb}
	Suppose that $\underline{\alpha}_\ma>0$. Then under Assumptions \ref{assum:statn}(i) and \ref{assum:gap}, there exists a  constant $c_{\bm{\omega}}= \min (2,   c /\alpha)>0$ such that for any  $\bm{\omega}\in\bm{\Omega}$ with $\|\bm{\omega} - \bm{\omega}^*\|_2\leq c_{\bm{\omega}}$, it holds $\|\bm{g}-\bm{g}^*\|_{2} + \underline{\alpha}_\ma \|\bm{\omega} - \bm{\omega}^*\|_2 \lesssim \|\bm{a}-\bm{a}^*\|_{2}^2\lesssim \|\bm{g}-\bm{g}^*\|_{2} + \overline{\alpha}_\ma \|\bm{\omega} - \bm{\omega}^*\|_2$, where $\bm{a}=(\bm{L}(\bm{\omega})\otimes \bm{I}_{N^2})\bm{g}$.
\end{proposition}

Proposition \ref{prop:perturb} shows that the mapping  $(\bm{\omega}, \bm{g})\rightarrow\bm{a}$ is linear within a constant-radius neighborhood of $\bm{\omega}^*$. Then, since the squared loss of our model is convex with respect to $\bm{a}$, it is also convex with respect to $(\bm{\omega}, \bm{g})$ jointly within the local region of $\bm{\omega}^*$. Note that the radius $c_{\bm{\omega}}$ is a  constant independent of $N$ and $T$  under the mild condition that $\underline{\alpha}_\ma \asymp \overline{\alpha}_\ma$, in which case 
$\{\|\bm{G}_{p+j}^*\|_{\Fr}\}_{j=1}^{r+2s}$  are of the same order of magnitude.

Since Proposition \ref{prop:perturb} relies on confining  $\bm{\omega}$ to a local neighborhood of $\bm{\omega}^*$, the theoretical guarantees derived in this paper are applicable to local estimators. That is, to derive nonasymptotic error bounds, we need to  assume that the estimator $\bm{\widehat{\omega}}$ obtained from \eqref{eq:lasso} lies within the local region of  $\bm{\omega}^*$ defined in Proposition \ref{prop:perturb}. We will discuss the practical aspect of this assumption after stating the main result. 
For simplicity, denote
\[
\eta_{T}= \sqrt{\frac{\kappa_2 \lambda_{\max}(\bm{\Sigma}_{\varepsilon})\log \{N(p\vee 1)\}}{\kappa_1^2 T}} \quad\text{and}\quad \varpi =   \frac{\lambda_{\max}(\bm{\Sigma}_{\varepsilon})}{\kappa_2 (p\vee 1)}.
\]

\begin{theorem}\label{thm:lasso}
	Suppose that Assumptions \ref{assum:statn}--\ref{assum:sparse} hold with $\sum_{j=0}^{\infty} \|\bm{\Psi}_j^* \|_{\op}^2 <\infty$, $R_q \lesssim \varpi/\eta_{T}^{2-q}$,  
	$\alpha^2 \lesssim R_q/R_q^\ma$,  $\varpi \lesssim  \overline{\alpha}_\ma^2 R_q/R_q^\ma$, and $\underline{\alpha}_\ma>0$.
	In addition, assume that  $\log N \gtrsim (\kappa_2/\kappa_1)^2$,
	$T\gtrsim \max\{\kappa_2 (p\vee1)^4, (\kappa_2/\kappa_1)^2  (p\vee1) \log\{ (\kappa_2/\kappa_1)\alpha N (p\vee1)\}\}$, and we solve \eqref{eq:lasso} with  $\lambda_g \asymp \sqrt{\kappa_2 \lambda_{\max}(\bm{\Sigma}_{\varepsilon}) \log \{N(p\vee 1)\} /T}$. If $\|\bm{\widehat{\omega}} - \bm{\omega}^*\|_{2}\leq c_{\bm{\omega}}$, then  with probability at least $1 - C (p\vee 1) e^{-c (\kappa_1 /\kappa_2)^2 \log N}$,
	\[
	\|\bm{\widehat{a}}-\bm{a}^*\|_{2} \lesssim \eta_{T}^{1-q/2} \sqrt{R_q}
	\quad\text{and}\quad
	\frac{1}{T}\sum_{t=1}^{T} \left \| \sum_{h=1}^{t-1}(\bm{\widehat{A}}_h-\bm{A}_h^*) \bm{y}_{t-h} \right \|_2^2 \lesssim \frac{\eta_{T}^{2-q}  R_q}{\kappa_1^{1-q}}.
	\]
\end{theorem}

Combining Theorem \ref{thm:lasso} with Proposition \ref{prop:perturb}, we immediately have the estimation error bounds $\|\bm{\widehat{g}}-\bm{g}^*\|_2\lesssim \eta_{T}^{1-q/2} \sqrt{R_q}$ and $\|\bm{\widehat{\omega}}-\bm{\omega}^*\|_2\lesssim \underline{\alpha}_\ma^{-1} \eta_{T}^{1-q/2} \sqrt{R_q}$. In particular, under exact sparsity, when $r=s=0$, the bound for $\|\bm{\widehat{a}}-\bm{a}^*\|_{2}$ in Theorem \ref{thm:lasso} matches that for the Lasso estimator of VAR($p$) models in \cite{basu2015regularized}, while the Gaussian assumption is relaxed. Also note that we do not require the uniqueness of the optimal solution to \eqref{eq:lasso}, that is,  Theorem \ref{thm:lasso} is valid for all local optima within the constant-radius neighborhood of $\bm{\omega}^*$. 

The JE   can be efficiently implemented via the block coordinate descent algorithm; see Section S1.1 of the supplementary file for details.
While the value of $c_{\bm{\omega}}$ is unknown in practice, it is known to be independent of $N$ and $T$  under the mild condition that $\underline{\alpha}_\ma \asymp \overline{\alpha}_\ma$. 
The practical implication of the condition $\|\bm{\widehat{\omega}} - \bm{\omega}^*\|_{2}\leq c_{\bm{\omega}}$ is that a reasonably good  initialization for $\bm{\omega}$ will be needed for the optimization algorithm of \eqref{eq:lasso}. 
For nonconvex estimators, to meet such requirements, commonly a convex preliminary estimator is used to initialize the algorithm  \citep[e.g.,][]{JG2021}.
However, for our model, the initialization task can  be  simplified, because the $r$ values  $\lambda_1,\dots, \lambda_r \in(-1,1)$ and the $s$ values $\bm{\eta}_1,\dots,\bm{\eta}_s \in [0,1)\times (0,\pi)$ are restricted to bounded spaces and must be well separated from one another; see Assumptions \ref{assum:statn}(i) and \ref{assum:gap}(i). In fact, when $r$ and $s$ are larger, the initialization of  $\bm{\omega}$ will be even easier, as the selected $r$ and $s$ values  will be denser on the bounded space and hence naturally tend to be closer to the true values. In practice, we recommend considering several different initial values for $\bm{\omega}$ and   selecting the solution of the optimization with minimum in-sample squared loss;  see Section S1.2 of the supplementary file for details.

\begin{remark}\label{remark:Sigma}
	Following the method for sparse VAR($P$) models in \cite{KP2021}, under a weak sparsity assumption on $\bm{\Sigma}_\varepsilon$, we can construct a high-dimensional estimator of $\bm{\Sigma}_\varepsilon$ as $\bm{\widehat{\Sigma}}_\varepsilon=\textrm{THR}_{\lambda_\varepsilon}(T^{-1}\sum_{t=1}^T \bm{\widehat{\varepsilon}}_t \bm{\widehat{\varepsilon}}_t^\top)$, where the residuals $\bm{\widehat{\varepsilon}}_t$ are obtained based on $\bm{\widehat{A}}_h$'s, and  $\textrm{THR}_{\lambda_\varepsilon}(\cdot)$ is the entrywise thresholding function with a chosen threshold parameter $\lambda_\varepsilon>0$; see  \cite{KP2021} for details. Then, based on $\bm{\widehat{\Sigma}}_\varepsilon$ and $\bm{\widehat{A}}_h$'s, we can estimate $\var(\bm{y}_t)$, so  the instantaneous cross-sectional dependence can be interpreted. We leave a rigorous theoretical study of this estimation for future research.
\end{remark}

\begin{remark}\label{remark:algoconv}
	While Theorem \ref{thm:lasso} establishes statistical error bounds, an interesting avenue for future research is to develop a more comprehensive estimation theory that integrates both statistical and algorithmic convergence analyses; see similar works such as \cite{ANW2012} and \cite{Loh2017}. To tackle the theoretical challenges arising from the nonconvexity of the loss function, Proposition \ref{prop:perturb} may be leveraged to transform the problem into a convex one within a local region around $\bm{\omega}^*$. 
\end{remark}

\subsection{$\ell_1$-regularized rowwise estimator}\label{subsec:RE}

While Theorem \ref{thm:lasso} allows $R_q$ to grow with $N$, it requires $R_q \lesssim \varpi/\eta_{T}^{2-q}$; e.g., if $q=0$, then this essentially will become $R_0 \lesssim T/ \log \{N(p\vee 1)\}$. However, this requirement could be stringent when $T$ is relatively small. To relax the sparsity requirement, we further introduce  a rowwise estimator (RE) based on separately fitting each row of the proposed model.

For $1\leq i\leq N$,  the $i$th row of model \eqref{eq:model-scalar} is $y_{i,t}=\sum_{h=1}^{\infty}\bm{a}_{i,h}^\top \bm{y}_{t-h} + \varepsilon_{i,t}$, 
where $\bm{a}_{i,h}=\sum_{k=1}^{d} \ell_{h,k}(\bm{\omega}) \bm{g}_{i,k}\in\mathbb{R}^N$ is the $i$th row of $\bm{A}_h$, and $\bm{g}_{i,k}\in\mathbb{R}^N$ is the $i$th row of $\bm{G}_k$. Then, the squared loss for the $i$th row is $\mathbb{L}_{i,T}(\bm{\omega}, \bm{g}_i)=T^{-1}\sum_{t=1}^{T} ( y_{i,t} - \sum_{h=1}^{\infty}\bm{a}_{i,h}^\top \bm{y}_{t-h})^2= T^{-1} \sum_{t=1}^{T} \{ y_{i,t} - \sum_{k=1}^{d}\bm{g}_{i,k}^\top \sum_{h=1}^{\infty}\ell_{h,k}(\bm{\omega}) \bm{y}_{t-h} \}^2$, where $\bm{g}_i=(\bm{g}_{i,1}^\top, \dots, \bm{g}_{i,d}^\top)^\top\in\mathbb{R}^{Nd}$ is the $i$th row of $\bm{G}=(\bm{G}_1, \dots, \bm{G}_d)$.
Note that joint loss function as defined in the previous section can be decomposed as
$\mathbb{L}_{T}(\bm{\omega}, \bm{g})=\sum_{i=1}^{N}\mathbb{L}_{i,T}(\bm{\omega}, \bm{g}_i)$. Thus, the rowwise losses $\mathbb{L}_{i,T}(\cdot)$'s  can be minimized separately with respect to $\bm{g}_i$ for $1\leq i\leq N$. Meanwhile, since $\bm{\omega}$ is shared by all $\mathbb{L}_{i,T}(\cdot)$'s, each rowwise minimization can yield a consistent estimator of $\bm{\omega}$. This motivates us to consider the following $\ell_1$-regularized RE for $1\leq i \leq N$:
\begin{equation}\label{eq:lassorow}
	(\bm{\widehat{\omega}}_{i}, \bm{\widehat{g}}_{i})=\argmin_{\bm{\omega}\in \bm{\Omega}, \, \bm{g}_i \in\mathbb{R}^{Nd} }\left \{\widetilde{\mathbb{L}}_{i,T}(\bm{\omega}, \bm{g}_i)+\lambda_g \|\bm{g}_i \|_1 \right \},
\end{equation}
where $\lambda_g >0$ is the regularization parameter, and $\widetilde{\mathbb{L}}_{i,T}(\bm{\omega}, \bm{g}_i)$ is  defined by setting the initial values $\{y_{i,s}, s\leq 0\}$  to zero, i.e., 
$\widetilde{\mathbb{L}}_{i,T}(\bm{\omega}, \bm{g}_i) =T^{-1}\sum_{t=1}^{T}  (y_{i,t} - \sum_{h=1}^{t-1}\bm{a}_{i,h}^\top \bm{y}_{t-h} )^2
=T^{-1}\sum_{t=1}^{T} \{ y_{i,t} - \sum_{k=1}^{d}\bm{g}_{i,k}^\top \sum_{h=1}^{t-1}\ell_{h,k}(\bm{\omega}) \bm{y}_{t-h}\}^2$.	
Let $\bm{a}_i=(\bm{a}_{i,1}^\top,  \bm{a}_{i,2}^\top, \dots)^\top\in\mathbb{R}^{\infty}$ be the $i$th row of $\bm{A}=(\bm{A}_1, \bm{A}_2, \dots)$ for $1\leq i\leq N$. 
Note that $\bm{a}_i= (\bm{L}(\bm{\omega})\otimes \bm{I}_N) \bm{g}_{i}$. Based on \eqref{eq:lassorow}, we have $\bm{\widehat{a}}_{i}=(\bm{\widehat{a}}_{i,1}^\top, \bm{\widehat{a}}_{i,2}^\top, \dots)^\top=(\bm{L}(\bm{\widehat{\omega}})\otimes \bm{I}_N) \bm{\widehat{g}}_{i}$, where $\bm{\widehat{g}}_{i}=(\bm{\widehat{g}}_{i,1}^\top, \dots, \bm{\widehat{g}}_{i,d}^\top)^\top$, and $\bm{\widehat{a}}_{i,h} =\sum_{k=1}^{d} \ell_{h,k}(\bm{\widehat{\omega}}_{i}) \bm{\widehat{g}}_{i,k}$. The algorithm for the RE is provided in Section S1.1 of  the supplementary file.

Similar to the previous section, we can derive the nonasymptotic error bounds for the RE. For $1\leq i\leq N$, let $\bm{g}_{i,\ar}=(\bm{g}_{i,1}^\top,\dots, \bm{g}_{i,p}^\top)^\top \in\mathbb{R}^{Np}$ and $\bm{g}_{i,\ma}=(\bm{g}_{i,p+1}^\top,\dots, \bm{g}_{i,d}^\top)^\top  \in\mathbb{R}^{N(r+2s)}$. To define the  weak sparsity of $\bm{g}_{i,\ar}^*$ and $\bm{g}_{i,\ma}^*$, we consider the $\ell_q$-``balls'', $\mathbb{B}_q(R_{i,q}^\ar):=\{\bm{g}_{i,\ar}\in\mathbb{R}^{Np}\mid \sum_{k=1}^{p}\sum_{j=1}^{N} |g_{i,j,k}|^q\leq R_{i,q}^\ar\}$ and $\mathbb{B}_q(R_{i,q}^\ma):=\{\bm{g}_{i,\ma}\in\mathbb{R}^{N(r+2s)}\mid \sum_{k=p+1}^{d}\sum_{j=1}^{N} |g_{i,j,k}|^q\leq R_{i,q}^\ma\}$. The following is the row-wise counterpart of Assumption \ref{assum:sparse}.

\renewcommand{\theassumption}{\arabic{assumption}$^\prime$}
\setcounter{assumption}{3}
\begin{assumption}[Rowwise weak sparsity]\label{assum:rowsparse}
	For $1\leq i\leq N$, there exists $q\in[0,1]$ such that $\bm{g}_{i,\ar}^*\in\mathbb{B}_q(R_{i,q}^\ar)$ and $\bm{g}_{i,\ma}^*\in\mathbb{B}_q(R_{i,q}^\ma)$ for some radii  $R_{i,q}^\ar, R_{i,q}^\ma>0$.
\end{assumption}

Let $R_{i,q}=R_{i,q}^\ar + R_{i,q}^\ma$, and then by Assumption \ref{assum:rowsparse}, $\bm{g}_i^*\in \mathbb{B}_q(R_{i,q}):=\{\bm{g}_i\in\mathbb{R}^{Nd}\mid \sum_{k=1}^{d}\sum_{j=1}^{N} |g_{i,j,k}|^q\leq R_{i,q}\}$. Moreover, Assumption \ref{assum:rowsparse}  implies the overall sparsity level in Assumption \ref{assum:sparse}, since it leads to $\bm{g}_\ar^*\in \mathbb{B}_q(R_{q}^\ar)$, $\bm{g}_\ma^*\in \mathbb{B}_q(R_{q}^\ma)$, and consequently $\bm{g}^*\in \mathbb{B}_q(R_{q})$, where $R_{q}^\ar=\sum_{i=1}^{N}R_{i,q}^\ar$, $R_{q}^\ma=\sum_{i=1}^{N}R_{i,q}^\ma$, and $R_q=R_{q}^\ma+R_{q}^\ar=\sum_{i=1}^N R_{i,q}$.

For  $1\leq i\leq N$, let $\underline{\alpha}_{i,\ma} = \min_{1\leq j\leq r+2s} \|\bm{g}_{i,p+j}^*\|_{2}$ and  $\overline{\alpha}_{i,\ma} = \max_{1\leq j\leq r+2s} \|\bm{g}_{i,p+j}^*\|_{2}$, which are both allowed to grow with $N$. Denote $\alpha_i= \overline{\alpha}_{i,\ma}/\underline{\alpha}_{i,\ma}$. The rowwise counterparts of   Proposition \ref{prop:perturb} and Theorem \ref{thm:lasso} are established as follows.

\begin{proposition}\label{prop:perturbrow}
	Fix $1\leq i\leq N$. Suppose that $\underline{\alpha}_{i,\ma}>0$.	Then under Assumptions \ref{assum:statn}(i) and \ref{assum:gap}, there exists a constant $c_{i,\bm{\omega}}= \min (2,   c/ \alpha_i)>0$ such that for any $\bm{\omega}\in\bm{\Omega}$ with $\|\bm{\omega} - \bm{\omega}^*\|_2\leq c_{i,\bm{\omega}}$, it holds $\|\bm{g}_i-\bm{g}_i^*\|_{2} + \underline{\alpha}_{i,\ma} \|\bm{\omega} - \bm{\omega}^*\|_2 \lesssim \|\bm{a}_i-\bm{a}_i^*\|_{2}^2\lesssim \|\bm{g}_i-\bm{g}_i^*\|_{2} + \overline{\alpha}_{i,\ma} \|\bm{\omega} - \bm{\omega}^*\|_2$, where  $\bm{a}_i= (\bm{L}(\bm{\omega})\otimes \bm{I}_N) \bm{g}_{i}$.
\end{proposition}

\begin{theorem}\label{thm:lassorow}
	Suppose that Assumptions \ref{assum:statn}--\ref{assum:error}  and \ref{assum:rowsparse} hold with $\sum_{j=0}^{\infty} \|\bm{\Psi}_j^* \|_{\op}^2 <\infty$, $R_{i,q} \lesssim \varpi/\eta_{T}^{2-q}$, $\alpha_i^2 \lesssim R_{i,q}/R_{i,q}^\ma$, $\varpi \lesssim  \overline{\alpha}_{i,\ma}^2 R_{i,q}/R_{i,q}^\ma$, and $\underline{\alpha}_{i,\ma}>0$, for $1\leq i\leq N$. In addition, assume that  $\log N \gtrsim (\kappa_2/\kappa_1)^2$, 	$T\gtrsim \max\{\kappa_2 (p\vee1)^4, (\kappa_2/\kappa_1)^2  (p\vee1) \log\{ (\kappa_2/\kappa_1)\alpha_{\max} N (p\vee1)\}\}$, with $\alpha_{\max}=\max_{1\leq i\leq N}\alpha_i$, and we solve \eqref{eq:lassorow} with  $\lambda_g \asymp \sqrt{\kappa_2 \lambda_{\max}(\bm{\Sigma}_{\varepsilon}) \log \{N(p\vee 1)\} /T}$. For $1\leq i\leq N$, if $\|\bm{\widehat{\omega}}_i - \bm{\omega}^*\|_{2}\leq c_{i,\bm{\omega}}$, then  with probability at least $1 - C (p\vee 1) e^{-c (\kappa_1 /\kappa_2)^2 \log N}$,
	\[
	\|\bm{\widehat{a}}_{i}-\bm{a}_i^*\|_{2} \lesssim \eta_{T}^{1-q/2} \sqrt{R_{i, q}} \quad\text{and}\quad 
	\frac{1}{T}\sum_{t=1}^{T} \left \| \sum_{h=1}^{t-1}(\bm{\widehat{a}}_{i,h}-\bm{a}_{i,h}^*)^\top \bm{y}_{t-h} \right \|_2^2 \lesssim \frac{\eta_{T}^{2-q} R_{i, q}}{\kappa_1^{1-q}}.
	\]
\end{theorem}

Compared to Theorem \ref{thm:lassorow}, the sparsity condition in Theorem \ref{thm:lassorow} is much weaker, i.e., $R_{i,q} \lesssim \varpi/\eta_{T}^{2-q}$ for $1\leq i\leq N$; or essentially, $R_{i,0} \lesssim T/ \log \{N(p\vee 1)\}$ when $q=0$. Thus, the RE may be preferred in practice when $T$ is relatively small. 

Moreover, by Theorem \ref{thm:lassorow} and Proposition \ref{prop:perturbrow}, we have $\|\bm{\widehat{g}}_i-\bm{g}_i^*\|_2\lesssim \eta_{T}^{1-q/2} \sqrt{R_{i,q}}$ and $\|\bm{\widehat{\omega}}_i-\bm{\omega}^*\|_2\lesssim \underline{\alpha}_{i,\ma}^{-1} \eta_{T}^{1-q/2} \sqrt{R_{i,q}}$ for $1\leq i\leq N$.  
Note that each RE $\bm{\widehat{\omega}}_i$ is a consistent estimator of $\bm{\omega}^*$, and the estimation error is proportional to $\underline{\alpha}_{i,\ma}^{-1}\sqrt{R_{i,q}}$. On the other hand, as implied by Theorem \ref{thm:lasso}, the estimation error of the JE for $\bm{\omega}^*$ is proportional to  $\underline{\alpha}_{\ma}^{-1}\sqrt{R_q}$. For example, if $R_{i,q}\asymp R_q/N$ and $\underline{\alpha}_{i,\ma}^2 \asymp \underline{\alpha}_{\ma}^2/N$, then the two bounds will be comparable. However, intuitively, allowing different estimators $\bm{\widehat{\omega}}_i$ for different rows may enhance the flexibility in practice, although it may also increase the risk of overfitting.
In addition, combining the results for $\bm{\widehat{a}}_i$, $\bm{\widehat{g}}_i$ and the prediction error across all rows, we have  $\|\bm{\widehat{a}}-\bm{a}^*\|_{2} \lesssim \eta_{T}^{1-q/2} \sqrt{R_q}$, $\|\bm{\widehat{g}}-\bm{g}^*\|_{2} \lesssim \eta_{T}^{1-q/2} \sqrt{R_q}$, and $T^{-1}\sum_{t=1}^{T} \| \sum_{h=1}^{t-1}(\bm{\widehat{A}}_h-\bm{A}_h^*) \bm{y}_{t-h} \|_2^2 \lesssim \eta_{T}^{2-q}  R_q / \kappa_1^{1-q}$. Here, with a slight abuse of  notation,  $\bm{\widehat{a}}$,  $\bm{\widehat{g}}$ and $\bm{\widehat{A}}_h$'s represent the estimates obtained based on  merging the RE $\bm{\widehat{a}}_i$ or $\bm{\widehat{g}}_i$ for $1\leq i\leq N$. Note that these bounds match exactly those of the JE in the previous section.

In addition to the above upper bounds analysis, we numerically assess the actual comparative performance of RE and JE via simulations in  Section S2.2 of the supplementary file. It is shown that they can perform very similarly for the estimation of $\bm{g}^*$, while RE may outperform JE for the estimation of $\bm{\omega}^*$, resulting in an overall advantage for the estimation of $\bm{a}^*$. However, as long as $T$ is not too small compared to $R_q$, JE and RE tend to have similar out-of-sample forecast accuracy; see the empirical analysis in Section \ref{sec:empirical} and the simulation study in Section S2.4 of the supplementary file for details. Furthermore, as commented by one referee, the competitive numerical performance of the JE might hint that its more stringent sparsity condition could be an artifact of the proof technique.

\section{Model order selection}\label{sec:BIC}

In this section, we introduce a Bayesian information criterion (BIC) based approach to selecting the model orders for the proposed high-dimensional SPVAR($\infty$) model. 

Let $\pazocal{M}^*=(p^*,r^*,s^*)$ denote the true orders.
For the feasibility of order selection, it is crucial to  ensure that  $\pazocal{M}^*$ is irreducible; i.e., if $\{\bm{y}_t\}$ is generated with orders $\pazocal{M}^*$,   there is no alternative parameterization with reduced orders. As  established in Lemma S14 in the supplementary file, the irreducibility of $r^*$ and $s^*$ is  guaranteed if $\lambda_j^*$'s, $\gamma_m^*$'s, and  $\underline{\alpha}_\ma$ are  nonzero. On the other hand, $p^*$ is irreducible under the following assumption.

\renewcommand{\theassumption}{\arabic{assumption}}
\begin{assumption}[Irreducibility] \label{assum:irred}
	$\bm{G}_{p^*}\neq\sum_{j=1}^{r^*}\bm{G}_{p^*+j} +\sum_{m=1}^{s^*} \bm{G}_{p^*+r^*+2m-1}$.
\end{assumption}

To select the model orders, for any $\pazocal{M}=(p,r,s)$, we define the high-dimensional BIC,
\begin{equation}\label{eq:bic}
	\textup{BIC}(\pazocal{M}) = \log \widetilde{\mathbb{L}}_T(\widehat{\bm{\omega}}_{\pazocal{M}},\widehat{\bm{g}}_{\pazocal{M}})  + \tau_N d \left [\frac{\log\{N(p\vee1)\}}{T}\right ]^{1-q/2} \log T,
\end{equation}
where $\bm{\widehat{\omega}}_{\pazocal{M}}$ and $\widehat{\bm{g}}_{\pazocal{M}}$ denote estimates obtained by fitting the model with orders $\pazocal{M}$ using either the JE in \eqref{eq:lasso} or the RE in \eqref{eq:lassorow}.  In particular, if the RE is employed, then $\widetilde{\mathbb{L}}_T(\widehat{\bm{\omega}}_{\pazocal{M}},\widehat{\bm{g}}_{\pazocal{M}})=\sum_{i=1}^{N}\mathbb{L}_{i,T}(\widehat{\bm{\omega}}_{i,\pazocal{M}}, \widehat{\bm{g}}_{i,\pazocal{M}})$, where $\widehat{\bm{\omega}}_{\pazocal{M}}$ and $\widehat{\bm{g}}_{\pazocal{M}}$ denote collections of $\widehat{\bm{\omega}}_{i,\pazocal{M}}$'s and $\widehat{\bm{g}}_{i,\pazocal{M}}$'s, respectively. Note that for notational simplicity, we suppress the dependence of $\widetilde{\mathbb{L}}_T(\cdot)$ and $\mathbb{L}_T(\cdot)$ on $\pazocal{M}$   in this section.
Additionally, $\tau_{N}>0$ is a sequence possibly dependent on $N$ satisfying the following condition.

\begin{assumption}[Penalty parameter]\label{assum:penalty}
	$\tau_{N}\gtrsim N^{-1} R_q \{\kappa_2 \lambda_{\max}(\bm{\Sigma}_{\varepsilon})\}^{1-q/2}/\kappa_1^{3-2q}$.
\end{assumption}

Assumption \ref{assum:penalty} ensures that the proposed BIC can rule out any overspecified model, $\pazocal{M}\in\mathcal{M}_{\textup{over}}= \{\pazocal{M}\in \mathcal{M}\mid  p\geq p^*,  r\geq r^* \text{ and } s\geq s^*\}\setminus \pazocal{M}^*$.
When the constants $\kappa_1, \kappa_2$ and $\lambda_{\max}(\bm{\Sigma}_{\varepsilon})$  are fixed, Assumption \ref{assum:penalty} can be simplified to $\tau_{N}\gtrsim  N^{-1} R_q$. While $R_q$ is unknown in practice, to set a reasonable $\tau_{N}$, we may assume that $R_q\lesssim N$; e.g., this will hold if $\bm{G}_k^*$'s are (weakly) row-sparse. Then it would suffice to fix $\tau_{N}\equiv\tau>0$. 
In practice, we may simply set $q=0$. We recommend $\tau=0.05$, which performs well in our simulations.  

Based on \eqref{eq:bic}, we estimate the model orders by
\[
\widehat{\pazocal{M}}=(\widehat{p}, \widehat{r}, \widehat{s}) = \argmin_{\pazocal{M}\in \mathcal{M}} \textup{BIC}(\pazocal{M}),
\]
where $\mathcal{M}=\{(p,r,s)\mid 0\leq p \leq \overline{p}, 0\leq r \leq \overline{r}, 0\leq s \leq \overline{s} \}$, with $\overline{\pazocal{M}}:=(\overline{p},\overline{r},\overline{s})$ being predetermined maximum orders. Since the true orders are usually small in practice, $\overline{\pazocal{M}}$ need not be large; e.g. $\overline{p}= \overline{r}=  \overline{s}=6$ may be sufficient for most applications. Our simulations show that $\widehat{\pazocal{M}}$ is insensitive to the choice of $\overline{\pazocal{M}}$  as long as it is large enough compared to $\pazocal{M}^*$.

Let
$\mathcal{M}_{\textup{mis}}= \{\pazocal{M}\in \mathcal{M}\mid  p< p^*,  r< r^* \text{ or } s< s^*\}$.
To  establish the conditions that prevent the  proposed BIC from selecting any misspecified model, we need to accurately quantify the minimum difference  between any $\pazocal{M}\in \mathcal{M}_{\textup{mis}}$ and $\pazocal{M}^*$. This analysis is challenging since there is no monotonic nested ordering over $\mathcal{M}$ due to the involvement of three different orders,  $p,r$ and $s$. Particularly,
$\pazocal{M}\in\mathcal{M}_{\textup{mis}}$ may not be nested within $\pazocal{M}^*$ regarding all   three orders. For instance, if $\pazocal{M}^*=(1,1,0)$, then a misspecified model may be  $\pazocal{M}_1=(\overline{p},0,0)$ or $\pazocal{M}_2=(0,\overline{r},\overline{s})$, where, e.g., $\overline{p}=\overline{r}=\overline{s}=6$. Clearly, we cannot simply treat $\pazocal{M}_1$ or $\pazocal{M}_2$ as a smaller model than  $\pazocal{M}^*$, as they possess  orders as large as  $\overline{p}$, $\overline{r}$, or $\overline{s}$.

To uniformly accommodate the possibly nonnested relationship between $\pazocal{M}\in\mathcal{M}_{\textup{mis}}$ and  $\pazocal{M}^*$,  we leverage their connections  with a common model, $\overline{\pazocal{M}}=(\overline{p},\overline{r}, \overline{s})$. Specifically, we can show that model \eqref{eq:model-scalar} with any orders $\pazocal{M}=(p,r,s)\in \mathcal{M}$ can be reparameterized as the model with  $\overline{\pazocal{M}}=(\overline{p},\overline{r}, \overline{s})$. In addition, the corresponding parameter vectors, denoted $\bm{\overline{\omega}}\in (-1,1)^{\overline{r}}\times \bm{\varPi}^{\overline{s}}$ and $\bm{\overline{g}} \in\mathbb{R}^{N\times N \overline{d}}$, satisfy the following equality constraints: 
\begin{equation}\label{eq:Cmats}
	\bm{\overline{C}}_1^{\pazocal{M}} \bm{\overline{\omega}}=\bm{0} \quad\text{and}\quad \left (\bm{\overline{C}}_2^{\pazocal{M}}( \bm{\overline{\omega}})  \otimes \bm{I}_{N^2} \right ) \bm{\overline{g}}=\bm{0},
\end{equation}
where $\bm{\overline{C}}_1^{\pazocal{M}}\in\mathbb{R}^{(\delta_r+2\delta_s)\times (\overline{r}+2\overline{s})}$ is a constant matrix encoding  $(\delta_r+2\delta_s)$   constraints on $\bm{\overline{\omega}}$, specifying which elements are restricted to zero, and the matrix function $\bm{\overline{C}}_2^{\pazocal{M}}(\bm{\overline{\omega}})\in \mathbb{R}^{\delta_d \times \overline{d}}$ encodes $\delta_d$ equality constraints on $\bm{\overline{g}}$ for any given $\bm{\overline{\omega}}$, with  $\delta_r=\overline{r}-r$, $\delta_s=\overline{s}-s$,  and $\delta_d=\overline{d}-d$; see Section S7.3 in the supplementary file for detailed definitions of $\bm{\overline{C}}_1^{\pazocal{M}}$ and $\bm{\overline{C}}_2^{\pazocal{M}}( \cdot)$. In particular, increasing $p$ by one amounts to deleting a particular row from the constraint matrix $\bm{\overline{C}}_2^{\pazocal{M}}( \cdot)$. On the other hand, increasing $r$ (or $s$) by one is equivalent to deleting a particular row (or a pair of rows) from both $\bm{\overline{C}}_1^{\pazocal{M}}$ and $\bm{\overline{C}}_2^{\pazocal{M}}( \cdot)$.

Note that $\bm{\overline{C}}_2^{\pazocal{M}}(\cdot)$ cannot reduce to a constant matrix independent of $\bm{\overline{\omega}}$ except in the special cases where $p=\overline{p}-1$ or $r=s=0$. In particular, when $p=\overline{p}-1$,   the second equation in \eqref{eq:Cmats}   is essentially the reducibility  condition of $\overline{p}$, which resembles that for $p^*$ in Assumption \ref{assum:irred}(i). However, in general, this equation represents  much more intricate constraints, since $\bm{\overline{C}}_2^{\pazocal{M}}(\cdot)$ is a nonlinear function.
The complexity of this form can be understood from two perspectives. First, due to the nonlinearity of  model  \eqref{eq:model-scalar} in $\bm{\omega}$, the effect of any underspecification in $r$ or $s$ will be highly nonlinear. Second,  the order $p$ plays a special role in the definition of $\ell_{h,k}(\cdot)$'s as it is involved in $\mathbb{I}_{\{h\geq p+1\}}\lambda_j^{h-p}$ and $\mathbb{I}_{\{h\geq p+1\}}\gamma_{m}^{h-p}$; see \eqref{eq:linearcomb2}. Then, whenever $p\neq p^*$, the exponent $h-p$  will differ from that under $\pazocal{M}^*$ for all lags  $h\geq p+1$, thereby affecting all $\ell_{h,k}(\cdot)$'s.  Consequently, due to the interplay between $p$ and  $\ell_{h,k}(\cdot)$'s, an  underspecification in $p$ generally will also have a nonlinear effect.

Let $\bm{\Gamma}_{\pazocal{M}}=\{ \bm{\overline{\omega}}\in (-1,1)^{\overline{r}} \times \bm{\varPi}^{\overline{s}}, \; \bm{\overline{g}} \in\mathbb{R}^{N^2 \overline{d}}: \bm{\overline{C}}_1^{\pazocal{M}} \bm{\overline{\omega}}=\bm{0} \text{ and }
(\bm{\overline{C}}_2^{\pazocal{M}}( \bm{\overline{\omega}})  \otimes \bm{I}_{N^2}) \bm{\overline{g}}=\bm{0} \}$ denote the restricted parameter space for any candidate model $\pazocal{M}$. By leveraging \eqref{eq:Cmats}, we can characterize the minimum difference  between the true model and the approximated model of orders $\pazocal{M}\in \mathcal{M}_{\textup{mis}}$ via the quantity $\delta_{\pazocal{M}}:=\kappa_1 \inf_{(\bm{\omega}, \bm{g})\in \bm{\Gamma}_{\pazocal{M}}} \|(\bm{L}(\bm{\omega})\otimes \bm{I}_{N^2})\bm{g}-\bm{a}^*\|_2^2$; see Proposition S1 and the proof of Theorem \ref{thm:selection} in Section S7 of the supplementary file for details. We may regard $\delta_{\pazocal{M}}$ as the signal strength of the misspecification. The following assumption guarantees that $\delta_{\pazocal{M}}$ is large enough for the BIC to detect the misspecification.

\begin{assumption}[Minimum signal strength]\label{assum:signaljoint}
	(i) $\min_{\pazocal{M}\in \mathcal{M}_{\textup{mis}}}\delta_{\pazocal{M}} /N\gg    (T^{-1}\log N)^{1-q/2} \tau_N  \log T$; and (ii) $\max_{\pazocal{M}\in \mathcal{M}_{\textup{mis}}} \delta_{\pazocal{M}}^{-1}  |\widetilde{\mathbb{L}}_T(\widehat{\bm{\omega}}_{\pazocal{M}},\widehat{\bm{g}}_{\pazocal{M}})- \mathbb{E}\{\mathbb{L}_T(\bm{\omega}_{\pazocal{M}}^{\circ}, \bm{g}_{\pazocal{M}}^{\circ})\}|=o_p(1)$, where $(\bm{\omega}_{\pazocal{M}}^{\circ}, \bm{g}_{\pazocal{M}}^{\circ})$ is the minima of $\mathbb{E}\{\mathbb{L}_T(\bm{\omega}_{\pazocal{M}}, \bm{g}_{\pazocal{M}})\}$ over the parameter space $\bm{\omega}_{\pazocal{M}}\in (-1,1)^r\times \bm{\varPi}^s$ and $\bm{g}_{\pazocal{M}}\in\mathbb{R}^{N^2d}$.
\end{assumption}

Note that $\delta_{\pazocal{M}}/N$ can be viewed as the average level of misspecification across $N$ rows of the model equation. As mentioned earlier, we may let $\tau_N\equiv \tau$ under mild condition. Thus, the lower bound in Assumption \ref{assum:signaljoint}(i) tends to zero as $T\rightarrow\infty$. Assumption \ref{assum:signaljoint}(ii) requires that the empirical loss for any fitted misspecified model converges to some population loss at a rate faster than $\delta_{\pazocal{M}}$ as $T\rightarrow\infty$. Here the mispecified model with parameters $(\bm{\omega}_{\pazocal{M}}^{\circ}, \bm{g}_{\pazocal{M}}^{\circ})$  can be understood as the best approximation of the process $\{\bm{y}_t\}$ under the misspecification. Now we are ready to establish the consistency of the estimator $\widehat{\pazocal{M}}$.

\begin{theorem}\label{thm:selection}
	If the JE (or the RE) is used,  suppose that for any $\pazocal{M}\in\mathcal{M}_{\textup{over}}$, there is a subvector  $\bm{\widehat{\omega}}_{\pazocal{M}^*}\in(-1,1)^{r^*}\times\bm{\varPi}^{s^*}$  of $\bm{\widehat{\omega}}_{\pazocal{M}}$ (or $\bm{\widehat{\omega}}_{i,\pazocal{M}^*}\in(-1,1)^{r^*}\times\bm{\varPi}^{s^*}$  of $\bm{\widehat{\omega}}_{i,\pazocal{M}}$ with $1\leq i\leq N$) such that  $\|\bm{\widehat{\omega}}_{\pazocal{M}^*} - \bm{\omega}^*\|_2\leq c_{\bm{\omega}}$ 
	(or $\|\bm{\widehat{\omega}}_{i,\pazocal{M}^*} - \bm{\omega}^*\|_2\leq c_{i,\bm{\omega}}$ with $1\leq i\leq N$), and the conditions in  Theorem \ref{thm:lasso} (or \ref{thm:lassorow}) hold with $\pazocal{M}=\pazocal{M}^*$.
	In addition, suppose that $\overline{\pazocal{M}}$ is fixed, with $\overline{p}\geq p^*, \overline{r}\geq r^*$ and $\overline{s}\geq s^*$. Under Assumptions \ref{assum:irred}--\ref{assum:signaljoint}, $\mathbb{P}( \widehat{\pazocal{M}} =\pazocal{M}^*) \to 1$ as $N,T\rightarrow\infty$.
\end{theorem} 

%%%%%%%%%%%%%%%%%%%%%%%%%%%%%%%%%%%%%%%%%%%%%%%%%%%%%%%%%%%%%%%%%%%%%%%%%%%%%%%%%%%%%%
\section{Simulation experiments}\label{sec:sim}
In this section, we present two simulation experiments to verify the estimation error rates of the JE and the consistency of the BIC. Four additional  experiments on the estimation error  of the RE, its comparison with the JE, sensitivity analysis of the initialization for $\{\bm{y}_t, t\leq 0\}$, and comparison of the proposed estimators with competing approaches  are provided in Section S2 of the supplementary file.

Throughout this section, we generate $\{\bm{y}_t\}$ from model \eqref{eq:model-scalar}, where $\{\bm{\varepsilon}_t\}$ are generated independently from $N(\bm{0}, \sigma^2\bm{I}_N)$ with $\sigma=0.2$, and each $\bm{G}_k$ is exactly sparse with $cN$ nonzero entries for $1\leq k\leq d$, so the overall sparsity level is $R_0= cdN$. We generate $\{\bm{G}_k\}_{k=1}^d$ by drawing their nonzero entries independently from the uniform distribution on $[-0.5,0.5]$.  Then, to ensure the  stationarity of $\{\bm{y}_t\}$, after setting $\bm{\omega}$, we rescale all $\bm{G}_k$'s by a common factor such that $\rho(\underline{\bm{G}}_1)+\bar{\rho} \sum_{k=1}^{r+2s}\rho(\bm{G}_{p+k}) / (1-\bar{\rho})=0.8$; see Theorem \ref{thm:stationary}. 

\begin{figure}[t]
	\centering
	\includegraphics[width=0.9\textwidth]{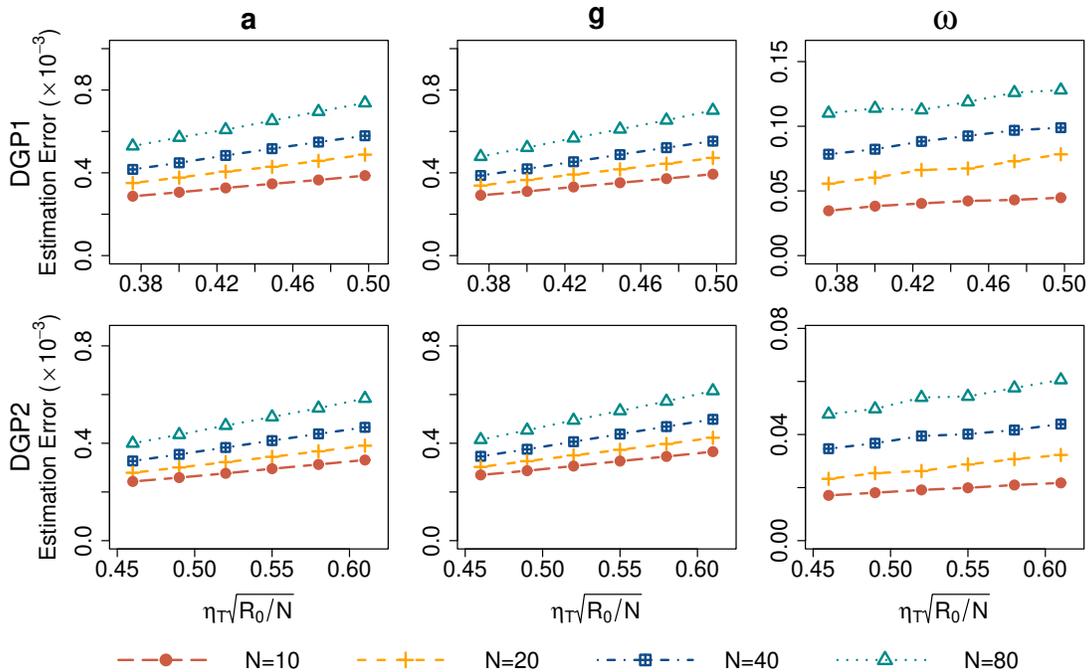}
	\caption{Plots of scaled estimation errors   $\|\bm{\widehat{a}}-\bm{a}^*\|_2 /\sqrt{N}$ (left panel), $\|\bm{\widehat{g}}-\bm{g}^*\|_2 /\sqrt{N}$ (middle panel), and $\underline{\alpha}_\ma \|\bm{\widehat{\omega}}-\bm{\omega}^*\|_2 /\sqrt{N}$ (right panel) against theoretical rate $\eta_T\sqrt{R_0/N}$ for JE. }
	\label{fig:exp1}
\end{figure}

In the first experiment, we examine the estimation error rates for the JE.   Two data generating processes are considered: $(p,r,s)=(1,1,0)$ (DGP1) and $(1,0,1)$ (DGP2), where $\lambda_1 = -0.6$ for DGP1, and $(\gamma_1,\theta_1)=(0.6, \pi/4)$ for DGP2. We let all $\bm{G}_k$'s be row-sparse matrices with three nonzero entries in each row, i.e., $R_0=3dN$, where $N=10, 20, 40$ or 80. Note that by Theorem \ref{thm:lasso}, we have $\|\bm{\widehat{a}}-\bm{a}^*\|_2 /\sqrt{N}\lesssim \eta_T\sqrt{R_0/N}$,
$\|\bm{\widehat{g}}-\bm{g}^*\|_2 /\sqrt{N} \lesssim \eta_T\sqrt{R_0/N}$, and $\underline{\alpha}_\ma \|\bm{\widehat{\omega}}-\bm{\omega}^*\|_2 /\sqrt{N} \lesssim \eta_T\sqrt{R_0/N}$, where $\eta_T=\sqrt{T^{-1}\log N}$.
To verify these bounds, we choose a grid of equally spaced values for the theoretical rate $\eta_T\sqrt{R_0/N}=\sqrt{3 T^{-1}d\log N}$ within the range of  $\mathcal{I}_1=[0.3756, 0.4981]$ for DGP1 and $\mathcal{I}_2=[0.46,0.61]$ for DGP2. Then we compute $T$ given the theoretical rate, $N$ and $d$. The selected ranges $\mathcal{I}_1$ and $\mathcal{I}_2$ lead to the same range of $T$ for both DGPs under any $N$; i.e., the ranges of the x-axis in Figure \ref{fig:exp1} are set such that the corresponding points in upper and lower panels share the same $T$. Across all settings, $T$ falls in the range of $[55,186]$. Figure \ref{fig:exp1} plots the scaled  estimation errors $\|\bm{\widehat{a}}-\bm{a}^*\|_2 /\sqrt{N}$, $\|\bm{\widehat{g}}-\bm{g}^*\|_2 /\sqrt{N}$, and $\underline{\alpha}_\ma \|\bm{\widehat{\omega}}-\bm{\omega}^*\|_2 /\sqrt{N}$, averaged over 500 replications, against the theoretical rate $\eta_T\sqrt{R_0/N}$. An approximately linear relationship  can be observed across all settings, confirming our theoretical results.

In the second experiment, we verify the consistency of the proposed BIC. Three cases of true model orders are considered: $(p^*,r^*,s^*)=(0,0,1)$, $(0,1,1)$, and $(1,0,1)$, referred to as DGPs 1, 2, and 3, respectively. We set $N=40$, $\theta_1=\pi/4$, and $\lambda_{1}=-\gamma_1=\bar{\rho}$, where three choices of the decay rate are considered: $\bar{\rho}\in\{0.45,0.5,0.5\}$. For $1\leq k\leq d$, each $\bm{G}_k$ contains $3N$ nonzero entries, so $R_0=3dN$, but unlike the first experiment, we do not restrict each row of $\bm{G}_k$ to have exactly three nonzero entries. We set $\tau=0.05$ and $\overline{p}=\overline{r}=\overline{s}=9$; the results are found to be unchanged if the maximum orders are 3. Figure \ref{fig:exp2} displays the proportion of  correct order selection based on 500 replications for each setting, with the models fitted by the JE; the results for the RE are very similar and hence omitted. It shows that the BIC generally performs better as $T$ or $\bar{\rho}$ increases, and the proportion of  correct order selection eventually becomes close to one with sufficiently large $T$. Thus, the consistency of the BIC is verified. Additionally, the required sample size for achieving accurate order selection follows this order among the three DGPs: DGP1 $<$ DGP3 $<$ DGP2. To understand this, first note that $R_0=6N, 9N$, and $9N$ for DGPs 1, 2, and 3, respectively. Thus, the estimation accuracy is highest for DGP1, and so is the order selection accuracy. Moreover, since DGP2 has a more complex temporal structure than DGP3, it leads to greater challenges in estimating $\bm{\omega}$ and, consequently, in order selection.

\begin{figure}[t]
	\centering
	\includegraphics[width=0.9\textwidth]{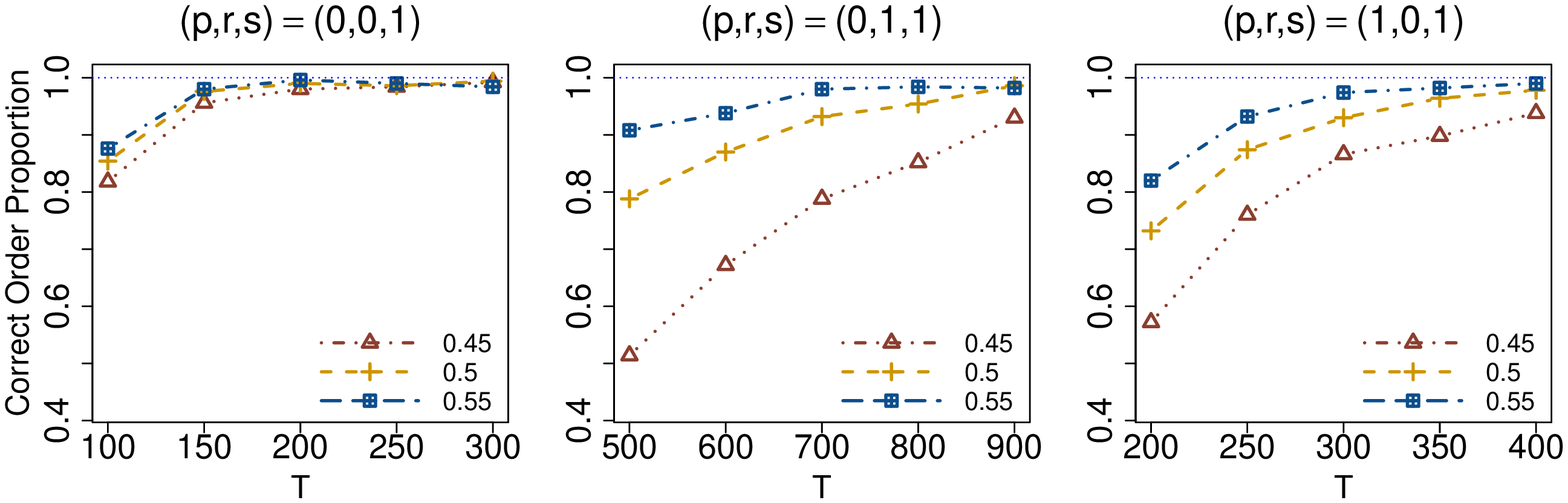}
	\caption{Proportion of correct model order selection for three DGPs and three choices of decay rates, $\bar{\rho}\in\{0.45,0.5,0.55\}$.}
	\label{fig:exp2}
\end{figure}

\section{Empirical analysis}\label{sec:empirical}
We analyze  $N=20$ quarterly macroeconomic variables of the United States from the first quarter of 1969 to the fourth quarter  of 2007. These are key economic and financial indicators collected by \cite{Koop13}, seasonally adjusted as needed. We conduct the transformations following \cite{Koop13} to make all series stationary, resulting in a sample of length $T=194$.  Then  each series is normalized to have zero mean and unit variance; see Table S1 in the supplementary file for detailed descriptions of the twenty variables.

\begin{figure}[t]
	\centering
	\includegraphics[width = 0.9\textwidth]{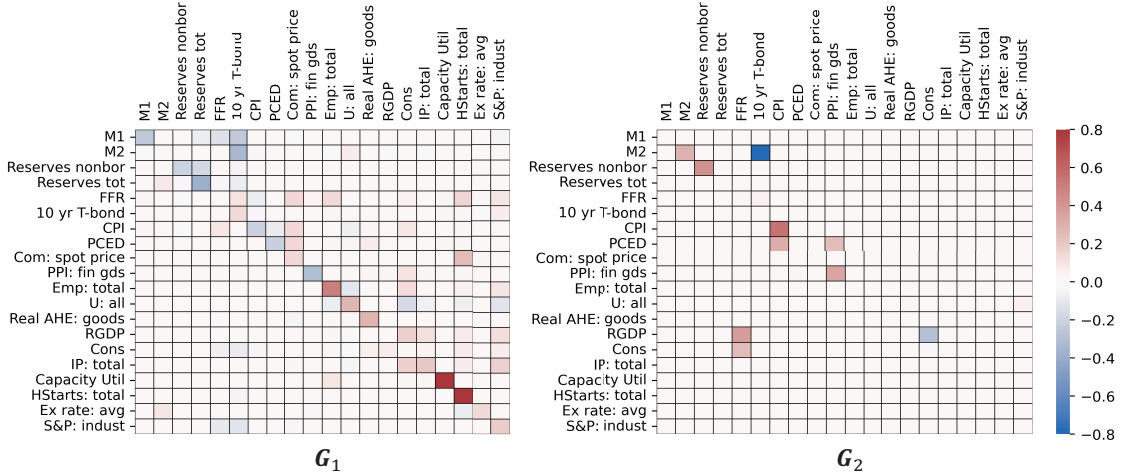}
	\caption{Estimates of $\bm{G}_1$ and $\bm{G}_2$ for the proposed model based on JE.}
	\label{fig:Gmat}
\end{figure}

We first fit the proposed model  to the entire dataset. Using the JE and the proposed BIC, we select $(p,r,s)=(1,1,0)$, so $d=2$, and the fitted model is 
%\[
$\bm{y}_t=\bm{\widehat{G}}_1\bm{y}_{t-1}+  \sum_{h=2}^\infty  (-0.45)^{h-1} \bm{\widehat{G}}_{2} \bm{y}_{t-h}+\bm{\varepsilon}_t$, %\widehat{\lambda}
%\]
where $\bm{\widehat{G}}_1$ and $\bm{\widehat{G}}_2$ are displayed in Figure \ref{fig:Gmat}; the estimation results based on the RE are roughly similar and provided in the supplementary file. The stationarity of the model is confirmed by the method in Remark \ref{remark:stat}. As discussed in Section \ref{subsec:model}, $\bm{\widehat{G}}_1$ and   $\bm{\widehat{G}}_2$ captures lag-one (or short-term) and higher-lag (or long-term) dependence, respectively. Note that $\bm{\widehat{G}}_1$ is much denser than $\bm{\widehat{G}}_2$, suggesting that many dynamic interactions are short-term. However, most of the nonzero entries in  $\bm{\widehat{G}}_2$ are fairly large in absolute value, supporting the necessity of a VARMA-type model.  For the Granger causal (GC) interpretation, take the model equation for real GDP (RGDP) as an example:
\begin{align*}
	y_{\text{RGDP}, t} & = 0.17 y_{\text{Cons}, t-1}+ 0.11 y_{\text{IP:total}, t-1} + 	0.07 y_{\text{HStarts:total}, t-1}  + 0.12 y_{\text{S\&P:indust}, t-1} \\
	&\hspace{5mm}+ \sum_{h=2}^\infty  (-0.45)^{h-1} (0.39 y_{\text{FFR}, t-h} - 0.30 y_{\text{Cons}, t-h}  )+\varepsilon_{\text{RGDP},t}, 
\end{align*}
suppressing other lag-one terms with coefficients less than 0.014 in absolute value  for brevity. The above equation indicates that five time series are GC for RGDP and can be categorized as follows: (1) the industrial production index (IP: total), housing starts (HStarts: total), and S\&P stock price index (S\&P: indust) only have short-term influence on RGDP; (2) the federal funds rate (FFR) only has long-term influence on RGDP; (3) the real personal consumption expenditures (Cons) has both short-term and long-term influence on RGDP. For other insights from the estimation results, see Section S3 in the supplementary file for more discussions.

Next we evaluate the forecasting performance via a rolling procedure: First set the forecast origin to $t=166$  (Q4-2000). For each $k=1,\dots, 28$, fit the  model using the data  of $1\leq t\leq  T_{\textrm{train}}=165+k$, and then compute the one-step ahead forecast for $t=166+k$. Thus, rolling forecasts over the period of Q1-2001 to Q4-2007 are obtained. We measure the  forecast error  by $\|\bm{\widehat{y}}_t-\bm{y}_t\|_2$; our findings based on the $\ell_1$-norm are similar and hence are omitted. For the proposed model, we consider both JE and RE, and implement them using a fixed  regularization parameter $\lambda_{g}$  throughout the forecasting period. 
Five other competing approaches are considered as follows:
\begin{itemize}[itemsep=2pt,parsep=2pt,topsep=4pt,partopsep=2pt]
	\item [(i)] VAR OLS: As a low-dimensional baseline, we consider the VAR($4$) model fitted via the OLS method, where the lag order $4$ is employed following \cite{Koop13}.
	\item [(ii)] VAR Lasso: Since the VAR($\infty$) model can be approximated by the VAR($P$) with $P\rightarrow\infty$ as $T\rightarrow\infty$, we fit the sparse VAR($P$) model via the Lasso with $P=\lfloor 1.5\sqrt{T_{\textrm{train}}}\rfloor$  following the first-stage estimation in \cite{WBBM21}.
	
	\item[(iii)] VAR HLag: Same as (ii) except that the hierarchical lag (HLag) regularization in \cite{NWBM20} is used instead of the $\ell_1$-regularization.
	
	\item [(iv)] VARMA $\ell_1$: Sparse VARMA($p,q$) \citep{WBBM21} with the $\ell_1$-regularization for the second stage and $p=q=\lfloor 0.75\sqrt{T_{\textrm{train}}}\rfloor$ as in the above paper.
	
	\item [(v)] VARMA HLag: Same as (iv) except that the HLag regularization is used at the second stage.
\end{itemize}

We implement (ii)--(v) by the R package \texttt{bigtime} which offers two regularization parameter selection methods, cross validation (CV) and BIC. We observe that neither one of these two  methods uniformly outperforms the other throughout the forecasting period. To better ensure the competitiveness of (ii)--(v), we obtain the forecast errors under both CV and BIC and only report the smaller value for each rolling step. 

The average forecast error over the entire forecast period is 5.367, 4.307,	4.069, 4.318, 4.144, 3.971,	and 3.968 for VAR OLS, VAR Lasso, VAR HLag, VARMA $\ell_1$, VARMA HLag, SPVAR($\infty$) JE, and SPVAR($\infty$) RE, respectively. Among the 28 rolling steps, each of these approaches performs best 4, 4, 0, 2, 2, 10, and 6 times, respectively. Thus, based on these measures, SPVAR($\infty$) has the highest  overall forecast accuracy among all models, and the performance of JE and RE are very similar; see Table S2 in the supplementary file for the forecast errors of all seven methods for each rolling step. Moreover, to check whether the advantage of the SPVAR($\infty$)-based forecasts is statistically significant, we conduct the model confidence set (MCS) procedure of \cite{HLN2011} implemented by the  R package \texttt{MCS}. We find that based on either the  Tmax or TR statistic, the 97.5\% MCS only includes  SPVAR($\infty$) JE and SPVAR($\infty$) RE, confirming that the   proposed model indeed  outperforms the competing ones in terms of forecasting   for the data.

\section{Conclusion and discussion}\label{sec:conclusion}

This paper develops the SPVAR($\infty$) model as a tractable variant of the VARMA model for high-dimensional time series. It overcomes the drawbacks in identification, computation, and interpretation of the latter, while greater statistical efficiency and Granger causal interpretations are achieved by imposing sparsity  on the parameter matrices capturing the cross-sectional dependence. To the best of our knowledge, it is the first high-dimensional sparse VARMA- or VAR($\infty$)-type model with all of the above advantages.

There is a vast literature on nonlinear and nonstationary VAR models \citep[e.g.,][]{KMS16, ZW21}, factor-augmented VAR \citep{MPS22}, and other extensions. The method in this paper can be extended to develop corresponding VAR($\infty$) counterparts; e.g., \eqref{eq:model-scalar} can be extended to the nonlinear model:
$\bm{y}_t=f(\bm{x}_t^{[1]}, \dots, \bm{x}_t^{[d]})+\bm{\varepsilon}_t$,
where 
$\bm{x}_t^{[k]}=\sum_{h=1}^\infty\ell_{h,k}(\bm{\omega}) \bm{y}_{t-h}$ for $1\leq k\leq d$ parsimoniously summarize the temporal information over all lags into $d$ predictors. Other interesting extensions include imposing group sparsity on $\bm{G}_k$'s to capture group-wise homogeneity \citep{BSM15}, extending $\ell_{h,k}(\bm{\omega})$'s to polynomial decay functions for long-memory time series \citep{Chung02}, and incorporating dynamic factor structures \citep{Wang2021High}.  Lastly,  it is important to study the high-dimensional statistical inference under the proposed model, e.g., hypothesis testing for Granger causality \citep{CHHW21,BGS22}.

%\section*{Acknowledgements}
%I am grateful to  the Editor, Associate Editor, and two anonymous referees for their valuable comments which led to  substantial improvement of this paper. This research was partially supported by NSF grant DMS-2311178.

%%%%%%%%%%%%%%%%%%%%%%%%%%%%%%%%%%%%%%%%%%%%%%%%%%%%%%%%%%%%%%%%%%%%%%%%%%%%%%%%%%%%%%%%%%%%%%%%%%%%%%%%%%%%%%%%%%%%%%%%%%%%%

%\bibliography{SparseSARMA}

\putbib[SparseSARMA]
\end{bibunit}

\newpage
\renewcommand{\arraystretch}{0.85}

\begin{bibunit}[apalike]
\vspace*{10pt}	
\begin{center}
{\Large \bf Supplementary Material: An Interpretable and Efficient Infinite-Order Vector Autoregressive Model for High-Dimensional Time Series}
\end{center}
\vspace{10pt}

\begin{abstract}
This supplementary file is organized into eight sections. Section \ref{section:algo} presents the algorithms for the proposed estimators. Section \ref{asec:sim} provides four additional simulation experiments, while Section \ref{asec:emp} offers more details for the empirical example discussed in the main paper. Sections  \ref{sec:proofmain}--\ref{asec:bic} contain the proofs of (1) Proposition \ref{prop:VARMA} and Theorem \ref{thm:stationary}, (2) Proposition \ref{prop:perturb} and Theorem \ref{thm:lasso}, (3) Proposition \ref{prop:perturbrow} and Theorem \ref{thm:lassorow}, and  Theorem \ref{thm:selection}, respectively. Finally, Section \ref{asec:aux} provides the proofs of all auxiliary lemmas.
\end{abstract}	

\renewcommand{\thesection}{S\arabic{section}}
\renewcommand{\thesubsection}{S\arabic{section}.\arabic{subsection}}
\renewcommand{\theequation}{S\arabic{equation}}
\renewcommand{\thetable}{S\arabic{table}}
\renewcommand{\thefigure}{S\arabic{figure}}
\renewcommand{\thelemma}{S\arabic{lemma}}
\renewcommand{\theproposition}{S\arabic{proposition}}
\setcounter{lemma}{0}
\setcounter{section}{0}

\section{Algorithm and implementation}\label{section:algo}
\subsection{Block coordinate descent algorithms}\label{subsec:algo}
We present the block coordinate descent algorithms for implementing the proposed estimators in this section.

First consider the JE in Section \ref{subsec:JE}. Observe that if  $\bm{\omega}$ is given, then the optimization problem in \eqref{eq:lasso} will simply become the $\ell_1$-regularized least squares optimization for multivariate linear regression, which can be efficiently solved by the proximal gradient descent  (i.e., iterative soft-thresholding)  algorithm \citep{ANW2012}. On the other hand, if $\bm{g}$ is given, we can rewrite $\widetilde{\mathbb{L}}_T(\bm{\omega},\bm{g})$ in the form of
\begin{equation}\label{eq:sepa}
	\widetilde{\mathbb{L}}_T(\bm{\omega}) = \frac{1}{T}
	\sum_{t=1}^{T} \Big\| \bm{y}_t^- - \sum_{j=1}^r  F_t^I(\lambda_{j}) - \sum_{m=1}^s F_t^{II}(\bm{\eta}_{m}) \Big\|_2^2,
\end{equation}
where  $F_t^I(\lambda_{j}) =\bm{G}_{p+j} f^I(\bm{\widetilde{x}}_t; \lambda_{j})$, $F_t^{II}(\bm{\eta}_{m})=  \sum_{\iota=1}^{2}\bm{G}_{p+r+2(m-1)+\iota} f^{II,\iota}(\bm{\widetilde{x}}_t; \bm{\eta}_{m})$, and
$\bm{y}_t^{-}=\bm{y}_t - \sum_{k=1}^{p}\bm{G}_k\bm{y}_{t-k}$, with $\bm{\widetilde{x}}_{t}= (\bm{y}_{t-1}^\top,\dots,\bm{y}_1^\top,0,0,\dots)^\top$ being the initialized version of the infinite-dimensional vector $\bm{x}_t=(\bm{y}_{t-1}^\top,\bm{y}_{t-2}^\top,\dots)^\top$. Here, $f^I(\bm{\widetilde{x}}_t; \lambda_{j})=\sum_{h=p+1}^{t-1}\lambda_j^{h-p} \bm{y}_{t-h}$, $f^{II,1}(\bm{\widetilde{x}}_t; \bm{\eta}_{m})=\sum_{h=p+1}^{t-1} \gamma_m^{h-p} \cos\{(h-p) \theta_m\} \bm{y}_{t-h}$, and $f^{II,2}(\bm{\widetilde{x}}_t; \bm{\eta}_{m})=\sum_{h=p+1}^{t-1} \gamma_m^{h-p} \sin\{(h-p) \theta_m\} \bm{y}_{t-h}$. Since each $\lambda_j$ or $\bm{\eta}_m$  appears in only one of the summands in \eqref{eq:sepa}, this  structure  allows for  acceleration via parallel implementation across $r+s$ machines. In addition, since each  $\lambda_j$ or $\bm{\eta}_m$ is only one- or two-dimensional, the computation cost of updating each $\lambda_j$ and $\bm{\eta}_m$ will be very low. 

\begin{algorithm}[t]
	\caption{Block coordinate descent algorithm for the JE}
	\label{alg:JE}
	\textbf{Input:}  model orders $(p,r,s)$, regularization parameter $\lambda_g$, initialization $\bm{\omega}^{(0)}$, $\bm{{g}}^{(0)}$, step length $\alpha$, constraint sets $\pazocal{C}_{\lambda}$, $\pazocal{C}_{\bm{\eta}}$.\\
	\textbf{repeat} $\iota=0,1,2,\dots$\\\vspace{2mm}
	\hspace*{5mm}\textbf{for} $j=1,\dots,r$:\\\vspace{2mm}
	\hspace*{10mm}$\displaystyle \lambda_j^{(\iota+1)} \leftarrow P_{\pazocal{C}_{\lambda}}\Big(\lambda_{j}^{(\iota)}-\alpha\times \nabla_{\lambda_{j}}\widetilde{\mathbb{L}}_{T}( \bm{\omega}^{(\iota)},\bm{g}^{(\iota)})\Big)$\\\vspace{2mm}
	\hspace*{5mm}\textbf{for} $m=1,\dots,s$:\\\vspace{2mm}
	\hspace*{10mm}	$\displaystyle\bm{\eta}_m^{(\iota+1)} \leftarrow P_{\pazocal{C}_{\bm{\eta}}}\Big(\bm{\eta}_m^{(\iota)}-\alpha\times \nabla_{\bm{\eta}_{m}}\widetilde{\mathbb{L}}_{T}(\bm{\omega}^{(\iota)},\bm{g}^{(\iota)})\Big)$\\\vspace{2mm}
	
	\hspace*{5mm}$\displaystyle \bm{g}^{(\iota+1)} \leftarrow S_{\alpha\lambda_g}\Big(\bm{g}^{(\iota)}-\alpha\times \nabla_{\bm{g}}\widetilde{\mathbb{L}}_{T}(\bm{\omega}^{(\iota+1)},\bm{g}^{(\iota)})\Big)$\\
	
	\textbf{until convergence}
\end{algorithm}

The above discussion motivates us to propose the block coordinate descent algorithm for the JE as displayed in Algorithm \ref{alg:JE}.  At each iteration, the following two steps are conducted: (S1) fixing $\bm{g}$, update  $\lambda_j$'s and $\bm{\eta}_m$'s by projected gradient descent; (S2) fixing $\bm{\omega}$, get the proximal gradient update of $\bm{g}$ via soft-thresholding. Both (S1) and (S2) can be implemented either successively or in parallel.
That is, in Algorithm \ref{alg:JE}, lines 3--6 can be realized on $r+s$ nodes, and the update of $\bm{g}$ in line 7 can be realized coordinate-wisely on $N^2d$ nodes.
In addition, since the projected  gradient descent requires the constraint set to be closed, we search $\lambda_j$ within  $\pazocal{C}_\lambda=[-1+\epsilon, 1-\epsilon]$ and $\bm{\eta}_m$ within $\pazocal{C}_{\bm{\eta}}= [0,1-\epsilon]\times[\epsilon,\pi-\epsilon]$, for a small $\epsilon>0$, e.g., $\epsilon=0.05$.  In Algorithm \ref{alg:JE},  $P_{\pazocal{C}}(\bm{x})=\argmin_{\bm{z}\in\pazocal{C}}\|\bm{x}-\bm{z}\|_2^2$ is the projection operator for any set $\pazocal{C}$, and   $S_{\tau}(\bm{z})$ is the soft-thresholding operator with coordinates
$[S_{\tau}(\bm{z})]_j=\textrm{sign}(z_j)\max\{|z_j|-\tau, 0\}$ for any threshold $\tau>0$.

\begin{algorithm}[t]
	\caption{Block coordinate descent algorithm for the RE}
	\label{alg:RE}
	\textbf{Input:}  model orders $(p,r,s)$, regularization parameter $\lambda_g$, initialization $\bm{\omega}_i^{(0)}=\bm{\omega}^{(0)}$ for $1\leq i\leq N$, $\bm{{g}}^{(0)}$, step length $\alpha$, constraint sets $\pazocal{C}_{\lambda}$, $\pazocal{C}_{\bm{\eta}}$.\\
	\textbf{for} $i=1,\dots,N$:\\\vspace{2mm}
	\hspace*{5mm}\textbf{repeat} $\iota=0,1,2,\dots$\\\vspace{2mm}
	\hspace*{10mm}\textbf{for} $j=1,\dots,r$:\\\vspace{2mm}
	\hspace*{15mm}$\displaystyle \lambda_{i,j}^{(\iota+1)} \leftarrow P_{\pazocal{C}_{\lambda}}\Big(\lambda_{i,j}^{(\iota)}-\alpha\times \nabla_{\lambda_{i,j}}\widetilde{\mathbb{L}}_{i,T}( \bm{\omega}_i^{(\iota)},\bm{g}_i^{(\iota)})\Big)$\\\vspace{2mm}
	\hspace*{10mm}\textbf{for} $m=1,\dots,s$:\\\vspace{2mm}
	\hspace*{15mm}	$\displaystyle\bm{\eta}_{i,m}^{(\iota+1)} \leftarrow P_{\pazocal{C}_{\bm{\eta}}}\Big(\bm{\eta}_{i,m}^{(\iota)}-\alpha\times \nabla_{\bm{\eta}_{i,m}}\widetilde{\mathbb{L}}_{i,T}(\bm{\omega}_i^{(\iota)},\bm{g}_i^{(\iota)})\Big)$\\\vspace{2mm}
	
	\hspace*{10mm}$\displaystyle \bm{g}_i^{(\iota+1)} \leftarrow S_{\alpha\lambda_g}\Big(\bm{g}_i^{(\iota)}-\alpha\times \nabla_{\bm{g}_i}\widetilde{\mathbb{L}}_{i,T}(\bm{\omega}_i^{(\iota+1)},\bm{g}_i^{(\iota)})\Big)$\\
	
	\hspace*{5mm}\textbf{until convergence}
\end{algorithm}

For the RE in Section \ref{subsec:RE}, a similar  block coordinate descent algorithm can be applied to each rowwise minimization \eqref{eq:lassorow}; see Algorithm \ref{alg:RE} for details. Here we denote $\lambda_{i,j}^{(\iota)}$ for $1\leq j\leq r$ and $\bm{\eta}_{i,m}^{(\iota)}$ for $1\leq m\leq s$ as the parameters in $\bm{\omega}_i^{(\iota)}$, where $1\leq i\leq N$, and $\iota$ is the iteration number.
Note that the $N$ rowwise minimizations can alternatively be implemented in parallel, allowing further acceleration. From our simulation studies in Sections \ref{subsec:compare} and  \ref{subsec:time}, we observe that the minimization for each individual row in Algorithm \ref{alg:RE}  tends to converge more quickly than the joint minimization in Algorithm \ref{alg:JE}. Nonetheless, the total computation time of  Algorithm \ref{alg:RE} across all $N$ rows tends to be higher than that of Algorithm \ref{alg:JE} if the $N$ rowwise minimizations are implemented successively rather than in parallel. In addition, especially when $N$ is relatively large, Algorithm \ref{alg:RE} is usually more stable than Algorithm \ref{alg:JE}, which is likely due to the weaker sparsity requirement for RE; see Section \ref{subsec:RE}.

\subsection{Algorithm initialization}\label{subsec:init}

We discuss the model parameter initialization  for  Algorithms \ref{alg:JE} and \ref{alg:RE} as follows. 
First, as shown in  Section \ref{sec:BIC}, the orders  $(p,r,s)$ can be selected by the proposed BIC. Meanwhile, for any fixed $(p,r,s)$, the corresponding optimal regularization parameter $\lambda_{g}$ can be selected using the high-dimensional BIC in \cite{WZ11}. Combining the two methods, we can select the model orders together with $\lambda_{g}$. 

Recall that the nonasymptotic error bounds in Theorems \ref{thm:lasso} and  \ref{thm:lassorow} are established for a local region of  $\bm{\omega}^*$. Algorithmically, this means we need a reasonably good initial value $\bm{\omega}^{(0)}$, although it need not be a consistent estimator of $\bm{\omega}^*$. For our model, it turns out that the boundedness of the parameter space of $\bm{\omega}$ makes finding a good initialization easier than  general nonconvex estimation problems. This is because $\lambda_1,\dots, \lambda_{r}$ must be well separated and lie within $(-1,0)\cup(0,1)$. Similarly, $(\gamma_1,\theta_1),\dots, (\gamma_s,\theta_s)$ must be well separated and lie within $(0,1)\times (0,\pi)$. Thus, given $r$ and $s$, setting initial values for these parameters is essentially the same as defining a grid of values on bounded intervals. Moreover, when $r$ and $s$ are larger, the grid will be denser and consequently even more likely to be closer to the true parameter values.
In practice,  we recommend the following procedure:
\begin{enumerate}
	\item Set a grid of initial values  for each element of $\bm{\omega}$ within their respective bounded intervals. For example, if $r,s\leq4$, then we may consider $\lambda_j \in \{\pm0.3, \pm0.6\}$, $\gamma_m\in\{0.3,0.6\}$, and $\theta_m\in\{\pi/4,3\pi/4\}$, for $1\leq j\leq r$ and $1\leq m\leq s$. Or, if $r=1$ or $s=1$, then we may consider denser grids such as $\lambda_1\in\{\pm 0.2, \pm 0.4, \pm 0.6, \pm 0.8\}$, $\gamma_1=\{0.2, 0.4, 0.6, 0.8\}$, and $\theta_1=\{\pi/4, \pi/2, 3\pi/4\}$.
	
	Then,  by considering  all combinations of distinct initial values chosen from the grids, we form the set of candidate initial values for $\bm{\omega}$.
	
	\item Run the algorithm with each candidate initial value   $\bm{\omega}^{(0)}$, and  select the solution with the minimum squared loss.
\end{enumerate}
Our simulations suggest that the above selection procedure  performs almost as well as initializing $\bm{\omega}$ with the true value.

To improve the stability of the algorithm, we recommend setting  $\bm{g}^{(0)}$ based on a preliminary estimator $\bm{a}^{(0)}$ of $\bm{a}$, given  any candidate initial value $\bm{\omega}^{(0)}$. Specifically, we first fit a sparse VAR($P$) model via the Lasso with $P=\lfloor 1.5\sqrt{T}\rfloor$ to obtain $\bm{A}_1^{(0)},\dots, \bm{A}_P^{(0)}$, and set $\bm{A}_h^{(0)}=\bm{0}$ for $h>P$. Note that it is infeasible to exactly solve for $\bm{g}$ given $\bm{a}$ and $\bm{\omega}$.
As a remedy, we define the pseudoinverse of $\bm{L}(\bm{\omega}^{(0)})$ as $\bm{L}^+(\bm{\omega}^{(0)})=[\{ \bm{L}^\top(\bm{\omega}^{(0)}) \bm{L}(\bm{\omega}^{(0)})\}^{-1} \bm{L}^\top(\bm{\omega}^{(0)})] \in\mathbb{R}^{d\times \infty}$. Then, we can obtain $\bm{g}^{(0)}=(\bm{L}^{+}(\bm{\omega}^{(0)}) \otimes\bm{I}_{N^2} ) \bm{a}^{(0)}$. 

%Note that by matrix algebra, we can show that $\bm{g}=(\bm{L}_{\text{trunc}}^{-1}(\bm{\omega}) \otimes\bm{I}_{N^2} ) \bm{a}_{\text{trunc}}$, where $\bm{L}_{\text{trunc}}(\bm{\omega})\in\mathbb{R}^{d\times d}$ denotes the first $d$ rows of $\bm{L}(\bm{\omega})$, and $\bm{a}_{\text{trunc}}=\vect((\bm{A}_1,\dots, \bm{A}_d))$. Thus, we can set   $\bm{g}^{(0)}=(\bm{L}_{\text{trunc}}^{-1}(\bm{\omega}^{(0)}) \otimes\bm{I}_{N^2} ) \bm{a}_{\text{trunc}}^{(0)}$.

%%%%%%%%%%%%%%%%%%%%%%%%%%%%%%%%%%%%%%%%%%%%%%%%%%%%%%%%%%%%%%%%%%%%%%%%%%%%%%%%%%%%%%%%%%%%%%%%%%

\section{Additional simulation experiments}\label{asec:sim}

We  provide four additional simulation experiments to (1) verify the estimation error rates of the RE, (2) compare the estimation errors of JE and RE, (3) investigate the sensitivity of the estimation to the initialization $\bm{y}_t=\bm{0}$ for $t\leq 0$, and (4) compare the computational and forecasting performance of the proposed estimators  to competing ones in high dimensions.

\subsection{Finite-sample performance of the RE}

In the first experiment, we examine the estimation error rates for the RE.  The data are generated under the same settings as those in the first experiment in Section \ref{sec:sim} of the main paper. That is,  two data generating processes with  $N=10, 20, 40$ or 80 are considered: $(p,r,s)=(1,1,0)$ (DGP1) and $(1,0,1)$ (DGP2), where $\lambda_1 = -0.6$ for DGP1, and $(\gamma_1,\theta_1)=(0.6, \pi/4)$ for DGP2. 
In addition, each $\bm{G}_k$ is a row-sparse matrix with three nonzero entries in each row, i.e., $R_{i,0}=3d$ for $1\leq i\leq N$ and $R_{\max,0}=\max_{1\leq i\leq N} R_{i,0}=3d$. 

\begin{figure}[t]
	\centering
	\includegraphics[width=0.9\textwidth]{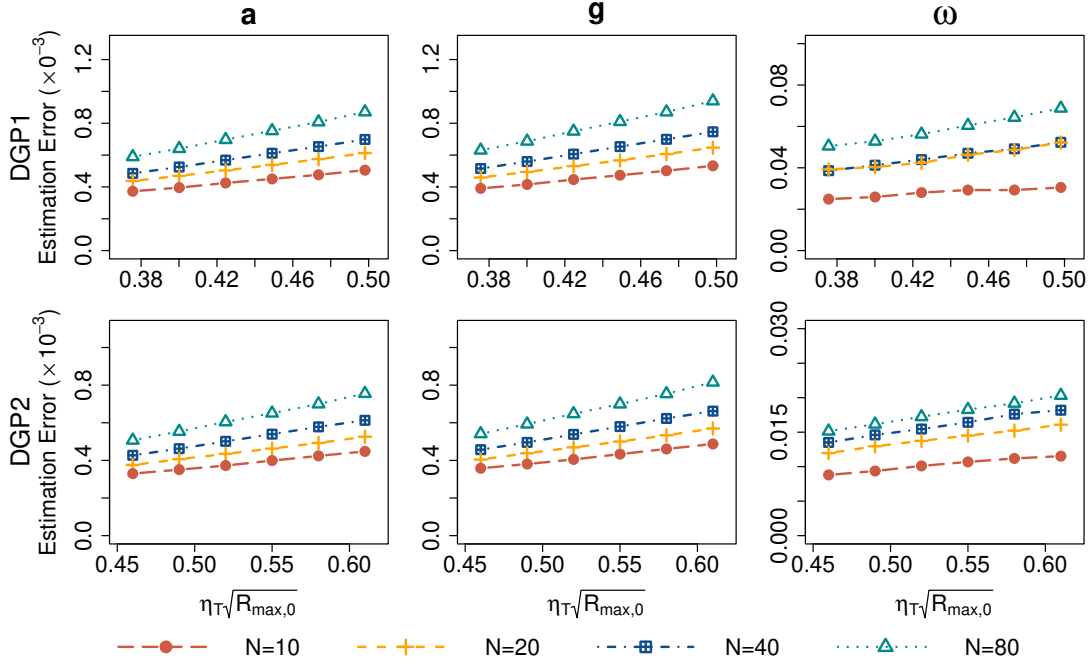}
	\caption{Plots of maximum estimation errors  $\max_{1\leq i\leq N}\|\bm{\widehat{a}}_{i}-\bm{a}_i^*\|_{2}$ (left panel), $\max_{1\leq i\leq N}\|\bm{\widehat{g}}_{i}-\bm{g}_i^*\|_{2}$ (middle panel), and $\max_{1\leq i\leq N} \underline{\alpha}_{i,\ma} \|\bm{\widehat{\omega}}_i-\bm{\omega}^*\|_2$ (right panel) against the theoretical rate $\eta_T\sqrt{R_{\max,0}}$ for the RE. }
	\label{fig:expS1}
\end{figure}

We aim to verify the following error bounds as implied by Theorem \ref{thm:lassorow}: $\max_{1\leq i\leq N}\|\bm{\widehat{a}}_{i}-\bm{a}_i^*\|_{2}\lesssim \eta_{T} \sqrt{R_{\max,0}}$, $\max_{1\leq i\leq N}\|\bm{\widehat{g}}_{i}-\bm{g}_i^*\|_{2}\lesssim \eta_{T} \sqrt{R_{\max,0}}$, and $\max_{1\leq i\leq N} \underline{\alpha}_{i,\ma} \|\bm{\widehat{\omega}}_i-\bm{\omega}^*\|_2\lesssim \eta_{T} \sqrt{R_{\max,0}}$, where $\eta_T=\sqrt{T^{-1}\log N}$. 
We consider a grid of equally spaced values for the theoretical rate $\eta_{T} \sqrt{R_{\max,0}}=\sqrt{3 T^{-1}d\log N}$ within the range of  $\mathcal{I}_1=[0.3756, 0.4981]$ for DGP1 and $\mathcal{I}_2=[0.46,0.61]$ for DGP2, and then obtain $T$ based on the theoretical rate, $N$ and $d$. This leads to the same set of values for  $T\in[55,186]$ as in the first experiment in Section \ref{sec:sim}. Figure \ref{fig:exp1} displays the maximum estimation errors $\max_{1\leq i\leq N}\|\bm{\widehat{a}}_{i}-\bm{a}_i^*\|_{2}$, $\max_{1\leq i\leq N}\|\bm{\widehat{g}}_{i}-\bm{g}_i^*\|_{2}$, and $\max_{1\leq i\leq N} \underline{\alpha}_{i,\ma} \|\bm{\widehat{\omega}}_i-\bm{\omega}^*\|_2$, averaged over 500 replications, against the theoretical rate $\eta_{T} \sqrt{R_{\max,0}}$. 
We observe a linear relationship between the empirical and theoretical rates across all settings. confirming the error rates suggested by  Theorem \ref{thm:lassorow}.

%%%%%%%%%%%%%%%%%%%%%%%%%%%%%%%%%%%%%%%%%%%%%%%%%%%%%%%%%%%%%%%%%%%%%%%%%%%%%%%%%%%%%%%%%%%%%%%%%%
\subsection{Comparison between JE and RE}\label{subsec:compare}

In this experiment, we compare the estimation accuracy of  JE and RE. The data are generated from the proposed model with $(p,r,s)=(1,1,0)$, $\lambda_1=0.6$, $N=20$ or 60, and $T=50,100,150,300$ or $500$, using the same method as in Section \ref{sec:sim}.  Each $\bm{G}_k$ is a row-sparse matrix with two or four nonzero entries in each row, i.e., $R_{i,0}=2d$ or $4d$ for $1\leq i\leq N$.

\begin{figure}[t]
	\centering
	\includegraphics[width=0.9\textwidth]{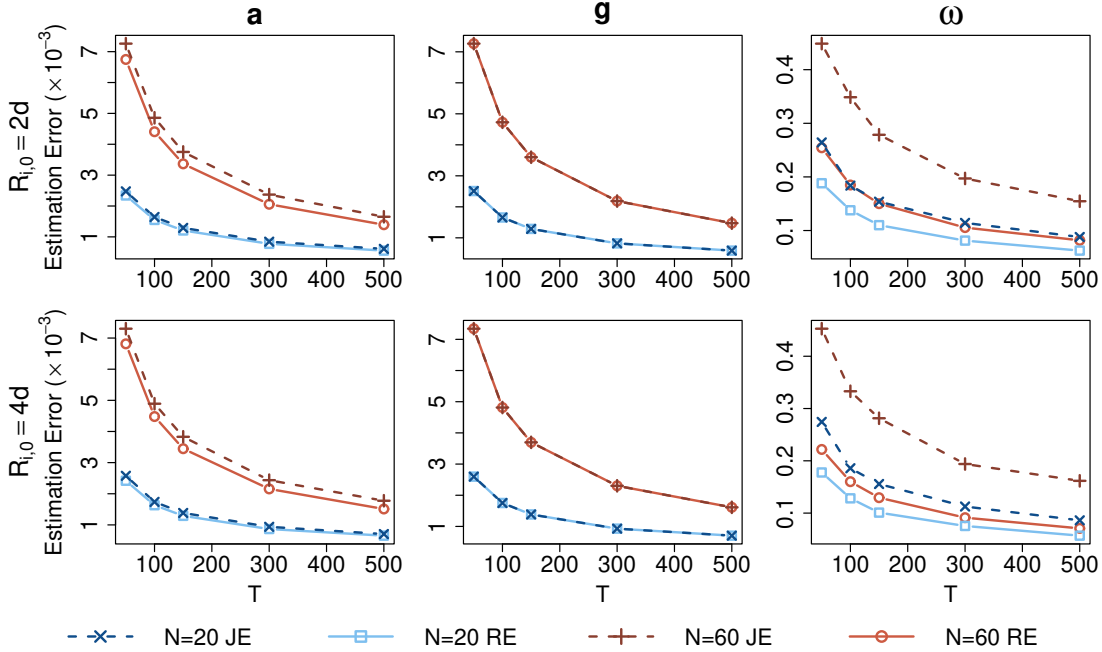}
	\caption{Plots of estimation errors for $\bm{a}$ (left panel), $\bm{g}$ (middle panel) and $\bm{\omega}$ (right panel)  against $T$ for JE and RE  when $R_{i,0}=2d$ (upper panel) or $4d$ (lower panel). }
	\label{fig:expS2}
\end{figure}

By Section \ref{sec:HDmethod} of the main paper, JE and RE result in the error bounds for the overall estimation errors $\|\bm{\widehat{a}}-\bm{a}^*\|_2\lesssim  \eta_{T} \sqrt{R_0}$ and  $\|\bm{\widehat{g}}-\bm{g}^*\|_2\lesssim  \eta_{T} \sqrt{R_0}$, where $R_0=\sum_{i=1}^N R_{i,0}$. However,  from the error bounds $\|\bm{\widehat{\omega}}_i-\bm{\omega}^*\|_2 \lesssim  \underline{\alpha}_{i,\ma}^{-1}  \eta_{T} \sqrt{R_{i,0}}$ for the RE and   $\|\bm{\widehat{\omega}}-\bm{\omega}^*\|_2 \lesssim \underline{\alpha}_{\ma}^{-1}  \eta_{T} \sqrt{R_0}$ for the JE, it is unclear which one will actually perform better in practice. We aim to provide numerical evidence for these questions. 
Figure \ref{fig:expS2} displays the estimation errors, averaged over 500 replications, against $T$. Here the estimation errors for $\bm{a}$ and $\bm{g}$ are computed as $\|\bm{\widehat{a}}-\bm{a}^*\|_2$ and $\|\bm{\widehat{g}}-\bm{g}^*\|_2$, respectively, for both JE and RE. The estimation error for $\bm{\omega}$ is computed as $\|\bm{\widehat{\omega}}-\bm{\omega}^*\|_2$ for the JE and  $\max_{1\leq i\leq N} \|\bm{\widehat{\omega}}_i-\bm{\omega}^*\|_2$ for the RE.  From Figure \ref{fig:expS2}, it can be seen that the estimation errors for $\bm{g}$ based on JE and RE are nearly identical across all settings. However, the RE generally results in smaller estimation errors for $\bm{\omega}$ than the JE. In addition, the estimation errors for $\bm{a}$ based on JE and RE are similar, with RE being slightly superior. This is also expected, because although  JE and RE have the same theoretical error rates for $\|\bm{\widehat{a}}-\bm{a}^*\|_2$, they can differ by a constant factor. Since the RE estimates $\bm{\omega}$ more accurately than the JE, it will naturally lead to smaller estimation errors for $\bm{a}$, as the two estimators yield the nearly identical estimates for $\bm{g}$. 
%In summary, the findings from this simulation experiment suggests that RE may outperform JE for the estimation of $\bm{\omega}$, while they may have very similar  performance for the estimation of $\bm{g}$. 
Overall, RE tends to slightly outperform the JE for the estimation of $\bm{a}$, especially when $N$ is large, which is equivalent to say that $R_q$ is large in this experiment.

%are consistent with our theoretical results. 

%%%%%%%%%%%%%%%%%%%%%%%%%%%%%%%%%%%%%%%%%%%%%%%%%%%%%%%%%%%%%%%%%%%%%%%%%%%%%%%%%%%%%%%%%%%%%%%%%%

\subsection{Sensitivity analysis for initialization of $\{\bm{y}_t, t\leq 0\}$}

\begin{figure}[t]
	\centering
	\includegraphics[width=0.9\textwidth]{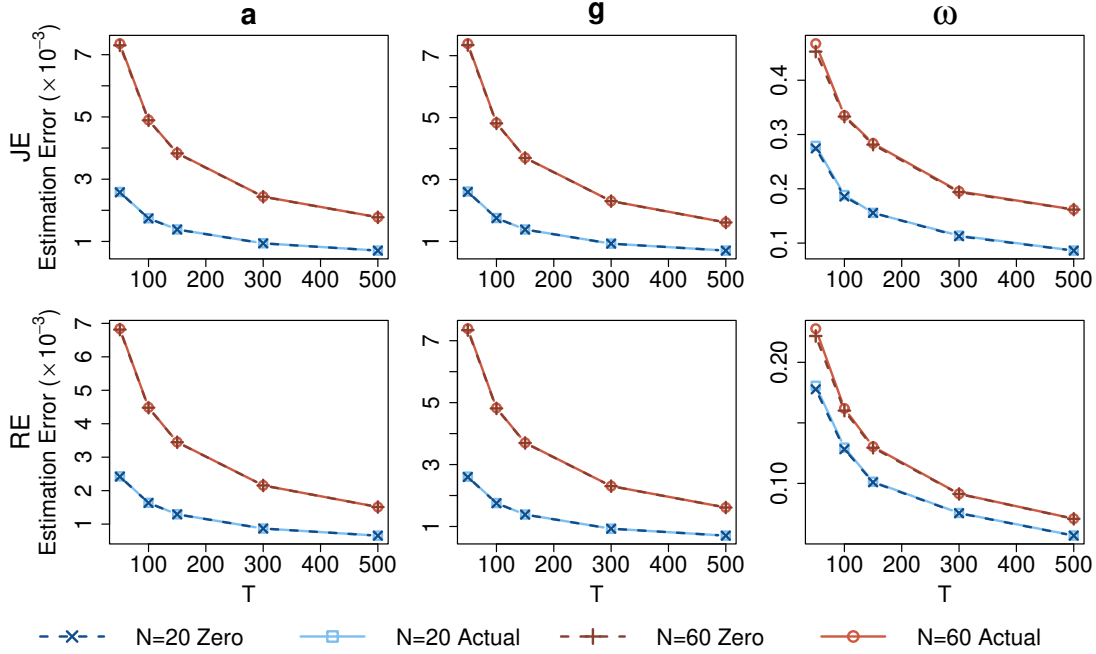}
	\caption{Plots of estimation errors for $\bm{a}$ (left panel), $\bm{g}$ (middle panel) and $\bm{\omega}$ (right panel)  against $T$  based on two initialization methods for JE (upper panel) and RE (lower panel). Zero: initializing $\bm{y}_t=\bm{0}$ for $t\leq 0$; Actual: initializing $\bm{y}_t$ for $t\leq 0$ by their actual values.}
	\label{fig:expS3}
\end{figure}

The aim of the third experiment is to assess the impact of initializing $\bm{y}_t=\bm{0}$ for $t\leq 0$ on the estimation  in finite samples. The data are generated as in Section \ref{subsec:compare}. For both JE and RE, we consider two initialization methods: (a) setting $\bm{y}_t=\bm{0}$ for $t\leq 0$, which is employed in this paper; and (b) setting them to their actual values obtained by generated a longer series. Note that Method (b) serves as a benchmark but is infeasible in practice. The estimation errors are computed as in  Section \ref{subsec:compare}, averaged over 500 replications. Figure \ref{fig:expS3} displays the results under the row sparsity level $R_{i,0}=4d$; the results for the sparser case $R_{i,0}=2d$  are similar and hence omitted. It can be observed that the estimation errors based on the two initialization methods are nearly identical across all settings for both JE and SE. In fact, there are only small visible differences when $T=50$ for the estimation of $\bm{\omega}$. This confirms that the initialization effect is negligible numerically.

%%%%%%%%%%%%%%%%%%%%%%%%%%%%%%%%%%%%%%%%%%%%%%%%%%%%%%%%%%%%%%%%%%%%%%%%%%%%%%%%%%%%%%%%%%%%%%%%%%

\subsection{Computation time and forecast accuracy}\label{subsec:time}
In the last experiment, we assess the computational efficiency and forecast accuracy of the proposed SPVAR($\infty$) model. To highlight its capability to capture VARMA dynamics, instead of generating data from the proposed model, we consider the VARMA($1,1$) process, 
\[
\bm{y}_t = \bm{\Phi}\bm{y}_{t-1}+\bm{\varepsilon}_t - \bm{\Theta}\bm{\varepsilon}_{t-1},
\]
where $\bm{\Phi}=0.5 \bm{I}_N$, $\{\bm{\varepsilon}_t\}$ are $i.i.d.$ following $N(\bm{0}, \sigma^2\bm{I}_N)$ with $\sigma=0.2$, $N\in[10,60]$, and $T=125$. As shown in the proof of Proposition \ref{prop:VARMA}, this process can be written as model \eqref{eq:model-scalar} with order  $p=1$ if we generate $\bm{\Theta}$ according to the Jordan decomposition $\bm{\Theta}=\bm{B}\bm{J}\bm{B}^{-1}$, where $\bm{J}$ is defined as in \eqref{eq:realJ} and $\bm{B}$ is an invertible matrix. Hence, we specify $\bm{J}$ from  $\bm{\omega}$ by setting  $(r,s)=(1,0)$ and $\lambda_1=-0.7$. In addition, we set $\bm{B}=\diag\{\bm{B}_{0}, \bm{I}\}$, where $\bm{B}_{0}\in\mathbb{R}^{3\times 3}$ is a randomly generated orthogonal matrix. Then, based on $\bm{J}, \bm{B}$ and $\bm{\Phi}$, we get the corresponding $\bm{g}$ for model \eqref{eq:model-scalar}, which contains $R_0=N+15$ nonzero entries. The total number of nonzero entries in $\bm{\Phi}$ and $\bm{\Theta}$ is $N+9$. 
The following five competing methods will be compared to JE and RE for the proposed model:
\begin{itemize}[itemsep=2pt,parsep=2pt,topsep=4pt,partopsep=2pt]
	\item [(i)] VAR OLS: As a low-dimensional baseline, we consider the VAR(2) model fitted via the ordinary least squares (OLS) method.
	\item [(ii)] VAR Lasso: Since the VAR($\infty$) process can be approximated by the VAR($P$) model with $P\rightarrow\infty$ as $T\rightarrow\infty$, we consider the sparse VAR($P$) model fitted via the Lasso with $P=\lfloor 1.5\sqrt{T}\rfloor$, following the Stage I estimation in \cite{WBBM21}.
	
	\item[(iii)] VAR HLag: Same as (ii) except that the hierarchical lag (HLag) regularization in \cite{NWBM20} is used instead of the $\ell_1$-regularization.
	
	\item [(iv)] VARMA $\ell_1$: Sparse VARMA($1,1$) model fitted via the two-stage procedure in \cite{WBBM21} with the $\ell_1$-regularization for Stage II.
	
	\item [(v)] VARMA HLag: Same as (iv) except that the HLag regularization is used at Stage II.
\end{itemize}

To assess the out-of-sample forecast accuracy, we compute the $\ell_2$-norm of the prediction error for the one-step ahead forecast at time $T+1$ for the fitted models. All programs are run on a PC with the Intel$^{\text{\textregistered}}$ Core$^\textrm{TM}$ i7 processor with CPU up to 3.00GHz and 16.0GB RAM.  
Methods (i) and (ii)--(v) are implemented by the R packages \texttt{vars} and \texttt{bigtime}, respectively. In the latter package, all estimation procedures are accelerated using C++ via \texttt{Rcpp}. The program for our methods is written entirely in Python.  For a more transparent comparison, we also take into account the following issues:
\begin{itemize}[itemsep=2pt,parsep=2pt,topsep=4pt,partopsep=2pt]
	\item [(a)] For iterative algorithms, the running time depends on both the time per iteration and the number of iterations. However, we are unable to determine the optimal stopping rule  for (ii)--(v) since the existing estimating functions in \texttt{bigtime} do not offer the option of specifying or outputting the number of iterations, which prohibits us from monitoring the performance over iterations. 
	
	\item[(b)] Users can directly control the termination of the algorithms for (ii)--(v) by specifying the convergence threshold value. However, since the convergence criteria are defined for different quantities under different models, they are not comparable across various methods. 
	
	\item [(c)]  All the high-dimensional estimators require certain additional procedures like tuning parameter selection and initialization. They can be time-consuming due to multiple rounds of estimation. The time required is influenced by factors such as grid density and selection criteria, which are not comparable across different methods.
\end{itemize}

In view of the above complications, we adopt the following procedure to simplify the comparison:
\begin{itemize}
	\item For (ii)--(v), we first select the optimal tuning parameters using the cross validation method provided by the \texttt{bigtime} package. This step is not counted towards the reported computation time. Then, fixing the selected tuning parameters, we run two rounds of estimation:
	\begin{itemize}
		\item [R1.] In the first round, by setting the convergence threshold to 
		a very large value (\texttt{eps} $=10^5$), we ensure that the algorithm terminates right after one iteration. We record the computation time of the single iteration\footnote{For methods (iv) and (v), the function for Stage II estimation of the VARMA model in the \texttt{bigtime} package requires specifying a list of at least two candidate values for the tuning parameter. We set both values to the pre-selected optimal tuning parameter. Then by dividing the computation time by two, we record the time corresponding to a single run. In addition, since the Stage I estimation of (iv) (resp. (v)) is exactly the VAR model fitting conducted in (ii) (resp. (iii)), we only report the computation time of Stage II estimation for (iv) (resp. (v)), which is calculated by subtracting the time consumed by (ii) (resp. (iii)).}, which is regarded as the minimum time required for the algorithm. This allows us to optimistically assess the computation time for  (ii)--(v), circumventing the lack of control due to (a) and (b).
		
		\item [R2.] In the second round, we use the default convergence threshold (\texttt{eps} $=10^{-3}$) and let the algorithm run until convergence. Then, we compute the one-step ahead forecast error based on this optimal result.
	\end{itemize}
	
	\item Similarly, for the proposed estimators, we pre-specify the tuning parameter and initial values of our algorithms according to Section \ref{subsec:init}. However, unlike (ii)--(v) for which  we record the computation time of a single iteration due to the unknown optimal stopping rule, we let our algorithms run until convergence. We record the total computation time together with the corresponding one-step ahead forecast error.
\end{itemize}

\begin{figure}[t]
	\centering
	\includegraphics[width=0.85\textwidth]{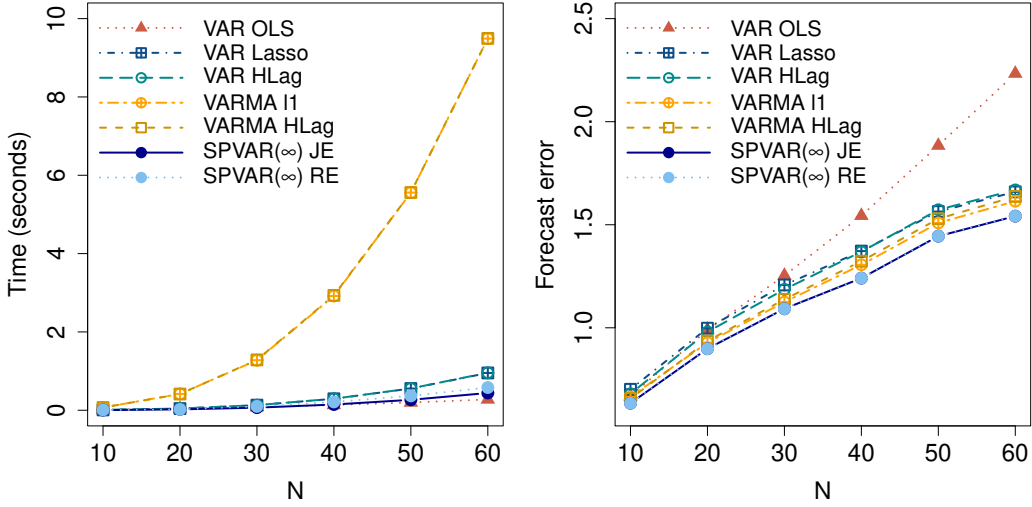}
	\caption{Plots of computation time (left panel) and out-of-sample forecast error  (right panel)  against $N$ for seven methods.}
	\label{fig:exp3}
\end{figure}

Figure \ref{fig:exp3} displays the average computation time  and forecast error based on  100 replications against $N$. According to the left panel, the computation time is ordered as follows: 
\[
\text{VAR OLS $<$ SPVAR($\infty$) $<$ VAR Lasso $\approx$ VAR HLag $\ll$ VARMA l1 $\approx$ VARMA HLag},
\]
where the RE computes slightly slower than the JE, especially for larger $N$. Note that the computation time for the VARMA estimators grows much faster with $N$ than the other methods. 
From the right panel of Figure \ref{fig:exp3}, the forecast error can be ordered as follows:
\[
\text{SPVAR($\infty$) $<$  VARMA l1 $\approx$ VARMA HLag $<$ VAR Lasso $\approx$ VAR HLag $\ll$ VAR OLS},
\]
and the forecast errors based on the JE and RE are nearly identical. 
As expected, the VAR OLS has the worst performance due to overparameterization. Among the high-dimensional methods, those incorporating VARMA dynamics  forecast more accurately than the pure VAR models. In short, this experiment shows that the proposed SPVAR($\infty$) model has the best out-of-sample forecasting performance among all competing models, while enjoying favorable computational efficiency especially compared to the sparse VARMA models. 

%Lastly, we note that ``SPVAR($\infty$)'' in Figure \ref{fig:exp3} corresponds to the estimator introduced in the main paper, which is based on jointly fitting all component series of $\bm{y}_t$. We also conducted the same procedure using an alternative estimator, the single-equation estimator (RE), to be discussed in Section \ref{sec:RE} in detail.  As shown in the left panel of Figure \ref{fig:exp3}, the computation time of the RE is slightly slower. However, its advantage over other high-dimensional methods is still clear. In addition, it can be seen from the right panel  that the forecast errors based on the two estimators are almost identical.

%%%%%%%%%%%%%%%%%%%%%%%%%%%%%%%%%%%%%%%%%%%%%%%%%%%%%%%%%%%%%%%%%%%%%%%%%%%%%%%%%%%%%%%%%%%%%%%%%%
\section{More details for the empirical example}\label{asec:emp}

\begin{table}[t]
	\small
	\centering
	\caption{Description of twenty macroeconomic variables, where T represents types of transformation: 1 = no transformation, 2 = first difference, 3 = second difference, 4 = log, 5 = first difference of logged variables, 6 = second difference of logged variables.}
	\label{tbl:macro}
%	\resizebox{0.8\columnwidth}{!}{
		\begin{tabular}{llll}
			\toprule
			Short name&Mnemonic&T&Description\\
			\midrule
			M1	&	FM1	&	6	&	 Money stock: M1 (bil\$)    \\                    
			M2	&	FM2	&	6	&	 Money stock: M2 (bil\$)      \\           
			Reserves nonbor	&	FMRNBA	&	3	&	 Depository inst reserves: nonborrowed (mil\$)\\
			Reserves tot	&	FMRRA	&	6	&	 Depository inst reserves: total (mil\$) \\
			FFR	&	FYFF	&	2	&	 Interest rate: federal funds (\% per annum)  \\
			10 yr T-bond&FYGT10&2&Interest rate: US treasury const. mat., 10 yr\\
			CPI	&	CPIAUCSL	&	6	&	 CPI: all items    \\                       
			PCED	&	GDP273	&	6	&	 Personal consumption exp.: price index\\  
			Com: spot price (real)	&	PSCCOMR	&	5	&	 Real spot market price index: all commodities\\
			PPI: fin gds	&	PWFSA	&	6	&	 Producer price index: finished goods    \\
			Emp: total	&	CES002	&	5	&	 Employees, nonfarm: total private\\       
			U: all	&	LHUR	&	2	&	 Unemp. rate: All workers, 16 and over (\%) \\  
			Real AHE: goods	&	CES275R	&	5	&	 Real avg hrly earnings, non-farm prod. workers\\
			RGDP	&	GDP251	&	5	&	 Real GDP, quantity index (2000=100)         \\   
			Cons	&	GDP252	&	5	&	 Real personal cons. exp.: quantity Index       \\
			IP: total	&	IPS10	&	5	&	 Industrial production index: total           \\  
			Capacity Util	&	UTL11	&	1	&	 Capacity utilization: manufacturing (SIC)      \\
			HStarts: total	&	HSFR	&	4	&	 Housing starts: total (thousands)           \\   
			Ex rate: avg	&	EXRUS	&	5	&	 US effective exchange rate: index number\\
			S\&P: indust	&	FSPIN	&	5	&	 S\&P's common stock price index: industrials \\
			\bottomrule
		\end{tabular}
%	}
\end{table} 

Table \ref{tbl:macro} provides a detailed description  of the twenty macroeconomic variables. 
More discussions about the fitted model based on the proposed JE as reported in the main paper are given as follows. 

As another example, consider the fitted model for the money stock (M2):
\begin{align*}%\label{eq:M2}
	y_{\text{M2},t}&= -0.34 y_{\text{10 yr T-bond},t-1} + 0.07 y_{\text{U: all},t-1} \\
	&\hspace{5mm}+ \sum_{h=2}^\infty  (-0.45)^{h-1} (0.29 y_{\text{M2}, t-h} -0.85 y_{\text{10 yr T-bond}, t-h}) +\varepsilon_{\text{M2},t},
\end{align*}
where other lag-one terms with coefficients less than 0.032 in absolute value are suppressed for brevity. Note that $y_{\text{M2},t}$ has an infinite-order AR structure. Moreover, based on the fitted model, two time series are Granger causal (GC) for M2:  the 10-year treasury rate (10 yr T-bond) and the unemployment rate (U: all). The former has both short-term and long-term influence on M2, while the latter's influence on M2 is only short-term.

Other findings about the long-term interactions based on $\bm{\widehat{G}}_2$ are summarized as follows. Firstly, 
there are pronounced long-term interactions among the trio: federal funds rate (FFR), real GDP (RGDP), and real personal consumption expenditures (Cons). The directions of influence are FFR $\rightarrow$ RGDP, FFR $\rightarrow$ Cons, and Cons $\rightarrow$ RGDP. Second, the personal consumption expenditures price index (PCED) is influenced by both the Producer Price Index (PPI)  and the Consumer Price Index (CPI), which is intuitive as they are all price indices. Third, in addition to M2 mentioned above, the diagonal of $\bm{\widehat{G}}_2$ indicates that the following variables are influenced by their own lagged values throughout the past: Reserves tot, CPI, and PPI. In addition, as discussed in Section \ref{subsec:model}, the fitted model suggests that the following variables are GC for RGDP: Cons,  IP: total,  HStarts: total,  S\&P: indust, and FFR. 
However, interestingly, since the columns for RGDP in both $\bm{\widehat{G}}_1$ and $\bm{\widehat{G}}_2$ contain all zeros, RGDP is not GC for any other variables. Thus,  the fitted model suggests that RGDP is driven by  the above fundamental economic and financial indicators but may not be a driving force of any other variables under consideration.

\begin{figure}[t]
	\centering
	\includegraphics[width =0.7\textwidth]{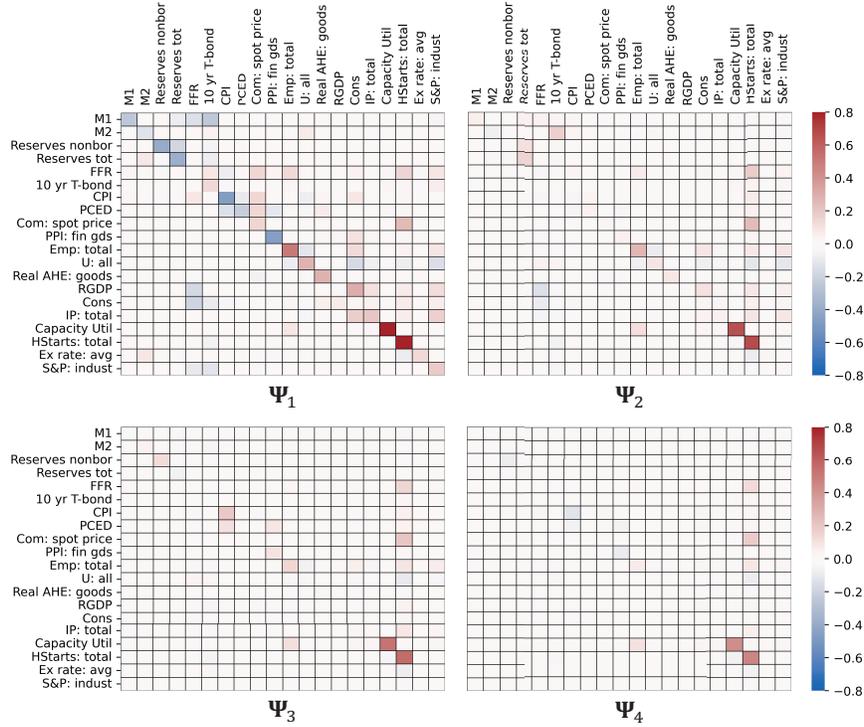}
	\caption{Estimates of $\bm{\Psi}_j$ for $j=1,\dots, 4$ for the VMA($\infty$) representation of the fitted model based on JE.}
	\label{fig:Psimat}
\end{figure}

\begin{figure}[t]
	\centering
	\includegraphics[width = 0.9\textwidth]{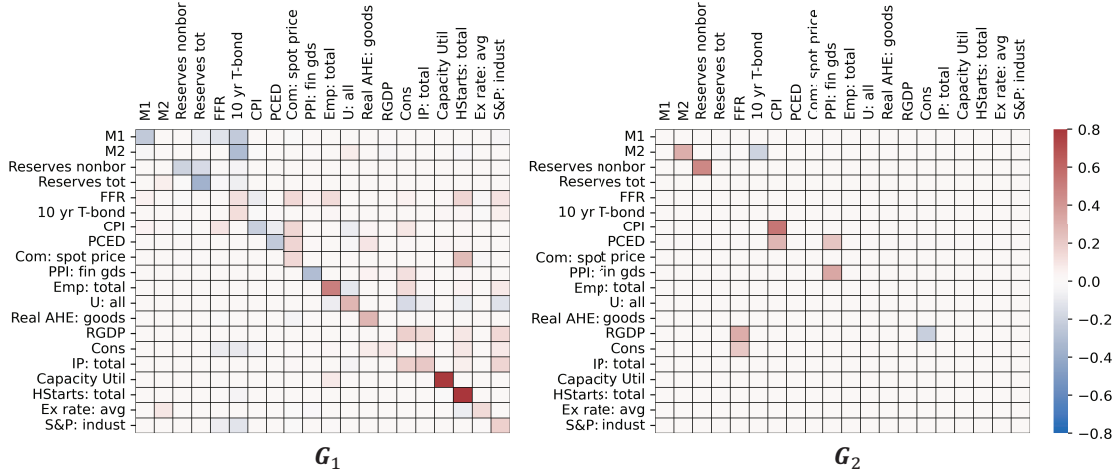} %\vspace{-3mm}
	\caption{Estimates of $\bm{G}_1$ and $\bm{G}_2$ for the proposed model based on RE.}
	\label{fig:GmatRE}
\end{figure}

In addition, as noted in Remark \ref{remark:IRA} in the main paper,  we may alternatively consider  the VMA($\infty$) form of the fitted model for the purpose of   impulse response analysis. For the fitted model reported in the main paper, we give the corresponding estimates of $\bm{\Psi}_j$ with $j=1,\dots, 4$  in Figure \ref{fig:Psimat}. It can be observed that the estimated coefficient matrices are all sparse. For example, by examining  $\bm{\Psi}_3$ and  $\bm{\Psi}_4$, we can see that  HStarts: total is particularly influential, as a shock to it will  impact a number of other variables such as FFR, Com: spot price, Emp: total, U: all, and IP: total.

We have also fitted the model using the RE. The estimates of $\bm{G}_1$ and $\bm{G}_2$ based on the RE exhibit a high degree of similarity to those obtained through the JE; see Figure \ref{fig:GmatRE}. Specifically, the estimates of $\bm{G}_1$  based on JE and RE are nearly identical. While the sparsity pattern and signs of the nonzero entries in $\bm{G}_2$  based on the two estimators are very similar, the magnitude of the nonzero entries derived from RE is generally smaller than those obtained from JE. This discrepancy arises from the impact of different estimates of $\lambda_1$. Note that RE provides distinct estimates of $\lambda_1$ across  rows, while JE only has a single estimate of $\lambda_1$ for all rows.

Finally,  Table \ref{tbl:forecast} displays the forecast errors $\|\bm{\widehat{y}}_t-\bm{y}_t\|_2$ for all competing methods over the rolling forecast period $167\leq t\leq 194$; see the main paper for the detailed procedure.

\begin{table}[H] 
	\renewrobustcmd{\bfseries}{\fontseries{b}\selectfont}
	\newrobustcmd{\B}{\bfseries} % abbreviation
	% uses booktabs package for \addlinespace[0.5ex]
	\centering\small
	\caption{Forecast error (in $\ell_2$ norm) of one-step ahead forecasts for twenty quarterly macroeconomic series. The smallest number in each row is marked in bold.}
	\label{tbl:forecast}
	\begin{tabularx}{0.8\textwidth}{l*{3}{Y}c*{2}{Y}c*{2}{Y}}
		\toprule
		& \multicolumn{3}{c}{VAR} && \multicolumn{2}{c}{VARMA} && \multicolumn{2}{c}{SPVAR($\infty$)}\\\cline{2-4} \cline{6-7} \cline{9-10}
		&	\multicolumn{1}{c}{OLS}	&	\multicolumn{1}{c}{Lasso}	&	\multicolumn{1}{c}{HLag}	&&	\multicolumn{1}{c}{$\ell_1$}	&	\multicolumn{1}{c}{HLag}	&&	\multicolumn{1}{c}{JE}	&	\multicolumn{1}{c}{RE}\\
		\midrule
		Q1-2001	&	4.54	&	4.49	&	4.20	&&	4.11	&	\B 3.81	&&	3.94	&	3.91	\\
		Q2-2001	&	3.29	&	3.44	&	3.38	&&	3.42	&	3.36	&&	\B 3.19	&	3.21	\\
		Q3-2001	&	10.36	&	8.78	&	8.71	&&	8.85	&	8.72	&&	\B 8.68	&	8.69	\\
		Q4-2001	&	12.01	&	11.93	&	11.7	&&	11.65	&	11.84	&&	\B 11.58&	11.62	\\
		Q1-2002	&	6.44	&	\B 3.53	&	4.22	&&	4.42	&	4.42	&&	4.15	&	4.11	\\
		Q2-2002	&	11.55	&	\B 4.15	&	4.26	&&	4.72	&	4.72	&&	5.25	&	4.70	\\
		Q3-2002	&	8.02	&	5.23	&	4.78	&&	5.19	&	4.66	&&	4.82	&	\B 4.65	\\
		Q4-2002	&	8.59	&	2.67	&	2.37	&&	3.33	&	3.33	&&	\B 2.19	&	2.33	\\
		Q1-2003	&	6.38	&	3.60	&	3.62	&&	4.10	&	4.10	&&	3.61	&	\B 3.52	\\
		Q2-2003	&	\B 4.00	&	5.18	&	4.72	&&	5.26	&	4.37	&&	4.42	&	4.47	\\
		Q3-2003	&	6.11	&	4.89	&	4.37	&&	5.25	&	5.16	&&	4.22	&	\B 4.16	\\
		Q4-2003	&	\B 5.36	&	7.09	&	6.17	&&	5.87	&	5.41	&&	5.98	&	5.96	\\
		Q1-2004	&	5.59	&	3.98	&	2.97	&&	4.45	&	3.47	&&	3.12	&	\B 2.92	\\
		Q2-2004	&	5.67	&	\B 3.44	&	3.60	&&	3.76	&	3.76	&&	3.53	&	3.63	\\
		Q3-2004	&	4.09	&	3.46	&	2.99	&&	3.78	&	3.46	&&	\B 2.65	&	2.75	\\
		Q4-2004	&	3.80	&	3.39	&	3.04	&&	\B 2.65	&	2.71	&&	2.96	&	2.98	\\
		Q1-2005	&	3.56	&	3.14	&	2.79	&&	3.45	&	3.32	&&	\B 2.74	&	2.80	\\
		Q2-2005	&	3.64	&	2.66	&	2.54	&&	3.04	&	2.84	&&	\B 2.49	&	2.54	\\
		Q3-2005	&	3.44	&	3.80	&	3.45	&&	3.00	&	\B 2.88	&&	3.10	&	3.23	\\
		Q4-2005	&	3.62	&	2.38	&	2.20	&&	2.84	&	2.37	&&	\B 1.91	&	2.02	\\
		Q1-2006	&	5.38	&	3.29	&	3.23	&&	\B 3.04	&	3.29	&&	3.17	&	3.20	\\
		Q2-2006	&	3.01	&	2.91	&	2.72	&&	3.20	&	3.17	&&	2.58	&	\B 2.54	\\
		Q3-2006	&	2.54	&	2.39	&	2.17	&&	2.39	&	2.39	&&	2.14	&	\B 2.11	\\
		Q4-2006	&	5.90	&	5.08	&	5.03	&&	5.01	&	4.96	&&	\B 4.78	&	4.89	\\
		Q1-2007	&	\B 2.69	&	4.77	&	4.16	&&	3.59	&	3.32	&&	3.73	&	3.71	\\
		Q2-2007	&	4.01	&	\B 2.85	&	3.00	&&	2.96	&	3.03	&&	3.10	&	3.06	\\
		Q3-2007	&	2.96	&	2.82	&	2.38	&&	2.75	&	2.57	&&	\B 2.28	&	2.37	\\
		Q4-2007	&	\B 3.73	&	5.26	&	5.18	&&	4.81	&	4.59	&&	4.89	&	5.05	\\
		\addlinespace[0.5ex]
		Average &	5.367	&	4.307	&	4.069	&&	4.318	&	4.144	&&	3.971	&  \B 3.968	 \\
		%Rank & 6 &	4 &	2 &&	5 &	3 & 1\\
		\bottomrule
	\end{tabularx}
\end{table} 

%%%%%%%%%%%%%%%%%%%%%%%%%%%%%%%%%%%%%%%%%%%%%%%%%%%%%%%%%%%%%%%%%%%%%%%%%%%%%%%%%%%%%%%%%%%%%%%%%%

%%%%%%%%%%%%%%%%%%%%%%%%%%%%%%%%%%%%%%%%%%%%%%%%%%%%%%%%%%%%%%%%%%%%%%%%%%%%%%%%%%%%%%%%%%%%%%%%%%
\section{Proofs of Proposition  \ref{prop:VARMA} and Theorem \ref{thm:stationary}}\label{sec:proofmain}
\subsection{Proof of Proposition \ref{prop:VARMA}}\label{sec1}

Consider the general VARMA$(p,q)$ model with $p,q\geq0$:
\[
\bm{y}_t =\sum_{i=1}^{p} \bm{\Phi}_i\bm{y}_{t-i}+\bm{\varepsilon}_t - \sum_{j=1}^{q}\bm{\Theta}_j\bm{\varepsilon}_{t-j},\hspace{5mm} t\in\mathbb{Z}.
\] 
Since it will reduce to the VAR($p$) model when $q=0$, in what follows we only need to consider the case where $q\geq1$.
Note that the model above can be written equivalently as
\begin{equation} \label{aeq:VARMA}
	\bm{\varepsilon}_t = \bm{\Theta}_1\bm{\varepsilon}_{t-1}-\cdots-\bm{\Theta}_q\bm{\varepsilon}_{t-q}+\bm{\Phi}(B) \bm{y}_t,
\end{equation} 
where 
%$\bm{\Phi}(B) \bm{y}_t=\bm{y}_t - \bm{\Phi}_1\bm{y}_{t-1}+\cdots+\bm{\Phi}_p\bm{y}_{t-1}$, with 
$\bm{\Phi}(B) = \bm{I}-\sum_{i=1}^{p}\bm{\Phi}_i B^i=-\sum_{i=0}^{p}\bm{\Phi}_i B^i$, with $\bm{\Phi}_0=-\bm{I}$. Then we have
\begin{align*}
	\underbrace{\left(\begin{matrix}
			\bm{\varepsilon}_t\\\bm{\varepsilon}_{t-1}\\\bm{\varepsilon}_{t-2}\\\vdots\\\bm{\varepsilon}_{t-q+1}\\
		\end{matrix} \right ) }_{\underline{\bm{\varepsilon}}_t}	
	= \underbrace{\left(\begin{matrix}
			\bm{\Theta}_1&\bm{\Theta}_2&\cdots&\bm{\Theta}_{q-1}&\bm{\Theta}_q\\
			\bm{I}&\bm{0}&\cdots&\bm{0}&\bm{0}\\
			\bm{0}&\bm{I}&\cdots&\bm{0}&\bm{0}\\
			\vdots&\vdots&\ddots&\vdots&\vdots\\
			\bm{0}&\bm{0}&\cdots&\bm{I}&\bm{0}
		\end{matrix} \right )}_{\underline{\bm{\Theta}}}  \underbrace{\left(\begin{matrix}
			\bm{\varepsilon}_{t-1}\\\bm{\varepsilon}_{t-2}\\\bm{\varepsilon}_{t-3}\\\vdots\\\bm{\varepsilon}_{t-q}\\
		\end{matrix} \right ) }_{\underline{\bm{\varepsilon}}_{t-1}} + \underbrace{\left(\begin{matrix}
			\bm{\Phi}(B) \bm{y}_t\\\bm{0}\\\bm{0}\\\vdots\\\bm{0}\\
		\end{matrix} \right ) }_{\underline{\bm{y}}_t},
\end{align*}
where $\underline{\bm{\Theta}}\in\mathbb{R}^{Nq\times Nq}$ is the MA companion matrix. By recursion, we have $\underline{\bm{\varepsilon}}_t=\sum_{j=0}^{\infty} \underline{\bm{\Theta}}^j \underline{\bm{y}}_{t-j}$. Let $\bm{P} = (\bm{I}_{N}, \bm{0}_{N\times N(q-1)})$. 
Note that $\bm{P}\underline{\bm{\varepsilon}}_t=\bm{\varepsilon}_t$, and  $\underline{\bm{y}}_t=\bm{P}^\top \bm{\Phi}(B)\bm{y}_t$. Thus,
\begin{equation}\label{eq:VARMA_inf0}
	\bm{\varepsilon}_t=\sum_{j=0}^{\infty}\bm{P} \underline{\bm{\Theta}}^j \bm{P}^\top \bm{\Phi}(B)\bm{y}_{t-j}
	=-\sum_{j=0}^{\infty} \bm{P} \underline{\bm{\Theta}}^j \bm{P}^\top \sum_{i=0}^{p}\bm{\Phi}_i\bm{y}_{t-j-i} 
	=-\sum_{k=0}^{\infty} \left ( \sum_{i=0}^{p\wedge k}\bm{P} \underline{\bm{\Theta}}^{k-i} \bm{P}^\top \bm{\Phi}_i \right )\bm{y}_{t-k}. 
\end{equation} 
Since $\bm{P} \bm{P}^\top = \bm{I}_N$, it follows from \eqref{eq:VARMA_inf0} that the VAR($\infty$) representation of the  VARMA($p,q$)  model can be written as
\begin{equation}\label{eq:VARMA_infa}
	\bm{y}_t =  \sum_{h=1}^{\infty} \underbrace{ \left ( \sum_{i=0}^{p\wedge h} \bm{P} \underline{\bm{\Theta}}^{h-i} \bm{P}^\top \bm{\Phi}_i \right )}_{\bm{A}_h}\bm{y}_{t-h}+\bm{\varepsilon}_t.
\end{equation}
First, we simply set 
\begin{align} \label{eq:init}
	\bm{G}_j=\sum_{i=0}^{j}\bm{P} \underline{\bm{\Theta}}^{j-i} \bm{P}^\top \bm{\Phi}_i=\bm{A}_j,
	\hspace{5mm}\text{for } 1\leq j\leq p,
\end{align}
and then  we only need to focus on the reparameterization of $\bm{A}_h$ for $h>p$. By \eqref{eq:VARMA_infa}, for $j\geq 1$, we have
\begin{equation} \label{eq:VARMA-Am}
	\bm{A}_{p+j} = \bm{P} \underline{\bm{\Theta}}^{j}\left (\sum_{i=0}^{p} \underline{\bm{\Theta}}^{p-i} \bm{P}^\top \bm{\Phi}_i \right ).
\end{equation}
Next we  derive an alternative parameterization for  $\bm{A}_{p+j}$ with $j\geq 1$.

Under the conditions of this proposition,  $\underline{\bm{\Theta}}$ can be decomposed as $\underline{\bm{\Theta}} = \bm{B}\bm{J}\bm{B}^{-1}$,
where $\bm{B}\in\mathbb{R}^{Nq \times Nq}$ is an invertible matrix, and  $\bm{J}=\diag\{\lambda_1, \dots, \lambda_r, \bm{C}_1, \dots, \bm{C}_s, \bm{0}\}$ is the  real Jordan form, which is a real block diagonal matrix with 
\begin{equation*}
	\bm{C}_k=
	\gamma_k\cdot \left(\begin{matrix}
		\cos (\theta_k)& \sin (\theta_k)\\
		- \sin (\theta_k)& \cos (\theta_k)
	\end{matrix}\right) \in\mathbb{R}^{2\times 2}, \hspace{5mm} 1\leq k\leq s;
\end{equation*}
see Chapter 3 in \cite{HJ12}.

Denote $\bm{\widetilde{B}}=\bm{P} \bm{B}$
and
$\bm{\widetilde{B}}_-=\bm{B}^{-1}\left (\sum_{i=0}^{p} \underline{\bm{\Theta}}^{p-i} \bm{P}^\top \bm{\Phi}_i \right )$.
Note that in the special case that $q=1$, we simply have  $\bm{\widetilde{B}}=\bm{B}$; in addition, $\bm{\widetilde{B}}_-=-\bm{B}^{-1}$ if $p=0$ and $\bm{\widetilde{B}}_-=\bm{B}^{-1}(\bm{\Phi}_1-\bm{\Theta}_1)$ if $p=1$. 

Then by \eqref{eq:VARMA-Am} and the Jordan decomposition, for $j\geq1$, we have 
\begin{equation}\label{eq:Ajordan}
	\bm{A}_{p+j} = \bm{\widetilde{B}} \bm{J}^{j} \bm{\widetilde{B}}_-.
\end{equation}

According to the block form of $\bm{J}$, we can partition the $Nq\times Nq$ matrix $\bm{\widetilde{B}}$ vertically and the $Nq\times Nq$ matrix $\bm{\widetilde{B}}_-$ horizontally as
\[
\bm{\widetilde{B}} = (\bm{\widetilde{b}}_1,\dots,\bm{\widetilde{b}}_r, \bm{\widetilde{B}}_{r+1},\dots  \bm{\widetilde{B}}_{r+s},\bm{\widetilde{B}}_{r+s+1})
\]
and
\[
\bm{\widetilde{B}}_-=(\bm{\widetilde{b}}_{-1},\dots, \bm{\widetilde{b}}_{-r}, \bm{\widetilde{B}}_{-(r+1)},\dots, \bm{\widetilde{B}}_{-(r+s)},\bm{\widetilde{B}}_{-(r+s+1)})^\top
\]
where $\bm{\widetilde{b}}_{k}$ and $\bm{\widetilde{b}}_{-k}$ are $N\times 1$ column vectors for $1\leq k\leq r$, $\bm{\widetilde{B}}_{r+k}$ and $\bm{\widetilde{B}}_{-(r+k)}$ are $N\times 2$ matrices for $1\leq k \leq s$, and $\bm{\widetilde{B}}_{r+s+1}$ and $\bm{\widetilde{B}}_{-(r+s+1)}$ are $N\times\big(Nq-(r+2s)\big)$ matrices. Notice that for any $j\geq 1$, $\bm{J}^j=\diag\{\lambda_1^j, \dots, \lambda_r^j, \bm{C}_1^j, \dots, \bm{C}_s^j, \bm{0}\}$, where 
\begin{equation*}
	\bm{C}_k^j=
	\gamma_k^j\cdot \left(\begin{matrix}
		\cos (j\theta_k)& \sin (j\theta_k)\\
		- \sin (j\theta_k)& \cos (j\theta_k)
	\end{matrix}\right) \in\mathbb{R}^{2\times 2}, \hspace{5mm} 1\leq k\leq s.
\end{equation*}

Let $\bm{\widetilde{b}}_{r+k}^{(i)}$ and $\bm{\widetilde{b}}_{-(r+k)}^{(i)}$ be the $i$th column of $\bm{\widetilde{B}}_{r+k}$ and $\bm{\widetilde{B}}_{-(r+k)}$, respectively, where $1\leq k\leq s$ and $i=1,2$. In addition, denote $\bm{\eta}_k=(\gamma_k,\theta_k)$ for $1\leq k\leq s$. Then by \eqref{eq:Ajordan}, , for $j\geq 1$, we can show that
\begin{align}
	\begin{split}\label{eq:A_extend}
		\bm{A}_{p+j}&=\sum_{k=1}^{r}\lambda_k^j\bm{\widetilde{b}}_{k}\bm{\widetilde{b}}_{-k}^\top+\sum_{k=1}^s\bm{\widetilde{B}}_{r+k}\bm{C}_k^j\bm{\widetilde{B}}_{-(r+k)}^\top\\
		&=\sum_{k=1}^r \lambda_k^j \bm{G}_{p+j}+\sum_{m=1}^{s}\Big\{ \gamma_m^j\cos(j\theta_m)\bm{G}_{p+r+2m-1}+ \gamma_m^j\sin(j\theta_m)\bm{G}_{p+r+2m}\Big\}.
	\end{split}
\end{align}
where 
\begin{align*}
	%	\begin{split}
		\bm{G}_{p+j} &= \bm{\widetilde{b}}_{k}\bm{\widetilde{b}}_{-k}^\top, \quad 1\leq k \leq r,\\
		\bm{G}_{p+r+2m-1} &= \bm{\widetilde{b}}_{r+m}^{(1)}\bm{\widetilde{b}}_{-(r+m)}^{(1)^\top}+\bm{\widetilde{b}}_{r+m}^{(2)}\bm{\widetilde{b}}_{-(r+m)}^{(2)^\top}, \quad 1\leq m\leq s,\\
		\bm{G}_{p+r+2m} &= \bm{\widetilde{b}}_{r+m}^{(1)}\bm{\widetilde{b}}_{-(r+m)}^{(2)^\top}-\bm{\widetilde{b}}_{r+m}^{(2)}\bm{\widetilde{b}}_{-(r+m)}^{(1)^\top}, \quad 1\leq m\leq s.
		%	\end{split}
\end{align*}
Combining \eqref{eq:init} and \eqref{eq:A_extend} , the proof of this proposition is complete.

%%%%%%%%%%%%%%%%%%%%%%%%%%%%%%%%%%%%%%%%%%%%%%%%%%%%%%%%%%%%%%%%%%%%%%%%%%%%%%%%%%%%%%%%%%%%%%%%%%%%%%%%%%%%%%%
\subsection{Proof of Theorem \ref{thm:stationary}}

The proof of  Theorem \ref{thm:stationary} relies on the following lemma.
\begin{lemma}\label{lem:fm}
	For any positive integer $m$, define the function
	\[
	f_m(x)=\sum_{l=2m}^{\infty} {l-m-1 \choose m-1} x^{l-m}.
	\]
	For $0<x<1$, the function $f_m(x)$ takes values on $(0,\infty)$ and can be written as $f_m(x)=x^{m}(1-x)^{-m}$.
\end{lemma}
\begin{proof}[Proof of Lemma \ref{lem:fm}]
	For any positive integer $m$, by the Taylor expansion of the function $g_m(x)=(1-x)^{-m}(m-1)!$ at $x=0$, it can be shown that
	\begin{equation*}%\label{eq:infsum}
		g_m(x)=\sum_{n=0}^{\infty}\frac{(n+m-1)! \,x^n}{n!},
	\end{equation*}
	and the above infinite sum converges for $0<x<1$. As a result,
	\begin{align*}
		f_m(x)=\sum_{l=2m}^{\infty} {l-m-1 \choose l-2m} x^{l-m} = \sum_{n=0}^{\infty} {n+m-1 \choose n} x^{n+m}&=\frac{x^{m}}{(m-1)!}\sum_{n=0}^{\infty} \frac{(n+m-1)! \, x^{n}}{n!}\\
		&=	x^{m}(1-x)^{-m},
	\end{align*}
	which takes values on $(0,\infty)$ for $0<x<1$.
\end{proof}

\begin{proof}[Proof of  Theorem \ref{thm:stationary}]
	It can be readily shown that the VMA($\infty$) representation of the VAR($\infty$)  model is
	\begin{equation}\label{aeq:VMA}
		\bm{y}_t = \bm{\varepsilon}_t+ \sum_{h=1}^{\infty}\bm{\Psi}_h\bm{\varepsilon}_{t-h},\quad\text{with}\quad	\bm{\Psi}_h=\sum_{k=1}^{h}\sum_{\substack{\iota_1+\cdots+ \iota_k=h,\\\iota_1,\dots, \iota_k\geq1}}\bm{A}_{\iota_1} \bm{A}_{\iota_2} \cdots \bm{A}_{\iota_k}, \quad h\geq1.
	\end{equation}
	In particular, $\bm{\Psi}_1=\bm{A}_1$.  Note that the process in \eqref{aeq:VMA} is stationary if
	\begin{equation}\label{aeq:Psi}
		\sum_{h=1}^{\infty}\|\bm{\Psi}_h\| < \infty,
	\end{equation}
	where $\|\cdot\|$ is any submultiplicative matrix norm. Thus, we just need to show that \eqref{aeq:Psi} holds under the conditions of Theorem \ref{thm:stationary}. 
	
	When $p=0$,  the condition  that $\max\{|\lambda_1|,\ldots,|\lambda_r|, \gamma_1,\ldots,\gamma_s\}\leq \bar{\rho}$ implies
	$\|\bm{A}_h\| \leq \bar{\rho}^h \sum_{k=1}^{r+2s}\|\bm{G}_{k}\|$ for $h\geq1$. Then, we can show that 
	\begin{align*}
		\sum_{h=1}^{\infty}\|\bm{\Psi}_h\| \leq  \sum_{k=1}^{\infty} \left \{\sum_{\iota_1=1}^\infty \bar{\rho}^{\iota_1} (\sum_{k=1}^{r+2s}\|\bm{G}_{k}\|) \right \}^k = \sum_{k=1}^{\infty} \left \{\frac{\bar{\rho}}{1-\bar{\rho}}  (\sum_{k=1}^{r+2s}\|\bm{G}_{k}\|) \right \}^k <\infty,
	\end{align*}
	under the condition of this theorem.

	Next we consider the case with $p=1$. On the one hand,
	for any $h\geq 2$, we have
	\[
	\bm{A}_h=\sum_{k=1}^{r}\lambda_k^{h-1}\bm{G}_{1+k}+\sum_{k=1}^{s}\gamma_{k}^{h-1}\cos\{(h-1)\theta_k\}\bm{G}_{1+r+2k-1}+\sum_{k=1}^{s}\gamma_{k}^{h-1}\sin\{(h-1)\theta_k\}\bm{G}_{1+r+2k},
	\]
	and hence the condition that
	$\max\{|\lambda_1|,\ldots,|\lambda_r|, \gamma_1,\ldots,\gamma_s\}\leq \bar{\rho}$ implies
	\begin{equation}\label{aeq:Ahub}
		\|\bm{A}_h\| \leq \bar{\rho}^{h-1}\sum_{k=1}^{r+2s}\|\bm{G}_{1+k}\|, \quad h\geq2.
	\end{equation}
	On the other hand,  $\bm{\Psi}_1=\bm{A}_1=\bm{G}_1$. 
	Then, in view of the expression of $\bm{\Psi}_h$ in \eqref{aeq:VMA}, we consider all possible choices of the  indices $\iota_1,\dots, \iota_k\geq1$ and integer $1\leq k\leq h$ such that  $\iota_1+\cdots+ \iota_k=h$. We can categorize them according to how many of $\iota_1,\dots, \iota_k$ are equal to one. First, note that there are at most $h$ ones among them, since their sum must be $h$. In fact, if there are indeed $h$ ones, then we must have $k=h$ and $\iota_1=\cdots=\iota_h=1$,  which corresponds to  $\bm{A}_{\iota_1} \bm{A}_{\iota_2} \cdots \bm{A}_{\iota_h}=\bm{G}_1^h$. Second, it is impossible that exactly $h-1$ of them are equal to one: e.g., if $\iota_1=\cdots=\iota_{h-1}=1$, then we must have $\iota_h=1$, since they must add up to $h$.
	However, it is possible that exactly $h-l$ of  $\iota_1,\dots, \iota_k$ are equal to one, for any $2\leq l\leq h$. In such cases, the other $m=k-(h-l)$ indices (i.e., indices whose values  are no less than two) must add up to $l$. Let the values of these $m$ indices be $\tau_1,\dots, \tau_{m}\geq 2$, which satisfy   $\tau_1+\cdots+\tau_{m}=l$. Then $\bm{A}_{\iota_1} \bm{A}_{\iota_2} \cdots \bm{A}_{\iota_k}$ has the following form:
	\[
	\bm{G}_1^{i_0}\bm{A}_{\tau_1}\bm{G}_1^{i_1}\bm{A}_{\tau_2} \bm{G}_1^{i_2}\bm{A}_{\tau_3}\cdots \bm{G}_1^{i_{{m}-1}}\bm{A}_{\tau_{m}} \bm{G}_1^{i_{m}},
	\]
	where 
	$i_0, i_1,\dots, i_{m}$ are  nonnegative integers such that $i_0+i_1+\cdots+i_{m}=h-l$.
	According to the above categorization, we can rewrite $\bm{\Psi}_h$ for any $h\geq 2$ as
	\[
	\bm{\Psi}_h=\bm{G}_1^h+\sum_{l=2}^{h}\sum_{m=1}^{\lfloor l/2\rfloor}    \sum_{\substack{i_0+i_1+\cdots+ i_{m}=h-l,\\i_0, i_1,\dots, i_{m}\geq0}} \sum_{\substack{\tau_1+\cdots+ \tau_{m}=l,\\\tau_1,\dots, \tau_{m}\geq2}} \bm{G}_1^{i_0}\bm{A}_{\tau_1}\bm{G}_1^{i_1}\bm{A}_{\tau_2}\cdots \bm{G}_1^{i_{{m}-1}}\bm{A}_{\tau_{m}} \bm{G}_1^{i_{m}}.
	\]
	Thus, to prove \eqref{aeq:Psi}, we only need to show that 
	\begin{equation}\label{aeq:Psi1}
		S_1:=\sum_{h=1}^{\infty}\|\bm{G}_1^h\|<\infty
	\end{equation}
	and 
	\begin{align}\label{aeq:Psi2}
		S_2&:=\sum_{h=1}^{\infty}\sum_{l=2}^{h}\sum_{m=1}^{\lfloor l/2\rfloor}    \sum_{\substack{i_0+i_1+\cdots+ i_{m}=h-l,\\i_0, i_1,\dots, i_{m}\geq0}} \sum_{\substack{\tau_1+\cdots+ \tau_{m}=l,\\\tau_1, \dots, \tau_{m}\geq2}} \|\bm{G}_1^{i_0}\| \|\bm{A}_{\tau_1}\| \|\bm{G}_1^{i_1}\| \|\bm{A}_{\tau_2}\|\cdots \|\bm{G}_1^{i_{{m}-1}}\| \|\bm{A}_{\tau_{m}}\| \|\bm{G}_1^{i_{m}}\| \notag\\
		& <\infty.
	\end{align}
	By Theorem  5.6.15 in \cite{HJ12}, \eqref{aeq:Psi1} holds if $\rho(\bm{G}_1)<1$, which is guaranteed under the condition of Theorem \ref{thm:stationary}. Thus, we next focus on $S_2$.
	By \eqref{aeq:Ahub}, $S_2$ is upper bounded by 
	\begin{align}\label{aeq:infsum}
		&\sum_{h=1}^{\infty}\sum_{l=2}^{h}\sum_{m=1}^{\lfloor l/2\rfloor}   \bar{\rho}^{l-m} (\sum_{k=1}^{r+2s}\|\bm{G}_{1+k}\|)^{m} \sum_{\substack{i_0+i_1+\cdots+ i_{m}=h-l,\\i_0, i_1,\dots, i_{m}\geq0}} \sum_{\substack{\tau_1+\cdots+ \tau_{m}=l,\\\tau_1, \dots, \tau_{m}\geq2}} \|\bm{G}_1^{i_0}\| \|\bm{G}_1^{i_1}\| \cdots  \|\bm{G}_1^{i_{m}}\| \notag\\
		&\hspace{5mm} = \sum_{h=1}^{\infty}\sum_{l=2}^{h}\sum_{m=1}^{\lfloor l/2\rfloor}   {l-m-1 \choose m-1} \bar{\rho}^{l-m} (\sum_{k=1}^{r+2s}\|\bm{G}_{1+k}\|)^{m} \sum_{\substack{i_0+i_1+\cdots+ i_{m}=h-l,\\i_0, i_1,\dots, i_{m}\geq0}} \|\bm{G}_1^{i_0}\| \|\bm{G}_1^{i_1}\| \cdots \|\bm{G}_1^{i_{m}}\| \notag\\
		&\hspace{5mm}  = \sum_{m=1}^{\infty}\sum_{l=2m}^{\infty}   {l-m-1 \choose m-1} \bar{\rho}^{l-m} (\sum_{k=1}^{r+2s}\|\bm{G}_{1+k}\|)^{m} \sum_{h=l}^{\infty}   \sum_{\substack{i_0+i_1+\cdots+ i_{m}=h-l,\\i_0, i_1,\dots, i_{m}\geq0}} \|\bm{G}_1^{i_0}\| \|\bm{G}_1^{i_1}\| \cdots \|\bm{G}_1^{i_{m}}\| \notag\\
		&\hspace{5mm}  = \sum_{m=1}^{\infty}  f_m(\bar{\rho}) (\sum_{k=1}^{r+2s}\|\bm{G}_{1+k}\|)^{m} \sum_{i=0}^{\infty}   \sum_{\substack{i_0+i_1+\cdots+ i_{m}=i,\\i_0, i_1,\dots, i_{m}\geq0}}  \|\bm{G}_1^{i_0}\| \|\bm{G}_1^{i_1}\| \cdots \|\bm{G}_1^{i_{m}}\|  \notag\\
		&\hspace{5mm} =  S_1 \sum_{m=1}^{\infty}  \left (\frac{\bar{\rho} }{1-\bar{\rho}} \sum_{k=1}^{r+2s}\|\bm{G}_{1+k}\| S_1  \right )^{m},
	\end{align}
	where $f_m(\cdot)$ is defined as in Lemma \ref{lem:fm}.
	In the first equality above, to calculate the number of cases for $\tau_1,\dots,\tau_{m}$, we exploit the one-to-one correspondence between the  partition $(\tau_1,\dots, \tau_{m})$ such that $\tau_1+\cdots+ \tau_{m}=l$ with $\tau_1\geq2,\dots, \tau_{m}\geq2$ and the partition $(\tau_1^\prime,\dots, \tau_{m}^\prime)$ such that $\tau_1^\prime+\cdots+ \tau_{m}^\prime=l-m$ with $\tau_1^\prime\geq1,\dots, \tau_{m}^\prime\geq1$, where $\tau_1^\prime=\tau_1-1,\dots, \tau_{m}^\prime=\tau_{m}-1$. Thus, the number of partitions $(\tau_1^\prime,\dots, \tau_{m}^\prime)$ as described above is ${l-m-1 \choose m-1}$.
	
	By the condition of Theorem \ref{thm:stationary} and Lemma 5.6.10 in \cite{HJ12}, there exists some small $\epsilon>0$ such that
	\[
	\frac{\bar{\rho}}{1-\bar{\rho}} \sum_{k=1}^{r+2s}\|\bm{G}_{1+k}\| +\epsilon \leq  \frac{\bar{\rho}}{1-\bar{\rho}} \sum_{k=1}^{r+2s}\rho(\bm{G}_{1+k}) + 2\epsilon <1- \rho(\bm{G}_1).
	\]
	Moreover, 
	\[
	S_1\leq (1-\|\bm{G}_1\|)^{-1} < (1- \rho(\bm{G}_1) -\epsilon)^{-1}.
	\]
	As a result, the power series in \eqref{aeq:infsum} is convergent, and then \eqref{aeq:Psi2} is verified. This completes the proof of \eqref{aeq:Psi} in the case with $p=1$.
	
	Lastly, we consider the general case with $p\geq 1$. The proof is similar to that for the case with $p=1$. The key is to recognize the following stacked representation of the model:
	\begin{equation}\label{aeq:stackmodel}
		\bm{\bar{y}}_t= \bm{\underline{G}}_1 \bm{\bar{y}}_{t-1} + \sum_{h=p+1}^{\infty}  \bm{\underline{A}}_h  \bm{\bar{y}}_{t-h} + \bm{\bar{\varepsilon}}_{t},
	\end{equation}
	where 
	\[
	\bm{\bar{y}}_t = \left( \begin{array}{c}
		\bm{y}_t\\\bm{y}_{t-1}\\\vdots \\\bm{y}_{t-p+1}
	\end{array}\right ), \quad 
	\bm{\bar{\varepsilon}}_t = \left( \begin{array}{c}
		\bm{\varepsilon}_t\\\bm{\varepsilon}_{t-1}\\\vdots \\\bm{\varepsilon}_{t-p+1}
	\end{array}\right ), \quad 
	\underline{\bm{G}}_1= \left(\begin{matrix}
		\bm{G}_1&\bm{G}_2&\cdots&\bm{G}_{p-1}&\bm{G}_p\\
		\bm{I}&\bm{0}&\cdots&\bm{0}&\bm{0}\\
		\bm{0}&\bm{I}&\cdots&\bm{0}&\bm{0}\\
		\vdots&\vdots&\ddots&\vdots&\vdots\\
		\bm{0}&\bm{0}&\cdots&\bm{I}&\bm{0}
	\end{matrix} \right ),
	\]
	and 
	\[
	\bm{\underline{A}}_h =
	\left(\begin{matrix}
		\bm{A}_h&\bm{0}&\cdots&\bm{0}\\
		\bm{0}&\bm{0}&\cdots&\bm{0}\\
		\vdots&\vdots&&\vdots\\
		\bm{0}&\bm{0}&\cdots&\bm{0}
	\end{matrix} \right ), \quad h\geq p+1,
	\]
	where 
	\[
	\bm{A}_h=\sum_{k=1}^{r}\lambda_k^{h-p}\bm{G}_{p+j}+\sum_{k=1}^{s}\gamma_{k}^{h-p}\cos\{(h-p)\theta_k\}\bm{G}_{p+r+2k-1}+\sum_{k=1}^{s}\gamma_{k}^{h-p}\sin\{(h-p)\theta_k\}\bm{G}_{p+r+2k}.
	\]
	Observe that the form of $\bm{\bar{y}}_t$ in \eqref{aeq:stackmodel}  is similar to the model equation for $\bm{y}_t$ in the case with $p=1$, where $\underline{\bm{G}}_1$ plays the same role as $\bm{G}_1$. Similar to \eqref{aeq:Ahub}, we have
	\begin{equation*}
		\|\bm{\underline{A}}_h\| \leq \bar{\rho}^{h-p}\sum_{k=1}^{r+2s}\|\bm{G}_{p+j}\|, \quad h\geq p+1.
	\end{equation*}
	Then, by arguments similar to those of \eqref{aeq:Psi1} and \eqref{aeq:Psi2}, to prove \eqref{aeq:Psi}, it suffices to show that 
	\begin{equation*}
		\underline{S}_1:=\sum_{h=1}^{\infty}\|\bm{\underline{G}}_1^h\|<\infty
	\end{equation*}
	and 
	\begin{align*}
		\underline{S}_2&:=\sum_{h=1}^{\infty}\sum_{l=2}^{h}\sum_{m=1}^{\lfloor l/(p+1)\rfloor}    \sum_{\substack{i_0+i_1+\cdots+ i_{m}=h-l,\\i_0, i_1,\dots, i_{m}\geq0}} \sum_{\substack{\tau_1+\cdots+ \tau_{m}=l,\\\tau_1, \dots, \tau_{m}\geq p+1}} \|\bm{\underline{G}}_1^{i_0}\| \|\bm{\underline{A}}_{\tau_1}\| \|\bm{\underline{G}}_1^{i_1}\| \|\bm{\underline{A}}_{\tau_2}\|\cdots \|\bm{\underline{G}}_1^{i_{{m}-1}}\| \|\bm{\underline{A}}_{\tau_{m}}\| \|\bm{\underline{G}}_1^{i_{m}}\| \notag\\
		& <\infty.
	\end{align*}
	Similar to \eqref{aeq:infsum}, we can show that  $\underline{S}_2$ is upper bounded by 
	\begin{align*}
		&\sum_{h=1}^{\infty}\sum_{l=2}^{h}\sum_{m=1}^{\lfloor l/(p+1)\rfloor}   \bar{\rho}^{l-pm} (\sum_{k=1}^{r+2s}\|\bm{G}_{p+j}\|)^{m} \sum_{\substack{i_0+i_1+\cdots+ i_{m}=h-l,\\i_0, i_1,\dots, i_{m}\geq0}} \sum_{\substack{\tau_1+\cdots+ \tau_{m}=l,\\\tau_1, \dots, \tau_{m}\geq p+1}} \|\bm{\underline{G}}_1^{i_0}\| \|\bm{\underline{G}}_1^{i_1}\| \cdots  \|\bm{\underline{G}}_1^{i_{m}}\|\\
		&\hspace{5mm}  = \sum_{m=1}^{\infty}\sum_{l=(p+1)m}^{\infty}   {l-pm-1 \choose m-1} \bar{\rho}^{l-pm} (\sum_{k=1}^{r+2s}\|\bm{G}_{p+j}\|)^{m} \sum_{h=l}^{\infty}   \sum_{\substack{i_0+i_1+\cdots+ i_{m}=h-l,\\i_0, i_1,\dots, i_{m}\geq0}} \|\bm{\underline{G}}_1^{i_0}\| \|\bm{\underline{G}}_1^{i_1}\| \cdots \|\bm{\underline{G}}_1^{i_{m}}\| \notag\\
		&\hspace{5mm}  = \sum_{m=1}^{\infty}  f_m(\bar{\rho}) (\sum_{k=1}^{r+2s}\|\bm{G}_{p+j}\|)^{m} \sum_{i=0}^{\infty}   \sum_{\substack{i_0+i_1+\cdots+ i_{m}=i,\\i_0, i_1,\dots, i_{m}\geq0}}  \|\bm{\underline{G}}_1^{i_0}\| \|\bm{\underline{G}}_1^{i_1}\| \cdots \|\bm{\underline{G}}_1^{i_{m}}\|  \notag\\
		&\hspace{5mm} =  \underline{S}_1 \sum_{m=1}^{\infty}  \left (\frac{\bar{\rho} }{1-\bar{\rho}} \sum_{k=1}^{r+2s}\|\bm{G}_{p+j}\| \underline{S}_1  \right )^{m}.
	\end{align*}
	Following the same arguments as those for the case with $p=1$, we accomplish the proof of this theorem.
\end{proof}

\section{Proofs of Proposition \ref{prop:perturb} and Theorem \ref{thm:lasso}}\label{asec:prop2}
\subsection{Notations}\label{subsec:notations}
This section collects the notations to be used repeatedly in the proofs of  Proposition \ref{prop:perturb} and Theorem \ref{thm:lasso}. Recall that
\[
\bm{a}=(\bm{L}(\bm{\omega})\otimes \bm{I}_{N^2})\bm{g}, \quad\text{or equivalently,}\quad \bm{A}=\bm{G} (\bm{L}(\bm{\omega})\otimes \bm{I}_N)^\top,
\]
where $\bm{a}=\vect(\bm{A})$ and $\bm{g}=\vect(\bm{g})$, with $\bm{A}=(\bm{A}_1,\bm{A}_2,\dots)\in\mathbb{R}^{N\times \infty}$ and  $\bm{G}=(\bm{G}_1,\dots, \bm{G}_d)\in\mathbb{R}^{N\times Nd}$ being the horizontal concatenations of $\{\bm{A}_h\}_{h=1}^\infty$
and  $\{\bm{G}_k\}_{k=1}^d$, respectively, and
\begin{equation*}%\label{eq:tensorpartition1}
	\bm{L}(\bm{\omega})=\left (\begin{matrix}
		\bm{I}_{p} & \bm{0}_{p\times (r+2s)} \\
		\bm{0}_{\infty\times p} & \bm{L}^{\ma}(\bm{\omega})
	\end{matrix}\right ) = \left ( 
	\begin{array}{ccc}
		\bm{I}_p&  \bm{0}_{p\times r } & \bm{0}_{p\times 2s}  \\
		\bm{0}_{\infty\times p}  &\bm{L}^{I}(\bm{\lambda})&\bm{L}^{II}(\bm{\eta})
	\end{array}
	\right ),  
\end{equation*}
where $\bm{L}^{\ma}(\bm{\omega})=(\bm{L}^{I}(\bm{\lambda}), \bm{L}^{II}(\bm{\eta}))$, with 
\[
\bm{L}^{I}(\bm{\lambda})=(\bm{\ell}^{I}(\lambda_1), \dots, \bm{\ell}^{I}(\lambda_r)) \quad\text{and}\quad
\bm{L}^{II}(\bm{\eta})=(\bm{\ell}^{II}(\bm{\eta}_1), \dots, \bm{\ell}^{II}(\bm{\eta}_s)).
\] 
For $h\geq 1$, the $h$th entry of  $\bm{\ell}^{I}(\lambda_j)\in\mathbb{R}^\infty$ is $\ell_{h}^{I}(\lambda_j)=\lambda_j^h$ and the $h$th  row of $\bm{\ell}^{II}(\bm{\eta}_m)\in\mathbb{
	R}^{\infty\times 2}$ is $\ell_{h}^{II}(\bm{\eta}_m)=(\ell_{h}^{II,1}(\bm{\eta}_m),\ell_{h}^{II,2}(\bm{\eta}_m))=(\gamma_m^h\cos(h\theta_m),\gamma_m^h\sin(h\theta_m) )$, where $1\leq j\leq r$ and $1\leq m\leq s$.

Let  $\nabla \bm{L}^{I}(\bm{\lambda})=(	\nabla \bm{\ell}^{I}(\lambda_1), \dots, \nabla \bm{\ell}^{I}(\lambda_r))$ and  $\nabla_{\theta}\bm{L}^{II}(\bm{\eta}) =(\nabla_{\theta}\bm{\ell}^{II}(\bm{\eta}_1), \dots, \nabla_{\theta}\bm{\ell}^{II}(\bm{\eta}_s))$, where 
$\nabla \bm{\ell}^{I}(\lambda_j)$ is the first-order derivative of  $\bm{\ell}^{I}(\lambda_j)$ with respect to $\lambda_j$, and $\nabla_{\theta}\bm{\ell}^{II}(\bm{\eta}_m)$ is the first-order partial derivative of $\bm{\ell}^{II}(\bm{\eta}_m)$ with respect to $\theta_m$.
Define the $\infty\times (d+r+2s)$  matrix  by augmenting $\bm{L}(\bm{\omega})$ with $(r+2s)$ extra columns: 
\begin{equation}\label{eq:Lstk}
	\bm{L}_{\rm{stack}}(\bm{\omega})
	=\left ( 
	\begin{array}{cccc}
		\bm{I}_p& \bm{0}_{p\times r } & \bm{0}_{p\times 2s} & \bm{0}_{p\times (r+2s)}\\
		\bm{0}_{\infty\times p} &\bm{L}^{I}(\bm{\lambda})&\bm{L}^{II}(\bm{\eta})&\bm{P}(\bm{\omega})
	\end{array}
	\right ), \quad \bm{P}(\bm{\omega})=\left ( \nabla \bm{L}^{I}(\bm{\lambda}),  \nabla_{\theta}\bm{L}^{II}(\bm{\eta}) \right ).
\end{equation}
Note that  since $\colsp\{\nabla_{\gamma}\bm{L}^{II}(\bm{\eta})\}=\colsp\{\nabla_{\theta}\bm{L}^{II}(\bm{\eta})\}$,  $\nabla_{\gamma}\bm{L}^{II}(\bm{\eta})$ is not included in $\bm{P}(\bm{\omega})$ to prevent singularity. 

For any $h\geq1$, let $\bm{\Delta}_h = \bm{A}_h-\bm{A}_h^*$. For any $1\leq k\leq d$, let $\bm{D}_k=\bm{G}_k-\bm{G}_k^*$. Define the corresponding horizontal concatenations 
\[
\bm{\Delta}=(\bm{\Delta}_1, \bm{\Delta}_2, \dots)=\bm{A}-\bm{A}^* \quad\text{and}\quad
\bm{D}=(\bm{D}_1,\dots, \bm{D}_d)=\bm{G}-\bm{G}^*. 
\]
Their vectorizations are 
\[
\bm{\delta}=\vect(\bm{\Delta})=\bm{a}-\bm{a}^*  \quad\text{and}\quad   \bm{d}=\vect(\bm{D})=\bm{g}-\bm{g}^*.
\]
In addition, let
\[
\bm{\phi}=\bm{\omega}-\bm{\omega}^*.
\]

Let $\bm{g}_{\rm{stack}}(\bm{\phi}, \bm{d})=\vect(\bm{G}_{\rm{stack}}(\bm{\phi}, \bm{d}))$, where the $N\times N(d+r+2s)$ matrix
\[
\bm{G}_{\rm{stack}}(\bm{\phi}, \bm{d}) = (\bm{D}, \bm{M}(\bm{\phi}))
\]
is formed by concatenating the $N\times Nd$ matrix $\bm{D}$ and the $N\times N (r+2s)$ matrix
\begin{align*}
	\bm{M}(\bm{\phi}) & = \Big( (\lambda_1 - \lambda_1^*)\bm{G}_{p+1}^{*}, \dots,  (\lambda_r - \lambda_r^*)\bm{G}_{p+r}^{*},  \\
	& \hspace{8mm} (\theta_1 - \theta_1^*)\bm{G}_{p+r+1}^{*}- \frac{\gamma_1 -\gamma_1^*}{\gamma_1^*}\bm{G}_{p+r+2}^{*}, (\theta_1- \theta_1^*) \bm{G}_{p+r+2}^{*} + \frac{\gamma_1 -\gamma_1^*}{\gamma_1^*}\bm{G}_{p+r+1}^{*}, \dots\\
	& \hspace{8mm}(\theta_s - \theta_s^*)\bm{G}_{p+r+2s-1}^{*}- \frac{\gamma_s -\gamma_s^*}{\gamma_s^*}\bm{G}_{p+r+2s}^{*}, (\theta_s- \theta_s^*) \bm{G}_{p+r+2s}^{*} + \frac{\gamma_s -\gamma_s^*}{\gamma_s^*}\bm{G}_{p+r+2s-1}^{*}  \Big),
	%	\\& \in\mathbb{R}^{N\times N (r+2s)}
\end{align*}
i.e., $\bm{M}(\bm{\phi})$ is  the horizontal concatenation of  $(\lambda_j - \lambda_j^*)\bm{G}_{p+j}^{*}$ for $1\leq j\leq r$ and $(\theta_m - \theta_m^*)\bm{G}_{p+r+2m-1}^{*}- \frac{\gamma_m -\gamma_m^*}{\gamma_m^*}\bm{G}_{p+r+2m}^{*}$ and $(\theta_m - \theta_m^*) \bm{G}_{p+r+2m}^{*} + \frac{\gamma_m -\gamma_m^*}{\gamma_m^*}\bm{G}_{p+r+2m-1}^{*}$ for $1\leq m\leq s$. Note that given $\bm{\omega}^*$ and $\bm{g}^*$, the function $\bm{M}(\bm{\phi})$ is linear in $\bm{\phi}$. Thus, $\bm{G}_{\rm{stack}}(\bm{\phi}, \bm{d})$ is bilinear in $\bm{\phi}$ and $\bm{d}$.

As will be shown in the proof of Theorem \ref{thm:lasso}, the following terms quantify the effect of initializing $\bm{y}_s=\bm{0}$ for $s\leq 0$:
\begin{align}\label{eq:notation_init}
	\begin{split}
		&S_1(\bm{\Delta}) = \frac{2}{T}\sum_{t=1}^{T}\langle \bm{\varepsilon}_t, \sum_{h=t}^{\infty}\bm{\Delta}_h \bm{y}_{t-h} \rangle,\\
		&S_2(\bm{\Delta})  = \frac{2}{T}\sum_{t=2}^{T}\langle \sum_{h=t}^{\infty}\bm{A}_h^* \bm{y}_{t-h}, \sum_{k=1}^{t-1}\bm{\Delta}_k \bm{y}_{t-k} \rangle,\\
		& S_3(\bm{\Delta}) = \frac{3}{T}\sum_{t=1}^{T}\Big \|\sum_{k=t}^{\infty}\bm{\Delta}_k\bm{y}_{t-k} \Big \|_2^2.
	\end{split}
\end{align}
Let $\bm{x}_t=(\bm{y}_{t-1}^\top,\bm{y}_{t-2}^\top,\dots)^\top$, and $\bm{\widetilde{x}}_{t}= (\bm{y}_{t-1}^\top,\dots,\bm{y}_1^\top,0,0,\dots)^\top$ is the initialized version of $\bm{x}_t$. For any $h\geq1$, let $\bm{\widehat{\Delta}}_h = \bm{\widehat{A}}_h-\bm{A}_h^*$. For any $1\leq k\leq d$, let $\bm{\widehat{D}}_k=\bm{\widehat{G}}_k-\bm{G}_k^*$. Define the corresponding horizontal concatenations 
\[
\bm{\widehat{\Delta}}=(\bm{\widehat{\Delta}}_1, \bm{\widehat{\Delta}}_2, \dots)=\bm{\widehat{A}}-\bm{A}^* \quad\text{and}\quad
\bm{\widehat{D}}=(\bm{\widehat{D}}_1,\dots, \bm{\widehat{D}}_d)=\bm{\widehat{G}}-\bm{G}^*,
\]
and their vectorizations
\[
\bm{\widehat{\delta}}=\vect(\bm{\widehat{\Delta}})=\bm{\widehat{a}}-\bm{a}^*  \quad\text{and}\quad   \bm{\widehat{d}}=\vect(\bm{\widehat{D}})=\bm{\widehat{g}}-\bm{g}^*,
\]
where $\bm{\widehat{a}}=\vect(\bm{\widehat{A}})$ and $\bm{\widehat{g}}=\vect(\bm{\widehat{g}})$, with $\bm{\widehat{A}}=(\bm{\widehat{A}}_1,\bm{\widehat{A}}_2,\dots)\in\mathbb{R}^{N\times \infty}$ and  $\bm{\widehat{G}}=(\bm{\widehat{G}}_1,\dots, \bm{\widehat{G}}_d)\in\mathbb{R}^{N\times Nd}$ being the horizontal concatenations of $\{\bm{\widehat{A}}_h\}_{h=1}^\infty$
and  $\{\bm{\widehat{G}}_k\}_{k=1}^d$, respectively.
Let
\[
\bm{\widehat{\phi}}=\bm{\widehat{\omega}}-\bm{\omega}^*.
\]
Moreover, denote
\[
\bm{\widehat{D}}_\ar=(\bm{\widehat{D}}_{1},\dots, \bm{\widehat{D}}_p)=\bm{\widehat{G}}_\ar-\bm{G}_\ar^*\quad\text{and}\quad
\bm{\widehat{D}}_\ma=(\bm{\widehat{D}}_{p+1},\dots, \bm{\widehat{D}}_d)=\bm{\widehat{G}}_\ma-\bm{G}_\ma^*,
\]
and their vectorizations
\[
\bm{\widehat{d}}_\ar=\vect(\bm{\widehat{D}}_\ar)=\bm{\widehat{g}}_\ar-\bm{g}^*_\ar \quad\text{and}\quad   \bm{\widehat{d}}_\ma=\vect(\bm{\widehat{D}}_\ma)=\bm{\widehat{g}}_\ma-\bm{g}^*_\ma,
\]

Given the constant $c_{\bm{\omega}}>0$ chosen as in \eqref{eq:comega}, we define the local neighborhood of $\bm{\omega}^*$,
\[\bm{\Omega}_1 = \{\bm{\omega}\in\bm{\Omega}\mid \|\bm{\omega}-\bm{\omega}^*\|_2 \leq c_{\bm{\omega}} \}.
\]
In addition, let 
\[
\bm{\Phi}= \{\bm{\phi}=\bm{\omega}-\bm{\omega}^* \mid \bm{\omega}\in\bm{\Omega}\} \quad\text{and}\quad \bm{\Phi}_1= \{\bm{\phi}=\bm{\omega}-\bm{\omega}^* \mid \bm{\omega}\in\bm{\Omega}_1\}. 
\]
Then under the conditions of Theorem \ref{thm:lasso}, we have $\widehat{\bm{\omega}}\in \bm{\Omega}_1$, $\bm{\widehat{\phi}}=\widehat{\bm{\omega}}-\bm{\omega}^*\in\bm{\Phi}_1$, and $\bm{\widehat{\delta}} = \bm{\widehat{a}}-\bm{a}^* \in 
\bm{\Upsilon}$,
where 
\begin{align*}%\label{eq:Upsilon}
	\bm{\Upsilon} = \left \{\bm{\delta} = \bm{a}-\bm{a}^*\in\mathbb{R}^{\infty} \mid \bm{a}=(\bm{L}(\bm{\omega})\otimes \bm{I}_{N^2})\bm{g}, \text{ where } \bm{\omega}\in\bm{\Omega}_1 \text{ and } \bm{g}\in\mathbb{R}^{N^2 d} \right \}. 
\end{align*}

Let
\begin{equation}\label{eq:kappatilde}
	\widetilde{\kappa}_1= \kappa_1	\min\{1, \sigma_{\min, L}^2\}  \quad\text{and}\quad \widetilde{\kappa}_2= \kappa_2\max\{1, \sigma_{\max, L}^2\},
\end{equation}
where 
\[
\sigma_{\min, L}=\sigma_{\min}(\bm{L}_{\rm{stack}}(\bm{\omega}^*))\quad\text{and}\quad
\sigma_{\max, L}=\sigma_{\max}(\bm{L}_{\rm{stack}}(\bm{\omega}^*)).
\]
Note that  $\widetilde{\kappa}_1 \leq  \kappa_1\leq \kappa_2\leq  \widetilde{\kappa}_2$, and as will be shown by Lemma \ref{lemma:fullrank},
\[
\widetilde{\kappa}_1 \asymp \kappa_1 \quad\text{and}\quad \kappa_2 \asymp \widetilde{\kappa}_2.
\]

Lastly, we use $C, C_1, C_2, \ldots >0$ (or $c, c_1, c_2, \ldots >0$) to denote generic large (or small) absolute constants whose values can vary from place to place. 
For any matrix $\bm{X}$, let $\sigma_{\max}(\bm{X})$ and $\sigma_{\min}(\bm{X})$  denote its largest and smallest singular values, respectively.

\subsection{Preliminary results}\label{subsec:prelim}

In this section, we provide the important lemmas that are directly used in the proofs of  Proposition \ref{prop:perturb} and Theorem \ref{thm:lasso}. The proofs of these lemmas are relegated to Section \ref{asec:aux}.

The goal of Proposition \ref{prop:perturb} is to establish the local linearity of $\bm{\delta}(\bm{\phi}, \bm{d})$ with respect to $\bm{\phi}$ and $\bm{d}$. Specifically, within a local neighborhood of $\bm{\omega}^*$, we aim to show that
\begin{equation}\label{eq:linearize}
	\bm{\Delta}(\bm{\phi}, \bm{d}) =\bm{A}(\bm{\omega}, \bm{g})-\bm{A}^* \approx  \bm{G}_{\rm{stack}}(\bm{\phi}, \bm{d}) (\bm{L}_{\rm{stack}}(\bm{\omega}^*)\otimes \bm{I}_N)^\top,
\end{equation}
or in vector form,
\begin{equation*} 
	\bm{\delta}(\bm{\phi}, \bm{d}) =\bm{a}(\bm{\omega}, \bm{g})-\bm{a}^* \approx (\bm{L}_{\rm{stack}}(\bm{\omega}^*)\otimes \bm{I}_{N^2}) \bm{g}_{\rm{stack}}(\bm{\phi}, \bm{d}).
\end{equation*}
Note that  $\bm{G}_{\rm{stack}}(\bm{\phi}, \bm{d})$ (or $\bm{g}_{\rm{stack}}(\bm{\phi}, \bm{d})$) is  bilinear in $\bm{\phi}$ and $\bm{d}$; see Section \ref{subsec:notations}. Moreover, it is necessary to show that the $\bm{L}_{\rm{stack}}(\bm{\omega}^*)$ is bounded. This is guaranteed by Assumptions  \ref{assum:statn}(i) and \ref{assum:gap}, as established by  Lemma \ref{lemma:fullrank} below, which is built upon Lemma \ref{cor1}. 

\begin{lemma}\label{cor1}
	Under Assumption \ref{assum:statn}(i), there exists an absolute constant $C_{\ell}\geq 1$ such that for all $\bm{\omega}\in\bm{\Omega}$, $h\geq 1$, $1\leq k\leq r$, $1\leq m\leq s$, and  $\iota=1,2$, it holds
	$|\nabla\ell_{h}^{I}(\lambda_j)|\leq C_{\ell}\bar{\rho}^{h}$, $\|\nabla \ell_{h}^{II,\iota}(\bm{\eta}_m)\|_2\leq C_{\ell}\bar{\rho}^{h}$, $|\nabla^2\ell_{h}^{I}(\lambda_j)|\leq C_{\ell}\bar{\rho}^{h}$, and $\|\nabla^2 \ell_{h}^{II,\iota}(\bm{\eta}_m)\|_{\Fr}\leq C_{\ell}\bar{\rho}^{h}$.
\end{lemma}

%%%%%%%%%%%%%%%%%%%%%%%%%%%%%%%%%%%%%%%%%%%%%%%%%%%%%%%%%%%%%%%%%%%%%%%%%%%%%%%%%%%%%%%%%
\begin{lemma} \label{lemma:fullrank}
	Under Assumption \ref{assum:statn}(i), the matrix $\bm{L}_{\rm{stack}}(\bm{\omega}^*)$ has full rank, and its largest and smallest singular values satisfy 
	\begin{equation*}
		0<	1 \wedge c_{\bar{\rho}}\leq \sigma_{\min}(\bm{L}_{\rm{stack}}(\bm{\omega}^*)) \leq \sigma_{\max}(\bm{L}_{\rm{stack}}(\bm{\omega}^*)) \leq 1\vee C_{\bar{\rho}}.
	\end{equation*}
	where $	C_{\bar{\rho}}=C_{\ell}\sqrt{J}\bar{\rho}(1-\bar{\rho})^{-1}$ and $c_{\bar{\rho}}=0.25^s (\nu_{\mathrm{lower}}^*)^{3J/2}(\nu_{\mathrm{gap}}^*)^{J(J/2-1)}/C_{\bar{\rho}}^{J-1}$, with $J=2(r+2s)$. Moreover, if Assumption \ref{assum:gap} further holds, then $C_{\bar{\rho}}\asymp 1$ and $c_{\bar{\rho}}\asymp 1$.
	%	Moreover, if  there exist an absolute constant $C>0$ such that $0\leq r, s \leq C$, then $C_{\bar{\rho}}\asymp 1$ and $c_{\bar{\rho}}\asymp 1$.
\end{lemma}
%%%%%%%%%%%%%%%%%%%%%%%%%%%%%%%%%%%%%%%%%%%%%%%%%%%%%%%%%%%%%%%%%%%%%%%%%%%%%%%%%%%%%%%%%

The proof of Theorem \ref{thm:lasso} directly relies on Lemmas \ref{lemma:devb}--\ref{lemma:init3} below.

\begin{lemma}[Deviation bound]\label{lemma:devb}
	Under Assumptions \ref{assum:statn} and \ref{assum:error}, if $\|\bm{\widehat{\omega}} - \bm{\omega}^*\|_{2}\leq c_{\bm{\omega}}$, $\sum_{j=0}^{\infty} \|\bm{\Psi}_j^* \|_{\op}^2<\infty$, and $T\gtrsim \log \{N(p\vee 1)\}$, then  with probability at least $1-C e^{-c\log N}$,
	\[
	\frac{1}{T} \left |\sum_{t=1}^{T}\langle \bm{\varepsilon}_t, \bm{\widehat{\Delta}}\bm{x}_t \rangle \right | \leq C_{\dev} \sqrt{\frac{\kappa_2\lambda_{\max}(\bm{\Sigma}_{\varepsilon})\log \{N (p\vee 1)\}}{T}} \left (\|\bm{\widehat{d}}\|_1 +\|\bm{g}_{\ma}^{*}\|_1 \| \bm{\widehat{\phi}}\|_2 \right ), 
	\]
	where  $C_{\dev}>0$ is an absolute constant.
	%If $p=0$, then the above result holds with $\log(N^2p)$ replaced by $\log N$.
	%$C_{\dev}=\max_{1\leq i\leq 3}C_i>0$, with $C_i$'s being the absolute constants in Lemma \ref{lemma:dev_lasso},
\end{lemma}

\begin{lemma}[Restricted strong convexity] \label{lemma:rsclasso}
	Under Assumptions \ref{assum:statn}--\ref{assum:error}, if $\|\bm{\widehat{\omega}} - \bm{\omega}^*\|_{2}\leq c_{\bm{\omega}}$ and $T\gtrsim(\kappa_2/\kappa_1)^2 \log\{ (\kappa_2/\kappa_1)(\overline{\alpha}_\ma/\underline{\alpha}_\ma) N (p\vee 1)\}$, then 	with probability at least $1-Ce^{-c\kappa_1^2 T/\kappa_2^2}$,
	\begin{equation*}%\label{eq:scs}
		\frac{1}{T}\sum_{t=1}^{T}\|\bm{\widehat{\Delta}}\bm{x}_t\|_{2}^2 \geq  C_{\rsc} \left [\kappa_1 \|\bm{\widehat{\Delta}}\|_{\Fr}^2  - \frac{  \kappa_2^2  \log \{N(p\vee1)\}}{\kappa_1 T}  \|\bm{\widehat{d}}\|_1^2 \right ],
	\end{equation*}
	where $C_{\rsc}>0$ is an absolute constant.  
\end{lemma}

\begin{lemma}[Effect of initial values I]\label{lemma:init1}
	Under Assumptions \ref{assum:statn} and \ref{assum:error}, if  $\|\bm{\widehat{\omega}} - \bm{\omega}^*\|_{2}\leq c_{\bm{\omega}}$, $\sum_{j=0}^{\infty} \|\bm{\Psi}_j^* \|_{\op}^2<\infty$, and  $T\gtrsim \log N$, then with probability at least $1 - C(p\vee1) e^{-c\log N}$, 
	\[
	|S_1(\bm{\widehat{\Delta}})| \leq \frac{C_{\init1}\sqrt{\kappa_2\lambda_{\max}(\bm{\Sigma}_{\varepsilon})(p\vee1)\log N}}{T} \left (\|\bm{\widehat{d}}\|_1 + \|\bm{g}_{\ma}^{*}\|_1 \|\bm{\widehat{\phi}}\|_2 \right),
	\]
	where $C_{\init1}>0$ is an absolute constant.
\end{lemma}

\begin{lemma}[Effect of initial values II]\label{lemma:init2}
	Under Assumptions \ref{assum:statn}--\ref{assum:error}, if $T\gtrsim \log \{N (p\vee1)\}$ and $\|\bm{\widehat{\omega}} - \bm{\omega}^*\|_{2}\leq c_{\bm{\omega}}$, then 
	with probability at least $1 - C(p\vee1) e^{-c\log \{N(p\vee1)\}}$, 
	\[
	|S_2(\bm{\widehat{\Delta}})| \leq \frac{C_{\init2} \kappa_2 (p\vee1)^2  }{T} \left (\|\bm{\widehat{d}}\|_1 + \|\bm{g}_{\ma}^{*}\|_1 \|\bm{\widehat{\phi}}\|_2 \right),
	\]
	where $C_{\init2}>0$ is an absolute constant.
\end{lemma}

\begin{lemma}[Effect of initial values III] \label{lemma:init3}
	Under Assumptions \ref{assum:statn}--\ref{assum:error}, if $\log N\gtrsim  (\kappa_2/\kappa_1)^2$ and $\|\bm{\widehat{\omega}} - \bm{\omega}^*\|_{2}\leq c_{\bm{\omega}}$, then	with probability at least $1-C  e^ {-c \kappa_1^2 (p\vee1) \log \{N(p\vee1)\}/\kappa_2^2}$,
	\begin{equation*}%\label{eq:scs}
		|S_{3}(\bm{\widehat{\Delta}})| \leq  \frac{ C_{\init3}  \kappa_2 (p\vee1)}{T}  \left[\|\bm{\widehat{\Delta}}\|_{\Fr}^2 \log \{N(p\vee1)\} + \|\bm{\widehat{d}}\|_1^2 \right],
	\end{equation*}
	where $C_{\init3}>0$ is an absolute constant.
\end{lemma}

%%%%%%%%%%%%%%%%%%%%%%%%%%%%%%%%%%%%%%%%%%%%%%%%%%%%%%%%%%%%%%%%%%%%%%%%%%%%%%%%%%%%%%%%%
\subsection{Proof of Proposition \ref{prop:perturb}}

Note that $\bm{A}_k=\bm{G}_k$ for $1\leq k\leq p$, and for any $h\geq1$,
\begin{align}\label{eq:linear}
	\bm{A}_{p+h} =  \sum_{j=1}^{r}\ell_{h}^{I}(\lambda_j)\bm{G}_{p+j}+\sum_{m=1}^{s}\left\{\ell_{h}^{II,1}(\bm{\eta}_m)\bm{G}_{p+r+2m-1}+\ell_{h}^{II,2}(\bm{\eta}_m)\bm{G}_{p+r+2m}\right\}.
\end{align}
Then $\bm{\Delta}_k=\bm{G}_k-\bm{G}_k^*$ for $1\leq k\leq p$.
Moreover, for any $h\geq 1$, by \eqref{eq:linear} and the Taylor expansion,
\begin{align}\label{eq:delta}
	\bm{\Delta}_{p+h}&=\bm{A}_{p+h}-\bm{A}_{p+h}^* \notag \\
	&=\sum_{j=1}^{r}\Bigg \{\ell_{h}^{I}(\lambda_j^*) +\nabla\ell_{h}^{I}(\lambda_j^*) (\lambda_j-\lambda_j^*) +\frac{1}{2}\nabla^2\ell_{h}^{I}(\widetilde{\lambda}_j) (\lambda_j-\lambda_j^*)^2 \Bigg \}\bm{G}_{p+j}\notag \\
	&\hspace{5mm} +\sum_{m=1}^{s}\Bigg \{\ell_{h}^{II,1}(\bm{\eta}_m^*) +(\bm{\eta}_m-\bm{\eta}_m^*)^\top \nabla \ell_{h}^{II,1}(\bm{\eta}_m^*) \notag\\
	&\hspace{33mm} +\frac{1}{2}(\bm{\eta}_m-\bm{\eta}_m^*)^\top\nabla^2 \ell_{h}^{II,1}(\widetilde{\bm{\eta}}_m)(\bm{\eta}_m-\bm{\eta}_m^*)\Bigg \}\bm{G}_{p+r+2m-1}\notag \\
	&\hspace{5mm} +\sum_{m=1}^{s}\Bigg \{\ell_{h}^{II,2}(\bm{\eta}_m^*) + (\bm{\eta}_m-\bm{\eta}_m^*)^\top \nabla \ell_{h}^{II,2}(\bm{\eta}_m^*) \notag\\ &\hspace{33mm}+\frac{1}{2}(\bm{\eta}_m-\bm{\eta}_m^*)^\top\nabla^2 \ell_{h}^{II,2}(\widetilde{\bm{\eta}}_m) (\bm{\eta}_m-\bm{\eta}_m^*)\Bigg \}\bm{G}_{p+r+2m} -\bm{A}_{p+h}^* \notag\\
	&:=\bm{H}_h+\bm{R}_h,
\end{align}
where $\widetilde{\lambda}_j$ lies between $\lambda_j^*$ and $\lambda_j$ for $1 \leq j \leq r$, $\widetilde{\bm{\eta}}_m$ lies between $\bm{\eta}^*_k$ and $\bm{\eta}_m$ for $1 \leq m \leq s$, the first-order approximation is
\begin{align}\label{eq:Hh}
	\bm{H}_h &=\sum_{j=1}^{r}\ell_{h}^{I}(\lambda_j^*) (\bm{G}_{p+j}-\bm{G}_{p+j}^*) +\sum_{m=1}^{s}\sum_{\iota=1}^2\ell_{h}^{II,\iota}(\bm{\eta}_m^*) (\bm{G}_{p+r+2(m-1)+\iota}-\bm{G}_{p+r+2(m-1)+\iota}^*) \notag\\
	&\hspace{5mm} +\sum_{j=1}^{r}(\lambda_j-\lambda_j^*)\nabla\ell_{h}^{I}(\lambda_j^*) \bm{G}_{p+j}^* +\sum_{m=1}^{s}\sum_{\iota=1}^2(\bm{\eta}_m-\bm{\eta}_m^*)^\top \nabla \ell_{h}^{II,\iota}(\bm{\eta}_m^*)\bm{G}_{p+r+2(m-1)+\iota}^*,
\end{align}
and the remainder is
\begin{align} \label{eq:Rh}
	\bm{R}_h &= \sum_{i=1}^{r}\nabla\ell_{h}^{I}(\lambda_j^*) (\lambda_j-\lambda_j^*) (\bm{G}_{p+j} - \bm{G}_{p+j}^*) \notag\\
	&\hspace{5mm}+ \sum_{m=1}^{s}\sum_{\iota=1}^2(\bm{\eta}_m-\bm{\eta}_m^*)^\top \nabla \ell_{h}^{II,\iota}(\bm{\eta}_m^*)(\bm{G}_{p+r+2(m-1)+\iota} - \bm{G}_{p+r+2(m-1)+\iota}^*) \notag\\
	&\hspace{5mm} +\frac{1}{2} \sum_{j=1}^{r}\nabla^2\ell_{h}^{I}(\widetilde{\lambda}_j) (\lambda_j-\lambda_j^*)^2 \bm{G}_{p+j} \notag\\
	&\hspace{5mm} 
	+\frac{1}{2} \sum_{m=1}^{s}\sum_{\iota=1}^2(\bm{\eta}_m-\bm{\eta}_m^*)^\top\nabla^2 \ell_{h}^{II,\iota}(\widetilde{\bm{\eta}}_m)(\bm{\eta}_m-\bm{\eta}_m^*)\bm{G}_{p+r+2(m-1)+\iota}.
\end{align}
Here for notational simplicity, we have suppressed the dependence of  $\widetilde{\lambda}_j$'s and $\widetilde{\bm{\eta}}_m$'s on $h$.

We first consider $\bm{R}_h$. Denote $\bm{R}_h=\bm{R}_{1h} +\bm{R}_{2h}+\bm{R}_{3h}$, where
\begin{align}\label{eq:Rhs}
	\bm{R}_{1h}=& \sum_{j=1}^{r}\nabla\ell_{h}^{I}(\lambda_j^*)  (\lambda_j-\lambda_j^*) (\bm{G}_{p+j} - \bm{G}_{p+j}^*) \notag\\
	&
	+ \sum_{m=1}^{s}\sum_{\iota=1}^2(\bm{\eta}_m-\bm{\eta}_m^*)^\top \nabla \ell_{h}^{II,\iota}(\bm{\eta}_m^*)(\bm{G}_{p+r+2(m-1)+\iota} - \bm{G}_{p+r+2(m-1)+\iota}^*),\notag\\
	\bm{R}_{2h} =&\frac{1}{2}\sum_{j=1}^{r}\nabla^2\ell_{h}^{I}(\widetilde{\lambda}_j) (\lambda_j-\lambda_j^*)^2 (\bm{G}_{p+j} - \bm{G}_{p+j}^* ) \notag\\
	&+\frac{1}{2} \sum_{m=1}^{s}\sum_{\iota=1}^2(\bm{\eta}_m-\bm{\eta}_m^*)^\top\nabla^2 \ell_{h}^{II,\iota}(\widetilde{\bm{\eta}}_m)(\bm{\eta}_m-\bm{\eta}_m^*)(\bm{G}_{p+r+2(m-1)+\iota} - \bm{G}_{p+r+2(m-1)+\iota}^*), \notag\\
	\bm{R}_{3h} =&\frac{1}{2}\sum_{j=1}^{r}\nabla^2\ell_{h}^{I}(\widetilde{\lambda}_j) (\lambda_j-\lambda_j^*)^2 \bm{G}_{p+j}^* \notag\\
	&+ \frac{1}{2} \sum_{m=1}^{s}\sum_{\iota=1}^2(\bm{\eta}_m-\bm{\eta}_m^*)^\top\nabla^2 \ell_{h}^{II,\iota}(\widetilde{\bm{\eta}}_m)(\bm{\eta}_m-\bm{\eta}_m^*) \bm{G}_{p+r+2(m-1)+\iota}^*.
\end{align}

Note that for any matrix $\bm{Y} = \sum_{k=1}^{d}a_k\bm{X}_k$, 
$\|\bm{Y}\|_{\op}\leq \|\bm{Y}\|_{\Fr}\leq  (\sum_{k=1}^{d}\|\bm{X}_k\|_{\Fr}^2)^{1/2}(\sum_{k=1}^{d}a_k^2)^{1/2}=\|\bm{X}\|_{\Fr}\|\bm{a}\|_2$,
and $\sum_{k=1}^da_k^4\leq (\sum_{k=1}^da_k^2)^2$, where $\bm{a} = (a_1,\dots,a_d)^\top\in\mathbb{R}^d$, and $\bm{X}=(\bm{X}_1,\dots,\bm{X}_{d})$. 
Then, by Lemma \ref{cor1},
\begin{align*}%\label{eq:R1}
	\| \bm{R}_{1h}\|_{\Fr}
	&\leq C_{\ell}\bar{\rho}^h \sqrt{\|\bm{\lambda} - \bm{\lambda}^*\|_2^2 + 2\|\bm{\eta} - \bm{\eta}^*\|_2^2} \\
	&\hspace{5mm}\cdot\sqrt{ \sum_{j=1}^{r}\|\bm{G}_{p+j} - \bm{G}_{p+j}^*\|_{\Fr}^2 + \sum_{m=1}^{s}\sum_{\iota=1}^2 \|\bm{G}_{p+r+2(m-1)+\iota} - \bm{G}_{p+r+2(m-1)+\iota}^*\|_{\Fr}^2 } \notag \\
	%	 \sqrt{\sum_{j=1}^{r}\|\bm{G}_{p+j} - \bm{G}_{p+j}^* \|_{\Fr}^2 +  \sum_{m=1}^{s}\sum_{\iota=1}^2\|\bm{G}_{p+r+2(m-1)+\iota} - \bm{G}_{p+r+2(m-1)+\iota}^*\|_{\Fr}^2} 
	%	 \notag \\
	&\leq \sqrt{2}C_{\ell}\bar{\rho}^h \|\bm{\phi}\|_2 \cdot \|\bm{G}_\ma - \bm{G}_\ma^*\|_{\Fr} \leq \sqrt{2}C_{\ell}\bar{\rho}^h \|\bm{\phi}\|_2 \|\bm{d}\|_{2},
\end{align*}
and similarly,
\begin{align*}%\label{eq:R2}
	\| \bm{R}_{2h}\|_{\Fr}
	%	&\leq \frac{\sqrt{2}}{2}C_{\ell}\bar{\rho}^h \|\bm{\phi}\|_2^2\sqrt{ \sum_{j=1}^{r}\|\bm{G}_{p+j} - \bm{G}_{p+j}^*\|_{\Fr}^2 + \sum_{m=1}^{s}\sum_{\iota=1}^2 \|\bm{G}_{p+r+2(m-1)+\iota} - \bm{G}_{p+r+2(m-1)+\iota}^*\|_{\Fr}^2 } \notag\\
	&\leq \frac{\sqrt{2}}{2} C_{\ell}\bar{\rho}^h \|\bm{\phi}\|_2^2 \cdot \|\bm{G}_\ma - \bm{G}_\ma^*\|_{\Fr} \leq \frac{\sqrt{2}}{2} C_{\ell}\bar{\rho}^h \|\bm{\phi}\|_2^2 \|\bm{d}\|_{2},
\end{align*}
where $\bm{G}_\ma=(\bm{G}_{p+1},\dots, \bm{G}_d)$.
%In addition, let $\bm{g}_\ar=\vect(\bm{G}_\ar)$ and $\bm{g}_\ma=\vect(\bm{G}_\ma)$, where  $\bm{G}_\ar=(\bm{G}_{1},\dots, \bm{G}_p)\in\mathbb{R}^{N\times Np}$ and $\bm{G}_\ma=(\bm{G}_{p+1},\dots, \bm{G}_d)\in\mathbb{R}^{N\times N(r+2s)}$.
Moreover, by Lemma \ref{cor1} again, we can show that
\begin{equation*}%\label{eq:R3}
	\| \bm{R}_{3h}\|_{\Fr} \leq  \frac{\sqrt{2}}{2} C_{\ell}\overline{\alpha}_\ma \bar{\rho}^h \|\bm{\phi}\|_2^2.
\end{equation*}
As a result,
\begin{align} \label{eq:Rnorm1}
	\|\bm{R}_h\|_{\Fr} &\leq \|\bm{R}_{1h}\|_{\Fr} + \|\bm{R}_{2h}\|_{\Fr} + \|\bm{R}_{3h}\|_{\Fr} \notag\\ &\leq    C_{\ell} \bar{\rho}^h \|\bm{\phi}\|_2 \left ( \sqrt{2}  \|\bm{d}\|_{2} + \frac{\sqrt{2}}{2} \|\bm{\phi}\|_2\|\bm{d}\|_{2}+ \frac{\sqrt{2}}{2}  \overline{\alpha}_\ma \|\bm{\phi}\|_2 \right ).
\end{align}

Now consider $\bm{H}_h$ in \eqref{eq:Hh}. 
Notice that for any $h\geq1$ and $1\leq m\leq s$, 
\begin{align*}%\label{eq:eqsp}
	&\nabla_\gamma\ell_{h}^{II,1}(\bm{\eta}_m)=h\gamma_m^{h-1}\cos(h\theta_m)=\frac{1}{\gamma_m}\nabla_\theta\ell_{h}^{II,2}(\bm{\eta}_m),\\
	&\nabla_\gamma\ell_{h}^{II,2}(\bm{\eta}_m)=h\gamma_m^{h-1}\sin(h\theta_m)=-\frac{1}{\gamma_m}\nabla_\theta\ell_{h}^{II,1}(\bm{\eta}_m).
\end{align*}
Thus, the last term on the right side of  \eqref{eq:Hh} can be simplified to 
\begin{align} \label{eq:linearcomb}
	&\sum_{m=1}^{s}\sum_{\iota=1}^2(\bm{\eta}_m-\bm{\eta}_m^*)^\top \nabla \ell_{h}^{II,\iota}(\bm{\eta}_m^*)\bm{G}_{p+r+2(m-1)+\iota}^* \notag\\
	&\hspace{5mm}= \sum_{m=1}^{s}\left[ (\theta_m - \theta_m^*) \bm{G}_{p+r+2m-1}^* - \frac{1}{\gamma_m^*}(\gamma_m -\gamma_m^*)\bm{G}_{p+r+2m}^*\right]\nabla_{\theta} \ell_{h}^{II,1}(\bm{\eta}_m^*)\notag\\
	&\hspace{10mm}+\sum_{m=1}^{s} \left[ (\theta_m - \theta_m^*) \bm{G}_{p+r+2m}^* + \frac{1}{\gamma_m^*}(\gamma_m -\gamma_m^*)\bm{G}_{p+r+2m-1}^*\right]\nabla_{\theta} \ell_{h}^{II,2}(\bm{\eta}_m^*).
\end{align}
Let $\bm{H}=(\bm{H}_1, \bm{H}_2, \dots)$ and  $\bm{R}=(\bm{R}_1, \bm{R}_2, \dots)$. Then by \eqref{eq:Hh} and \eqref{eq:linearcomb} it can be verified that 
\begin{align}\label{eq:stackH}
	\bm{\widetilde{H}}:= (\bm{G}_1-\bm{G}_1^*,\cdots,\bm{G}_p-\bm{G}_p^*,\bm{H})
	&=\bm{D} (\bm{L}(\bm{\omega}^*)\otimes \bm{I}_N)^\top +\bm{M}(\bm{\phi}) ( \bm{P}(\bm{\omega}^*) \otimes \bm{I}_N)^\top \notag\\
	&=\bm{G}_{\rm{stack}}(\bm{\phi},\bm{d})(\bm{L}_{\rm{stack}}(\bm{\omega}^*)\otimes \bm{I}_N)^\top.
\end{align}

Note that 
\begin{equation}\label{eq:Delta}
	\bm{\Delta}=\bm{\widetilde{H}}+(\bm{0}_{N\times N p}, \bm{R}).
\end{equation}
Moreover,
\begin{align*}%\label{eq:DFr2}
	\|\bm{M}(\bm{\phi})\|_{\Fr}^2 &= \sum_{j=1}^{r}(\lambda_j - \lambda_j^*)^2\|\bm{G}_{p+j}^*\|_{\Fr}^2 + \sum_{m=1}^{s} \left\| (\theta_m - \theta_m^*) \bm{G}_{p+r+2m-1}^* - \frac{\gamma_m -\gamma_m^*}{\gamma_m^*}\bm{G}_{p+r+2m}^*\right\|_{\Fr}^2 \notag\\
	&\hspace{5mm}+ \sum_{m=1}^{s}\left\| (\theta_m - \theta_m^*) \bm{G}_{p+r+2m}^* + \frac{\gamma_m -\gamma_m^*}{\gamma_m^*}\bm{G}_{p+r+2m-1}^*\right\|_{\Fr}^2 \notag\\
	&= \sum_{j=1}^{r}(\lambda_j - \lambda_j^*)^2\|\bm{G}_{p+j}^*\|_{\Fr}^2 + \sum_{m=1}^{s} (\theta_m - \theta_m^*)^2(\|\bm{G}_{p+r+2m-1}^*\|_{\Fr}^2 + \|\bm{G}_{p+r+2m}^*\|_{\Fr}^2 ) \notag\\
	&\hspace{5mm}+ \sum_{m=1}^{s}\frac{(\gamma_m - \gamma_m^*)^2}{\gamma_m^{*2}} (\|\bm{G}_{p+r+2m-1}^*\|_{\Fr}^2 + \|\bm{G}_{p+r+2m}^*\|_{\Fr}^2 ),
\end{align*}
which leads to
\begin{equation}\label{eq:DFr}
	\underline{\alpha}_\ma \|\bm{\phi}\|_2  \leq \|\bm{M}(\bm{\phi})\|_{\Fr} \leq \frac{\sqrt{2}\overline{\alpha}_\ma}{\min_{1\leq k\leq s}\gamma_{k}^*} \|\bm{\phi}\|_2.
\end{equation}
By the simple inequalities $ (|x| + |y|) / 2 \leq \sqrt{x^2 + y^2} \leq |x| + |y|$, we have
$0.5(\|\bm{d}\|_{2} + \|\bm{M}(\bm{\phi})\|_{\Fr}) \leq  \|\bm{G}_{\rm{stack}}(\bm{\phi},\bm{d})\|_{\Fr}	\leq \|\bm{d}\|_{2} + \|\bm{M}(\bm{\phi})\|_{\Fr}$,
%\end{equation*} 
and thus in view of \eqref{eq:DFr} we further have
\begin{equation}\label{eq:Gstacknorm}
	\frac{1}{2}(\|\bm{d}\|_{2} +\underline{\alpha}_\ma\|\bm{\phi}\|_2) \leq  \|\bm{G}_{\rm{stack}}(\bm{\phi},\bm{d})\|_{\Fr}	\leq \|\bm{d}\|_{2} + \frac{\sqrt{2}\overline{\alpha}_\ma}{\min_{1\leq k\leq s}\gamma_{k}^*} \|\bm{\phi}\|_2.
\end{equation} 
%By Lemma \ref{lemma:fullrank}, $\sigma_{\min, L}=\sigma_{\min}(\bm{L}_{\rm{stack}}(\bm{\omega}^*))>0$. 
Then  it follows from \eqref{eq:Gstacknorm} that
\[
\frac{\sigma_{\min, L}}{2} (\|\bm{d}\|_{2} +\underline{\alpha}_\ma \|\bm{\phi}\|_2 )\leq \|\bm{\widetilde{H}}\|_{\Fr} \leq  \sigma_{\max, L}\left (\|\bm{d}\|_{2} + \frac{\sqrt{2}\overline{\alpha}_\ma}{\min_{1\leq k\leq s}\gamma_{k}^*}\|\bm{\phi}\|_2\right ),
\]
where  $\sigma_{\min, L}=\sigma_{\min}(\bm{L}_{\rm{stack}}(\bm{\omega}^*))$ and  $\sigma_{\max, L}= \sigma_{\max}(\bm{L}_{\rm{stack}}(\bm{\omega}^*))$.
Combining this with \eqref{eq:Rnorm1}, \eqref{eq:Delta}, \eqref{eq:DFr}, as well as the fact that $\|\bm{G}_\ma - \bm{G}_\ma^*\|_{\Fr}\leq \|\bm{d}\|_{2}$, we have
\begin{align*}
	\|\bm{\Delta}\|_{\Fr} &\leq \|\bm{\widetilde{H}}\|_{\Fr} + \|\bm{R}\|_{\Fr} \\
	&\leq \left\{\sigma_{\max, L} + \frac{\sqrt{2}C_{\ell}}{1-\bar{\rho}} \left (\|\bm{\phi}\|_2+\frac{\|\bm{\phi}\|_2^2}{2}\right )\right\} \|\bm{d}\|_{2}
	+\left(\frac{\sqrt{2} \sigma_{\max, L}}{\min_{1\leq k\leq s}\gamma_{k}^*} + \frac{\sqrt{2}}{2}\cdot \frac{C_{\ell}}{1-\bar{\rho}} \|\bm{\phi}\|_2\right) \overline{\alpha}_\ma \|\bm{\phi}\|_2
\end{align*}
and
\begin{align*}
	\|\bm{\Delta}\|_{\Fr} &\geq \|\bm{\widetilde{H}}\|_{\Fr} - \|\bm{R}\|_{\Fr}\\
	&\geq \left\{\frac{\sigma_{\min, L}}{2}-\frac{\sqrt{2}C_{\ell}}{1-\bar{\rho}}  \left (\|\bm{\phi}\|_2+\frac{\|\bm{\phi}\|_2^2}{2}\right )\right\} \|\bm{d}\|_{2} +\left( \frac{ \sigma_{\min, L}}{2}- \frac{\sqrt{2}}{2}\cdot\frac{C_{\ell}\overline{\alpha}_\ma \|\bm{\phi}\|_2}{(1-\bar{\rho})\underline{\alpha}_\ma}  \right)  \underline{\alpha}_\ma \|\bm{\phi}\|_2. 
\end{align*}
Thus, as long as
\begin{equation}\label{eq:comega}
	\|\bm{\phi}\|_2\leq c_{\bm{\omega}}\leq  \min\left \{2, \frac{\underline{\alpha}_\ma(1-\bar{\rho})\sigma_{\min, L} }{8\sqrt{2}C_{\ell}\overline{\alpha}_\ma }\right \},
\end{equation}
we have
\begin{equation}\label{eq:prop2}
	c_{\Delta}	\left(\|\bm{d}\|_{2} + \underline{\alpha}_\ma\|\bm{\phi}\|_2\right) 
	\leq \|\bm{\Delta}\|_{\Fr} \leq         
	C_{\Delta}	\left(\|\bm{d}\|_{2} +  \overline{\alpha}_\ma\|\bm{\phi}\|_2\right),
\end{equation}
where
\[
c_{\Delta}= \sigma_{\min, L}/4 \quad \text{and}\quad
C_{\Delta}= \sigma_{\max, L} \left (1\vee \frac{\sqrt{2}}{\nu_{\mathrm{lower}}^*} \right ) + \frac{4\sqrt{2} C_{\ell} }{1-\bar{\rho}}.
\]
Finally, by Lemma  \ref{lemma:fullrank}, we have
\[
0<(1\wedge c_{\bar{\rho}})/4 \leq c_{\Delta} \leq C_{\Delta} \leq  (1\vee C_{\bar{\rho}}) \left (1\vee \frac{\sqrt{2}}{\nu_{\mathrm{lower}}^*} \right ) + \frac{4\sqrt{2} C_{\ell} }{1-\bar{\rho}},
\]
i.e., $c_{\Delta}\asymp1$ and $C_{\Delta}\asymp1$, and \eqref{eq:comega} is fulfilled by taking
\begin{equation}\label{eq:comg}
	c_{\bm{\omega}} =  \min\left \{2, \frac{\underline{\alpha}_\ma(1-\bar{\rho})(1\wedge c_{\bar{\rho}}) }{8\sqrt{2}C_{\ell}\overline{\alpha}_\ma }\right \}.
\end{equation}
The proof of this proposition is complete.
%\end{proof}

%%%%%%%%%%%%%%%%%%%%%%%%%%%%%%%%%%%%%%%%%%%%%%%%%%%%%%%%%%%%%%%%%%%%%%%%%%%%%%%%%%
\subsection{Proof of Theorem \ref{thm:lasso}}\label{asec:proofThm2}
Note that $\sum_{h=1}^{t-1}\bm{A}_h\bm{y}_{t-h}= \bm{A}\bm{\widetilde{x}}_{t}$, where $\bm{\widetilde{x}}_{t}= (\bm{y}_{t-1}^\top,\dots,\bm{y}_1^\top,0,0,\dots)^\top$ is the initialized version of $\bm{x}_t$. By the optimality of $\bm{\widehat{A}}$, we have
\[
\frac{1}{T}\sum_{t=1}^{T} \| \bm{y}_t - \bm{A}^*\bm{\widetilde{x}}_{t} - \bm{\widehat{\Delta}} \bm{\widetilde{x}}_{t}\|_2^2
\leq \frac{1}{T}\sum_{t=1}^{T} \| \bm{y}_t - \bm{A}^*\bm{\widetilde{x}}_{t}\|_2^2+\lambda_g(\|\bm{g}^*\|_1-\|\widehat{\bm{g}}\|_1),
\]
%Note that $\bm{\widehat{\Delta}}\bm{\widetilde{x}}_{t}= \sum_{k=1}^{t-1}\bm{\widehat{\Delta}}_k \bm{y}_{t-k}$ and $\bm{\widehat{\Delta}}\bm{x}_{t}= \sum_{k=1}^{\infty}\bm{\widehat{\Delta}}_k \bm{y}_{t-k}$.
Then, since $\bm{y}_t - \bm{A}^*\bm{\widetilde{x}}_{t}=\bm{\varepsilon}_t +   \sum_{h=t}^{\infty}\bm{A}_h^* \bm{y}_{t-h}$ and $\bm{\widehat{\Delta}}\bm{\widetilde{x}}_{t}=\bm{\widehat{\Delta}}\bm{x}_{t}-\sum_{k=t}^{\infty}\bm{\widehat{\Delta}}_k \bm{y}_{t-k}$, we have
\begin{align}\label{eq:thm1eq}
	\frac{1}{T}\sum_{t=1}^{T}\|\bm{\widehat{\Delta}}\bm{\widetilde{x}}_{t}\|_2^2 &\leq 
	\frac{2}{T}\sum_{t=1}^{T}\langle \bm{\varepsilon}_t, \bm{\widehat{\Delta}}\bm{\widetilde{x}}_{t} \rangle + \underbrace{\frac{2}{T}\sum_{t=1}^{T}\langle \sum_{h=t}^{\infty}\bm{A}_h^* \bm{y}_{t-h}, \bm{\widehat{\Delta}}\bm{\widetilde{x}}_{t}  \rangle}_{S_2(\bm{\widehat{\Delta}})} +\lambda_g(\|\bm{g}^*\|_1-\|\widehat{\bm{g}}\|_1 ) \notag \\
	&= \frac{2}{T}\sum_{t=1}^{T}\langle \bm{\varepsilon}_t, \bm{\widehat{\Delta}}\bm{x}_t \rangle +\lambda_g(\|\bm{g}^*\|_1-\|\widehat{\bm{g}}\|_1 )+S_2(\bm{\widehat{\Delta}}) - S_1(\bm{\widehat{\Delta}}),
\end{align}
where $S_1(\cdot)$ and $S_2(\cdot)$ are defined as in \eqref{eq:notation_init}.
Moreover, applying the inequality $\|\bm{a}-\bm{b}\|_2^2\geq (3/4)\|\bm{a}\|_2^2-3\|\bm{b}\|_2^2$ with $\bm{a}=\bm{\widehat{\Delta}}\bm{x}_{t}=\sum_{h=1}^{\infty}\bm{\widehat{\Delta}}_h\bm{y}_{t-h}$ and $\bm{b}=\sum_{k=t}^{\infty}\bm{\widehat{\Delta}}_k\bm{y}_{t-k}$, we can lower bound the left-hand side of  \eqref{eq:thm1eq} to further obtain that
\begin{align} \label{eq:thm1eq1}
	\frac{3}{4T}\sum_{t=1}^{T}\|\bm{\widehat{\Delta}}\bm{x}_t\|_{2}^2-S_3(\bm{\widehat{\Delta}})
	&\leq \frac{2}{T}\sum_{t=1}^{T}\langle \bm{\varepsilon}_t, \bm{\widehat{\Delta}}\bm{x}_t \rangle+\lambda_g(\|\bm{g}^*\|_1-\|\widehat{\bm{g}}\|_1 ) + S_2(\bm{\widehat{\Delta}}) - S_1(\bm{\widehat{\Delta}}),
\end{align}
where $S_3(\cdot)$ is defined as in \eqref{eq:notation_init}.
It is worth pointing out that  $S_i(\bm{\widehat{\Delta}})$ for $1\leq i\leq 3$ capture the initialization effect of $\bm{y}_s=\bm{0}$ for $s\leq 0$ on the estimation error, and their upper bounds are given in Lemmas \ref{lemma:init1}--\ref{lemma:init3}.

Next we assume that the high probability events in Lemmas \ref{lemma:devb}--\ref{lemma:init3} all hold and focus on the deterministic analysis. For a threshold $\eta>0$ to be chosen later, define the thresholded subsets
\begin{align*}
	S_\ar(\eta)&=\{ (i,j,k)\mid  |g_{i,j,k}^*|>\eta, i,j\in\{1,\dots, N\}, k\in\{1,\dots, p\}\},\\
	S_\ma(\eta)&=\{ (i,j,k)\mid  |g_{i,j,k}^*|>\eta, i,j\in\{1,\dots, N\}, k\in\{p+1,\dots, d\}\},
\end{align*}
and
\[
S(\eta)=S_\ar(\eta) \cup S_\ma(\eta)=\{ (i,j,k)\mid  |g_{i,j,k}^*|>\eta, i,j\in\{1,\dots, N\}, k\in\{1,\dots, d\}\}.
\]
Define  $S^\complement(\eta)=\{ (i,j,k)\mid i,j\in\{1,\dots, N\}, k\in\{1,\dots, d\}\}\setminus S(\eta)$ as the complementary set of $S(\eta)$. Similarly, the complementary set of $S_\ma(\eta)$ is  $S_\ma^\complement(\eta)=\{ (i,j,k)\mid i,j\in\{1,\dots, N\}, k\in\{p+1,\dots, d\}\}\setminus S_\ma(\eta)$. Let $|S|$ denote the cardinality of a set $S$.
Note that 
\[
R_q\geq \sum_{i=1}^{N}\sum_{j=1}^{N}\sum_{k=1}^{d} |g_{i,j,k}^*|^q \geq \sum_{(i,j,k)\in S(\eta)}  |g_{i,j,k}^*|^q \geq \eta^q |S(\eta)|,
\]
and
\begin{equation*}
	\|\bm{g}_{S^\complement(\eta)}^*\|_1 = \sum_{(i,j,k)\in S^\complement(\eta)} |g_{i,j,k}^*| =  \sum_{(i,j,k)\in S^\complement(\eta)} |g_{i,j,k}^*|^q |g_{i,j,k}^*|^{1-q}.
\end{equation*} 
Thus, we have
\begin{equation}\label{eq:ws1}
	|S(\eta)| \leq R_q \eta^{-q} \quad\text{and} \quad 		\|\bm{g}_{S^\complement(\eta)}^*\|_1 \leq R_q \eta^{1-q}.
\end{equation}
Similarly, we can show that 
\begin{equation}\label{eq:ws2}
	|S_\ma(\eta)| \leq R_q^\ma \eta^{-q} \quad\text{and}\quad \|(\bm{g}^*_{\ma})_{ S_\ma^\complement(\eta)}\|_1 \leq R_q^\ma \eta^{1-q}. 
\end{equation}
By \eqref{eq:ws2}, by choosing $\eta$ such  that
\begin{equation}\label{eq:etacond}
	\eta^{2-q} \leq \frac{(r+2s) \overline{\alpha}_\ma^2}{R_q^\ma},
\end{equation}
we have
\begin{align*}
	\|\bm{g}_{\ma}^{*}\|_1^2 \leq 2\|(\bm{g}_{\ma}^{*})_{S_\ma(\eta)}\|_1^2 +  2\|(\bm{g}_{\ma}^{*})_{S^\complement_\ma(\eta)}\|_1^2 & \leq 2 |S_\ma(\eta)| \|\bm{g}_{\ma}^{*}\|_2^2 + 2 (R_q^\ma \eta^{1-q})^2\notag \\
	&\leq 2 R_q^\ma  \eta^{-q} \left \{ (r+2s) \overline{\alpha}_\ma^2 + R_q^\ma \eta^{2-q} \right \} \notag\\
	& \leq 4 R_q^\ma  \eta^{-q} (r+2s) \overline{\alpha}_\ma^2.
\end{align*}
Then, since $r+2s\lesssim1$ and $(\overline{\alpha}_\ma/\underline{\alpha}_\ma)^2\lesssim R_q/R_q^\ma$,  we further have 
\begin{equation}\label{eq:ws3}
	\underline{\alpha}_\ma^{-2} \|\bm{g}_{\ma}^{*}\|_1^2 \lesssim R_q \eta^{-q}.
\end{equation}
%	\begin{equation}\label{eq:etacond1}
	%	\eta \leq \left \{ \frac{(r+2s) \overline{\alpha}_\ma^2}{R_q^\ma} \right \}^{\frac{1}{2-q}}
	%	\end{equation}

Consider the right-hand side of \eqref{eq:thm1eq}. By Lemma \ref{lemma:devb}, if we choose $\lambda_g$ such that 
\begin{equation}\label{eq:lgcond1}
	\frac{\lambda_g}{4}\geq   C_{\dev} \sqrt{\frac{\kappa_2\lambda_{\max}(\bm{\Sigma}_{\varepsilon})\log \{N(p\vee 1)\}}{T}},
\end{equation}
then we can show that
\begin{align}\label{eq:thm1eq2}
	&\frac{2}{T}\sum_{t=1}^{T}\langle \bm{\varepsilon}_t, \bm{\widehat{\Delta}}\bm{x}_t \rangle +\lambda_g(\|\bm{g}^*\|_1-\|\widehat{\bm{g}}\|_1 ) \notag\\
	&\hspace{5mm} \leq  \frac{\lambda_g}{2}(\|\bm{\widehat{d}}\|_1 +\|\bm{g}_{\ma}^{*}\|_1 \| \bm{\widehat{\phi}}\|_2)+ \lambda_g(\|\bm{g}^*\|_1-\|\bm{g}_{S(\eta)}^*+\widehat{\bm{d}}_{S^\complement(\eta)}\|_1 + \|\bm{g}_{S^\complement(\eta)}^*+\widehat{\bm{d}}_{S(\eta)}\|_1 ) \notag\\
	&\hspace{5mm} \leq \frac{\lambda_g}{2}(\|\widehat{\bm{d}}_{S(\eta)}\|_1+\|\widehat{\bm{d}}_{S^\complement(\eta)}\|_1 +\|\bm{g}_{\ma}^{*}\|_1 \| \bm{\widehat{\phi}}\|_2 ) +\lambda_g(2\|\bm{g}_{S^\complement(\eta)}^*\|_1+ \|\widehat{\bm{d}}_{S(\eta)}\|_1 -\|\widehat{\bm{d}}_{S^\complement(\eta)}\|_1) \notag\\
	&\hspace{5mm} \leq \frac{\lambda_g}{2}\left (4\|\bm{g}_{S^\complement(\eta)}^*\|_1 + 3\|\widehat{\bm{d}}_{S(\eta)}\|_1-\|\widehat{\bm{d}}_{S^\complement(\eta)}\|_1 +  \|\bm{g}_{\ma}^{*}\|_1 \| \bm{\widehat{\phi}}\|_2\right ).
\end{align}
In addition, since $T\gtrsim \kappa_2 (p\vee1)^4$, it follows from  Lemmas \ref{lemma:init1} and \ref{lemma:init2} that  
\begin{align}\label{eq:thm1eq7}
	S_2(\bm{\widehat{\Delta}}) - S_1(\bm{\widehat{\Delta}}) & \leq \frac{\lambda_{g}}{4}\left (\|\bm{\widehat{d}}\|_1 + \|\bm{g}_{\ma}^{*}\|_1 \|\bm{\widehat{\phi}}\|_2 \right) \notag\\
	& =\frac{\lambda_{g}}{4}\left (\|\widehat{\bm{d}}_{S(\eta)}\|_1+\|\widehat{\bm{d}}_{S^\complement(\eta)}\|_1 + \|\bm{g}_{\ma}^{*}\|_1 \|\bm{\widehat{\phi}}\|_2\right).
\end{align}
Combining \eqref{eq:thm1eq}, \eqref{eq:thm1eq2} and \eqref{eq:thm1eq7}, we have
\begin{align*}%\label{eq:thm1eq6}
	0\leq \frac{1}{T}\sum_{t=1}^{T}\|\bm{\widehat{\Delta}}\bm{\widetilde{x}}_t\|_{2}^2  &\leq \frac{2}{T}\sum_{t=1}^{T}\langle \bm{\varepsilon}_t, \bm{\widehat{\Delta}}\bm{x}_t \rangle +\lambda_g(\|\bm{g}^*\|_1-\|\widehat{\bm{g}}\|_1 )+S_2(\bm{\widehat{\Delta}}) - S_1(\bm{\widehat{\Delta}}) \notag \\
	&\leq \frac{\lambda_{g}}{4}\left (8\|\bm{g}_{S^\complement(\eta)}^*\|_1 + 7\|\widehat{\bm{d}}_{S(\eta)}\|_1-\|\widehat{\bm{d}}_{S^\complement(\eta)}\|_1 +  3\|\bm{g}_{\ma}^{*}\|_1 \| \bm{\widehat{\phi}}\|_2\right),
\end{align*}
which implies 
\[
\|\bm{\widehat{d}}\|_1=\|\widehat{\bm{d}}_{S(\eta)}\|_1+\|\widehat{\bm{d}}_{S^\complement(\eta)}\|_1 \leq   8\|\bm{g}_{S^\complement(\eta)}^*\|_1 + 8\|\widehat{\bm{d}}_{S(\eta)}\|_1 +  3\|\bm{g}_{\ma}^{*}\|_1 \| \bm{\widehat{\phi}}\|_2.
\]
Then, by the Cauchy-Schwarz inequalty, \eqref{eq:prop2}, \eqref{eq:ws1}, and \eqref{eq:ws3}, we can further show that
\begin{align}\label{eq:thm1eq4}
	\|\bm{\widehat{d}}\|_1^2 & \leq 3\left (64\|\bm{g}_{S^\complement(\eta)}^*\|_1^2 + 64\|\widehat{\bm{d}}_{S(\eta)}\|_1^2  + 9 \|\bm{g}_{\ma}^{*}\|_1^2 \| \bm{\widehat{\phi}}\|_2^2 \right ) \notag\\
	&\leq 192 \|\bm{g}_{S^\complement(\eta)}^*\|_1^2 + c_{\Delta}^{-2} \|\bm{\widehat{\Delta}}\|_{\Fr}^2 \left \{192 |S(\eta)| + 27 \underline{\alpha}_\ma^{-2} \|\bm{g}_{\ma}^{*}\|_1^2  \right \}\notag\\
	&\leq  192 \|\bm{g}_{S^\complement(\eta)}^*\|_1^2 + C_1 c_{\Delta}^{-2} R_q \eta^{-q} \|\bm{\widehat{\Delta}}\|_{\Fr}^2,
\end{align}
for an absolute constant $C_1>0$.
Similarly, from \eqref{eq:thm1eq2} and \eqref{eq:thm1eq7}, we can deduce that
\begin{align}\label{eq:thm1eq5}
	&\frac{2}{T}\sum_{t=1}^{T}\langle \bm{\varepsilon}_t, \bm{\widehat{\Delta}}\bm{x}_t \rangle +\lambda_g(\|\bm{g}^*\|_1-\|\widehat{\bm{g}}\|_1 )  + S_2(\bm{\widehat{\Delta}}) - S_1(\bm{\widehat{\Delta}}) \notag \\
	& \hspace{5mm}\leq \frac{\lambda_g}{4}\left (8\|\bm{g}_{S^\complement(\eta)}^*\|_1 + 8\|\widehat{\bm{d}}_{S(\eta)}\|_1 +  3\|\bm{g}_{\ma}^{*}\|_1 \| \bm{\widehat{\phi}}\|_2\right ) \notag\\
	&\hspace{5mm} \leq \frac{\lambda_g}{2} \left \{4\|\bm{g}_{S^\complement(\eta)}^*\|_1 + C_2 c_{\Delta}^{-1} R_q^{1/2} \eta^{-q/2} \|\bm{\widehat{\Delta}}\|_{\Fr} \right \},
\end{align}
for an absolute constant $C_2>0$.

By Lemmas \ref{lemma:rsclasso} and  \ref{lemma:init3}, we can show that 
\[
\frac{3}{4T}\sum_{t=1}^{T}\|\bm{\widehat{\Delta}}\bm{x}_t\|_{2}^2-S_3(\bm{\widehat{\Delta}}) \geq \frac{C_{\rsc} \kappa_1}{2}  \|\bm{\widehat{\Delta}}\|_{\Fr}^2 - \frac{\kappa_2}{T}\left \{ C_{\init3} (p\vee1) + \frac{3}{4} C_{\rsc} \frac{\kappa_2}{\kappa_1} \log \{N(p\vee1)\}\right \} \|\bm{\widehat{d}}\|_1^2.
\]
which, in conjunction with \eqref{eq:thm1eq4}, leads to
\begin{equation}\label{eq:thm1eq6}
	\frac{3}{4T}\sum_{t=1}^{T}\|\bm{\widehat{\Delta}}\bm{x}_t\|_{2}^2-S_3(\bm{\widehat{\Delta}}) \geq \frac{C_{\rsc} \kappa_1}{4}  \|\bm{\widehat{\Delta}}\|_{\Fr}^2 -  \frac{C_3\kappa_2^2 (p\vee1) \log \{N(p\vee1)\} }{\kappa_1 T}\|\bm{g}_{S^\complement(\eta)}^*\|_1^2, 
\end{equation}
where $C_3>0$ is an absolute constant, if we further have
%\begin{equation}\label{eq:Tcond3}
%T \gtrsim c_{\Delta}^{-2} R_q \eta^{-q} (\kappa_2/\kappa_1)^2 (p\vee1)  \log \{N(p\vee1)\}. \tag{T.3}
%\end{equation}
\begin{equation}\label{eq:Tcond3}
	T \gtrsim R_q \eta^{-q} (\kappa_2/\kappa_1)^2 (p\vee1)  \log \{N(p\vee1)\}. 
\end{equation}

Combining  \eqref{eq:thm1eq1}, \eqref{eq:thm1eq5}, and \eqref{eq:thm1eq6}, we have
\[
\frac{C_{\rsc} \kappa_1}{4}  \|\bm{\widehat{\Delta}}\|_{\Fr}^2 -  \frac{C_3\kappa_2^2(p\vee1) \log \{N(p\vee1)\} }{\kappa_1 T}\|\bm{g}_{S^\complement(\eta)}^*\|_1^2 
\leq  \frac{\lambda_g}{2} \left \{4\|\bm{g}_{S^\complement(\eta)}^*\|_1 + C_2 c_{\Delta}^{-1} R_q^{1/2} \eta^{-q/2} \|\bm{\widehat{\Delta}}\|_{\Fr} \right \}.
\]
Consider the following two cases.

\textit{Case (i):} First suppose that 
$\frac{C_{\rsc} \kappa_1}{8}  \|\bm{\widehat{\Delta}}\|_{\Fr}^2 \geq   \frac{C_3\kappa_2^2(p\vee1) \log \{N(p\vee1)\} }{\kappa_1 T}\|\bm{g}_{S^\complement(\eta)}^*\|_1^2$. 
Then 
\[
\frac{C_{\rsc} \kappa_1}{8}  \|\bm{\widehat{\Delta}}\|_{\Fr}^2 \leq \frac{\lambda_g}{2} \left \{4\|\bm{g}_{S^\complement(\eta)}^*\|_1 + C_2 c_{\Delta}^{-1} R_q^{1/2} \eta^{-q/2} \|\bm{\widehat{\Delta}}\|_{\Fr} \right \},
\]
which involves a quadratic form in $\|\bm{\widehat{\Delta}}\|_{\Fr}$. By computing the zeros of this quadratic form, we can show that
\[
\|\bm{\widehat{\Delta}}\|_{\Fr}^2 \leq \frac{32 C_2^2}{C_{\rsc}^2c_{\Delta}^2} \cdot \frac{\lambda_{g}^2R_q\eta^{-q}}{\kappa_1^2} + \frac{32}{C_{\rsc}} \cdot\frac{\lambda_{g} \|\bm{g}_{S^\complement(\eta)}^*\|_1}{\kappa_1}.
\]

\textit{Case (ii):} Otherwise, we must have $\frac{C_{\rsc} \kappa_1}{8}  \|\bm{\widehat{\Delta}}\|_{\Fr}^2 \leq   \frac{C_3\kappa_2^2(p\vee1) \log \{N(p\vee1)\} }{\kappa_1 T} \|\bm{g}_{S^\complement(\eta)}^*\|_1^2$.

Combining the two cases above, we can apply \eqref{eq:ws1} and \eqref{eq:Tcond3} to show  that
\begin{align*}
	\|\bm{\widehat{\Delta}}\|_{\Fr}^2 & \leq \frac{32 C_2^2}{C_{\rsc}^2c_{\Delta}^2} \cdot \frac{\lambda_{g}^2R_q\eta^{-q}}{\kappa_1^2} + \frac{32}{C_{\rsc}} \cdot\frac{\lambda_{g} \|\bm{g}_{S^\complement(\eta)}^*\|_1}{\kappa_1} + \frac{8C_3}{C_{\rsc}} \cdot \frac{\kappa_2^2 (p\vee1) \log \{N(p\vee1)\}}{\kappa_1^2 T} \|\bm{g}_{S^\complement(\eta)}^*\|_1^2\\
	&\leq \frac{32 C_2^2}{C_{\rsc}^2c_{\Delta}^2} \cdot \frac{\lambda_{g}^2R_q\eta^{-q}}{\kappa_1^2} + \frac{32}{C_{\rsc}} \cdot\frac{\lambda_{g} R_q \eta^{1-q}}{\kappa_1} +\frac{8C_3}{C_{\rsc}} \cdot (R_q \eta^{-q} )^{-1}(R_q \eta^{1-q})^2\\
	&\lesssim \left (\frac{\lambda_{g} }{\kappa_1} \right )^{2-q} R_q = \eta^{2-q} R_q,
\end{align*}
if we choose 
\[
\eta = \frac{\lambda_{g}}{\kappa_1}.
\]
Thus, taking $\lambda_{g}$ as its lower bound in \eqref{eq:lgcond1}, i.e., 
$\lambda_g \asymp \sqrt{\kappa_2\lambda_{\max}(\bm{\Sigma}_{\varepsilon})\log \{N(p\vee 1)\}/T}$, we have
\[
\|\bm{\widehat{\Delta}}\|_{\Fr}^2 \lesssim \left [\frac{\kappa_2 \lambda_{\max}(\bm{\Sigma}_{\varepsilon})\log \{N(p\vee 1)\}}{\kappa_1^2 T} \right ]^{1-q/2} R_q,
\]
and subsequently,
\[
\frac{1}{T}\sum_{t=1}^{T}\|\bm{\widehat{\Delta}}\bm{\widetilde{x}}_{t}\|_2^2 \lesssim \lambda_{g}  \eta^{1-q} R_q =\left [\frac{\kappa_2 \lambda_{\max}(\bm{\Sigma}_{\varepsilon})\log \{N(p\vee 1)\}}{\kappa_1^2 T} \right ]^{1-q/2}  \frac{R_q}{\kappa_1^{1-q}}, 
\]
where the latter follows from \eqref{eq:thm1eq} and \eqref{eq:thm1eq5}. 
On the one hand,  with the above choice of $\eta$, condition \eqref{eq:Tcond3} can be guaranteed if 
\begin{equation}\label{eq:Tcond}
	R_q \lesssim  \frac{\lambda_{\max}(\bm{\Sigma}_{\varepsilon})}{\kappa_2 (p\vee 1) } \cdot  \left [\frac{\kappa_1^2 T}{\kappa_2 \lambda_{\max}(\bm{\Sigma}_{\varepsilon}) \log \{N(p\vee1)\}} \right ]^{1-q/2}.
	%1\gtrsim \frac{\kappa_2 (p\vee 1) }{\lambda_{\max}(\bm{\Sigma}_{\varepsilon})} \cdot \left [\frac{\kappa_2 \lambda_{\max}(\bm{\Sigma}_{\varepsilon}) \log \{N(p\vee1)\}}{\kappa_1^2 T} \right ]^{1-q/2}  R_q.
\end{equation}
Under condition \eqref{eq:Tcond}, since  $r+2s \lesssim 1$, we can show that a sufficient condition for \eqref{eq:etacond} is
\begin{equation}\label{eq:Rqcond}
	\frac{\lambda_{\max}(\bm{\Sigma}_{\varepsilon})}{\kappa_2 (p\vee 1)} \lesssim  \overline{\alpha}_\ma^2  R_q/R_q^\ma.
\end{equation}
Finally, combining the tail probabilities in Lemmas \ref{lemma:devb}--\ref{lemma:init3} and the required conditions including \eqref{eq:Tcond} and \eqref{eq:Rqcond}, we accomplish the proof of this theorem.

%%%%%%%%%%%%%%%%%%%%%%%%%%%%%%%%%%%%%%%%%%%%%%%%%%%%%%%%%%%%%%%%%%%%%%%%%%%%%%%%%%%%%%%%%%%%%%%%%%

\section{Proofs of Proposition \ref{prop:perturbrow} and Theorem \ref{thm:lassorow}}\label{asec:RE}
% in Section \ref{sec:RE}
\subsection{Notations}\label{subsec:notationsrow}
%Note that 
%$\bm{a}_i= (\bm{L}(\bm{\omega})\otimes \bm{I}_N) \bm{g}_{i}$.
For $1\leq i\leq N$, denote $\bm{\delta}_i=\bm{a}_i-\bm{a}_i^*=(\bm{\delta}_{i,1}^\top,  \bm{\delta}_{i,2}^\top, \dots)^\top \in\mathbb{R}^{\infty}$ and $\bm{d}_i=\bm{g}_i-\bm{g}_i^*$, where $\bm{\delta}_{i,h}=\bm{a}_{i,h}-\bm{a}_{i,h}^{*} = \sum_{k=1}^{d} \ell_{h,k}(\bm{\omega}) \bm{g}_{i,k} - \sum_{k=1}^{d} \ell_{h,k}(\bm{\omega}^*) \bm{g}_{i,k}^{*}$ for $h\geq1$. Given $\bm{\omega}^*$ and $\bm{g}_i^*$, define
\begin{equation*}%\label{eq:Gstkrow}
	\bm{g}_{i,\rm{stack}}(\bm{\phi}, \bm{d}_i)=(\bm{d}_i^\top, (\bm{m}_i(\bm{\phi}))^\top )^\top \in\mathbb{R}^{N(d+r+2s)},
\end{equation*}
where $\bm{m}_i(\bm{\phi})\in\mathbb{R}^{N(r+2s)}$ is the following linear mapping of $\bm{\phi}$,
\begin{equation*}
	\bm{m}_i(\bm{\phi})=\left (\begin{matrix}
		(\lambda_1 - \lambda_1^*)\bm{g}_{i,p+1}^{*}\\
		\vdots\\
		(\lambda_r - \lambda_r^*)\bm{g}_{i,p+r}^{*}\\
		(\theta_1 - \theta_1^*)\bm{g}_{i,p+r+1}^{*}- \frac{\gamma_1 -\gamma_1^*}{\gamma_1^*}\bm{g}_{i,p+r+2}^{*}\\
		(\theta_1 - \theta_1^*) \bm{g}_{i,p+r+2}^{*} + \frac{\gamma_1 -\gamma_1^*}{\gamma_1^*}\bm{g}_{i,p+r+1}^{*}\\
		\vdots\\
		(\theta_s - \theta_s^*)\bm{g}_{i,p+r+2s-1}^{*}- \frac{\gamma_s -\gamma_s^*}{\gamma_s^*}\bm{g}_{i,p+r+2s}^{*}, \\
		(\theta_s - \theta_s^*) \bm{g}_{i,p+r+2s}^{*} + \frac{\gamma_s -\gamma_s^*}{\gamma_s^*}\bm{g}_{i,p+r+2s-1}^{*} 
	\end{matrix} \right ).
\end{equation*}
Note that $\bm{g}_{i,\rm{stack}}(\bm{\phi}, \bm{d}_i)$ and $\bm{m}_i(\bm{\phi})$ correspond to the $i$th row of $\bm{G}_{\rm{stack}}(\bm{\phi}, \bm{d}_i)$ and $\bm{M}(\bm{\phi})$, respectively; see Section \ref{subsec:notations}. In addition, for $1\leq i\leq N$, let  $\bm{\widehat{\delta}}_i=\bm{\widehat{a}}_i-\bm{a}_i^*$, where $\bm{\widehat{a}}_i=(\bm{\widehat{a}}_{i,1}^\top,  \bm{\widehat{a}}_{i,2}^\top, \dots)^\top\in\mathbb{R}^{\infty}$, 
$\bm{\widehat{d}}_i=\bm{\widehat{g}}_i-\bm{g}_i^*$, and $\bm{\widehat{\phi}}_i=\bm{\widehat{\omega}}_i-\bm{\omega}^*$. 

As will be shown in the proof of Theorem \ref{thm:lassorow}, the following terms quantify the effect of initializing $\bm{y}_s=\bm{0}$ for $s\leq 0$:
\begin{align}\label{eq:notation_initrow}
	\begin{split}
		&S_1(\bm{\delta}_i) = \frac{2}{T}\sum_{t=1}^{T}\langle \varepsilon_{i,t}, \sum_{h=t}^{\infty}\bm{\delta}_{i,h}^\top \bm{y}_{t-h} \rangle\\
		&S_2(\bm{\delta}_i)  = \frac{2}{T}\sum_{t=2}^{T}\langle \sum_{h=t}^{\infty}\bm{a}_{i,h}^{*\top} \bm{y}_{t-h}, \sum_{k=1}^{t-1}\bm{\delta}_{i,k}^\top \bm{y}_{t-k} \rangle\\
		& S_3(\bm{\delta}_i) = \frac{3}{T}\sum_{t=1}^{T}\Big (\sum_{k=t}^{\infty} \bm{\delta}_{i,k}^\top \bm{y}_{t-k} \Big )^2.
	\end{split}
\end{align}
Here we use the notations $S_i(\cdot)$'s for convenience, while their definitions in this section are different from those in \eqref{eq:notation_init}. 

\subsection{Preliminary results}\label{subsec:prelimrow}

The proofs of Proposition \ref{prop:perturbrow} and Theorem \ref{thm:lassorow} can be regarded as special cases of those of Proposition \ref{prop:perturb} and Theorem \ref{thm:lasso} with a univariate response variable. 

In Proposition \ref{prop:perturbrow}, the goal  is to establish the local linearity of $\bm{\delta}_i(\bm{\phi}, \bm{d})$ with respect to $\bm{\phi}$ and $\bm{d}_i$. That is,  within a local neighborhood of $\bm{\omega}^*$, we aim to show that
\begin{equation} \label{eq:linearizerow}
	\bm{\delta}_i(\bm{\phi}, \bm{d}_i) =\bm{a}_i(\bm{\omega}, \bm{g}_i)-\bm{a}_i^* \approx (\bm{L}_{\rm{stack}}(\bm{\omega}^*)\otimes \bm{I}_N) \bm{g}_{i,\rm{stack}}(\bm{\phi}, \bm{d}_i).
\end{equation}
Note that \eqref{eq:linearizerow} corresponds to the $i$th row of \eqref{eq:linearize}.

The proof of Theorem \ref{thm:lassorow} directly relies on Lemmas \ref{lemma:devbrow}--\ref{lemma:init3row} below. Their proofs  are straightforward univariate versions of those of Lemmas \ref{lemma:devb}--\ref{lemma:init3}, and hence are omitted.

%%%%%%%%%%%%%%%%%%%%%%%%%%%%%%%%%%%%%%%%%%%%%%%%%%%%%%%%%%%%%%%%%%%%%%%%%%%%%%%%%%
\begin{lemma}[Deviation bound]\label{lemma:devbrow}
	Under Assumptions \ref{assum:statn} and \ref{assum:error}, if  $\|\bm{\widehat{\omega}}_i - \bm{\omega}^*\|_{2}\leq c_{i,\bm{\omega}}$, $\sum_{j=0}^{\infty} \|\bm{\Psi}_j^* \|_{\op}^2<\infty$, and $T\gtrsim \log \{N(p\vee 1)\}$, then  with probability at least $1-C e^{-c\log N}$,
	\[
	\frac{1}{T} \left |\sum_{t=1}^{T}\langle \varepsilon_{i,t}, \bm{\widehat{\delta}}_i^\top \bm{x}_t \rangle \right | \leq C_{\dev} \sqrt{\frac{\kappa_2\lambda_{\max}(\bm{\Sigma}_{\varepsilon})\log \{N (p\vee 1)\}}{T}} \left (\|\bm{\widehat{d}}_i\|_1 +\|\bm{g}_{i, \ma}^{*}\|_1 \| \bm{\widehat{\phi}}_i\|_2 \right ), 
	\]
	where  $C_{\dev}>0$ is an absolute constant.
\end{lemma}

%%%%%%%%%%%%%%%%%%%%%%%%%%%%%%%%%%%%%%%%%%%%%%%%%%%%%%%%%%%%%%%%%%%%%%%%%%%%%%%%%%
\begin{lemma}[Restricted strong convexity] \label{lemma:rsc_lassorow}
	Under Assumptions \ref{assum:statn}--\ref{assum:error}, if $\|\bm{\widehat{\omega}}_i - \bm{\omega}^*\|_{2}\leq c_{i,\bm{\omega}}$ and $T\gtrsim(\kappa_2/\kappa_1)^2 \log\{ (\kappa_2/\kappa_1)(\overline{\alpha}_{i,\ma}/\underline{\alpha}_{i,\ma}) N (p\vee 1)\}$, then 	with probability at least $1-Ce^{-c\kappa_1^2 T/\kappa_2^2}$,
	\begin{equation*}%\label{eq:scs}
		\frac{1}{T}\sum_{t=1}^{T}(\bm{\widehat{\delta}}_i^\top \bm{x}_t)^2 \geq  C_{\rsc} \left [ \kappa_1 \|\bm{\widehat{\delta}}_i
		\|_{2}^2  - \frac{  \kappa_2^2 \log \{N(p\vee1)\} }{\kappa_1 T}  \|\bm{\widehat{d}}_i\|_1^2 \right ],
	\end{equation*}
	where $C_{\rsc}>0$ is an absolute constant.  
\end{lemma}

%%%%%%%%%%%%%%%%%%%%%%%%%%%%%%%%%%%%%%%%%%%%%%%%%%%%%%%%%%%%%%%%%%%%%%%%%%%%%%%%%%
\begin{lemma}[Effect of initial values I]\label{lemma:init1row}
	Under Assumptions \ref{assum:statn} and \ref{assum:error}, if  $\|\bm{\widehat{\omega}}_i - \bm{\omega}^*\|_{2}\leq c_{i,\bm{\omega}}$, $\sum_{j=0}^{\infty} \|\bm{\Psi}_j^* \|_{\op}^2<\infty$, and  $T\gtrsim \log N$, then with probability at least $1 - C(p\vee1) e^{-c\log N}$, 
	\[
	|S_1(\bm{\widehat{\delta}}_i
	)| \leq \frac{C_{\init1}\sqrt{\kappa_2\lambda_{\max}(\bm{\Sigma}_{\varepsilon})(p\vee1)\log N}}{T} \left (\|\bm{\widehat{d}}_i\|_1 + \|\bm{g}_{i,\ma}^{*}\|_1 \|\bm{\widehat{\phi}}_i\|_2 \right),
	\]
	where $C_{\init1}>0$ is an absolute constant.
\end{lemma}

%%%%%%%%%%%%%%%%%%%%%%%%%%%%%%%%%%%%%%%%%%%%%%%%%%%%%%%%%%%%%%%%%%%%%%%%%%%%%%%%%%
\begin{lemma}[Effect of initial values II]\label{lemma:init2row}
	Under Assumptions \ref{assum:statn}--\ref{assum:error}, if  $\|\bm{\widehat{\omega}}_i - \bm{\omega}^*\|_{2}\leq c_{i,\bm{\omega}}$ and $T\gtrsim \log \{N (p\vee1)\}$, then 
	with probability at least $1 - C(p\vee1) e^{-c\log \{N(p\vee1)\}}$, 
	\[
	|S_2(\bm{\widehat{\delta}}_i
	)| \leq \frac{C_{\init2} \kappa_2(p\vee1)^2}{T} \left (\|\bm{\widehat{d}}_i\|_1 + \|\bm{g}_{i,\ma}^{*}\|_1 \|\bm{\widehat{\phi}}_i\|_2 \right),
	\]
	where $C_{\init2}>0$ is an absolute constant.
\end{lemma}

%%%%%%%%%%%%%%%%%%%%%%%%%%%%%%%%%%%%%%%%%%%%%%%%%%%%%%%%%%%%%%%%%%%%%%%%%%%%%%%%%%
\begin{lemma}[Effect of initial values III] \label{lemma:init3row}
	Under Assumptions \ref{assum:statn}--\ref{assum:error}, if  $\|\bm{\widehat{\omega}}_i - \bm{\omega}^*\|_{2}\leq c_{i,\bm{\omega}}$ and $\log N\gtrsim  (\kappa_2/\kappa_1)^2$, then	with probability at least  $1-C  e^ {-c \kappa_1^2 (p\vee1) \log \{N(p\vee1)\}/\kappa_2^2}$,
	\begin{equation*}%\label{eq:scs}
		|S_{3}(\bm{\widehat{\delta}}_i
		)| \leq  \frac{ C_{\init3}  \kappa_2 (p\vee1)}{T}  \left [\|\bm{\widehat{\delta}}_i
		\|_{2}^2 \log \{N(p\vee1)\}+ \|\bm{\widehat{d}}_i\|_1^2 \right],
	\end{equation*}
	where $C_{\init3}>0$ is an absolute constant.
\end{lemma}
%%%%%%%%%%%%%%%%%%%%%%%%%%%%%%%%%%%%%%%%%%%%%%%%%%%%%%%%%%%%%%%%%%%%%%%%%%%%%%%%%%

%%%%%%%%%%%%%%%%%%%%%%%%%%%%%%%%%%%%%%%%%%%%%%%%%%%%%%%%%%%%%%%%%%%%%%%%%%%%%%%%%%%%%%%%%%%%%%%%%%
\subsection{Proof of Proposition \ref{prop:perturbrow}}\label{asec:prop2row}
Note that  $\bm{a}_{i,k}= \bm{g}_{i,k}$  for $1\leq k\leq p$, and
\begin{align*}
	\bm{a}_{i,p+h} =  \sum_{j=1}^{r}\ell_{h}^{I}(\lambda_j)\bm{g}_{i,p+j}+\sum_{m=1}^{s}\left\{\ell_{h}^{II,1}(\bm{\eta}_m)\bm{g}_{i,p+r+2m-1}+\ell_{h}^{II,2}(\bm{\eta}_m)\bm{g}_{i,p+r+2m}\right\},\quad\forall h\geq 1.
\end{align*}
Then $\bm{\delta}_{i,k}=\bm{g}_{i,k}-\bm{g}_{i,k}^{*}$ for $1\leq k\leq p$, and  by the Taylor expansion, for any $h\geq 1$, we have
\begin{align}\label{eq:deltarow}
	\bm{\delta}_{i,p+h}&=\bm{a}_{i,p+h}-\bm{a}_{i,p+h}^{*} \notag \\
	&=\sum_{j=1}^{r}\Bigg \{\ell_{h}^{I}(\lambda_j^*) +\nabla\ell_{h}^{I}(\lambda_j^*) (\lambda_j-\lambda_j^*) +\frac{1}{2}\nabla^2\ell_{h}^{I}(\widetilde{\lambda}_j) (\lambda_j-\lambda_j^*)^2 \Bigg \}\bm{g}_{i,p+j}\notag \\
	&\hspace{5mm} +\sum_{m=1}^{s}\Bigg \{\ell_{h}^{II,1}(\bm{\eta}_m^*) +(\bm{\eta}_m-\bm{\eta}_m^*)^\top \nabla \ell_{h}^{II,1}(\bm{\eta}_m^*) \notag\\
	&\hspace{33mm} +\frac{1}{2}(\bm{\eta}_m-\bm{\eta}_m^*)^{\top}\nabla^2 \ell_{h}^{II,1}(\widetilde{\bm{\eta}}_j)(\bm{\eta}_m-\bm{\eta}_m^*)\Bigg \}\bm{g}_{i,p+r+2m-1}\notag \\
	&\hspace{5mm} +\sum_{m=1}^{s}\Bigg \{\ell_{h}^{II,2}(\bm{\eta}_m^*) + (\bm{\eta}_m-\bm{\eta}_m^*)^\top \nabla \ell_{h}^{II,2}(\bm{\eta}_m^*) \notag\\ &\hspace{33mm}+\frac{1}{2}(\bm{\eta}_m-\bm{\eta}_m^*)^{\top}\nabla^2 \ell_{h}^{II,2}(\widetilde{\bm{\eta}}_j) (\bm{\eta}_m-\bm{\eta}_m^*)\Bigg \}\bm{g}_{i,p+r+2m} -\bm{a}_{i,p+h}^{*} \notag\\
	&:=\bm{h}_{i,h}+\bm{r}_{i,h},
\end{align}
where $\widetilde{\lambda}_j$ lies between $\lambda_j^*$ and $\lambda_j$ for $1 \leq j \leq r$, $\widetilde{\bm{\eta}}_j$ lies between $\bm{\eta}_m^*$ and $\bm{\eta}_m$ for $1 \leq m \leq s$, the first-order approximation is
\begin{align}\label{eq:Hhrow}
	\bm{h}_{i,h} &=\sum_{j=1}^{r}\ell_{h}^{I}(\lambda_j^*) (\bm{g}_{i,p+j}-\bm{g}_{i,p+j}^{*}) +\sum_{m=1}^{s}\sum_{\iota=1}^2\ell_{h}^{II,\iota}(\bm{\eta}_m^*) (\bm{g}_{i,p+r+2(m-1)+\iota}-\bm{g}_{i,p+r+2(m-1)+\iota}^{*}) \notag\\
	&\hspace{5mm} +\sum_{j=1}^{r}(\lambda_j-\lambda_j^*)\nabla\ell_{h}^{I}(\lambda_j^*) \bm{g}_{i,p+j}^{*} +\sum_{m=1}^{s}\sum_{\iota=1}^2(\bm{\eta}_m-\bm{\eta}_m^*)^\top \nabla \ell_{h}^{II,\iota}(\bm{\eta}_m^*)\bm{g}_{i,p+r+2(m-1)+\iota}^{*},
\end{align}
and the remainder is
\begin{align} \label{eq:Rhrow}
	\bm{r}_{i,h} &= \sum_{i=1}^{r}\nabla\ell_{h}^{I}(\lambda_j^*) (\lambda_j-\lambda_j^*) (\bm{g}_{i,p+j} - \bm{g}_{i,p+j}^{*}) \notag\\
	&\hspace{5mm}+ \sum_{m=1}^{s}\sum_{\iota=1}^2(\bm{\eta}_m-\bm{\eta}_m^*)^\top \nabla \ell_{h}^{II,\iota}(\bm{\eta}_m^*)(\bm{g}_{i,p+r+2(m-1)+\iota} - \bm{g}_{i,p+r+2(m-1)+\iota}^{*}) \notag\\
	&\hspace{5mm} +\frac{1}{2} \sum_{j=1}^{r}\nabla^2\ell_{h}^{I}(\widetilde{\lambda}_j) (\lambda_j-\lambda_j^*)^2 \bm{g}_{i,p+j} \notag\\
	&\hspace{5mm} 
	+\frac{1}{2} \sum_{m=1}^{s}\sum_{\iota=1}^2(\bm{\eta}_m-\bm{\eta}_m^*)^{\top}\nabla^2 \ell_{h}^{II,\iota}(\widetilde{\bm{\eta}}_j)(\bm{\eta}_m-\bm{\eta}_m^*)\bm{g}_{i,p+r+2(m-1)+\iota}.	
\end{align}
Here for notational simplicity, we have suppressed the dependence of  $\widetilde{\lambda}_j$'s and $\widetilde{\bm{\eta}}_j$'s on $i,h$.

We first consider $\bm{r}_{i,h}$. Denote $\bm{r}_{i,h}=\bm{r}_{i,1h} +\bm{r}_{i,2h}+\bm{r}_{i,3h}$, where
\begin{align}\label{eq:Rhsrow}
	\bm{r}_{i,1h}=& \sum_{j=1}^{r}\nabla\ell_{h}^{I}(\lambda_j^*)  (\lambda_j-\lambda_j^*) (\bm{g}_{i,p+j} - \bm{g}_{i,p+j}^{*}) \notag\\
	&
	+ \sum_{m=1}^{s}\sum_{\iota=1}^2(\bm{\eta}_m-\bm{\eta}_m^*)^\top \nabla \ell_{h}^{II,\iota}(\bm{\eta}_m^*)(\bm{g}_{i,p+r+2(m-1)+\iota} - \bm{g}_{i,p+r+2(m-1)+\iota}^{*}), \notag\\
	\bm{r}_{i,2h} =&\frac{1}{2}\sum_{j=1}^{r}\nabla^2\ell_{h}^{I}(\widetilde{\lambda}_j) (\lambda_j-\lambda_j^*)^2 (\bm{g}_{i,p+j} - \bm{g}_{i,p+j}^{*} ) \notag	\\
	&+\frac{1}{2} \sum_{m=1}^{s}\sum_{\iota=1}^2(\bm{\eta}_m-\bm{\eta}_m^*)^{\top}\nabla^2 \ell_{h}^{II,\iota}(\widetilde{\bm{\eta}}_j)(\bm{\eta}_m-\bm{\eta}_m^*)(\bm{g}_{i,p+r+2(m-1)+\iota} - \bm{g}_{i,p+r+2(m-1)+\iota}^{*}), \notag\\
	\bm{r}_{i,3h} =&\frac{1}{2}\sum_{j=1}^{r}\nabla^2\ell_{h}^{I}(\widetilde{\lambda}_j) (\lambda_j-\lambda_j^*)^2 \bm{g}_{i,p+j}^{*} \notag\\
	&+ \frac{1}{2} \sum_{m=1}^{s}\sum_{\iota=1}^2(\bm{\eta}_m-\bm{\eta}_m^*)^{\top}\nabla^2 \ell_{h}^{II,\iota}(\widetilde{\bm{\eta}}_j)(\bm{\eta}_m-\bm{\eta}_m^*) \bm{g}_{i,p+r+2(m-1)+\iota}^{*}.
\end{align}

Similar to the proof of Proposition \ref{prop:perturb}, by Lemma \ref{cor1}, we can show that
\begin{align*}%\label{eq:R1}
	\| \bm{r}_{i,1h}\|_{2}
	&\leq C_{\ell}\bar{\rho}^h \sqrt{\|\bm{\lambda} - \bm{\lambda}^*\|_2^2 + 2\|\bm{\eta} - \bm{\eta}^*\|_2^2} \\
	&\hspace{5mm}\cdot\sqrt{ \sum_{j=1}^{r}\|\bm{g}_{i,p+j} - \bm{g}_{i,p+j}^{*}\|_{2}^2 + \sum_{m=1}^{s}\sum_{\iota=1}^2 \|\bm{g}_{i,p+r+2(m-1)+\iota} - \bm{g}_{i,p+r+2(m-1)+\iota}^{*}\|_{2}^2 } \notag \\
	%	 \sqrt{\sum_{j=1}^{r}\|\bm{g}_{i,p+j} - \bm{g}_{i,p+j}^{*} \|_{2}^2 +  \sum_{m=1}^{s}\sum_{\iota=1}^2\|\bm{g}_{i,p+r+2(m-1)+\iota} - \bm{g}_{i,p+r+2(m-1)+\iota}^{*}\|_{2}^2} 
	%	 \notag \\
	&\leq \sqrt{2}C_{\ell}\bar{\rho}^h \|\bm{\phi}\|_2 \cdot \|\bm{g}_{i,\ma} - \bm{g}_{i,\ma}^{*}\|_{2} \leq \sqrt{2}C_{\ell}\bar{\rho}^h \|\bm{\phi}\|_2 \|\bm{d}_i\|_{2},
\end{align*}
and similarly,
\begin{align*}%\label{eq:R2}
	\| \bm{r}_{i,2h}\|_{2}
	&\leq \frac{\sqrt{2}}{2} C_{\ell}\bar{\rho}^h \|\bm{\phi}\|_2^2 \cdot \|\bm{g}_{i,\ma} - \bm{g}_{i,\ma}^{*}\|_{2} \leq \frac{\sqrt{2}}{2} C_{\ell}\bar{\rho}^h \|\bm{\phi}\|_2^2 \|\bm{d}_i\|_{2}.
\end{align*}
Moreover, by Lemma \ref{cor1} again, we can show that
\begin{equation*}%\label{eq:R3}
	\| \bm{r}_{i,3h}\|_{2} \leq  \frac{\sqrt{2}}{2} C_{\ell}\overline{\alpha}_{i,\ma} \bar{\rho}^h \|\bm{\phi}\|_2^2.
\end{equation*}
As a result,
\begin{align} \label{eq:Rnorm1row}
	\|\bm{r}_{i,h}\|_{2} &\leq \|\bm{r}_{i,1h}\|_{2} + \|\bm{r}_{i,2h}\|_{2} + \|\bm{r}_{i,3h}\|_{2} \notag\\ &\leq    C_{\ell} \bar{\rho}^h \|\bm{\phi}\|_2 \left ( \sqrt{2}  \|\bm{d}_i\|_{2} + \frac{\sqrt{2}}{2} \|\bm{\phi}\|_2\|\bm{d}_i\|_{2}+ \frac{\sqrt{2}}{2}  \overline{\alpha}_{i,\ma} \|\bm{\phi}\|_2 \right ).
\end{align}

Now consider $\bm{h}_{i,h}$ in \eqref{eq:Hhrow}. 
Notice that for any $h\geq1$ and $1\leq j\leq s$, 
\begin{align*}
	&\nabla_\gamma\ell_{h}^{II,1}(\bm{\eta}_m)=h\gamma_m^{h-1}\cos(h\theta_m)=\frac{1}{\gamma_m}\nabla_\theta\ell_{h}^{II,2}(\bm{\eta}_m),\\
	&\nabla_\gamma\ell_{h}^{II,2}(\bm{\eta}_m)=h\gamma_m^{h-1}\sin(h\theta_m)=-\frac{1}{\gamma_m}\nabla_\theta\ell_{h}^{II,1}(\bm{\eta}_m).
\end{align*}
Thus, the last term on the right side of  \eqref{eq:Hhrow} can be simplified to 
\begin{align} \label{eq:linearcombrow}
	&\sum_{m=1}^{s}\sum_{\iota=1}^2(\bm{\eta}_m-\bm{\eta}_m^*)^\top \nabla \ell_{h}^{II,\iota}(\bm{\eta}_m^*)\bm{g}_{i,p+r+2(m-1)+\iota}^{*} \notag\\
	&\hspace{5mm}= \sum_{m=1}^{s}\left[ (\theta_m - \theta_m^*) \bm{g}_{i,p+r+2m-1}^{*} - \frac{1}{\gamma_m^*}(\gamma_m -\gamma_m^*)\bm{g}_{i,p+r+2m}^{*}\right]\nabla_{\theta} \ell_{h}^{II,1}(\bm{\eta}_m^*) \notag\\
	&\hspace{10mm}+\sum_{m=1}^{s} \left[ (\theta_m - \theta_m^*) \bm{g}_{i,p+r+2m}^{*} + \frac{1}{\gamma_m^*}(\gamma_m -\gamma_m^*)\bm{g}_{i,p+r+2m-1}^{*}\right]\nabla_{\theta} \ell_{h}^{II,2}(\bm{\eta}_m^*).
\end{align}
Let $\bm{h}_i=(\bm{h}_{i,1}^\top, \bm{h}_{i,2}^\top, \dots)^\top$ and  $\bm{r}_i=(\bm{r}_{i,1}^\top, \bm{r}_{i,2}^\top, \dots)^\top$. Then by \eqref{eq:Hhrow} and \eqref{eq:linearcombrow} it can be verified that 
\begin{align}\label{eq:stacjHrow}
	\bm{\widetilde{h}}_i:= ((\bm{g}_{i,1}-\bm{g}_{i,1}^{*})^\top,\cdots,(\bm{g}_{i,p}-\bm{g}_{i,p}^{*})^\top,\bm{h}_i^\top)^\top&=(\bm{L}(\bm{\omega}^*) \otimes \bm{I}_N)\bm{d}_i + (\bm{P}(\bm{\omega}^*)\otimes \bm{I}_N)\bm{m}_i(\bm{\phi})\notag\\
	&=(\bm{L}_{\rm{stacj}}(\bm{\omega}^*)\otimes \bm{I}_N)\bm{g}_{i,\rm{stacj}}(\bm{\phi},\bm{d}_i).
\end{align}

Note that 
\begin{equation}\label{eq:Deltarow}
	\bm{\delta}_i=\bm{\widetilde{h}}_i+
	\left (\begin{matrix}
		\bm{0}_{Np}\\
		\bm{r}_i
	\end{matrix}\right )
\end{equation}
Moreover,
\begin{align*}%\label{eq:DFr2}
	\|\bm{m}_i(\bm{\phi})\|_{2}^2 &= \sum_{j=1}^{r}(\lambda_j - \lambda_j^*)^2\|\bm{g}_{i,p+j}^{*}\|_{2}^2 + \sum_{m=1}^{s} \left\| (\theta_m - \theta_m^*) \bm{g}_{i,p+r+2m-1}^{*} - \frac{\gamma_m -\gamma_m^*}{\gamma_m^*}\bm{g}_{i,p+r+2m}^{*}\right\|_{2}^2 \notag\\
	&\hspace{5mm}+ \sum_{m=1}^{s}\left\| (\theta_m - \theta_m^*) \bm{g}_{i,p+r+2m}^{*} + \frac{\gamma_m -\gamma_m^*}{\gamma_m^*}\bm{g}_{i,p+r+2m-1}^{*}\right\|_{2}^2 \notag\\
	&= \sum_{j=1}^{r}(\lambda_j - \lambda_j^*)^2\|\bm{g}_{i,p+j}^{*}\|_{2}^2 + \sum_{m=1}^{s} (\theta_m - \theta_m^*)^2(\|\bm{g}_{i,p+r+2m-1}^{*}\|_{2}^2 + \|\bm{g}_{i,p+r+2m}^{*}\|_{2}^2 ) \notag\\
	&\hspace{5mm}+ \sum_{m=1}^{s}\frac{(\gamma_m - \gamma_m^*)^2}{\gamma_m^{*2}} (\|\bm{g}_{i,p+r+2m-1}^{*}\|_{2}^2 + \|\bm{g}_{i,p+r+2m}^{*}\|_{2}^2 ),
\end{align*}
which leads to
\begin{equation}\label{eq:DFrrow}
	\underline{\alpha}_{i,\ma} \|\bm{\phi}\|_2  \leq \|\bm{m}_i(\bm{\phi})\|_{2} \leq \frac{\sqrt{2}\overline{\alpha}_{i,\ma}}{\min_{1\leq j\leq s}\gamma_m^*} \|\bm{\phi}\|_2.
\end{equation}
By the simple inequalities $ (|x| + |y|) / 2 \leq \sqrt{x^2 + y^2} \leq |x| + |y|$, we have
$0.5(\|\bm{d}_i\|_{2} + \|\bm{m}_i(\bm{\phi})\|_{2}) \leq  \|\bm{g}_{i,\rm{stacj}}(\bm{\phi},\bm{d}_i)\|_{2}	\leq \|\bm{d}_i\|_{2} + \|\bm{m}_i(\bm{\phi})\|_{2}$,
%\end{equation*} 
and thus in view of \eqref{eq:DFrrow} we further have
\begin{equation}\label{eq:Gstacjnormrow}
	\frac{1}{2}(\|\bm{d}_i\|_{2} +\underline{\alpha}_{i,\ma}\|\bm{\phi}\|_2) \leq  \|\bm{g}_{i,\rm{stacj}}(\bm{\phi},\bm{d}_i)\|_{2}	\leq \|\bm{d}_i\|_{2} + \frac{\sqrt{2}\overline{\alpha}_{i,\ma}}{\min_{1\leq j\leq s}\gamma_m^*} \|\bm{\phi}\|_2.
\end{equation} 
%By Lemma \ref{lemma:fullranj}, $\sigma_{\min, L}=\sigma_{\min}(\bm{L}_{\rm{stacj}}(\bm{\omega}^*))>0$. 
Then  it follows from \eqref{eq:Gstacjnormrow} that
\[
\frac{\sigma_{\min, L}}{2} (\|\bm{d}_i\|_{2} +\underline{\alpha}_{i,\ma} \|\bm{\phi}\|_2 )\leq \|\bm{\widetilde{h}}\|_{2} \leq  \sigma_{\max, L}\left (\|\bm{d}_i\|_{2} + \frac{\sqrt{2}\overline{\alpha}_{i,\ma}}{\min_{1\leq j\leq s}\gamma_m^*}\|\bm{\phi}\|_2\right ).
\]
Combining this with \eqref{eq:Rnorm1row}, \eqref{eq:Deltarow}, \eqref{eq:DFrrow}, as well as the fact that $\|\bm{g}_{i,\ma} - \bm{g}_{i,\ma}^{*}\|_{2}\leq \|\bm{d}_i\|_{2}$, we have
\begin{align*}
	\|\bm{\delta}_i\|_{2} &\leq \|\bm{\widetilde{h}}_i\|_{2} + \|\bm{r}_i\|_{2} \\
	&\leq \left\{\sigma_{\max, L} + \frac{\sqrt{2}C_{\ell}}{1-\bar{\rho}} \left (\|\bm{\phi}\|_2+\frac{\|\bm{\phi}\|_2^2}{2}\right )\right\} \|\bm{d}_i\|_{2}
	+\left(\frac{\sqrt{2}\overline{\alpha}_{i,\ma} \sigma_{\max, L}}{\min_{1\leq j\leq s}\gamma_m^*} + \frac{\sqrt{2}}{2}\cdot \frac{C_{\ell}\overline{\alpha}_{i,\ma}}{1-\bar{\rho}} \|\bm{\phi}\|_2\right)  \|\bm{\phi}\|_2
\end{align*}
and
\begin{align*}
	\|\bm{\delta}_i\|_{2} &\geq \|\bm{\widetilde{h}}_i\|_{2} - \|\bm{r}_i\|_{2}\\
	&\geq \left\{\frac{\sigma_{\min, L}}{2}-\frac{\sqrt{2}C_{\ell}}{1-\bar{\rho}}  \left (\|\bm{\phi}\|_2+\frac{\|\bm{\phi}\|_2^2}{2}\right )\right\} \|\bm{d}_i\|_{2} +\left( \frac{\underline{\alpha}_{i,\ma} \sigma_{\min, L}}{2}- \frac{\sqrt{2}}{2}\cdot\frac{C_{\ell}\overline{\alpha}_{i,\ma}}{1-\bar{\rho}} \|\bm{\phi}\|_2\right)  \|\bm{\phi}\|_2. 
\end{align*}
Thus, as long as
\begin{equation}\label{eq:comegarow}
	\|\bm{\phi}\|_2\leq c_{i,\bm{\omega}}\leq  \min\left \{2, \frac{\underline{\alpha}_{i,\ma}(1-\bar{\rho})\sigma_{\min, L} }{8\sqrt{2}C_{\ell}\overline{\alpha}_{i,\ma} }\right \},
\end{equation}
we have
\begin{equation}\label{eq:prop2row}
	c_{\Delta}	\left(\|\bm{d}_i\|_{2} +\|\bm{\phi}\|_2\right) 
	\leq \|\bm{\delta}_i\|_{2} \leq         
	C_{\Delta}	\left(\|\bm{d}_i\|_{2} +  \|\bm{\phi}\|_2\right),
\end{equation}
where $c_{\Delta}$ and $C_{\Delta}$ are absolute constants defined as in the proof of Proposition \ref{prop:perturb}. By Lemma  \ref{lemma:fullrank}, \eqref{eq:comegarow} is fulfilled by taking 
\begin{equation}\label{eq:comgrow}
	c_{i,\bm{\omega}} =  \min\left \{2, \frac{\underline{\alpha}_{i,\ma}(1-\bar{\rho})(1\wedge c_{\bar{\rho}}) }{8\sqrt{2}C_{\ell}\overline{\alpha}_{i,\ma}}\right \}.
\end{equation}
The proof of this proposition is complete.

%%%%%%%%%%%%%%%%%%%%%%%%%%%%%%%%%%%%%%%%%%%%%%%%%%%%%%%%%%%%%%%%%%%%%%%%%%%%%%%%%%
\subsection{Proof of Theorem \ref{thm:lassorow}}
The proof of  this theorem  closely mirrors  that of Theorem \ref{thm:lasso}.
Note that $\sum_{h=1}^{t-1}\bm{a}_{i,h}^\top\bm{y}_{t-h}= \bm{a}_{i}^\top\bm{\widetilde{x}}_{t}$, where $\bm{\widetilde{x}}_{t}= (\bm{y}_{t-1}^\top,\dots,\bm{y}_1^\top,0,0,\dots)^\top$ is the initialized version of $\bm{x}_t$. By the optimality of $\bm{\widehat{a}}_{i}$, we have
\[
\frac{1}{T}\sum_{t=1}^{T} ( y_{i,t} - \bm{a}_i^{*\top}\bm{\widetilde{x}}_{t} - \bm{\widehat{\delta}}_i^\top \bm{\widetilde{x}}_{t})^2
\leq \frac{1}{T}\sum_{t=1}^{T} ( y_{i,t} - \bm{a}_i^{*\top}\bm{\widetilde{x}}_{t})^2+\lambda_g(\|\bm{g}_i^*\|_1-\|\bm{\widehat{g}}_i\|_1),
\]
Then, since $y_{i,t} - \bm{a}_i^{*\top}\bm{\widetilde{x}}_{t}=\varepsilon_{i,t} +   \sum_{h=t}^{\infty}\bm{a}_{i,h}^{*\top} \bm{y}_{t-h}$ and $\bm{\widehat{\delta}}_i^\top\bm{\widetilde{x}}_{t}=\bm{\widehat{\delta}}_i^\top \bm{x}_{t}-\sum_{k=t}^{\infty}\bm{\widehat{\delta}}_{i,k}^\top \bm{y}_{t-k}$, we have
\begin{align}\label{eq:thm1eqrow}
	\frac{1}{T}\sum_{t=1}^{T}(\bm{\widehat{\delta}}_i^\top\bm{\widetilde{x}}_{t})^2 &\leq 
	\frac{2}{T}\sum_{t=1}^{T}\langle \varepsilon_{i,t}, \bm{\widehat{\delta}}_i^\top\bm{\widetilde{x}}_{t} \rangle + \underbrace{\frac{2}{T}\sum_{t=1}^{T}\langle \sum_{h=t}^{\infty}\bm{a}_{i,h}^{*\top} \bm{y}_{t-h}, \bm{\widehat{\delta}}_i^\top\bm{\widetilde{x}}_{t}  \rangle}_{S_{2}(\bm{\widehat{\delta}}_i)} +\lambda_g(\|\bm{g}_i^*\|_1-\|\bm{\widehat{g}}_i\|_1 ) \notag \\
	&= \frac{2}{T}\sum_{t=1}^{T}\langle \varepsilon_{i,t}, \bm{\widehat{\delta}}_i^\top\bm{x}_t \rangle +\lambda_g(\|\bm{g}_i^*\|_1-\|\bm{\widehat{g}}_i\|_1 )+S_{2}(\bm{\widehat{\delta}}_i) - S_1(\bm{\widehat{\delta}}_i),
\end{align}
where $S_1(\cdot)$ and $S_2(\cdot)$ are defined as in \eqref{eq:notation_initrow}.
Moreover, similar to \eqref{eq:thm1eq1}, we can lower bound the left-hand side of  \eqref{eq:thm1eqrow} to further obtain that
\begin{align} \label{eq:thm1eq1row}
	\frac{3}{4T}\sum_{t=1}^{T}(\bm{\widehat{\delta}}_i^\top\bm{x}_t)^2-S_3(\bm{\widehat{\delta}}_i)
	&\leq \frac{2}{T}\sum_{t=1}^{T}\langle \varepsilon_{i,t}, \bm{\widehat{\delta}}_i^\top\bm{x}_t \rangle+ \lambda_g(\|\bm{g}_i^*\|_1-\|\bm{\widehat{g}}_i\|_1 )+ S_2(\bm{\widehat{\delta}}_i) - S_1(\bm{\widehat{\delta}}_i),
\end{align}
where $S_3(\cdot)$ is defined as in \eqref{eq:notation_initrow}.

%Without loss of generality, we assume that $p\geq 1$ throughout our proof. 
Next we assume that the high probability events in Lemmas \ref{lemma:devbrow}--\ref{lemma:init3row} all hold and focus on the deterministic analysis. For a threshold $\eta>0$ to be chosen later, define the thresholded subsets
\begin{align*}
	S_{i,\ar}(\eta)&=\{ (j,k)\mid  |g_{i,j,k}^*|>\eta, j\in\{1,\dots, N\}, k\in\{1,\dots, p\}\},\\
	S_{i,\ma}(\eta)&=\{ (j,k)\mid  |g_{i,j,k}^*|>\eta, j\in\{1,\dots, N\}, k\in\{p+1,\dots, d\}\},
\end{align*}
and
\[
S_i(\eta)=S_{i,\ar}(\eta) \cup S_{i,\ma}(\eta)=\{ (j,k)\mid  |g_{i,j,k}^*|>\eta, j\in\{1,\dots, N\}, k\in\{1,\dots, d\}\}.
\]
Define  $S_i^\complement(\eta)=\{ (j,k)\mid j\in\{1,\dots, N\}, k\in\{1,\dots, d\}\}\setminus S_i(\eta)$ as the complementary set of $S_i(\eta)$. Similarly, the complementary set of $S_{i,\ma}(\eta)$ is  $S_{i,\ma}^\complement(\eta)=\{ (j,k)\mid j\in\{1,\dots, N\}, k\in\{p+1,\dots, d\}\}\setminus S_{i,\ma}(\eta)$. 

Note that 
\[
R_{i,q}\geq \sum_{j=1}^{N}\sum_{k=1}^{d} |g_{i,j,k}^*|^q \geq \sum_{(j,k)\in S_i(\eta)}  |g_{i,j,k}^*|^q \geq \eta^q |S_i(\eta)|,
\]
and
\begin{equation*}
	\|(\bm{g}_i^*)_{S_i^\complement(\eta)}\|_1 = \sum_{(j,k)\in S_i^\complement(\eta)} |g_{i,j,k}^*| =  \sum_{(j,k)\in S_i^\complement(\eta)} |g_{i,j,k}^*|^q |g_{i,j,k}^*|^{1-q}.
\end{equation*} 
Thus, we have
\begin{equation}\label{eq:ws1row}
	|S_i(\eta)| \leq R_{i,q} \eta^{-q} \quad\text{and} \quad 		\|(\bm{g}_i^*)_{S_i^\complement(\eta)}\|_1 \leq R_{i,q} \eta^{1-q}.
\end{equation}
Similarly, we can show that 
\begin{equation}\label{eq:ws2row}
	|S_{i,\ma}(\eta)| \leq R_{i,q}^\ma \eta^{-q} \quad\text{and}\quad \|(\bm{g}^*_{i,\ma})_{ S_{i,\ma}^\complement(\eta)}\|_1 \leq R_{i,q}^\ma \eta^{1-q}. 
\end{equation}

By \eqref{eq:ws2row}, by choosing $\eta$ such  that
\begin{equation}\label{eq:etacondrow}
	\eta^{2-q} \leq \frac{(r+2s) \overline{\alpha}_{i,\ma}^2}{R_{i,q}^\ma},
\end{equation}
we have
\begin{align*}
	\|\bm{g}_{i,\ma}^{*}\|_1^2 \leq 2\|(\bm{g}_{i,\ma}^{*})_{S_{i,\ma}(\eta)}\|_1^2 +  2\|(\bm{g}_{i,\ma}^{*})_{S_{i,\ma}^\complement(\eta)}\|_1^2 & \leq 2 |S_{i,\ma}(\eta)| \|\bm{g}_{i,\ma}^{*}\|_2^2 + 2 (R_{i,q}^\ma \eta^{1-q})^2\notag \\
	&\leq 2 R_{i,q}^\ma  \eta^{-q} \left \{ (r+2s) \overline{\alpha}_{i,\ma}^2 + R_{i,q}^\ma \eta^{2-q} \right \} \notag\\
	& \leq 4 R_{i,q}^\ma  \eta^{-q} (r+2s) \overline{\alpha}_{i,\ma}^2.
\end{align*}
Then, since $r+2s\lesssim1$ and $(\overline{\alpha}_{i,\ma}/\underline{\alpha}_{i,\ma})^2\lesssim R_{i,q}/R_{i,q}^\ma$,  we further have 
\begin{equation}\label{eq:ws3row}
	\underline{\alpha}_{i,\ma}^{-2} \|\bm{g}_{i,\ma}^{*}\|_1^2 \lesssim R_{i,q} \eta^{-q}.
\end{equation}
%	\begin{equation}\label{eq:etacond1}
	%	\eta \leq \left \{ \frac{(r+2s) \overline{\alpha}_{i,\ma}^2}{R_{i,q}^\ma} \right \}^{\frac{1}{2-q}}
	%	\end{equation}

Consider the right-hand side of \eqref{eq:thm1eqrow}. By Lemma \ref{lemma:devbrow}, if we choose $\lambda_g$ such that 
\begin{equation}\label{eq:lgcond1row}
	\frac{\lambda_g}{4}\geq   C_{\dev} \sqrt{\frac{\kappa_2\lambda_{\max}(\bm{\Sigma}_{\varepsilon})\log \{N(p\vee 1)\}}{T}},
\end{equation}
then we can show that 
\begin{align}\label{eq:thm1eq2row}
	& \frac{2}{T}\sum_{t=1}^{T}\langle \varepsilon_{i,t}, \bm{\widehat{\delta}}_i^\top\bm{x}_t \rangle+\lambda_g(\|\bm{g}_i^*\|_1-\|\bm{\widehat{g}}_i\|_1 ) \notag\\
	&\hspace{5mm} \leq  \frac{\lambda_g}{2}(\|\bm{\widehat{d}}_i\|_1 +\|\bm{g}_{i,\ma}^{*}\|_1 \| \bm{\widehat{\phi}}_i\|_2)+ \lambda_g(\|\bm{g}_i^*\|_1-\|\bm{g}_{S_i(\eta)}^*+(\widehat{\bm{d}}_i)_{S_i^\complement(\eta)}\|_1 + \|(\bm{g}_i^*)_{S_i^\complement(\eta)}+(\widehat{\bm{d}}_i)_{S_i(\eta)}\|_1 ) \notag\\
	&\hspace{5mm} \leq \frac{\lambda_g}{2}(\|(\widehat{\bm{d}}_i)_{S_i(\eta)}\|_1+\|(\widehat{\bm{d}}_i)_{S_i^\complement(\eta)}\|_1 +\|\bm{g}_{i,\ma}^{*}\|_1 \| \bm{\widehat{\phi}}_i\|_2 ) +\lambda_g(2\|(\bm{g}_i^*)_{S_i^\complement(\eta)}\|_1+ \|(\widehat{\bm{d}}_i)_{S_i(\eta)}\|_1 -\|(\widehat{\bm{d}}_i)_{S_i^\complement(\eta)}\|_1) \notag\\
	&\hspace{5mm} \leq \frac{\lambda_g}{2}\left (4\|(\bm{g}_i^*)_{S_i^\complement(\eta)}\|_1 + 3\|(\widehat{\bm{d}}_i)_{S_i(\eta)}\|_1-\|(\widehat{\bm{d}}_i)_{S_i^\complement(\eta)}\|_1 +  \|\bm{g}_{i,\ma}^{*}\|_1 \| \bm{\widehat{\phi}}_i\|_2\right ).
\end{align}
In addition, since $T\gtrsim \kappa_2 (p\vee1)^4$, it follows from  Lemmas \ref{lemma:init1row} and \ref{lemma:init2row} that  
\begin{align}\label{eq:thm1eq7row}
	S_2(\bm{\widehat{\delta}}_i) - S_1(\bm{\widehat{\delta}}_i) & \leq \frac{\lambda_{g}}{4}\left (\|\bm{\widehat{d}}_i\|_1 + \|\bm{g}_{i,\ma}^{*}\|_1 \|\bm{\widehat{\phi}}_i\|_2 \right) \notag\\
	& =\frac{\lambda_{g}}{4}\left (\|(\widehat{\bm{d}}_i)_{S_i(\eta)}\|_1+\|(\widehat{\bm{d}}_i)_{S_i^\complement(\eta)}\|_1 + \|\bm{g}_{i,\ma}^{*}\|_1 \|\bm{\widehat{\phi}}_i\|_2\right).
\end{align}
Combining \eqref{eq:thm1eqrow}, \eqref{eq:thm1eq2row} and \eqref{eq:thm1eq7row}, we have
\begin{equation*}%\label{eq:thm1eq6}
	0\leq \frac{1}{T}\sum_{t=1}^{T}(\bm{\widehat{\delta}}_i^\top\bm{\widetilde{x}}_t)^2 
	\leq \frac{\lambda_{g}}{4}\left (8\|(\bm{g}_i^*)_{S_i^\complement(\eta)}\|_1 + 7\|(\widehat{\bm{d}}_i)_{S_i(\eta)}\|_1-\|(\widehat{\bm{d}}_i)_{S_i^\complement(\eta)}\|_1 +  3\|\bm{g}_{i,\ma}^{*}\|_1 \| \bm{\widehat{\phi}}_i\|_2\right),
\end{equation*}
which implies 
\[
\|\bm{\widehat{d}}_i\|_1=\|(\widehat{\bm{d}}_i)_{S_i(\eta)}\|_1+\|(\widehat{\bm{d}}_i)_{S_i^\complement(\eta)}\|_1 \leq   8\|(\bm{g}_i^*)_{S_i^\complement(\eta)}\|_1 + 8\|(\widehat{\bm{d}}_i)_{S_i(\eta)}\|_1 +  3\|\bm{g}_{i,\ma}^{*}\|_1 \| \bm{\widehat{\phi}}_i\|_2.
\]
Then, by the Cauchy-Schwarz inequalty, \eqref{eq:prop2row}, \eqref{eq:ws1row}, and \eqref{eq:ws3row}, we can further show that
\begin{align}\label{eq:thm1eq4row}
	\|\bm{\widehat{d}}_i\|_1^2 & \leq 3\left (64\|(\bm{g}_i^*)_{S_i^\complement(\eta)}\|_1^2 + 64\|(\widehat{\bm{d}}_i)_{S_i(\eta)}\|_1^2  + 9 \|\bm{g}_{i,\ma}^{*}\|_1^2 \| \bm{\widehat{\phi}}_i\|_2^2 \right ) \notag\\
	&\leq 192 \|(\bm{g}_i^*)_{S_i^\complement(\eta)}\|_1^2 + c_{\Delta}^{-2} \|\bm{\widehat{\delta}}_i\|_{2}^2 \left \{192 |S_i(\eta)| + 27 \underline{\alpha}_{i,\ma}^{-2} \|\bm{g}_{i,\ma}^{*}\|_1^2  \right \}\notag\\
	&\leq  192 \|(\bm{g}_i^*)_{S_i^\complement(\eta)}\|_1^2 + C_1 c_{\Delta}^{-2} R_{i,q} \eta^{-q} \|\bm{\widehat{\delta}}_i\|_{2}^2,
\end{align}
for an absolute constant $C_1>0$.
Similarly, from \eqref{eq:thm1eq2row} and \eqref{eq:thm1eq7row}, we can deduce that
\begin{align}\label{eq:thm1eq5row}
	& \frac{2}{T}\sum_{t=1}^{T}\langle \varepsilon_{i,t}, \bm{\widehat{\delta}}_i^\top\bm{x}_t \rangle+\lambda_g(\|\bm{g}_i^*\|_1-\|\bm{\widehat{g}}_i\|_1 ) + S_2(\bm{\widehat{\delta}}_i) - S_1(\bm{\widehat{\delta}}_i) \notag \\
	& \hspace{5mm}\leq \frac{\lambda_g}{4}\left (8\|(\bm{g}_i^*)_{S_i^\complement(\eta)}\|_1 + 8\|(\widehat{\bm{d}}_i)_{S_i(\eta)}\|_1 +  3\|\bm{g}_{i,\ma}^{*}\|_1 \| \bm{\widehat{\phi}}_i\|_2\right ) \notag\\
	&\hspace{5mm} \leq \frac{\lambda_g}{2} \left \{4\|(\bm{g}_i^*)_{S_i^\complement(\eta)}\|_1 + C_2 c_{\Delta}^{-1} R_{i,q}^{1/2} \eta^{-q/2} \|\bm{\widehat{\delta}}_i\|_{2} \right \},
\end{align}
for an absolute constant $C_2>0$.

By Lemmas \ref{lemma:rsc_lassorow} and  \ref{lemma:init3row}, we can show that 
\[
\frac{3}{4T}\sum_{t=1}^{T}(\bm{\widehat{\delta}}_i^\top\bm{x}_t)^2-S_3(\bm{\widehat{\delta}}_i) \geq \frac{C_{\rsc} \kappa_1}{2}  \|\bm{\widehat{\delta}}_i\|_{2}^2 - \frac{\kappa_2}{T}\left \{ C_{\init3} (p\vee1) + \frac{3}{4} C_{\rsc} \frac{\kappa_2}{\kappa_1} \log \{N(p\vee1)\}\right \} \|\bm{\widehat{d}}_i\|_1^2.
\]
which, in conjunction with \eqref{eq:thm1eq4row}, leads to
\begin{equation}\label{eq:thm1eq6row}
	\frac{3}{4T}\sum_{t=1}^{T}(\bm{\widehat{\delta}}_i^\top\bm{x}_t)^2-S_3(\bm{\widehat{\delta}}_i) \geq \frac{C_{\rsc} \kappa_1}{4}  \|\bm{\widehat{\delta}}_i\|_{2}^2 -  \frac{C_3\kappa_2^2 (p\vee1) \log \{N(p\vee1)\} }{\kappa_1 T}\|(\bm{g}_i^*)_{S_i^\complement(\eta)}\|_1^2, 
\end{equation}
where $C_3>0$ is an absolute constant, if we further have
%\begin{equation}\label{eq:Tcond3}
%T \gtrsim c_{\Delta}^{-2} R_{i,q} \eta^{-q} (\kappa_2/\kappa_1)^2 (p\vee1)  \log \{N(p\vee1)\}. \tag{T.3}
%\end{equation}
\begin{equation}\label{eq:Tcond3row}
	T \gtrsim R_{i,q} \eta^{-q} (\kappa_2/\kappa_1)^2 (p\vee1)  \log \{N(p\vee1)\}. 
\end{equation}

Combining  \eqref{eq:thm1eq1row}, \eqref{eq:thm1eq5row}, and \eqref{eq:thm1eq6row}, we have
\[
\frac{C_{\rsc} \kappa_1}{4}  \|\bm{\widehat{\delta}}_i\|_{2}^2 -  \frac{C_3\kappa_2^2(p\vee1) \log \{N(p\vee1)\} }{\kappa_1 T}\|(\bm{g}_i^*)_{S_i^\complement(\eta)}\|_1^2 
\leq  \frac{\lambda_g}{2} \left \{4\|(\bm{g}_i^*)_{S_i^\complement(\eta)}\|_1 + C_2 c_{\Delta}^{-1} R_{i,q}^{1/2} \eta^{-q/2} \|\bm{\widehat{\delta}}_i\|_{2} \right \}.
\]
Consider the following two cases.

\textit{Case (i):} First suppose that 
$\frac{C_{\rsc} \kappa_1}{8}  \|\bm{\widehat{\delta}}_i\|_{2}^2 \geq   \frac{C_3\kappa_2^2(p\vee1) \log \{N(p\vee1)\} }{\kappa_1 T}\|(\bm{g}_i^*)_{S_i^\complement(\eta)}\|_1^2$. 
Then 
\[
\frac{C_{\rsc} \kappa_1}{8}  \|\bm{\widehat{\delta}}_i\|_{2}^2 \leq \frac{\lambda_g}{2} \left \{4\|(\bm{g}_i^*)_{S_i^\complement(\eta)}\|_1 + C_2 c_{\Delta}^{-1} R_{i,q}^{1/2} \eta^{-q/2} \|\bm{\widehat{\delta}}_i\|_{2} \right \},
\]
which involves a quadratic form in $\|\bm{\widehat{\delta}}_i\|_{2}$. By computing the zeros of this quadratic form, we can show that
\[
\|\bm{\widehat{\delta}}_i\|_{2}^2 \leq \frac{32 C_2^2}{C_{\rsc}^2c_{\Delta}^2} \cdot \frac{\lambda_{g}^2R_{i,q}\eta^{-q}}{\kappa_1^2} + \frac{32}{C_{\rsc}} \cdot\frac{\lambda_{g} \|(\bm{g}_i^*)_{S_i^\complement(\eta)}\|_1}{\kappa_1}.
\]

\textit{Case (ii):} Otherwise, we must have $\frac{C_{\rsc} \kappa_1}{8}  \|\bm{\widehat{\delta}}_i\|_{2}^2 \leq   \frac{C_3\kappa_2^2(p\vee1) \log \{N(p\vee1)\} }{\kappa_1 T} \|(\bm{g}_i^*)_{S_i^\complement(\eta)}\|_1^2$.

Combining the two cases above, we can apply \eqref{eq:ws1row} and \eqref{eq:Tcond3row} to show  that
\begin{align*}
	\|\bm{\widehat{\delta}}_i\|_{2}^2 & \leq \frac{32 C_2^2}{C_{\rsc}^2c_{\Delta}^2} \cdot \frac{\lambda_{g}^2R_{i,q}\eta^{-q}}{\kappa_1^2} + \frac{32}{C_{\rsc}} \cdot\frac{\lambda_{g} \|(\bm{g}_i^*)_{S_i^\complement(\eta)}\|_1}{\kappa_1} + \frac{8C_3}{C_{\rsc}} \cdot \frac{\kappa_2^2 (p\vee1) \log \{N(p\vee1)\}}{\kappa_1^2 T} \|(\bm{g}_i^*)_{S_i^\complement(\eta)}\|_1^2\\
	&\leq \frac{32 C_2^2}{C_{\rsc}^2c_{\Delta}^2} \cdot \frac{\lambda_{g}^2R_{i,q}\eta^{-q}}{\kappa_1^2} + \frac{32}{C_{\rsc}} \cdot\frac{\lambda_{g} R_{i,q} \eta^{1-q}}{\kappa_1} +\frac{8C_3}{C_{\rsc}} \cdot (R_{i,q} \eta^{-q} )^{-1}(R_{i,q} \eta^{1-q})^2\\
	&\lesssim \left (\frac{\lambda_{g} }{\kappa_1} \right )^{2-q} R_{i,q} = \eta^{2-q} R_{i,q},
\end{align*}
if we choose 
\[
\eta = \frac{\lambda_{g}}{\kappa_1}.
\]
Thus, taking $\lambda_{g}$ as its lower bound in \eqref{eq:lgcond1row}, i.e., 
$\lambda_g \asymp \sqrt{\kappa_2\lambda_{\max}(\bm{\Sigma}_{\varepsilon})\log \{N(p\vee 1)\}/T}$, we have
\[
\|\bm{\widehat{\delta}}_i\|_{2}^2 \lesssim \left [\frac{\kappa_2 \lambda_{\max}(\bm{\Sigma}_{\varepsilon})\log \{N(p\vee 1)\}}{\kappa_1^2 T} \right ]^{1-q/2} R_{i,q},
\]
and subsequently,
\[
\frac{1}{T}\sum_{t=1}^{T}(\bm{\widehat{\delta}}_i^\top\bm{\widetilde{x}}_{t})^2 \lesssim \lambda_{g}  \eta^{1-q} R_{i,q} =\left [\frac{\kappa_2 \lambda_{\max}(\bm{\Sigma}_{\varepsilon})\log \{N(p\vee 1)\}}{\kappa_1^2 T} \right ]^{1-q/2}  \frac{R_{i,q}}{\kappa_1^{1-q}}, 
\]
where the latter follows from \eqref{eq:thm1eqrow} and \eqref{eq:thm1eq5row}. 
On the one hand,  with the above choice of $\eta$, condition \eqref{eq:Tcond3row} can be guaranteed if 
\begin{equation}\label{eq:Tcondrow}
	R_{i,q} \lesssim  \frac{\lambda_{\max}(\bm{\Sigma}_{\varepsilon})}{\kappa_2 (p\vee 1) } \cdot  \left [\frac{\kappa_1^2 T}{\kappa_2 \lambda_{\max}(\bm{\Sigma}_{\varepsilon}) \log \{N(p\vee1)\}} \right ]^{1-q/2}.
	%1\gtrsim \frac{\kappa_2 (p\vee 1) }{\lambda_{\max}(\bm{\Sigma}_{\varepsilon})} \cdot \left [\frac{\kappa_2 \lambda_{\max}(\bm{\Sigma}_{\varepsilon}) \log \{N(p\vee1)\}}{\kappa_1^2 T} \right ]^{1-q/2}  R_{i,q}.
\end{equation}
Under condition \eqref{eq:Tcondrow}, since  $r+2s \lesssim 1$, we can show that a sufficient condition for \eqref{eq:etacondrow} is
\begin{equation}\label{eq:Rqcondrow}
	\frac{\lambda_{\max}(\bm{\Sigma}_{\varepsilon})}{\kappa_2 (p\vee 1)} \lesssim  \overline{\alpha}_{i,\ma}^2  R_{i,q}/R_{i,q}^\ma.
\end{equation}
Finally, combining the tail probabilities in Lemmas \ref{lemma:devbrow}--\ref{lemma:init3row} and the required conditions including \eqref{eq:Tcondrow} and \eqref{eq:Rqcondrow}, we accomplish the proof of this theorem.

%%%%%%%%%%%%%%%%%%%%%%%%%%%%%%%%%%%%%%%%%%%%%%%%%%%%%%%%%%%%%%%%%%%%%%%%%%%%%%%%%%%%%%%%%%%%%%%%%%

\section{Proof of Theorem \ref{thm:selection}}\label{asec:bic}
%%%%%%%%%%%%%%%%%%%%%%%%%%%%%%%%%%%%%%%%%%%%%%%%%%%%%%%%%%%%%%%%%%%%%%%%%%%%%%%%%%%%%%%%%%%%%%%%%%
\subsection{Irreducibility condition}

Lemma \ref{lem:irreduc} provides  the irreducibility condition for the orders $(p,r,s)$ of model \eqref{eq:model-scalar}. To better understand result (i) in this lemma,  it is worth noting that the order $p$ has a more intricate impact on the parameterization than $r$ and $s$, due to the dependence of the functions $\ell_{h,k}(\cdot)$'s on $p$. For example, suppose that $(p, r, s)=(1,1,0)$, i.e., $\bm{y}_t=\bm{G}_1\bm{y}_{t-1}+\sum_{h=2}^{\infty}\lambda_1^{h-1} \bm{G}_2\bm{y}_{t-h}+\bm{\varepsilon}_t$. Decreasing $p$ to zero leads to the reduced model
$\bm{y}_t=\sum_{h=1}^{\infty}\lambda_1^h \bm{G}\bm{y}_{t-h}+\bm{\varepsilon}_t$.
Note that the latter cannot be obtained by simply setting $\bm{G}_1=\bm{0}$. However, if the equality $\bm{G}_1=\bm{G}_2$ is satisfied, then the reduced model will be fulfilled with $\bm{G}=\lambda_1^{-1}\bm{G}_1$. 

%The following lemma  presents  a more general analysis.

%While  Assumption \ref{assum:irred}(ii) is straightforward, Assumption \ref{assum:irred}(i) may be less obvious. Indeed,  whereas the identification of $r^*$ and $s^*$ is essentially a variable selection problem, with a smaller order  corresponding to certain parameters being zero, 
%the order $p^*$ is rather special in nature. This is due to the dependence of $\ell_{h,k}(\cdot)$'s on $p$. For example, suppose that $(p^*, r^*, s^*)=(1,1,0)$, i.e., $\bm{y}_t=\bm{G}_1^*\bm{y}_{t-1}+\sum_{h=2}^{\infty}\lambda_1^{h-1} \bm{G}_2^*\bm{y}_{t-h}+\bm{\varepsilon}_t$. Decreasing $p^*$ to zero leads to the reduced model
%$\bm{y}_t=\sum_{h=1}^{\infty}\lambda_1^h \bm{G}\bm{y}_{t-h}+\bm{\varepsilon}_t$.
%Notably, the latter cannot be obtained by simply setting $\bm{G}_1^*=\bm{0}$. By contrast, if $\bm{G}_1^*=\bm{G}_2^*$, then the reduced model will be fulfilled with $\bm{G}=\lambda_1^{-1}\bm{G}_1^*$; see the proof of Lemma \ref{lem:irreduc}  for a more general analysis.

\begin{lemma}[Irreducibility of model orders]\label{lem:irreduc}
	Consider the parameterization of $\bm{A}_h$ for $h\geq1$ with model orders $(p,r,s)$ in \eqref{eq:linearcomb2}, i.e.,
	\begin{align}\label{aeq:linearcomb2}
		\begin{split}
			\bm{A}_{h} 
			&=\sum_{k=1}^{p}\mathbb{I}_{\{h=k\}}\bm{G}_{k}
			+ 
			\sum_{j=1}^{r} \mathbb{I}_{\{h\geq p+1\}} \lambda_j^{h-p} \bm{G}_{p+j}\\
			&\hspace{5mm} +\sum_{m=1}^{s} \mathbb{I}_{\{h\geq p+1\}} \gamma_{m}^{h-p} \left [ \cos\{(h-p) \theta_{m}\} \bm{G}_{p+r+2m-1} + \sin\{(h-p)\theta_{m}\} \bm{G}_{p+r+2m} \right],
		\end{split}
	\end{align}	
	where $\lambda_j\in(-1,1)$ for $1\leq j\leq r$ are distinct, and $\bm{\eta}_m=(\gamma_m,\theta_m)^\top \in \bm{\varPi}$ for $1\leq m\leq s$ are distinct, with $\bm{\varPi} = [0,1)\times (0,\pi)$.
	\begin{itemize}
		\item [(i)] If $\bm{G}_p=\sum_{j=1}^{r} \mathbb{I}_{\{\lambda_j\neq0\}} \bm{G}_{p+j} +\sum_{m=1}^{s} \mathbb{I}_{\{\gamma_m\neq0\}} \bm{G}_{p+r+2m-1}$, then the order $p$ can be reduced to $p-1$. Otherwise, the order $p$ is irreducible.
		\item [(ii)] If there exists $1\leq j\leq r$ such that $\lambda_j=0$ or $\bm{G}_{p+j}=\bm{0}$, then the order $r$ can be reduced to $r-1$. Otherwise, the order $r$ is irreducible.
		\item [(iii)] If there exists $1\leq m\leq s$ such that $\gamma_m=0$ or $\bm{G}_{p+r+2m-1}=\bm{G}_{p+r+2m}=\bm{0}$, then the order $s$ can be reduced to $s-1$.  Otherwise, the order $s$ is irreducible.
	\end{itemize}
\end{lemma}

\begin{proof}[Proof of Lemma \ref{lem:irreduc}]
	Let us first prove (i). Let $\widetilde{p}=p-1$. If $\bm{G}_p=\sum_{j=1}^{r}\bm{G}_{p+j}+\sum_{m=1}^{s} \bm{G}_{p+r+2m-1}$, then it can be readily verified that for $h\geq1$,
	\begin{align}\label{aeq:reduce}
		\begin{split}
			\bm{A}_{h} 
			&=\sum_{k=1}^{\widetilde{p}}\mathbb{I}_{\{h=k\}}\bm{\widetilde{G}}_{k}
			+ 
			\sum_{j=1}^{r} \mathbb{I}_{\{h\geq \widetilde{p}+1\}} \lambda_j^{h-\widetilde{p}} \bm{\widetilde{G}}_{\widetilde{p}+j}\\
			&\hspace{5mm} +\sum_{m=1}^{s} \mathbb{I}_{\{h\geq \widetilde{p}+1\}} \gamma_{m}^{h-\widetilde{p}} \left [ \cos\{(h-\widetilde{p}) \theta_{m}\} \bm{\widetilde{G}}_{\widetilde{p}+r+2m-1} + \sin\{(h-\widetilde{p})\theta_{m}\} \bm{\widetilde{G}}_{\widetilde{p}+r+2m} \right],
		\end{split}
	\end{align}
	where 	$\bm{\widetilde{G}}_{k}=\bm{G}_k$ for $1\leq k\leq \widetilde{p}$, 	$\bm{\widetilde{G}}_{\widetilde{p}+j}=  \mathbb{I}_{\{\lambda_j\neq0\}} \lambda_j^{-1}\bm{G}_{p+j}$ for $1\leq j\leq r$,  
	%\[
	%\bm{\widetilde{G}}_{k}=\bm{G}_k \quad\text{for}\quad 1\leq k\leq \widetilde{p}, 
	%\quad  \bm{\widetilde{G}}_{\widetilde{p}+j}=\lambda_j^{-1}\bm{G}_{p+j}  \quad\text{for}\quad 1\leq j\leq r,
	%\]
	and
	\begin{align*}
		\bm{\widetilde{G}}_{\widetilde{p}+r+2m-1} & =  \mathbb{I}_{\{\gamma_m\neq0\}} \gamma_{m}^{-1} \left \{ \cos (\theta_m) \bm{G}_{\widetilde{p}+r+2m-1} - \sin ( \theta_m) \bm{G}_{\widetilde{p}+r+2m} \right \},\\
		\bm{\widetilde{G}}_{\widetilde{p}+r+2m} & = \mathbb{I}_{\{\gamma_m\neq0\}} \gamma_{m}^{-1} \left \{ \sin( \theta_m) \bm{G}_{\widetilde{p}+r+2m-1} + \cos( \theta_m) \bm{G}_{\widetilde{p}+r+2m}\right \},
	\end{align*}
	for $1\leq m\leq s$. In other words, the order $p$ can be reduced to $\widetilde{p}$. 
	
	Now suppose that  $\bm{G}_p\neq \sum_{j=1}^{r}\bm{G}_{p+j}+\sum_{m=1}^{s} \bm{G}_{p+r+2m-1}$. If  \eqref{aeq:linearcomb2} can be reduced to the form in \eqref{aeq:reduce}, then we must have $\bm{G}_k=\bm{\widetilde{G}}_{k}$ for $1\leq k\leq \widetilde{p}$,  
	\[
	\bm{G}_p =	\sum_{j=1}^{r}  \lambda_j \bm{\widetilde{G}}_{\widetilde{p}+j}\\
	+\sum_{m=1}^{s} \left \{  (\gamma_{m} \cos\theta_m) \bm{\widetilde{G}}_{\widetilde{p}+r+2m-1} + (\gamma_{m} \sin\theta_m) \bm{\widetilde{G}}_{\widetilde{p}+r+2m} \right \},
	\]
	$\bm{G}_{p+j}=\lambda_j\bm{\widetilde{G}}_{\widetilde{p}+j}$ for $1\leq j\leq r$,  and 
	\begin{align*}
		\bm{G}_{\widetilde{p}+r+2m-1} & = \gamma_{m}\cos (\theta_m) \bm{\widetilde{G}}_{\widetilde{p}+r+2m-1} + \gamma_{m}\sin (\theta_m) \bm{\widetilde{G}}_{\widetilde{p}+r+2m},\\
		\bm{G}_{\widetilde{p}+r+2m} & = -\gamma_{m}\sin (\theta_m) \bm{\widetilde{G}}_{\widetilde{p}+r+2m-1} + \gamma_{m}\cos (\theta_m) \bm{\widetilde{G}}_{\widetilde{p}+r+2m},
	\end{align*}
	for $1\leq m\leq s$. However, this implies  $\bm{G}_p = \sum_{j=1}^{r}\bm{G}_{p+j}+\sum_{m=1}^{s} \bm{G}_{p+r+2m-1}$, resulting in a contradiction. Thus, (i) is proved.

	To establish (ii) and (iii), it is helpful to rewrite  \eqref{aeq:linearcomb2} in the form of
	\begin{align}\label{aeq:linearcomb2complex}
		\begin{split}
			\bm{A}_{h} 
			&=\sum_{k=1}^{p}\mathbb{I}_{\{h=k\}}\bm{G}_{k}
			+ 
			\sum_{j=1}^{r} \mathbb{I}_{\{h\geq p+1\}} \lambda_j^{h-p} \bm{G}_{p+j}\\
			&\hspace{5mm} +\sum_{m=1}^{s} \mathbb{I}_{\{h\geq p+1\}} \left \{v_{m}^{h-p}   \bm{H}_{p+r+2m-1} +  u_{m}^{h-p}   \bm{H}_{p+r+2m} \right\}, \quad h\geq1,
		\end{split}
	\end{align}	
	where $v_m=\gamma_m e^{i\theta_m}$, $u_m=\gamma_m e^{-i\theta_m}$,
	$\bm{H}_{p+r+2m-1} = (\bm{G}_{p+r+2m-1} - i  \bm{G}_{p+r+2m})/2$, 
	and 
	$\bm{H}_{p+r+2m} = (\bm{G}_{p+r+2m-1} + i \bm{G}_{p+r+2m})/2$, 
	for $1\leq m\leq s$, with $i$ denoting the imaginary unit. Note that  $\bm{H}_{p+r+2m-1}=\bm{H}_{p+r+2m}=\bm{0}$ if and only if $\bm{G}_{p+r+2m-1}=\bm{G}_{p+r+2m}=\bm{0}$. Then the first part of  (ii) and (iii)  is obvious.
	
	Lastly, note that if $\gamma_m\neq 0$ for $1\leq m\leq s$, then  $v_1,\dots, v_s, u_1,\dots, u_s$  are all distinct and nonzero. As a result,  the second part of (ii) and (iii)  is a straightforward consequence of the linear independence of exponential functions.
	%Now suppose that  $\lambda_j\neq 0$ for $1\leq j\leq r$. Then the second part of (ii)  is a straightforward consequence of the linear independence of exponential functions. Similarly, suppose that $\gamma_m\neq 0$ for $1\leq m\leq s$. Then   $v_1,\dots, v_s, u_1,\dots, u_s$  are all distinct and nonzero. 
\end{proof}

%%%%%%%%%%%%%%%%%%%%%%%%%%%%%%%%%%%%%%%%%%%%%%%%%%%%%%%%%%%%%%%%%%%%%%%%%%%%%%%%%%%%%%%%%%%%%%%%%%
\subsection{Reparameterization with maximum orders}
%Let $\overline{d}=\overline{p}+\overline{r}+2\overline{s}$.
We  show that any model of order $\pazocal{M}=(p,r,s)\in \mathcal{M}=\{(p,r,s)\mid 0\leq p \leq \overline{p}, 0\leq r \leq \overline{r}, 0\leq s \leq \overline{s} \}$  can be expressed as one of maximum orders   $\overline{\pazocal{M}}=(\overline{p},\overline{r}, \overline{s})$, with the corresponding parameters  determined by the original ones.
Let $\delta_p=\overline{p}-p$, $\delta_r=\overline{r}-r$, $\delta_s=\overline{s}-s$, and $\delta_d=\overline{d}-d$. The proof of Lemma \ref{lem:restrict} is straightforward by elementary algebra.  

\begin{lemma}[Reparameterization with maximum orders]\label{lem:restrict}
	Suppose that $\bm{A}_h=\bm{A}_h(\bm{\omega}, \bm{g})$ for $h\geq 1$ is parameterized as in \eqref{aeq:linearcomb2} with   model orders $\pazocal{M}=(p,r,s)\in \mathcal{M}$, where  $\bm{{\omega}} \in (-1,1)^r \times \bm{\varPi}^s$ and $\bm{{g}}\in\mathbb{R}^{N^2 d}$. 
	Then $\bm{A}_h$ for $h\geq1$ can be  expressed with orders   $\overline{\pazocal{M}}=(\overline{p},\overline{r}, \overline{s})$ as follows,
	\begin{align*}%\label{aeq:Mbar}
		\begin{split}
			\bm{A}_h(\bm{\overline{\omega}}, \bm{\overline{g}}) 
			&=\sum_{k=1}^{\overline{p}}\mathbb{I}_{\{h=k\}}\overline{\bm{G}}_{k}
			+ 
			\sum_{j=1}^{\overline{r}} \mathbb{I}_{\{h\geq \overline{p}+1\}} \overline{\lambda}_j^{h-\overline{p}} \overline{\bm{G}}_{\overline{p}+j}\\
			&\hspace{5mm} +\sum_{m=1}^{\overline{s}} \mathbb{I}_{\{h\geq \overline{p}+1\}} \overline{\gamma}_{m}^{h-\overline{p}} \left [ \cos\{(h-\overline{p}) \overline{\theta}_{m}\} \overline{\bm{G}}_{\overline{p}+\overline{r}+2m-1} + \sin\{(h-\overline{p})\overline{\theta}_{m}\} \overline{\bm{G}}_{\overline{p}+\overline{r}+2m} \right],
		\end{split}
	\end{align*}	
	where the  parameter vector $\bm{\overline{\omega}}=(\overline{\lambda}_1,\dots, \overline{\lambda}_{\overline{r}}, \bm{\overline{\eta}}_1^\top, \dots, \bm{\overline{\eta}}_{\overline{s}}^\top)^\top \in  (-1,1)^{\overline{r}}\times \bm{\varPi}^{\overline{s}}$ and the matrices $\overline{\bm{G}}_{k}$ for $1\leq k\leq \overline{d}$ are given by
	\begin{align*}
		\overline{\lambda}_j&=\mathbb{I}_{\{1\leq j\leq r\}} \lambda_j \quad \text{for} \quad  1\leq j\leq \overline{r},\quad
		\bm{\overline{\eta}}_m=\mathbb{I}_{\{1\leq m\leq s\}}\bm{\eta}_m  \quad \text{for} \quad 1\leq m\leq \overline{s},\\
		\overline{\bm{G}}_{k} & = \bm{G}_k \quad \text{for} \quad 1\leq k\leq p,\\
		\bm{\overline{G}}_{p+k}&= \sum_{j=1}^{r} \lambda_j^k \bm{G}_{p+j} \\
		&\hspace{5mm} +  \sum_{m=1}^{s} \gamma_m^k\left \{ \cos(k\theta_m) \bm{G}_{p+r+2m-1}+ \sin(k\theta_m) \bm{G}_{p+r+2m}\right\} \quad \text{for} \quad 1\leq k\leq  \delta_p,\\
		\bm{\overline{G}}_{\overline{p}+j} &= \mathbb{I}_{\{1\leq j\leq r\}}  \lambda_j^{\delta_p} \bm{G}_{p+j} \quad \text{for} \quad 1\leq j\leq \overline{r},\\
		\bm{\overline{G}}_{\overline{p}+\overline{r}+2m-1} &= \mathbb{I}_{\{1\leq m\leq s\}} \gamma_m^{\delta_p} \left \{ \cos(\delta_p \theta_m)\bm{G}_{p+r+2m-1} + \sin(\delta_p \theta_m)\bm{G}_{p+r+2m} \right \}  \quad \text{for} \quad 1\leq m\leq \overline{s},\\
		\bm{\overline{G}}_{\overline{p}+\overline{r}+2m} &= \mathbb{I}_{\{1\leq m\leq s\}} \gamma_m^{\delta_p} \left \{ - \sin(\delta_p \theta_m)\bm{g}_{i,p+r+2m} + \cos(\delta_p \theta_m)\bm{G}_{p+r+2m} \right \} \quad \text{for} \quad 1\leq m\leq \overline{s},
	\end{align*}
	and $\bm{\overline{g}}=\vect(\bm{\overline{G}})$ with $\bm{\overline{G}}=(\bm{\overline{G}}_1, \dots, \bm{\overline{G}}_{\overline{d}})\in\mathbb{R}^{N\times N \overline{d}}$.
\end{lemma}

%\begin{proof}[Proof of Lemma \ref{lem:restrict}]
%	The proof is straightforward by elementary algebra. 
%\end{proof}
%Here $\bm{\overline{\omega}}=(\bm{\overline{\lambda}}^\top, \bm{\overline{\eta}})^\top$  and $\bm{\overline{g}}=\vect(\bm{\overline{G}})$, with  $\bm{\overline{\lambda}}=(\overline{\lambda}_1,\dots, \overline{\lambda}_{\overline{r}})^\top \in (-1,1)^{\overline{r}}$, $\bm{\overline{\eta}}=(\bm{\overline{\eta}}_1^\top, \dots, \bm{\overline{\eta}}_{\overline{s}}^\top)^\top \in \bm{\varPi}^{\overline{s}}$.

\subsection{Restricted parameter space}\label{subsec:restrictedspace} 
Based on Lemma \ref{lem:restrict}, this section provides a useful intermediate result for the proof of Theorem \ref{thm:selection}. It allows us to establish a connection  between the parameter space of any $\pazocal{M}\in \mathcal{M}_{\textup{mis}}$ and that of $\pazocal{M}^*$; see Proposition \ref{prop:space} below.

The relationship between $(\bm{\overline{\omega}}, \bm{\overline{g}})$ and $(\bm{\omega}, \bm{g})$ in Lemma \ref{lem:restrict} can be equivalently written as
\begin{equation}\label{aeq:restrictvec}
	\overline{\bm{\omega}}=\bm{\overline{R}}_1^{\pazocal{M}} \bm{\omega} \quad\text{and}\quad \overline{\bm{g}}=(\bm{\overline{R}}_2^{\pazocal{M}}(\bm{\omega})  \otimes \bm{I}_{N^2})\bm{g}.
\end{equation}
Here $\bm{\overline{R}}_1^{\pazocal{M}}$ is a $(\overline{r}+2\overline{s})\times (r+2s)$ constant matrix,
\begin{equation*}%\label{eq:R1}
	\bm{\overline{R}}_1^{\pazocal{M}} = \left ( \begin{array}{ll}
		\bm{I}_{r} & \bm{0}_{r\times 2s}\\
		\bm{0}_{\delta_r\times r} & \bm{0}_{\delta_r\times 2s}\\
		\bm{0} &\bm{I}_{2s}\\
		\bm{0} &\bm{0}_{2\delta_s\times 2s}
	\end{array} \right ),
\end{equation*}
and the function $\bm{\overline{R}}_2^{\pazocal{M}}: (-1,1)^r \times \bm{\varPi}^s \rightarrow \mathbb{R}^{\overline{d}\times d}$ is defined as 
\begin{equation*}%\label{eq:R2gen}
	\bm{\overline{R}}_2^{\pazocal{M}}(\bm{\omega}) 
	= \left (
	\begin{array}{lll}
		\bm{I}_p & \bm{0}_{p\times r} & \bm{0}_{p\times 2s}\\
		\bm{0}_{\delta_p\times p} &\bm{L}_1(\bm{\lambda}) & \bm{L}_2(\bm{\eta})\\
		\bm{0}_{r\times p} &\bm{D}_1(\bm{\lambda}) &\bm{0}_{r\times 2s}\\
		\bm{0}_{\delta_r\times p} &\bm{0}_{\delta_r\times r} &\bm{0}_{\delta_r\times 2s}\\
		\bm{0} &\bm{0}_{2s\times r} & \bm{D}_2(\bm{\eta})\\
		\bm{0} &\bm{0}_{2\delta_s\times r}&\bm{0}_{2\delta_s\times 2s}
	\end{array} \right ),
\end{equation*}
where $\bm{L}_1(\bm{\lambda})$ is a $\delta_p\times r$  matrix whose $k$th row is $(\lambda_1^k,\dots, \lambda_r^k)$, $\bm{L}_2(\bm{\eta})$ is a $\delta_p\times 2s$ matrix whose $k$th row is $(\gamma_1^k\cos(k\theta_1), \gamma_1^k\sin(k\theta_1), \dots, \gamma_s^k\cos(k\theta_s), \gamma_s^k\sin(k\theta_s))$, for $1\leq k\leq \delta_p$, $\bm{D}_1(\bm{\lambda})=\diag\{\lambda_1^{\delta_p}, \dots, \lambda_r^{\delta_p}\}$ is an $r\times r$ diagonal matrix, and  $\bm{D}_2(\bm{\eta})=\diag\{\bm{B}(\bm{\eta}_1, \delta_p), \dots, \bm{B}(\bm{\eta}_s, \delta_p)\}$ is a $2s\times 2s$ block diagonal matrix whose $m$th block is  
\[
\bm{B}(\bm{\eta}_m, \delta_p) = \left (
\begin{matrix}
	\gamma_m^{\delta_p}\cos(\delta_p\theta_m) & \gamma_m^{\delta_p}\sin(\delta_p\theta_m)\\
	-\gamma_m^{\delta_p}\sin(\delta_p\theta_m) & \gamma_m^{\delta_p}\cos(\delta_p\theta_m) 
\end{matrix} \right ) \quad \text{for}\quad 1\leq m\leq s.
\]
In particular, when $\delta_r=0$  or $\delta_s=0$, the corresponding zero rows in $\bm{\overline{R}}_1^{\pazocal{M}}$ and $\bm{\overline{R}}_2^{\pazocal{M}}(\cdot)$ will disappear. When $\delta_p=0$,  $\bm{L}_1(\cdot)$ and $\bm{L}_2(\cdot)$ will disappear, while $\bm{D}_1(\cdot)=\bm{I}_r$ and $\bm{D}_2(\cdot)=\bm{I}_{2s}$, and then $\bm{\overline{R}}_2^{\pazocal{M}}(\cdot)$ will reduce to the constant block diagonal matrix, $\bm{\overline{R}}_2^{\pazocal{M}} =\diag\{\bm{I}_{\overline{p}}, \bm{\overline{R}}_1^{\pazocal{M}}\}$.
%\begin{equation*}
%	\bm{\overline{R}}_2^{\pazocal{M}} = \left ( \begin{array}{lll}
	%		\bm{I}_{\overline{p}} & \bm{0} & \bm{0}\\
	%		\bm{0}& \bm{I}_{r} & \bm{0}_{r\times 2s}\\
	%		\bm{0}& \bm{0}_{\delta_r\times r} &\bm{0}_{\delta_r\times 2s} \\
	%		\bm{0}& \bm{0} & \bm{I}_{2s}\\
	%		\bm{0}& \bm{0} & \bm{0}_{2\delta_s\times 2s}
	%	\end{array} \right ) 
%	=  \left ( \begin{array}{ll}
	%		\bm{I}_{\overline{p}} &\bm{0} \\
	%		\bm{0} & \bm{\overline{R}}_1^{\pazocal{M}}
	%	\end{array}\right ).
%\end{equation*}

By Lemma \ref{lem:restrict}, for any $\pazocal{M}=(p,r,s)\in \mathcal{M}$, the following constraints are satisfied by  $\bm{\overline{\omega}}$ and  $\overline{\bm{G}}_{k}$ for $1\leq k\leq \overline{d}$:
\begin{equation}\label{aeq:restrictomg}
	\overline{\lambda}_{r+1}=\cdots=\overline{\lambda}_{\overline{r}} =0, \quad 
	\bm{\overline{\eta}}_{s+1}=\cdots=\bm{\overline{\eta}}_{\overline{s}}=\bm{0},
\end{equation}
and 
\begin{align}\label{aeq:restrictg}
	\begin{split}
		\bm{\overline{G}}_{p+k} &= \sum_{j=1}^{\overline{r}} \overline{\lambda}_j^{k-\delta_p} \bm{\overline{G}}_{\overline{p}+j} +  \sum_{m=1}^{\overline{s}} \overline{\gamma}_m^{k-\delta_p}  \cos\{(k-\delta_p)\overline{\theta}_m\}\bm{\overline{G}}_{\overline{p}+\overline{r}+2m-1}  \\
		&\hspace{5mm} +  \sum_{m=1}^{\overline{s}} \overline{\gamma}_m^{k-\delta_p} \sin\{(k-\delta_p)\overline{\theta}_m\}\bm{\overline{G}}_{\overline{p}+\overline{r}+2m} \quad \text{for} \quad 1\leq k\leq  \delta_p,\\
		\bm{\overline{G}}_{\overline{p}+r+1} &=\cdots=\bm{\overline{G}}_{\overline{p}+\overline{r}}= \bm{0},\quad 
		\bm{\overline{G}}_{\overline{p}+\overline{r}+2s+1} =\cdots = \bm{\overline{G}}_{\overline{p}+\overline{r}+2\overline{s}} = \bm{0}.
	\end{split}
\end{align}
%\[
%\bm{\overline{G}}_{p+k}= \sum_{j=1}^{\overline{r}} \overline{\lambda}_j^{k-\delta_p} \bm{\overline{G}}_{\overline{p}+j} +  \sum_{m=1}^{\overline{s}} \overline{\gamma}_m^{k-\delta_p} \left \{ \cos\{(k-\delta_p)\overline{\theta}_m\}\bm{\overline{G}}_{\overline{p}+\overline{r}+2m-1}  + \sin\{(k-\delta_p)\overline{\theta}_m\}\bm{\overline{G}}_{\overline{p}+\overline{r}+2m}  \right\},
%\]
%for $1\leq k\leq  \delta_p$. 
These constraints can be written in vector form as
\begin{equation}\label{aeq:restrictvec2}
	\bm{\overline{C}}_1^{\pazocal{M}} \bm{\overline{\omega}}=\bm{0} \quad\text{and}\quad \left (\bm{\overline{C}}_2^{\pazocal{M}}( \bm{\overline{\omega}})  \otimes \bm{I}_{N^2} \right ) \bm{\overline{g}}=\bm{0}.
\end{equation}
Here $\bm{\overline{C}}_1^{\pazocal{M}}\in\mathbb{R}^{(\delta_r+2\delta_s)\times (\overline{r}+2\overline{s})}$ is a constant matrix encoding the $(\delta_r+2\delta_s)$   constraints on $\bm{\overline{\omega}}$ as stated in \eqref{aeq:restrictomg}, 
\[
\bm{\overline{C}}_1^{\pazocal{M}}=
\left ( \begin{matrix}
	\bm{0}_{\delta_r\times r} & \bm{I}_{\delta_r} & \bm{0} & \bm{0}\\
	\bm{0} & \bm{0} & \bm{0}_{2\delta_s\times 2s} & \bm{I}_{2\delta_s} 
\end{matrix} \right ),
\]
and $\bm{\overline{C}}_2^{\pazocal{M}}: (-1,1)^{\overline{r}} \times \bm{\varPi}^{\overline{s}}  \rightarrow \mathbb{R}^{\delta_d \times \overline{d}}$ encodes the $\delta_d$ constraints on $\bm{\overline{g}}$ for any given $\bm{\overline{\omega}}$ as stated in \eqref{aeq:restrictg},
\[
\bm{\overline{C}}_2^{\pazocal{M}}( \bm{\overline{\omega}}) =
\left ( 
\begin{matrix}
	\bm{0}_{\delta_p\times p} & \bm{I}_{\delta_p} & \bm{L}_3(\bm{\overline{\lambda}}) & \bm{0} & \bm{L}_4(\bm{\overline{\eta}}) & \bm{0}\\
	\bm{0} & \bm{0} & \bm{0} & \bm{I}_{\delta_r} & \bm{0} & \bm{0}\\
	\bm{0} & \bm{0} & \bm{0} & \bm{0} & \bm{0} & \bm{I}_{2\delta_s}
\end{matrix}
\right ),
\]
where  $\bm{L}_3(\bm{\overline{\lambda}})$ is a $\delta_p\times r$  matrix whose $k$th row is $(\overline{\lambda}_1^{k-\delta_p},\dots, \overline{\lambda}_r^{k-\delta_p})$, and $\bm{L}_4(\bm{\overline{\eta}})$ is a $\delta_p\times 2s$ matrix whose $k$th row is
\[
(\overline{\gamma}_1^{k-\delta_p}\cos\{(k-\delta_p)\overline{\theta}_1\}, \overline{\gamma}_1^{k-\delta_p}\sin\{(k-\delta_p)\overline{\theta}_1\}, \dots, \overline{\gamma}_s^{k-\delta_p}\cos\{(k-\delta_p)\overline{\theta}_s\}, \overline{\gamma}_s^{k-\delta_p}\sin\{(k-\delta_p)\overline{\theta}_s\}),
\]
for $1\leq k\leq \delta_p$.
Note that $\bm{\overline{C}}_1^{\pazocal{M}}$ and $\bm{\overline{C}}_2^{\pazocal{M}}(\cdot)$ are intrinsically determined by  $\bm{\overline{R}}_1^{\pazocal{M}}$ and $\bm{\overline{R}}_2^{\pazocal{M}}(\cdot)$ in \eqref{aeq:restrictvec}, respectively. In fact, it holds 
\[
\bm{L}_3(\bm{\overline{\lambda}}) = \bm{L}_1(\bm{\lambda}) \bm{D}_1^{-1}(\bm{\lambda})
\quad\text{and}\quad 
\bm{L}_4(\bm{\overline{\eta}}) = \bm{L}_2(\bm{\eta}) \bm{D}_2^{-1}(\bm{\eta}),
\]
since
$\overline{\lambda}_j=\mathbb{I}_{\{1\leq j\leq r\}} \lambda_j$ for $1\leq j\leq \overline{r}$, and 
$\bm{\overline{\eta}}_m=\mathbb{I}_{\{1\leq m\leq s\}}\bm{\eta}_m$ for $1\leq m\leq \overline{s}$.

As indicated by \eqref{aeq:restrictvec2}, increasing $p$ by one amounts to deleting a particular row from $\bm{\overline{C}}_2^{\pazocal{M}}( \bm{\overline{\omega}})$, while increasing $r$ (or $s$) by one is equivalent to deleting a particular row (or a pair of rows) from both $\bm{\overline{C}}_1^{\pazocal{M}}$ and $\bm{\overline{C}}_2^{\pazocal{M}}( \bm{\overline{\omega}})$. The following proposition is a direct consequence of the above discussion. It also establishes the  monotonicity of $\bm{\Gamma}_{\pazocal{M}}$ in $\pazocal{M}$ along a single direction of $p,r$ or $s$.

\begin{proposition}[Restricted parameter spaces]\label{prop:space}
	Any model \eqref{eq:model-scalar} with orders $\pazocal{M}=(p,r,s)\in \mathcal{M}$ can be reparameterized as the model with orders $\overline{\pazocal{M}}=(\overline{p},\overline{r}, \overline{s})$ and the corresponding parameter vectors $\bm{\overline{\omega}}$ and $\overline{\bm{g}}$ belonging to the restricted parameter space,
	\begin{align*}
		\bm{\Gamma}_{\pazocal{M}} &=\left \{ \bm{\overline{\omega}}\in (-1,1)^{\overline{r}} \times \bm{\varPi}^{\overline{s}} , \; \bm{\overline{g}} \in\mathbb{R}^{N^2\overline{d}}: \bm{\overline{C}}_1^{\pazocal{M}} \bm{\overline{\omega}}=\bm{0} \text{ and }
		(\bm{\overline{C}}_2^{\pazocal{M}}( \bm{\overline{\omega}}) \otimes \bm{I}_{N^2}  )  \bm{\overline{g}}=\bm{0} \right \}\\
		&=  \left \{ \bm{\overline{\omega}}=\bm{\overline{R}}_1^{\pazocal{M}} \bm{\omega}, \; \overline{\bm{g}}=(\bm{\overline{R}}_2^{\pazocal{M}}(\bm{\omega})  \otimes \bm{I}_{N^2})\bm{g}: \bm{\omega}\in (-1,1)^r \times \bm{\varPi}^s \text{ and }  \bm{g} \in\mathbb{R}^{N^2 d} \right \}.
	\end{align*}
	Moreover, 
	$\bm{\Gamma}_{\pazocal{M}}\subset \bm{\Gamma}_{\pazocal{M}^\prime}$, for any $\pazocal{M}^\prime$  obtained by increasing one of the $p,r,s$ in $\pazocal{M}$ by one.
\end{proposition}

%%%%%%%%%%%%%%%%%%%%%%%%%%%%%%%%%%%%%%%%%%%%%%%%%%%%%%%%%%%%%%%%%%%%%%%%%%%%%%%%%%%%%%%%%%%%%%%%%%
\subsection{Proof of Theorem \ref{thm:selection}}

In this proof, we will focus on the JE, since the proof for the RE will be similar.  
Since $\overline{p}, \overline{r}$ and $\overline{s}$ are assumed to be fixed, $\mathcal{M}$ contains a fixed number of candidate models. 
To prove this theorem, it suffices to show that for each $\pazocal{M}\in  \mathcal{M}_{\textup{over}}\cup\mathcal{M}_{\textup{mis}}$,\[
\mathbb{P}\left \{\textup{BIC} ({\pazocal{M}}) > \textup{BIC}(\pazocal{M}^*)\right \} \to 0 \quad\text{as}\quad T\to \infty,
\]
where
$\mathcal{M}_{\textup{over}}= \{\pazocal{M}\in \mathcal{M}\mid  p\geq p^*,  r\geq r^* \text{ and } s\geq s^*\}\setminus \pazocal{M}^*$
and 
$\mathcal{M}_{\textup{mis}}= \{\pazocal{M}\in \mathcal{M}\mid  p< p^*,  r< r^* \text{ or } s< s^*\}$. 
For any   $\pazocal{M}=(p,r,s)\in\mathcal{M}$, define the unregularized population  minimizer:
\[
(\bm{\omega}^{\circ}_{\pazocal{M}},\bm{g}^{\circ}_{\pazocal{M}}) = \argmin_{\bm{\omega}\in(-1,1)^r \times \bm{\varPi}^s, \bm{g}\in \mathbb{R}^{N^2d}} \mathbb{E}\{\mathbb{L}_T(\bm{\omega},\bm{g}) \}.
\]
Note that when $\pazocal{M}=\pazocal{M}^*$, we simply have $(\bm{\omega}^{\circ}_{\pazocal{M}}, \bm{g}^{\circ}_{\pazocal{M}})=(\bm{\omega}^*,\bm{g}^*)$.  In addition, denote
\[
\widetilde{\varphi}_{T,\pazocal{M}} = \tau_N \left [\frac{\log\{N(p\vee1)\}}{T}\right ]^{1-q/2}.
\]

Let $\widehat{\bm{\omega}}$ and $\widehat{\bm{g}}$ denote the estimators obtained from fitting the correctly specified model, i.e., $\pazocal{M}^*$.
Note that 
\begin{equation}\label{eq:bic1}
	\textup{BIC}(\pazocal{M})-\textup{BIC}(\pazocal{M}^*)=
	\log \left(1 + \frac{D_{\pazocal{M}}}{ \widetilde{\mathbb{L}}_{T}(\widehat{\bm{\omega}},\widehat{\bm{g}}) }\right) + (d \widetilde{\varphi}_{T,\pazocal{M}}- d^*\widetilde{\varphi}_{T,\pazocal{M}^*})\log T, 
\end{equation}
where 
\[
D_{\pazocal{M}}= \widetilde{\mathbb{L}}_{T}(\widehat{\bm{\omega}}_{\pazocal{M}},\widehat{\bm{g}}_{\pazocal{M}}) - \widetilde{\mathbb{L}}_{T}(\widehat{\bm{\omega}},\widehat{\bm{g}})= D_{\pazocal{M}, 1}-D_{\pazocal{M}^*, 2}+D_{\pazocal{M}, 3},
\]
with $D_{\pazocal{M},1} =\widetilde{\mathbb{L}}_T(\widehat{\bm{\omega}}_{\pazocal{M}},\widehat{\bm{g}}_{\pazocal{M}})- \mathbb{E}\{\mathbb{L}_T(\bm{\omega}^{\circ}_{\pazocal{M}},\bm{g}^{\circ}_{\pazocal{M}})\}$,  $D_{\pazocal{M}^*, 2} =\widetilde{\mathbb{L}}_T(\widehat{\bm{\omega}},\widehat{\bm{g}})- \mathbb{E}\{\mathbb{L}_T(\bm{\omega}^*,\bm{g}^*)\}$, and $D_{\pazocal{M},3} =\mathbb{E}\{\mathbb{L}_T(\bm{\omega}^{\circ}_{\pazocal{M}},\bm{g}^{\circ}_{\pazocal{M}})\} - \mathbb{E}\{\mathbb{L}_T(\bm{\omega}^*,\bm{g}^*)\}$. 
By the proof of Theorem \ref{thm:lasso} or \ref{thm:lassorow}, we can directly show that 
\begin{equation}\label{eq:bicerr2}
	D_{\pazocal{M}^*, 2}=O_p(N\widetilde{\varphi}_{T,\pazocal{M}^*}).
\end{equation}

Recall that  $\bm{a}=\vect(\bm{A})$,  where $\bm{A}=(\bm{A}_1,\bm{A}_2,\dots)$ is the horizontal concatenation of $\{\bm{A}_h\}_{h=1}^\infty$. Note that $\bm{a}=(\bm{L}(\bm{\omega})\otimes \bm{I}_{N^2})\bm{g}$. Throughout our proof, we will suppress the dependence of $\bm{L}(\cdot)$ on $\pazocal{M}$ for simplicity. Analogously, for any $\pazocal{M}\in\mathcal{M}$, we can define $\bm{\widehat{a}}_{\pazocal{M}}=\vect(\bm{\widehat{A}}_{\pazocal{M}})=(\bm{L}(\bm{\widehat{\omega}_{\pazocal{M}}})\otimes \bm{I}_{N^2})\bm{\widehat{g}}_{\pazocal{M}}$ and $\bm{a}_{\pazocal{M}}^{\circ}=\vect(\bm{A}_{\pazocal{M}}^{\circ})=(\bm{L}(\bm{\omega}_{\pazocal{M}}^{\circ})\otimes \bm{I}_{N^2})\bm{g}_{\pazocal{M}}^{\circ}$.  
Moreover, by Proposition \ref{prop:space}, we can write
\[
\mathbb{E}\{\mathbb{L}_T(\bm{\omega}^{\circ}_{\pazocal{M}},\bm{g}^{\circ}_{\pazocal{M}})\} =\mathbb{E}\{\|\bm{y}_t-(\bm{x}_t^\top \otimes \bm{I}_N)\bm{a}_{\pazocal{M}}^{\circ}\|_2^2\}
=\min_{(\bm{\omega}, \bm{g})\in \bm{\Gamma}_{\pazocal{M}}}\mathbb{E}\left \{\|\bm{y}_t-(\bm{x}_t^\top \otimes \bm{I}_N)\bm{a}(\bm{\omega}, \bm{g})\|_2^2\right \}.
\]

\smallskip
\noindent
\textbf{(i) Misspecified models:}
Let $\pazocal{M}\in\mathcal{M}_{\textup{mis}}$. The key of this analysis is to derive a lower bound for  $D_{\pazocal{M},3}$ based on Proposition \ref{prop:space} and then show that it dominates both $D_{\pazocal{M},1}$ and $D_{\pazocal{M}^*,2}$. 

Denote 
$\mathcal{L}(\bm{a})=\mathbb{E}\{\|\bm{y}_t-(\bm{x}_t^\top \otimes \bm{I}_N)\bm{a}\|_2^2\}$.  By Lemma \ref{lemma:Wcov}, $\lambda_{\min} \left \{\mathbb{E}(\bm{x}_t \bm{x}_t^\top )  \otimes \bm{I}_N \right \} = \lambda_{\min} \left \{\mathbb{E}(\bm{x}_t \bm{x}_t^\top ) \right \} \geq \kappa_1$. Then, by the Taylor expansion and Proposition \ref{prop:space}, we have
\begin{align*}%\label{eq:bicerr4}
	D_{\pazocal{M},3} =\mathcal{L}(\bm{a}_{\pazocal{M}}^{\circ}) -\mathcal{L}(\bm{a}^{*}) 
	&=(\bm{a}_{\pazocal{M}}^{\circ}-\bm{a}^*)^\top  \left \{\mathbb{E}(\bm{x}_t \bm{x}_t^\top )  \otimes \bm{I}_N \right \} (\bm{a}_{\pazocal{M}}^{\circ}-\bm{a}^*) \notag\\ &\geq \kappa_1 \|\bm{a}_{\pazocal{M}}^{\circ}-\bm{a}^*\|_2^2  \geq \delta_{\pazocal{M}},
\end{align*}
where $\delta_{\pazocal{M}}=   \kappa_1  \inf_{(\bm{\omega}, \bm{g})\in \bm{\Gamma}_{\pazocal{M}}} \|(\bm{L}(\bm{\omega})\otimes \bm{I}_{N^2})\bm{g}-\bm{a}^*\|_2^2$. Note that by Assumption \ref{assum:signaljoint}(i) and the boundedness of $d^*$, we have $\delta_{\pazocal{M}}\gg N d^* \widetilde{\varphi}_{T,\pazocal{M}^*} \log T$. As a result, it follows from \eqref{eq:bicerr2} that
%\begin{equation*}%\label{eq:bicerr3}
$D_{\pazocal{M}^*, 2}=o_p(\delta_{\pazocal{M}})$.
%\end{equation*}
Moreover, 
Assumption \ref{assum:signaljoint}(ii) implies
%\begin{equation*}%\label{eq:bicerr5}
$D_{\pazocal{M},1}=o_p(\delta_{\pazocal{M}})$.
%\end{equation*}

Lastly, since $\log(1+x)\geq \min\{0.5x, \log2\}$ for any $x>0$ and $\widetilde{\mathbb{L}}_{T}(\widehat{\bm{\omega}},\widehat{\bm{g}})=E(\|\bm{\varepsilon}_t\|_2^2)+D_{\pazocal{M}^*,2}=O_p(N)$, by combining \eqref{eq:bic1} with the results above, we can show that
\[
\textup{BIC}(\pazocal{M})-\textup{BIC}(\pazocal{M}^*)
\geq  \min\left \{\frac{0.5 D_{\pazocal{M}}}{ \widetilde{\mathbb{L}}_{T}(\widehat{\bm{\omega}},\widehat{\bm{g}})}, \log2 \right \} +  (d \widetilde{\varphi}_{T,\pazocal{M}}- d^*\widetilde{\varphi}_{T,\pazocal{M}^*})\log T > 0,
\]
as $T\rightarrow\infty$.

\smallskip
\noindent
\textbf{(ii) Overspecified models:}
Let $\pazocal{M}\in\mathcal{M}_{\textup{over}}$. First, we can show that 
\[
\min_{\bm{a}\in\mathbb{R}^\infty}\mathbb{E}\left \{\|\bm{y}_t-(\bm{x}_t^\top \otimes \bm{I}_N)\bm{a}\|_2^2\right \}=\mathbb{E}\{\|\bm{\varepsilon}_t\|_2^2 \}
\]
and this minimum is attained at $\bm{a}^*=\bm{a}(\bm{\omega}^*, \bm{g}^*)$. Moreover, since $(\bm{\omega}^*, \bm{g}^*)\in\bm{\Gamma}_{\pazocal{M}^*}\subset \bm{\Gamma}_{\pazocal{M}}$, we have 
%\begin{equation*}
$\mathbb{E}\{\mathbb{L}_T(\bm{\omega}^{\circ}_{\pazocal{M}},\bm{g}^{\circ}_{\pazocal{M}})\} 
=\min_{(\bm{\omega}, \bm{g})\in \bm{\Gamma}_{\pazocal{M}}}\mathbb{E}\left \{\|\bm{y}_t-(\bm{x}_t^\top \otimes \bm{I}_N)\bm{a}(\bm{\omega}, \bm{g})\|_2^2\right \}=\mathbb{E}\{\|\bm{\varepsilon}_t\|_2^2 \}$,
%\end{equation*}
with the minimum attained at some $(\bm{\omega}^{\circ}_{\pazocal{M}}, \bm{g}^{\circ}_{\pazocal{M}})$ such that  $\bm{a}_{\pazocal{M}}^{\circ}=\bm{a}^*$. 
Thus, 
\begin{equation}\label{eq:bicerr1}
	D_{\pazocal{M},3}=0.
\end{equation}

In addition, we can show that
\begin{equation}\label{eq:bicerr}
	D_{\pazocal{M},1}= O_p(N\widetilde{\varphi}_{T,\pazocal{M}}).
\end{equation}
Since  $\bm{A}_{\pazocal{M}}^{\circ}=\bm{A}^*$, by the optimality of $\bm{\widehat{A}}_{\pazocal{M}}$, we have
\begin{equation*}%\label{eq:opt}
	\frac{3}{4T}\sum_{t=1}^{T}\|\bm{\widehat{\Delta}}_{\pazocal{M}}\bm{x}_t\|_{2}^2-S_3(\bm{\widehat{\Delta}}_{\pazocal{M}})
	\leq \frac{2}{T}\sum_{t=1}^{T}\langle \bm{\varepsilon}_{t}, \bm{\widehat{\Delta}}_{\pazocal{M}}\bm{x}_t \rangle+\lambda_g(\|\bm{g}^{*}\|_1-\|\widehat{\bm{g}}_{\pazocal{M}}\|_1 ) + S_2(\bm{\widehat{\Delta}}_{\pazocal{M}}) - S_1(\bm{\widehat{\Delta}}_{\pazocal{M}}),
\end{equation*}
where  $\bm{\widehat{\Delta}}_{\pazocal{M}}=\bm{\widehat{A}}_{\pazocal{M}}-\bm{A}^*$, and $S_i(\cdot)$ for $1\leq i\leq 3$ are defined as in the proof of Theorem \ref{thm:lasso}. The remainder of the proof can be completed by modifying that of Theorem \ref{thm:lasso}. 
This involves adapting Proposition \ref{prop:perturb} for $\pazocal{M}\in\mathcal{M}_{\textup{over}}$.
To this end, we define the following notations:  Let
$\bm{g}_{\pazocal{M}}= (\bm{g}_{\pazocal{M},\ar}^\top, \bm{g}_{\pazocal{M},\ma}^\top)^\top\in\mathbb{R}^{N^2d}$, where $\bm{g}_{\pazocal{M},\ar}=\vect((\bm{G}_1,\dots, \bm{G}_p))$ and $\bm{g}_{\pazocal{M},\ma}=\vect((\bm{G}_{p+1},  \dots, \bm{G}_{d}))$. We can partition any $\bm{\omega}_{\pazocal{M}}\in (-1,1)^r \times \bm{\varPi}^s$ into two subvectors: $\bm{\omega}_{\pazocal{M}^*}\in(-1,1)^{r^*} \times \bm{\varPi}^{s^*}$ and $\bm{\omega}_{\pazocal{M}^{\delta}}\in(-1,1)^{\delta_r} \times \bm{\varPi}^{\delta_s}$, where $\delta_r=r-r^*$ and $\delta_s=s-s^*$. Accordingly, partition $\bm{g}_{\pazocal{M},\ma}$ into two subvectors: $\bm{g}_{\pazocal{M}^*,\ma}\in\mathbb{R}^{N^2 (r+2s)}$ and  $\bm{g}_{\pazocal{M}^{\delta},\ma}\in\mathbb{R}^{N^2 (\delta_r+2\delta_s)}$.
%$\bm{g}_{\pazocal{M}^*,\ma}=\vect((\bm{G}_{p+1},  \dots, \bm{G}_{p+r^*}, \bm{G}_{p+r+1},\dots, \bm{G}_{p+r+2s^*}))$ and $\bm{g}_{\pazocal{M}^{\delta},\ma} =\vect((\bm{G}_{p+r^*+1}, \dots, \bm{G}_{p+r},  \bm{G}_{p+r+2s^*+1},\dots,  \bm{G}_d))$. 
Then, 
%write  $\bm{a}_{\pazocal{M}}= (\bm{a}_{\pazocal{M},\ar}^\top, \bm{a}_{\pazocal{M},\ma}^\top)^\top\in\mathbb{R}^{\infty}$, where 
let $\bm{a}_{\pazocal{M},\ar}=\vect((\bm{A}_1,\dots, \bm{A}_p))$ and $\bm{a}_{\pazocal{M},\ma}=\vect((\bm{A}_{p+1}, \bm{A}_{p+2}, \dots))$.

Note that $\bm{a}_{\pazocal{M},\ar}=\bm{g}_{\pazocal{M},\ar}$ and $\bm{a}_{\pazocal{M},\ma}=(\bm{L}^{\ma}(\bm{\omega})\otimes \bm{I}_{N^2})\bm{g}_{\pazocal{M},\ma}=(\bm{L}^{\ma}(\bm{\omega}_{\pazocal{M}^*})\otimes \bm{I}_{N^2})\bm{g}_{\pazocal{M}^*,\ma}+\bm{a}_{\pazocal{M}^\delta,\ma}$, where $\bm{a}_{\pazocal{M}^\delta,\ma}=(\bm{L}^{\ma}(\bm{\omega}_{\pazocal{M}^\delta})\otimes \bm{I}_{N^2})\bm{g}_{\pazocal{M}^\delta,\ma}$.
By a method similar to that for deriving \eqref{aeq:restrictvec}, we can show that $\bm{\omega}^{\circ}_{\pazocal{M}^\delta}=\bm{0}$ and $\bm{g}^{\circ}_{\pazocal{M}^\delta,\ma}=\bm{0}$, which are subvectors of $\bm{\omega}^{\circ}_{\pazocal{M}}$ and $\bm{g}^{\circ}_{\pazocal{M}}$, respectively.
%we can show that $\bm{\omega}^{\circ}_{\pazocal{M}}=\bm{R}_1^{\pazocal{M}^*}\bm{\omega}^*$ and $\bm{g}^{\circ}_{\pazocal{M}}=(\bm{R}_2^{\pazocal{M}^*}(\bm{\omega}^*)  \otimes \bm{I}_{N^2})\bm{g}^*$, where $\bm{R}_1^{\pazocal{M}^*}$ is a $(r+2s)\times (r^*+2s^*)$ constant matrix defined analogously to $\bm{\overline{R}}_1^{\pazocal{M}}$, and $\bm{R}_2^{\pazocal{M}^*}: (-1,1)^{r^*} \times \bm{\varPi}^{s^*} \rightarrow \mathbb{R}^{d\times d^*}$ is defined similarly to $\bm{\overline{R}}_2^{\pazocal{M}}(\cdot)$. 
%This implies that 
%$\bm{\omega}^{\circ}_{\pazocal{M}}$ is formed by the subvectors $\bm{\omega}^{\circ}_{\pazocal{M}^*}=\bm{\omega}^*$ and  $\bm{\omega}^{\circ}_{\pazocal{M}^\delta}=\bm{0}$. In addition, the subvector  $\bm{g}^{\circ}_{\pazocal{M}^\delta,\ma}=\bm{0}$. 
Thus, $\bm{a}_{\pazocal{M}^\delta,\ma}^{\circ}=(\bm{L}^{\ma}(\bm{\omega}_{\pazocal{M}^\delta}^{\circ})\otimes \bm{I}_{N^2})\bm{g}_{\pazocal{M}^\delta,\ma}^{\circ}=\bm{0}$.
Then, by adapting the proof of Proposition \ref{prop:perturb}, under Assumptions \ref{assum:statn}(i) and \ref{assum:gap}, we can show that if $\|\bm{\omega}_{\pazocal{M}^*} - \bm{\omega}^*\|_2\leq c_{\bm{\omega}}$, then $\|\bm{a}_{\pazocal{M}^\delta,\ma}\|_{2}+\|\bm{g}_{\pazocal{M},\ar}-\bm{g}_{\pazocal{M},\ar}^{\circ}\|_{2} + \|\bm{g}_{\pazocal{M}^*,\ma}-\bm{g}_{\pazocal{M}^*,\ma}^{\circ}\|_{2} + \underline{\alpha}_\ma \|\bm{\omega}_{\pazocal{M}^*} - \bm{\omega}^*\|_2 \lesssim \|\bm{\Delta}_{\pazocal{M}}\|_{\Fr}^2\lesssim \|\bm{a}_{\pazocal{M}^\delta,\ma}\|_{2}+ \|\bm{g}_{\pazocal{M},\ar}-\bm{g}_{\pazocal{M},\ar}^{\circ}\|_{2} + \|\bm{g}_{\pazocal{M}^*,\ma}-\bm{g}_{\pazocal{M}^*,\ma}^{\circ}\|_{2} + \overline{\alpha}_\ma \|\bm{\omega}_{\pazocal{M}^*} - \bm{\omega}^*\|_2$. Along the lines of this adaptation, we can modify  the proof of Theorem \ref{thm:lasso} to show that 
\[
D_{\pazocal{M},1} \lesssim \left [\frac{\kappa_2 \lambda_{\max}(\bm{\Sigma}_{\varepsilon})\log \{N(p\vee 1)\}}{\kappa_1^2 T} \right ]^{1-q/2}  \frac{R_q}{\kappa_1^{1-q}} \lesssim \widetilde{\varphi}_{T,\pazocal{M}},
\]
with high probability, and hence \eqref{eq:bicerr}, provided that $\bm{\widehat{\omega}}_{\pazocal{M}}$ contains a subvector $\bm{\widehat{\omega}}_{\pazocal{M}^*}$ satisfying $\|\bm{\widehat{\omega}}_{\pazocal{M}^*} - \bm{\omega}^*\|_2\leq c_{\bm{\omega}}$.

Now using the inequality $\log(1+x)\leq x$, we have
\[
\log \left(1 + \frac{D_{\pazocal{M}}}{ \widetilde{\mathbb{L}}_{T}(\widehat{\bm{\omega}},\widehat{\bm{g}}) }\right)\geq -  \frac{D_{\pazocal{M}}}{ \widetilde{\mathbb{L}}_{T}(\widehat{\bm{\omega}},\widehat{\bm{g}}) }.
\]
Additionally, note that $\widetilde{\mathbb{L}}_{T}(\widehat{\bm{\omega}},\widehat{\bm{g}})=\mathbb{E}\{\mathbb{L}_T(\bm{\omega}^*,\bm{g}^*)\}+D_{\pazocal{M}^*,2}=E(\|\bm{\varepsilon}_t\|_2^2)+D_{\pazocal{M}^*,2}$, where $E(\|\bm{\varepsilon}_t\|_2^2) \asymp N$.  Finally, since $\widetilde{\varphi}_{T,\pazocal{M}} > \widetilde{\varphi}_{T,\pazocal{M}^*}$, it follows from \eqref{eq:bic1}--\eqref{eq:bicerr} that
\begin{align*}
	\textup{BIC}(\pazocal{M})-\textup{BIC}(\pazocal{M}^*)&\geq (d \widetilde{\varphi}_{T,\pazocal{M}}- d^*\widetilde{\varphi}_{T,\pazocal{M}^*})\log T - O_p(N (\widetilde{\varphi}_{T,\pazocal{M}}-\widetilde{\varphi}_{T,\pazocal{M}^*}) / N)\\
	& = (d-d^*) \widetilde{\varphi}_{T,\pazocal{M}}\log T + O_p((\widetilde{\varphi}_{T,\pazocal{M}}-\widetilde{\varphi}_{T,\pazocal{M}^*})(d^*\log T - 1)) >0,
\end{align*}
as $T\rightarrow\infty$.
The proof  of this theorem is complete.

%%%%%%%%%%%%%%%%%%%%%%%%%%%%%%%%%%%%%%%%%%%%%%%%%%%%%%%%%%%%%%%%%%%%%%%%%%%%%%%%%%%%%%%%%%%%%%%%%%

%%%%%%%%%%%%%%%%%%%%%%%%%%%%%%%%%%%%%%%%%%%%%%%%%%%%%%%%%%%%%%%%%%%%%%%%%%%%%%%%%%
%%%%%%%%%%%%%%%%%%%%%%%%%%%%%%%%%%%%%%%%%%%%%%%%%%%%%%%%%%%%%%%%%%%%%%%%%%%%%%%%%%
\section{Proofs of auxiliary lemmas}\label{asec:aux}
\subsection{Proof of Lemma \ref{cor1}}
By definition,  $\ell_{h}^{I}(\lambda_j)=\lambda_j^h$ for $1\leq j\leq r$, and $\ell_{h}^{II,1}(\bm{\eta}_m)=\gamma_m^h\cos(h\theta_m)$ and  $\ell_{h}^{II,2}(\bm{\eta}_m)=\gamma_m^h\sin(h\theta_m)$ for $1\leq m\leq s$. 
Then their first-order derivatives are $\nabla\ell_{h}^{I}(\lambda_j)=h\lambda_j^{h-1}$, $\nabla_\gamma\ell_{h}^{II,1}(\bm{\eta}_m)=h\gamma_m^{h-1}\cos(h\theta_m)$, 
$\nabla_\theta\ell_{h}^{II,1}(\bm{\eta}_m)=-h\gamma_m^{h}\sin(h\theta_m)$, $\nabla_\gamma\ell_{h}^{II,2}(\bm{\eta}_m)=h\gamma_m^{h-1}\sin(h\theta_m)$, and
$\nabla_\theta\ell_{h}^{II,2}(\bm{\eta}_m)=h\gamma_m^{h}\cos(h\theta_m)$. Their second-order derivatives are
$\nabla^2\ell_{h}^{I}(\lambda_j)=h(h-1)\lambda_j^{h-2}$, 
$\nabla^2_\gamma\ell_{h}^{II,1}(\bm{\eta}_m)=h(h-1)\gamma_m^{h-2}\cos(h\theta_m)$,
$\nabla^2_{\gamma\theta}\ell_{h}^{II,1}(\bm{\eta}_m)=-h^2\gamma_m^{h-1}\sin(h\theta_m)$, $\nabla^2_\theta\ell_{h}^{II,1}(\bm{\eta}_m)=-h^2\gamma_m^{h}\cos(h\theta_m)$, 
$\nabla^2_\gamma\ell_{h}^{II,2}(\bm{\eta}_m)=h(h-1)\gamma_m^{h-2}\sin(h\theta_m)$, 
$\nabla^2_{\gamma\theta}\ell_{h}^{II,2}(\bm{\eta}_m)=h^2\gamma_m^{h-1}\cos(h\theta_m)$, and $\nabla^2_\theta\ell_{h}^{II,2}(\bm{\eta}_m)=-h^2\gamma_m^{h}\sin(h\theta_m)$. By Assumption \ref{assum:statn}(i), there exists $\rho_1>0$ such that $\max\{|\lambda_1|,\ldots,|\lambda_r|, \gamma_1,\ldots,\gamma_s\} \leq \rho_1< \bar{\rho}$. Thus,
\[\max_{1\leq j\leq r, 1\leq m\leq s, \iota=1,2}\left \{|\nabla\ell_{h}^{I}(\lambda_j)|, \|\nabla \ell_{h}^{II,\iota}(\bm{\eta}_m)\|_2, |\nabla^2\ell_{h}^{I}(\lambda_j)|, \|\nabla^2 \ell_{h}^{II,\iota}(\bm{\eta}_m)\|_{\Fr}\right \}\leq C_{\ell}\bar{\rho}^{h}.
\]
by choosing $C_{\ell}$ dependent on $\rho_1$ and $\bar{\rho}$ such that $C_{\ell}\geq 2h^2(\rho_1/\bar{\rho})^{h-2}\bar{\rho}^{-2}$ for all $h\geq 1$. Note that such a $0<C_{\ell}<\infty$ exists and is an absolute constant.
%%%%%%%%%%%%%%%%%%%%%%%%%%%%%%%%%%%%%%%%%%%%%%%%%%%%%%%%%%%%%%%%%%%%%%%%%%%%%%%%%%

\subsection{Proof of Lemma \ref{lemma:fullrank}}
For simplicity, we omit the superscript ``*'' in all notations below. Consider the following partitions of the $\infty \times (p+J)$ matrix $\bm{L}_{\rm{stack}}(\bm{\omega})$:
\[
\bm{L}_{\rm{stack}}(\bm{\omega})=\left ( 
\begin{array}{cc}
	\bm{I}_p& \bm{0}\\
	\bm{0}&\bm{L}_{\rm{stack}}^{\ma}(\bm{\omega})
\end{array}
\right )
=\left (\begin{matrix}
	\bm{I}_p&\bm{0}\\
	\bm{0}&\bm{L}_{[1:J]}(\bm{\omega})\\
	\bm{0}&\bm{L}_{\rm{Rem}}(\bm{\omega})
\end{matrix}\right ),
\]
where $\bm{L}_{\rm{stack}}^{\ma}(\bm{\omega})=\left (\bm{L}^{I}(\bm{\lambda}), \bm{L}^{II}(\bm{\eta}), \nabla \bm{L}^{I}(\bm{\lambda}), \nabla_{\theta}  \bm{L}^{II}(\bm{\eta})\right )$ is further partitioned into two blocks, the $J\times J$ block
$\bm{L}_{[1:J]}(\bm{\omega})$ and  the $\infty\times J$ remainder block $\bm{L}_{\rm{Rem}}(\bm{\omega})$. Note that for $1\leq h\leq J$, the $h$th row of $\bm{L}_{[1:J]}(\bm{\omega})$ is 
\[
\bm{L}_{h}(\bm{\omega}):=\left (\left ( \bm{\ell}_h^{I}(\bm{\lambda}) \right )^\top,  \left (\bm{\ell}_h^{II}(\bm{\eta})\right )^\top, \left ( \nabla \bm{\ell}_h^{I}(\bm{\lambda}) \right )^\top,   \left (\nabla_{\theta}\bm{\ell}_h^{II}(\bm{\eta})\right )^\top \right ),
\]
where $\bm{\ell}_h^{I}(\bm{\lambda}) =(\lambda_1^h,\dots, \lambda_r^h)^\top$, $\nabla \bm{\ell}_h^{I}(\bm{\lambda}) =(h\lambda_1^{h-1},\dots, h\lambda_r^{h-1})^\top$, and
\begin{align*}
	\bm{\ell}_h^{II}(\bm{\eta}) &=\left (\gamma_1^h\cos(h\theta_1),\gamma_1^h\sin(h\theta_1), \dots, \gamma_s^h\cos(h\theta_s),\gamma_s^h\sin(h\theta_s) \right )^\top,\\
	\nabla_{\theta}\bm{\ell}_h^{II}(\bm{\eta}) &= \left (-h\gamma_1^h\sin(h\theta_1),h\gamma_1^h\cos(h\theta_1), \dots, -h\gamma_s^h\sin(h\theta_s),h\gamma_s^h\cos(h\theta_s) \right )^\top.
\end{align*}
For $h\geq 1$, the $h$th row of $\bm{L}_{\rm{Rem}}(\bm{\omega})$ is $\bm{L}_{J+h}(\bm{\omega})$.

By Lemma \ref{cor1}, we have $\|\bm{L}_{\rm{stack}}^{\ma}(\bm{\omega})\|_{\Fr}\leq \sqrt{J\sum_{h=1}^{\infty} C_L^2\bar{\rho}^{2h}}\leq  C_L\sqrt{J}\bar{\rho}(1-\bar{\rho})^{-1}=C_{\bar{\rho}}$. Then
\begin{align}\label{eq:sigmax}
	\sigma_{\max}(\bm{L}_{\rm{stack}}(\bm{\omega}))\leq \max\left \{1,\sigma_{\max}(\bm{L}_{\rm{stack}}^{\ma}(\bm{\omega}))\right \}\leq  \max\left \{1,\|\bm{L}_{\rm{stack}}^{\ma}(\bm{\omega})\|_{\Fr}\right \}
	\leq \max \{1, C_{\bar{\rho}} \}
\end{align}
and
\begin{align}\label{eq:maxsigmasub}
	\sigma_{\max}(\bm{L}_{[1:J]}(\bm{\omega}))\leq \|\bm{L}_{[1:J]}(\bm{\omega})\|_{\Fr}
	\leq\|\bm{L}_{\rm{stack}}^{\ma}(\bm{\omega})\|_{\Fr}	\leq  C_{\bar{\rho}}.
\end{align}

It remains to derive a lower bound of $\sigma_{\min}(\bm{L}_{\rm{stack}}(\bm{\omega}))$. To this end, we first derive a lower bound of $\sigma_{\min}(\bm{L}_{[1:J]}(\bm{\omega}))$ by lower bounding the determinant of $\bm{L}_{[1:J]}(\bm{\omega})$.
For any $(\gamma, \theta)\in[0,1)\times(-\pi/2, \pi/2)$, it can be verified that 
\[
\left (\gamma^h\cos(h\theta),\gamma^h\sin(h\theta) \right ) \underbrace{\left(\begin{matrix}1&1\\i&-i\end{matrix}\right)}_{:=\bm{C}_1} = \left ((\gamma e^{i\theta})^h,(\gamma e^{-i\theta})^h\right )
\]
and 
\[
\left (-h\gamma^h\sin(h\theta),h\gamma^h\cos(h\theta) \right ) \underbrace{\left(\begin{matrix}-i&i\\1&1\end{matrix}\right)}_{:=\bm{C}_2} =\left (h(\gamma e^{i\theta})^h,h(\gamma e^{-i\theta})^h\right ).
\]

Let $\bm{P}_1 = \diag( \bm{I}_r, \bm{C}_1, \dots, \bm{C}_1, \bm{I}_r, \bm{C}_2, \dots, \bm{C}_2)$ be a $J\times J$ block diagonal matrix consisting of two identity matrices $\bm{I}_r$ and $s$ repeated blocks of $\bm{C}_1$ and $\bm{C}_2$.
We then have $\det(\bm{P}_1)=(-2i)^{2s}=4^s$, and
\begin{align*}
	\bm{L}_{[1:J]}(\bm{\omega})\bm{P}_1 = \left(\begin{matrix}
		x_1&x_2&\cdots&x_{r+2s}&x_1&x_2&\cdots&x_{r+2s}\\
		x_1^2&x_2^2&\cdots&x_{r+2s}^2&2x_1^2&2x_2^2&\cdots&2x_{r+2s}^2\\
		\vdots&\vdots&\ddots&\vdots&\vdots&\vdots&\ddots&\vdots\\
		x_1^{J}&x_2^{J}&\cdots&x_{r+2s}^{J}&Jx_1^{J}&Jx_2^{J}&\cdots&Jx_{r+2s}^{J}
	\end{matrix}\right) := \bm{P}_2 \in \mathbb{R}^{J \times J},
\end{align*}
where $x_j = \lambda_j$ for $1\leq j\leq r$, while $x_{r+2m-1} = \gamma_m e^{i\theta_m}$ and $x_{r+2m} = \gamma_me^{-i\theta_m}$ for $ 1\leq m\leq s$, and $i$ is the imaginary unit.

We subtract the $h$th column of $\bm{P}_2$ from its  $(r+2s+h)$th column, for all $1\leq h \leq r+2s$, and obtain a matrix with the same determinant as $\bm{P}_2$ as follows,
\begin{align*}
	\bm{P}_3=\left(\begin{matrix}
		x_1&x_2&\cdots&x_{r+2s}&0&0&\cdots&0\\
		x_1^2&x_2^2&\cdots&x_{r+2s}^2&x_1^2&x_2^2&\cdots&x_{r+2s}^2\\
		\vdots&\vdots&\ddots&\vdots&\vdots&\vdots&\ddots&\vdots\\
		x_1^{J}&x_2^{J}&\cdots&x_{r+2s}^{J}&(J-1)x_1^{J}&(J-1)x_2^{J}&\cdots&(J-1)x_{r+2s}^{J}
	\end{matrix}\right).
\end{align*}
Note that $\bm{P}_3=\bm{P}_4\bm{P}_5$, where 
\[
\bm{P}_4=\left(\begin{matrix}
	1&1&\cdots&1&0&0&\cdots&0\\
	x_1&x_2&\cdots&x_{r+2s}&x_1&x_2&\cdots&x_{r+2s}\\
	x_1^2&x_2^2&\cdots&x_{r+2s}^2&2x_1^2&2x_2^2&\cdots&2x_{r+2s}^2\\
	\vdots&\vdots&\ddots&\vdots&\vdots&\vdots&\ddots&\vdots\\
	x_1^{J-1}&x_2^{J-1}&\cdots&x_{r+2s}^{J-1}&(J-1)x_1^{J-1}&(J-1)x_2^{J-1}&\cdots&(J-1)x_{r+2s}^{J-1}
\end{matrix}\right)
\] 
is a generalized Vandermonde matrix \citep{li2008special}, and $\bm{P}_5=\diag\{x_1, \dots, x_{r+2s},x_1,\dots,x_{r+2s}\}$. By \cite{li2008special}, 
$|\det(\bm{P}_4)|=\prod_{i=1}^{r+2s}x_i\prod_{1\leq k<h\leq r+2s}(x_h - x_k)^4$. 
As a result,
\begin{equation*}
	|\det(\bm{P}_2)|=|\det(\bm{P}_3)|=|\det(\bm{P}_4)||\det(\bm{P}_5)|= \prod_{h=1}^{r+2s}|x_h|^3\prod_{1\leq h<k\leq r+2s}(x_h - x_k)^4\geq \nu_{\mathrm{lower}}^{3J/2}\nu_{\mathrm{gap}}^{J(J/2-1)}.
\end{equation*}
It follows that
\begin{align}\label{eq:det2}
	|\det( \bm{L}_{[1:J]}(\bm{\omega}) )| =\frac{|\det(\bm{P}_2)|}{|\det(\bm{P}_1)|} \geq 0.25^s \nu_{\mathrm{lower}}^{3J/2}\nu_{\mathrm{gap}}^{J(J/2-1)}>0,
\end{align}
and hence $\bm{L}_{[1:J]}(\bm{\omega})$ is full-rank. Moreover, combining \eqref{eq:maxsigmasub} and \eqref{eq:det2}, we have
\begin{align}\label{eq:minsigmasub}
	\sigma_{\min}(\bm{L}_{[1:J]}(\bm{\omega})) \geq \frac{	|\det( \bm{L}_{[1:J]}(\bm{\omega}) )| }{\sigma_{\max}^{J-1}(\bm{L}_{[1:J]}(\bm{\omega}))} \geq \frac{0.25^s  \nu_{\mathrm{lower}}^{3J/2}\nu_{\mathrm{gap}}^{J(J/2-1)}}{C_{\bar{\rho}}^{J-1}}=c_{\bar{\rho}} >0.
\end{align}
%Note that the right side of \eqref{eq:minsigmasub} when $\bm{\omega}=\bm{\omega}$ is $c_{\bar{\rho}}$.

Finally, similar to \eqref{eq:sigmax}, by the Courant–Fischer theorem, it can be shown that
\begin{align*}%\label{eq:sigmin}
	\sigma_{\min}(\bm{L}_{\rm{stack}}(\bm{\omega}))\geq \min\left \{1,\sigma_{\min}(\bm{L}_{\rm{stack}}^{\ma}(\bm{\omega}))\right \}\geq  \min\left \{1, \sigma_{\min}(\bm{L}_{[1:J]}(\bm{\omega}))\right \},
\end{align*}
which, together with \eqref{eq:minsigmasub}, leads to a lower bound of 	$\sigma_{\min}(\bm{L}_{\rm{stack}}(\bm{\omega}))$.
In view of the aforementioned lower bound and  the upper bound in \eqref{eq:sigmax}, the inequalities in the lemma are verified. 
Lastly, when $r$ and $s$ are bounded from above, we immediately have $C_{\bar{\rho}}\asymp 1$ and $c_{\bar{\rho}}\asymp 1$. The proof of this lemma is complete.

%%%%%%%%%%%%%%%%%%%%%%%%%%%%%%%%%%%%%%%%%%%%%%%%%%%%%%%%%%%%%%%%%%%%%%%%%%%%%%%%%%
%%%%%%%%%%%%%%%%%%%%%%%%%%%%%%%%%%%%%%%%%%%%%%%%%%%%%%%%%%%%%%%%%%%%%%%%%%%%%%%%%%
\subsection{Proof of Lemma \ref{lemma:devb} (Deviation bound)}
Since $\widehat{\bm{\Delta}}_h=\bm{\widehat{G}}_h-\bm{G}_h^*=\bm{\widehat{D}}_h$ for $1\leq h\leq p$, we have
\begin{equation}\label{eq:devb1}
	\frac{1}{T} \left |\sum_{t=1}^{T}\langle \bm{\varepsilon}_t, \bm{\widehat{\Delta}}\bm{x}_t \rangle \right | \leq  \frac{1}{T}  \left | \sum_{t=1}^{T}\langle \bm{\varepsilon}_t, \sum_{h=1}^{p} \bm{\widehat{D}}_h\bm{y}_{t-h} \rangle \right |  +  \frac{1}{T} \left | \sum_{t=1}^{T}\langle \bm{\varepsilon}_t, \sum_{h=p+1}^\infty \widehat{\bm{\Delta}}_h \bm{y}_{t-h} \rangle \right |,
\end{equation}
where the first term on the right-hand side is suppressed if $p=0$. Without loss of generality, we assume that $p\geq1$ in what follows.
First, it can be verified that
\begin{align}\label{eq:devb2}
	\frac{1}{T}  \left | \sum_{t=1}^{T}\langle \bm{\varepsilon}_t, \sum_{h=1}^{p} \bm{\widehat{D}}_h\bm{y}_{t-h} \rangle \right |  =  \frac{1}{T}  \left | \sum_{t=1}^{T}\langle \bm{\varepsilon}_t, \widehat{\bm{D}}_{\ar}\bm{x}_t^p \rangle \right | 
	& = 	\left |\Big \langle \frac{1}{T} \sum_{t=1}^{T}\bm{\varepsilon}_t (\bm{x}_t^{p})^{\top}, \widehat{\bm{D}}_{\ar}  \Big \rangle \right | \notag\\
	&\leq \|\widehat{\bm{d}}_{\ar}\|_1  \left \| \frac{1}{T} \sum_{t=1}^{T}\bm{\varepsilon}_t (\bm{x}_t^{p})^{\top} \right \|_{\max},
\end{align}
where $\bm{x}_t^p=(\bm{y}_{t-1}^\top, \dots, \bm{y}_{t-p}^\top)^\top$.
For the second term on the right-hand side of \eqref{eq:devb1}, since 
\begin{align*}
	\sum_{h=p+1}^\infty \widehat{\bm{\Delta}}_h \bm{y}_{t-h} &=  \left [ \widehat{\bm{G}}_{\ma} \{\bm{L}^{\ma}(\widehat{\bm{\omega}})\otimes \bm{I}_N\}^\top - \bm{G}_{\ma}^{*} \{ \bm{L}^{\ma}(\bm{\omega}^*)\otimes \bm{I}_N\}^\top  \right ] \bm{x}_{t-p}\\
	&=  \widehat{\bm{D}}_{\ma} \{ \bm{L}^{\ma}(\widehat{\bm{\omega}}) \otimes \bm{I}_N\}^\top \bm{x}_{t-p} + \bm{G}_{\ma}^{*} \left [\left \{\bm{L}^{\ma}(\widehat{\bm{\omega}}) -\bm{L}^{\ma}(\bm{\omega}^*) \right \}\otimes\bm{I}_N\right ]^\top \bm{x}_{t-p},
\end{align*}
we have
\begin{align}\label{eq:devb3}
	&\frac{1}{T} \left | \sum_{t=1}^{T}\langle \bm{\varepsilon}_t, \sum_{h=p+1}^\infty \widehat{\bm{\Delta}}_h \bm{y}_{t-h} \rangle \right | \notag\\
	& \hspace{5mm}  \leq 
	\left |\Big \langle \frac{1}{T} \sum_{t=1}^{T}\bm{\varepsilon}_t\bm{x}_{t-p}^\top \{ \bm{L}^{\ma}(\widehat{\bm{\omega}}) \otimes \bm{I}_N\}, \widehat{\bm{D}}_{\ma} \Big \rangle \right | \notag\\
	&\hspace{10mm} + 
	\left |\Big\langle \frac{1}{T}\sum_{t=1}^{T} \bm{\varepsilon}_t \bm{x}_{t-p}^\top  \left [\left \{\bm{L}^{\ma}(\widehat{\bm{\omega}}) -\bm{L}^{\ma}(\bm{\omega}^*) \right \}\otimes\bm{I}_N\right ], \bm{G}_{\ma}^{*} \Big\rangle\right | \notag\\
	& \hspace{5mm}  \leq  \|\widehat{\bm{d}}_{\ma}\|_1 \sup_{\bm{\omega}\in\bm{\Omega}} \left \| \frac{1}{T} \sum_{t=1}^{T}\bm{\varepsilon}_t \bm{x}_{t-p}^\top \{ \bm{L}^{\ma}(\bm{\omega}) \otimes \bm{I}_N\}\right \|_{\max} \notag\\
	&\hspace{10mm} +  \|\bm{g}_{\ma}^{*}\|_1 \sup_{\bm{\phi}\in \bm{\Phi}_1} \left \| \frac{1}{T}\sum_{t=1}^{T} \bm{\varepsilon}_t \bm{x}_{t-p}^\top  \left [\left \{\bm{L}^{\ma}(\bm{\omega}^*+\bm{\phi}) -\bm{L}^{\ma}(\bm{\omega}^*) \right \}\otimes\bm{I}_N\right ] \right \|_{\max}, 
\end{align}
where we use the property that $\bm{\widehat{\phi}}\in\bm{\Phi}_1$.

To prove this lemma, it suffices to establish the following intermediate results:
\begin{itemize}
	\item [(i)] With probability at least $1-4 e^{-2\log(Np)}$,
	\begin{equation}\label{eq:lemdvb0}
		\left \| \frac{1}{T} \sum_{t=1}^{T}\bm{\varepsilon}_t (\bm{x}_t^{p})^{\top} \right \|_{\max} \leq C_1 \sqrt{\frac{\kappa_2\lambda_{\max}(\bm{\Sigma}_{\varepsilon})\log(Np)}{T}},
	\end{equation}
	where $C_1>0$ is an absolute constant.
	
	\item [(ii)] With probability at least  $1- 5  e^{-4\log N}$,
	\begin{equation}\label{eq:lemdvb1}
		\sup_{\bm{\omega}\in\bm{\Omega}} \left \| \frac{1}{T}   \sum_{t=1}^{T}\bm{\varepsilon}_t \bm{x}_{t-p}^\top \{ \bm{L}^{\ma}(\bm{\omega}) \otimes \bm{I}_N\}\right \|_{\max} \leq C_2 \sqrt{\frac{\kappa_2\lambda_{\max}(\bm{\Sigma}_{\varepsilon})\log N}{T}}
	\end{equation}
	and
	\begin{equation}\label{eq:lemdvb2}
		\sup_{\bm{\phi}\in\bm{\Phi}_1} \frac{ \left \| \sum_{t=1}^{T} \bm{\varepsilon}_t \bm{x}_{t-p}^\top  \left [\left \{\bm{L}^{\ma}(\bm{\omega}^*+\bm{\phi}) -\bm{L}^{\ma}(\bm{\omega}^*) \right \}\otimes\bm{I}_N\right ] \right \|_{\max} }{T \|\bm{\phi}\|_2 }  \leq C_3 \sqrt{\frac{\kappa_2\lambda_{\max}(\bm{\Sigma}_{\varepsilon})\log N}{T}},
	\end{equation}
	where $C_2, C_3>0$ are absolute constants.
\end{itemize}

\noindent\textbf{Proof of \eqref{eq:lemdvb0}:}  Note that 
\[
\left \|  \frac{1}{T}  \sum_{t=1}^{T}\bm{\varepsilon}_t (\bm{x}_t^{p})^{\top} \right \|_{\max} =  \max_{1\leq i, j\leq N,1\leq k\leq p} \left | \frac{1}{T} \sum_{t=1}^{T}\varepsilon_{i,t} y_{j,t-k} \right |. 
\]
We begin by considering any fixed triplet $(i, j, k)$ such that $1\leq i,j\leq N$ and $1\leq k\leq p$.
Let $\bm{\iota}_i\in\mathbb{R}^N$ be the $i$th unit vector, which consists of all zeros except that the $i$th entry is one. Applying Lemma \ref{lemma:hansonw} with $T_0=-k$, $T_1=T$, $\bm{w}_t=\bm{y}_t$, and $\bm{M}=\bm{\iota}_j^\top$, together with Lemma \ref{lemma:Wcov}(i), we have
\begin{equation*}
	\mathbb{P}\left \{\left |\frac{1}{T}\sum_{t=1}^{T} y_{j,t-k}^2 - \mathbb{E}(y_{j,t-k}^2)\right | \geq \eta \sigma^2 \lambda_{\max}(\bm{\Sigma}_\varepsilon)\mu_{\max}(\bm{\Psi}_*)\right \} \leq 2e^{-c_{\HW}\min(\eta, \eta^2)T},
\end{equation*}
for any $\eta>0$. In addition, by  Lemma \ref{lemma:Wcov}(i), $\mathbb{E}(y_{j,t-k}^2)=\bm{\iota}_j^\top\mathbb{E}(\bm{y}_{t-k}\bm{y}_{t-k}^\top) \bm{\iota}_j\leq 	\lambda_{\max}(\bm{\Sigma}_\varepsilon)\mu_{\max}(\bm{\Psi}_*) = \kappa_2$. 
%Note that 
%\[
%\kappa_1 \leq \lambda_{\min}(\bm{\Sigma}_\varepsilon)\mu_{\min}(\bm{\Psi}_*) \leq \lambda_{\max}(\bm{\Sigma}_\varepsilon)\mu_{\max}(\bm{\Psi}_*)\leq \kappa_2.
%\]
Thus, by taking $\eta=(2\sigma^2)^{-1}$, we have
\begin{equation}\label{eq:reunivar}
	\mathbb{P}\left (\frac{1}{T}\sum_{t=1}^{T} y_{j,t-k}^2 \geq  1.5\kappa_2 \right ) \leq 2e^{-c T},
\end{equation}
where $c=c_{\HW}\min\{(2\sigma^2)^{-1},(2\sigma^2)^{-2}\}$. Then  we can show that for any $K>0$,
\begin{align}\label{eq:dvb2}
	&\mathbb{P}\left (\left | \frac{1}{T} \sum_{t=1}^{T}\varepsilon_{i,t} y_{j,t-k} \right |  \geq K \right ) \notag\\
	& \hspace{5mm}\leq 	\mathbb{P}\left ( \left | \sum_{t=1}^{T}\varepsilon_{i,t} y_{j,t-k}\right |\geq KT, \; \sum_{t=1}^{T}y_{j,t-k}^2 \leq 1.5\kappa_2 T\right ) + \mathbb{P}\left (\frac{1}{T} \sum_{t=1}^{T}y_{j,t-k}^2 \geq 1.5\kappa_2 \right )\notag\\
	& \hspace{5mm}\leq 2e^{-K^2T/\{3\sigma^2 \kappa_2 \lambda_{\max}(\bm{\Sigma}_{\varepsilon})\}} + 2e^{-c T},
\end{align}
where we applied  Lemma \ref{lemma:martgl}(i) with $a=KT$ and $b=1.5\kappa_2 T$ in the last inequality. 
%it follows that
%\[
%\mathbb{P}\left ( \left | \frac{1}{T} \sum_{t=1}^{T}\varepsilon_{i,t} y_{j,t-k} \right |  \geq \frac{\kappa_1^{3/2}\lambda_{\max}^{1/2}(\bm{\Sigma}_{\varepsilon})}{\kappa_2} \right)\leq e^{-\sigma^{-2}(\kappa_1/\kappa_2)^2T} + 2e^{-c (\kappa_1/\kappa_2)^2T}\leq 3e^{-\min\{\sigma^{-2}, c\} (\kappa_1/\kappa_2)^2T}.
%\]
As a result, by applying \eqref{eq:dvb2} with
%$K=\kappa_1^{3/2}\lambda_{\max}^{1/2}(\bm{\Sigma}_{\varepsilon})/\kappa_2$
\[
K=\sqrt{\frac{6\sigma^2\kappa_2\lambda_{\max}(\bm{\Sigma}_{\varepsilon})\log(N^2 p)}{T}},
\]
if $T\geq 2c^{-1} \log(N^2 p)$, then it can be verified that
\begin{align}\label{eq:dvb3}
	&\mathbb{P}\left \{\max_{1\leq i, j\leq N,1\leq k\leq p} \left | \frac{1}{T} \sum_{t=1}^{T}\varepsilon_{i,t} y_{j,t-k} \right |  \geq \sqrt{\frac{6\sigma^2\kappa_2\lambda_{\max}(\bm{\Sigma}_{\varepsilon})\log(N^2 p)}{T}} \right \} \notag\\
	&\hspace{5mm}\leq N^2p  \max_{1\leq i, j\leq N,1\leq k\leq p} \mathbb{P}\left \{ \left | \frac{1}{T} \sum_{t=1}^{T}\varepsilon_{i,t} y_{j,t-k} \right |  \geq  \sqrt{\frac{6\sigma^2\kappa_2\lambda_{\max}(\bm{\Sigma}_{\varepsilon})\log(N^2 p)}{T}} \right \}\notag\\
	&\hspace{5mm} \leq 2e^{-\log(N^2p)}+ 2e^{-c T + \log(N^2p)}
	\leq 4 e^{-\log(N^2p)}.
\end{align}
Hence, \eqref{eq:lemdvb0} proved.

\smallskip
\noindent\textbf{Proof of \eqref{eq:lemdvb1}:}
Note that by Assumption \ref{assum:statn}(i), for all $\bm{\omega}\in\bm{\Omega}$, we have $0<|\ell_{h,k}(\bm{\omega})|\leq \bar{\rho}^{h-p}$ if $h\geq p+1$ and $p+1\leq k\leq d$. Then we can show that
\begin{align}\label{eq:dvb1}
	& \sup_{\bm{\omega}\in\bm{\Omega}} \left \| \frac{1}{T}   \sum_{t=1}^{T}\bm{\varepsilon}_t \bm{x}_{t-p}^\top \{ \bm{L}^{\ma}(\bm{\omega}) \otimes \bm{I}_N\}\right \|_{\max} \notag \\
	&\hspace{5mm}= \sup_{\bm{\omega}\in\bm{\Omega}}\max_{1\leq i, j\leq N,p+1\leq k\leq d}\left | \frac{1}{T} \sum_{t=1}^{T}\varepsilon_{i,t} \sum_{h=p+1}^\infty\ell_{h,k}(\bm{\omega})y_{j,t-h} \right |\notag \\
	&\hspace{5mm}= \sup_{\bm{\omega}\in\bm{\Omega}}\max_{1\leq i, j\leq N,p+1\leq k\leq d}\left | \sum_{h=p+1}^\infty\ell_{h,k}(\bm{\omega}) \left (\frac{1}{T}  \sum_{t=1}^{T} \varepsilon_{i,t}y_{j,t-h} \right ) \right |\notag\\
	&\hspace{5mm}\leq  \sum_{h=p+1}^\infty \sup_{\bm{\omega}\in\bm{\Omega}}\max_{p+1\leq k\leq d}|\ell_{h,k}(\bm{\omega})| \max_{1\leq i, j\leq N}  \left | \frac{1}{T} \sum_{t=1}^{T}\varepsilon_{i,t} y_{j,t-h} \right |\notag \\
	&\hspace{5mm}\leq 
	\sum_{h=p+1}^\infty \bar{\rho}^{h-p} \max_{1\leq i, j\leq N}  \left | \frac{1}{T} \sum_{t=1}^{T}\varepsilon_{i,t} y_{j,t-h} \right |. 
\end{align}

To establish an upper bound for the weighted infinite sum in \eqref{eq:dvb1}, we first consider a fixed  triplet $(i, j, h)$ such that  $1\leq i,j\leq N$ and $h\geq p+1$. By the same arguments as  those for \eqref{eq:reunivar} except that we take $\eta=h-p$, we can show that
\begin{equation}\label{eq:reunivarj}
	\mathbb{P}\left \{\frac{1}{T}\sum_{t=1}^{T} y_{j,t-h}^2 \geq  \{(h-p)\sigma^2+1\}\kappa_2 \right \} \leq 2e^{-c(h-p) T}.
\end{equation}
Similar to \eqref{eq:dvb2}, for any $K>0$, it follows that
\begin{align*}
	&\mathbb{P}\left (\left | \frac{1}{T} \sum_{t=1}^{T}\varepsilon_{i,t} y_{j,t-h} \right |  \geq K \right ) \notag\\
	& \hspace{5mm}\leq 	\mathbb{P}\left [\left |\sum_{t=1}^{T}\varepsilon_{i,t} y_{j,t-h} \right |\geq KT, \; \sum_{t=1}^{T}y_{j,t-h}^2 \leq \{(h-p)\sigma^2+1\}\kappa_2 T\right ] + \mathbb{P}\left [\frac{1}{T} \sum_{t=1}^{T}y_{j,t-h}^2 \geq \{(h-p)\sigma^2+1\}\kappa_2 \right ]\notag\\
	& \hspace{5mm}\leq 2 e^{-K^2T/[2\{(h-p)\sigma^2+1\}\sigma^2 \kappa_2 \lambda_{\max}(\bm{\Sigma}_{\varepsilon})]} + 2e^{-c_{\HW} (h-p)T}.
\end{align*}
Applying the above result with 
\[
K=\sqrt{\frac{4\{(h-p)\sigma^2+1\}(h-p+1)\sigma^2\kappa_2\lambda_{\max}(\bm{\Sigma}_{\varepsilon})\log (N^2)}{T}},
\]
if $T\geq 4c^{-1} \log (N^2)$, similar to \eqref{eq:dvb3}, for any fixed  $h\geq p+1$, we have
\begin{align*}
	&\mathbb{P}\left [\max_{1\leq i, j\leq N} \left | \frac{1}{T} \sum_{t=1}^{T}\varepsilon_{i,t} y_{j,t-h} \right |  \geq \sqrt{\frac{4\{(h-p)\sigma^2+1\}(h-p+1)\sigma^2\kappa_2\lambda_{\max}(\bm{\Sigma}_{\varepsilon})\log (N^2)}{T}} \right ]\notag\\
	&\hspace{5mm}\leq N^2  \max_{1\leq i, j\leq N} \mathbb{P}\left [ \left | \frac{1}{T} \sum_{t=1}^{T}\varepsilon_{i,t} y_{j,t-h} \right |  \geq \sqrt{\frac{4\{(h-p)\sigma^2+1\}(h-p+1)\sigma^2\kappa_2\lambda_{\max}(\bm{\Sigma}_{\varepsilon})\log (N^2)}{T}} \right ]\notag\\
	&\hspace{5mm} \leq 2e^{-2(h-p+1)\log(N^2)+\log (N^2)}+ 2e^{-c_{\HW} (h-p)T + \log (N^2)}
	\leq 4 e^{-2(h-p)\log (N^2)}.
\end{align*}
Note that $\{(h-p)\sigma^2+1\}(h-p+1)\sigma^2\leq \{2(h-p)\sigma^2+1\}^2$. Thus,
\[
\mathbb{P}\left [\max_{1\leq i, j\leq N} \left | \frac{1}{T} \sum_{t=1}^{T}\varepsilon_{i,t} y_{j,t-h} \right |  \geq \{2(h-p)\sigma^2+1\} \sqrt{\frac{4\kappa_2\lambda_{\max}(\bm{\Sigma}_{\varepsilon})\log (N^2)}{T}} \right ]  \leq  4 e^{-2(h-p)\log (N^2)},
\]
which  can be further strengthened to a union bound for all $h\geq p+1$ as follows:
\begin{align}\label{eq:dvb4}
	&\mathbb{P}\left [\forall h\geq p+1:  \max_{1\leq i, j\leq N} \left | \frac{1}{T} \sum_{t=1}^{T}\varepsilon_{i,t} y_{j,t-h} \right |  \geq  \{2(h-p)\sigma^2+1\} \sqrt{\frac{4\kappa_2\lambda_{\max}(\bm{\Sigma}_{\varepsilon})\log (N^2)}{T}}  \right ]\notag\\
	&\hspace{5mm}  \leq \sum_{h=p+1}^{\infty} 2 e^{-2(h-p)\log (N^2)} \leq 5 e^{-4\log N},
\end{align}
where the last inequality holds as long as $N\geq 2$.
Combining \eqref{eq:dvb1} with \eqref{eq:dvb4}, we have
\begin{align*}
	\sup_{\bm{\omega}\in\bm{\Omega}} \left \| \frac{1}{T} \sum_{t=1}^{T}\bm{\varepsilon}_t \bm{x}_{t-p}^\top \{ \bm{L}^{\ma}(\bm{\omega}) \otimes \bm{I}_N\}\right \|_{\max} 
	&\leq \sum_{h=p+1}^{\infty}  \bar{\rho}^{h-p} \{2(h-p)\sigma^2+1\} \sqrt{\frac{4\kappa_2\lambda_{\max}(\bm{\Sigma}_{\varepsilon})\log (N^2)}{T}} \\
	& \lesssim \sqrt{\frac{\kappa_2\lambda_{\max}(\bm{\Sigma}_{\varepsilon})\log N}{T}},
\end{align*}
with probability at least $1- 5 e^{-4\log N}$. Thus, \eqref{eq:lemdvb1} is proved.

\smallskip
\noindent\textbf{Proof of \eqref{eq:lemdvb2}:}
For any $h\geq1$ and $1\leq k\leq r$, by the Taylor expansion, we have
\[
\ell_{h}^{I}(\lambda_k)-\ell_{h}^{I}(\lambda_k^*)=\nabla\ell_{h}^{I}(\lambda_k^*) (\lambda_k-\lambda_k^*) +\frac{1}{2}\nabla^2\ell_{h}^{I}(\widetilde{\lambda}_k) (\lambda_k-\lambda_k^*)^2, 
\] 
where $\widetilde{\lambda}_k$ lies between $\lambda_k^*$ and $\lambda_k$. Then, by Lemma \ref{cor1}, for any $\bm{\omega}=\bm{\omega}^*+\bm{\phi}$ with $\bm{\phi}\in \bm{\Phi}_1$, 
\[
\max_{1\leq k\leq r}|\ell_{h}^{I}(\lambda_k)-\ell_{h}^{I}(\lambda_k^*)|\leq C_{\ell}\bar{\rho}^h \|\bm{\phi}\|_2 +\frac{1}{2}C_{\ell}\bar{\rho}^h \|\bm{\phi}\|_2^2 \leq  2C_{\ell}\bar{\rho}^h \|\bm{\phi}\|_2, \quad \forall h\geq1,
\]
where we used the fact that $\|\bm{\phi}\|_2\leq c_{\bm{\omega}}\leq 2$ for all $\bm{\phi}\in \bm{\Phi}_1$. By a similar argument, for any $\bm{\omega}=\bm{\omega}^*+\bm{\phi}$ with $\bm{\phi}\in \bm{\Phi}_1$, we can show that
\[
\max_{1\leq k\leq s, \iota=1,2}|\ell_{h}^{II,\iota}(\bm{\eta}_k)-\ell_{h}^{II,\iota}(\bm{\eta}_k^*)|\leq  2C_{\ell}\bar{\rho}^h \|\bm{\phi}\|_2, \quad \forall h\geq1.
\]
As a result,
\begin{equation*}
	\sup_{\bm{\phi}\in\bm{\Phi}_1}\max_{p+1\leq k\leq d} \frac{|\ell_{h,k}(\bm{\omega}^*+\bm{\phi})-\ell_{h,k}(\bm{\omega}^*)|}{ \|\bm{\phi}\|_2} \leq 2C_{\ell}\bar{\rho}^{h-p}, \quad \forall h\geq p+1.
\end{equation*}
Then it follows that 
\begin{align*}
	&\sup_{\bm{\phi}\in\bm{\Phi}_1} \frac{ \left \| \sum_{t=1}^{T} \bm{\varepsilon}_t \bm{x}_{t-p}^\top  \left [\left \{\bm{L}^{\ma}(\bm{\omega}^*+\bm{\phi}) -\bm{L}^{\ma}(\bm{\omega}^*) \right \}\otimes\bm{I}_N\right ] \right \|_{\max} }{T \|\bm{\phi}\|_2 }  \\
	&\hspace{5mm}
	= \sup_{\bm{\phi}\in\bm{\Phi}_1} \max_{1\leq i, j\leq N,p+1\leq k\leq d}\frac{\left | \sum_{t=1}^{T}\varepsilon_{i,t} \sum_{h=p+1}^\infty\{\ell_{h,k}(\bm{\omega}^*+\bm{\phi})-\ell_{h,k}(\bm{\omega}^*)\}y_{j,t-h} \right | }{T \|\bm{\phi}\|_2 }\\
	&\hspace{5mm}\leq 
	\sum_{h=p+1}^\infty \sup_{\bm{\phi}\in\bm{\Phi}_1}	\max_{p+1\leq k\leq d} \frac{|\ell_{h,k}(\bm{\omega}^*+\bm{\phi})-\ell_{h,k}(\bm{\omega}^*)|}{ \|\bm{\phi}\|_2}  \max_{1\leq i, j\leq N}\left | \frac{1}{T} \sum_{t=1}^{T}\varepsilon_{i,t} y_{j,t-h} \right | \\
	& \hspace{5mm} \leq 2 C_{\ell}\sum_{h=p+1}^\infty \bar{\rho}^{h-p}  \max_{1\leq i, j\leq N}\left | \frac{1}{T} \sum_{t=1}^{T}\varepsilon_{i,t} y_{j,t-h} \right |,
\end{align*}
which is similar to \eqref{eq:dvb1}. Similar to the method for \eqref{eq:lemdvb1},  we accomplish the proof of \eqref{eq:lemdvb2} by combining the above result with \eqref{eq:dvb4}.

\smallskip

Lastly, in view of \eqref{eq:devb1}--\eqref{eq:lemdvb2}, and the fact that $\|\widehat{\bm{d}}_{\ar}\|_1+\|\widehat{\bm{d}}_{\ma}\|_1=\|\widehat{\bm{d}}\|_1$, we accomplish the proof of this lemma by taking $C_{\dev}=\max_{1\leq i\leq 3}C_i>0$ and combining the tail probabilities for \eqref{eq:lemdvb0}--\eqref{eq:lemdvb2}.

%%%%%%%%%%%%%%%%%%%%%%%%%%%%%%%%%%%%%%%%%%%%%%%%%%%%%%%%%%%%%%%%%%%%%%%%%%%%%%%%%%
\subsection{Proof of Lemma \ref{lemma:rsclasso} (Restricted strong convexity)}\label{asec:rsc}
By  the proof of Proposition \ref{prop:perturb}, we can write 
\[
\bm{\Delta}=\bm{D}\{  \bm{L}(\bm{\omega}^*) \otimes \bm{I}_N \}^\top + \bm{M}(\bm{\phi}) \{ \bm{P}(\bm{\omega}^*)\otimes \bm{I}_N\}^\top +(\bm{0}_{N\times Np}, \bm{R}),
\]  
where the remainder term $\bm{R}$ depends on both $\bm{\phi}$ and $\bm{D}$; see \eqref{eq:stackH} and \eqref{eq:Delta} for details.

Let  $\bm{Q}(\bm{\phi})=\left ( q_{h,j}(\bm{\phi}) \right )$ and $\bm{S}(\bm{\phi})=\left ( s_{h,j}(\bm{\phi}) \right )$ be $\infty\times (r+2s)$ matrices whose entries are
\begin{align*}
	q_{h,j}(\bm{\phi})&=\nabla\ell_{h}^{I}(\lambda_j^*)  (\lambda_j-\lambda_j^*)+\frac{1}{2}\nabla^2\ell_{h}^{I}(\widetilde{\lambda}_j) (\lambda_j-\lambda_j^*)^2,\\
	s_{h,j}(\bm{\phi})&= \frac{1}{2}\nabla^2\ell_{h}^{I}(\widetilde{\lambda}_j) (\lambda_j-\lambda_j^*)^2,\\
	q_{h,r+2(m-1)+\iota}(\bm{\phi})&=(\bm{\eta}_m-\bm{\eta}_m^*)^\top \nabla \ell_{h}^{II,\iota}(\bm{\eta}_m^*) +\frac{1}{2}(\bm{\eta}_m-\bm{\eta}_m^*)^{\prime}\nabla^2 \ell_{h}^{II,\iota}(\widetilde{\bm{\eta}}_m)(\bm{\eta}_m-\bm{\eta}_m^*),\\
	s_{h,r+2(m-1)+\iota}(\bm{\phi})&=  \frac{1}{2}(\bm{\eta}_m-\bm{\eta}_m^*)^{\prime}\nabla^2 \ell_{h}^{II,\iota}(\widetilde{\bm{\eta}}_m)(\bm{\eta}_m-\bm{\eta}_m^*),
\end{align*}
where $h\geq1$, $1\leq j\leq r$, $1\leq m\leq s$, $\iota=1,2$, and  $\widetilde{\lambda}_j$'s and $\widetilde{\bm{\eta}}_m$'s are defined as in \eqref{eq:delta}; that is, $\widetilde{\lambda}_j$ lies between $\lambda_j^*$ and $\lambda_j$ for $1 \leq j \leq r$, and $\widetilde{\bm{\eta}}_m$ lies between $\bm{\eta}^*_m$ and $\bm{\eta}_m$ for $1 \leq m \leq s$, and we suppress their dependence on $h$ for notational simplicity. 
Then, by the definition of $\bm{R}_h$'s in \eqref{eq:Rhs}, we can write
\[
\bm{R}=\bm{D}_{\ma}\{\bm{Q}(\bm{\phi}) \otimes\bm{I}_N\}^\top+ \bm{G}_{\ma}^{*}\{ \bm{S}(\bm{\phi})\otimes\bm{I}_N\}^\top.
\]

Denote 
\begin{align}\label{eq:ts-z}
	\begin{split}
		\bm{Z}&=(\bm{z}_1,\dots, \bm{z}_T), \quad \bm{z}_t =\left \{\bm{L}(\bm{\omega}^*)\otimes \bm{I}_N \right \}^\top\bm{x}_{t},\\
		\bm{V}&=(\bm{v}_1,\dots, \bm{v}_T), \quad \bm{v}_t =\left \{\bm{P}(\bm{\omega}^*)\otimes \bm{I}_N \right \}^\top\bm{x}_{t},\\
		\bm{H}(\bm{\phi})&=(\bm{h}_{1}(\bm{\phi}),\dots, \bm{h}_{T}(\bm{\phi})), \quad \bm{h}_t(\bm{\phi})=\left \{\bm{Q}(\bm{\phi})\otimes \bm{I}_N \right \}^\top\bm{x}_{t-p},\\ 
		\bm{B}(\bm{\phi})&=(\bm{b}_{1}(\bm{\phi}),\dots, \bm{b}_{T}(\bm{\phi})), \quad \bm{b}_{t}(\bm{\phi})=\left \{\bm{S}(\bm{\phi})\otimes \bm{I}_N \right \}^\top\bm{x}_{t-p},
	\end{split}
\end{align}
and $\bm{X}=(\bm{x}_1, \dots, \bm{x}_T)$. Combining  all results above, we have
\begin{align*}%\label{eq:split}
	\bm{\Delta}\bm{x}_t&=\left [\bm{D}\{\bm{L}(\bm{\omega}^*)\otimes\bm{I}_N\}^\top+\bm{M}(\bm{\phi})\{\bm{P}(\bm{\omega}^*)\otimes\bm{I}_N\}^\top \right ]\bm{x}_t\\
	&\hspace{5mm} +\left [\bm{D}_{\ma}\{ \bm{Q}(\bm{\phi})\otimes\bm{I}_N\}^\top 
	+\bm{G}_{\ma}^{*}\{ \bm{S}(\bm{\phi})\otimes\bm{I}_N\}^\top\right ]\bm{x}_{t-p}\\
	&=\bm{D} \bm{z}_t + \bm{M}(\bm{\phi})\bm{v}_t +\bm{D}_{\ma}\bm{h}_{t}(\bm{\phi}) + \bm{G}_{\ma}^{*} \bm{b}_{t}(\bm{\phi}),
\end{align*}
or equivalently,
\[
\bm{\Delta}\bm{X}=\bm{D} \bm{Z} + \bm{M}(\bm{\phi}) \bm{V} +\bm{D}_{\ma}\bm{H}(\bm{\phi})+\bm{G}_{\ma}^{*}\bm{B}(\bm{\phi}).
\]
By the triangle inequality and the fact that $(|x| + |y|) / 2 \leq \sqrt{x^2 + y^2}$ for any $x,y\in\mathbb{R}$, we have
\begin{equation}\label{eq:triangle}
	\|\bm{\Delta}\bm{X}\|_{\Fr} \geq 0.5\|\bm{D} \bm{Z}\|_{\Fr} +0.5 \|\bm{M}(\bm{\phi}) \bm{V}\|_{\Fr}-\|\bm{D}_{\ma}\bm{H}(\bm{\phi})\|_{\Fr}-\|\bm{G}_{\ma}^{*}\bm{B}(\bm{\phi})\|_{\Fr}.
\end{equation}

We need to lower bound the first term and upper bound the other three terms on the right-hand side of \eqref{eq:triangle}. We state the following intermediate results for deriving these bounds and relegate their proofs to the end of this subsection:
\begin{itemize}
	\item [(i)] If 
	$T\geq 4c_1^{-1} (r+2s)^2(\kappa_2/\widetilde{\kappa}_1)^2 \log(Nd)$, with probability at least $1-2e^{-0.5c_1\widetilde{\kappa}_1^2 T/\{(r+2s)^2\kappa_2^2\}}$,
	\[
	\frac{1}{\sqrt{T}}\|\bm{D} \bm{Z}\|_{\Fr} \geq\frac{\sqrt{\widetilde{\kappa}_1} }{2} \|\bm{d}\|_{2} - \sqrt{\frac{ (r+2s)^2\kappa_2^2 \log(Nd)}{c_1 \widetilde{\kappa}_1  T} } \|\bm{d}\|_1, \quad \forall \bm{d}\in \mathbb{R}^{N^2 d},
	\]
	where $c_1>0$ is an absolute constant, and $\bm{d}=\vect(\bm{D})$.
	\item [(ii)] If
	$T\geq 2  c_2^{-1} (r+2s)^3 (\kappa_2/\widetilde{\kappa}_1)^2
	\max\left\{\log (12 u_\phi^3/l_\phi^3)+0.5\log(3 \widetilde{\kappa}_2/\widetilde{\kappa}_1), \log (6 u_\phi/l_\phi) \right\}$, with probability at least $1- 2e^{-0.5 c_2\widetilde{\kappa}_1^2 T/\{(r+2s)^2\kappa_2^2\}}$,
	\[
	\frac{ \widetilde{\kappa}_1 l_\phi^2}{8 u_\phi^4}
	\leq 
	\inf_{\bm{\phi}\in\bm{\Phi}} \frac{\|\bm{M}(\bm{\phi}) \bm{V}\|_{\Fr}^2}{T\| \bm{\phi}\|_2^2}
	\leq 
	\sup_{\bm{\phi}\in\bm{\Phi}} \frac{\|\bm{M}(\bm{\phi}) \bm{V}\|_{\Fr}^2}{T\| \bm{\phi}\|_2^2} 
	\leq
	\frac{6\widetilde{\kappa}_2 u_\phi^2}{l_\phi^4}, 
	\] 
	where $c_2>0$ is an absolute constant, $l_\phi=(\sqrt{2}\overline{\alpha}_\ma)^{-1} \min_{1\leq k\leq s}\gamma_{k}^*$, and $u_\phi=\underline{\alpha}_\ma^{-1}$.
	\item [(iii)]  If  $T \geq 4c_{\HW}^{-1}\log\{N(r+2s)\}$, then with probability at least $ 1-4e^ {-0.5c_{\HW} T}$, 
	\begin{equation*}
		\sup_{\bm{\phi}\in \bm{\Phi}_1} \frac{ \|\bm{D}_{\ma} \bm{H}(\bm{\phi})\|_{\Fr}^2 }{T \|\bm{\phi}\|_2^2} \leq C_4 (r+2s) \widetilde{\kappa}_2 \left [ \|\bm{d}_{\ma}\|_{2}^2+ \frac{ 4 \log\{N(r+2s)\}}{c_{\HW} T} \|\bm{d}_{\ma}\|_1^2 \right ], \quad \forall \bm{d}_{\ma}\in \mathbb{R}^{N^2 (r+2s)},
	\end{equation*}
	where  $c_{\HW}>0$ is defined as in Lemma \ref{alem:rsc3aux}, and $C_4>0$ is an absolute constant.
	\item [(iv)] If  $T\geq 2c_{\HW}^{-1}\log N$, then with probability at least  $1-4 e^{-0.5c_{\HW} T}$,
	\[
	\sup_{\bm{\phi}\in \bm{\Phi}_1}	\frac{\|\bm{G}_{\ma}^{*} \bm{B}(\bm{\phi})\|_{\Fr}^2}{T \|\bm{\phi}\|_2^4} \leq  C_4 \overline{\alpha}_\ma^2  (r+2s)^2 \widetilde{\kappa}_2.
	\] 
\end{itemize}

Now we prove this lemma based on the above results. First note that $\bm{\Delta}=\bm{\Delta}(\bm{\phi},\bm{d})$ is linear in $\bm{d}$ for any fixed $\bm{\phi}$. That is, for any $\alpha\neq 0$, it holds 
\[
\alpha\bm{\Delta}(\bm{\phi},\bm{d})=(\alpha\bm{D}+\alpha\bm{G}^*)\{\bm{L}(\bm{\phi}+\bm{\omega}^*)\otimes\bm{I}_N\}^\top-\alpha\bm{G}^*\{\bm{L}(\bm{\omega}^*)\otimes\bm{I}_N\}^\top=\bm{\Delta}(\bm{\phi},\alpha\bm{d}),
\] 
where we suppress the dependence of $\bm{\Delta}$ on $\bm{\omega}^*$ and $\bm{g}^*$ (or $\alpha \bm{g}^*$) since they are fixed. As a result, it suffices  to show that the conclusion stated in this lemma holds uniformly over the intersection of $\bm{\Upsilon}$ and $\pazocal{S}(\delta)$ with high probability, where $\pazocal{S}(\delta)= \{\bm{\Delta} \in \mathbb{R}^{N\times \infty} \mid\|\bm{\Delta}\|_{\Fr}=\delta\}$ is a sphere, for some radius $\delta>0$ such that $\bm{\Upsilon}\cap\pazocal{S}(\delta)$ is nonempty. The reason is that the same conclusion will remain true if we multiply $\bm{\Delta}$ by any $\alpha\neq 0$.

We restrict our attention to $\bm{\Delta}=\bm{\Delta}(\bm{\phi},\bm{d})\in\bm{\Upsilon}\cap\pazocal{S}(\delta)$ with the  radius $\delta\in(0, c_{\Delta} c_{\bm{\omega}})$, where $c_{\Delta}>0$ is defined as in \eqref{eq:prop2} in the proof of Proposition \ref{prop:perturb}. The specific $\delta$ will be chosen later.  
Note that by \eqref{eq:prop2},  for a sufficiently small $\delta$, if $\|\bm{\Delta} \|_{\Fr}= \delta$, then
\begin{equation}\label{eq:rsc-norm}
	\delta C_{\Delta}^{-1} \leq \|\bm{d}\|_{2} \leq \delta c_{\Delta}^{-1} 
	\quad\text{and} \quad
	\delta C_{\Delta}^{-1} \overline{\alpha}_\ma^{-1} \leq \| \bm{\phi}\|_2 \leq \delta c_{\Delta}^{-1} \underline{\alpha}_\ma^{-1} \leq c_{\bm{\omega}}.
\end{equation}
The second inequality in \eqref{eq:rsc-norm} indicates that $\bm{\Upsilon}\cap\pazocal{S}(\delta)\neq \emptyset$. 

%\[
%c_{\Delta}	\left(\|\bm{d}\|_{2} + \underline{\alpha}_\ma\|\bm{\phi}\|_2\right) 
%\leq \|\bm{\Delta}\|_{\Fr} \leq         
%C_{\Delta}	\left(\|\bm{d}\|_{2} +  \overline{\alpha}_\ma\|\bm{\phi}\|_2\right)
%\]

Note that $0<\kappa_2\leq  \widetilde{\kappa}_2$. Combining  the high probability events in claims (i)--(iv) with \eqref{eq:triangle} and \eqref{eq:rsc-norm}, we have the following result that holds uniformly for all $\bm{\Delta}=\bm{\Delta}(\bm{\phi},\bm{d})\in\bm{\Upsilon}\cap\pazocal{S}(\delta)$: 
\begin{align*}
	\frac{\|\bm{\Delta}\bm{X}\|_{\Fr}}{\sqrt{T}}
	&\geq \frac{1}{2}\left \{ \frac{\sqrt{\widetilde{\kappa}_1} }{2} \|\bm{d}\|_{2} - \sqrt{\frac{ (r+2s)^2\widetilde{\kappa}_2^2 \log(Nd)}{c_1 \widetilde{\kappa}_1  T} } \|\bm{d}\|_1 + \sqrt{\frac{ \widetilde{\kappa}_1 l_\phi^2}{8 u_\phi^4} } \| \bm{\phi}\|_2 \right \} \\
	&\hspace{5mm} -  \sqrt{ C_4 (r+2s) \widetilde{\kappa}_2 } \left [ \|\bm{d}_{\ma}\|_{2}+ \sqrt{\frac{ 4 \log\{N(r+2s)\}}{c_{\HW} T} } \|\bm{d}_{\ma}\|_1 +  \sqrt{\overline{\alpha}_\ma^2  (r+2s)} \|\bm{\phi}\|_2 \right ] \|\bm{\phi}\|_2
	\\
	&\geq \frac{C_{\Delta}^{-1}  \left(2 +\overline{\alpha}_\ma^{-1}\sqrt{2l_\phi^2/u_\phi^4}\right) }{8} \sqrt{\widetilde{\kappa}_1} \cdot \delta  -  c_{\Delta}^{-2} \sqrt{C_4 \left \{ 1 + (\overline{\alpha}_\ma /\underline{\alpha}_\ma)^2 (r+2s) \right \} (r+2s) \widetilde{\kappa}_2 } \cdot \delta^2\\
	&\hspace{5mm} 
	-  \left (\sqrt{\frac{ (r+2s) \widetilde{\kappa}_2}{c \widetilde{\kappa}_1} } +   \sqrt{\frac{ 4  C_4   C_{\Delta}^{-2} }{c_{\HW}} } \cdot \delta \right )\sqrt{\frac{(r+2s)\widetilde{\kappa}_2  \log(Nd)}{T}}\|\bm{d}\|_1,
\end{align*}
where we used the fact that $\sqrt{x^2+y^2}\leq |x|+|y|$ in the first inequality. Since $\overline{\alpha}_\ma^{-1}\sqrt{2l_\phi^2/u_\phi^4} = (\underline{\alpha}_\ma/\overline{\alpha}_\ma)^2 \min_{1\leq k\leq s}\gamma_{k}^* \leq \bar{\rho}<1$, by choosing
%\[
%0<\delta \leq \min \left [ \frac{C_{\Delta}^{-1}  \left(2+\overline{\alpha}_\ma^{-1}\sqrt{2l_\phi^2/u_\phi^4}\right) \sqrt{\widetilde{\kappa}_1 / \widetilde{\kappa}_2} }{16    c_{\Delta}^{-2}  \sqrt{C_4\left \{ 1 + (\overline{\alpha}_\ma/\underline{\alpha}_\ma)^2  (r+2s) \right \}  }}, \; \sqrt{\frac{c_{\HW} (r+2s) \widetilde{\kappa}_2 /  \widetilde{\kappa}_1 }{16 C_4 C_{\Delta}^{-2} c }},\; c_{\Delta}\underline{\alpha}_\ma c_{\bm{\omega}}\right ]
%\]
\[
0<\delta \leq \min \left [ \frac{3 C_{\Delta}^{-1} \sqrt{\widetilde{\kappa}_1 / \widetilde{\kappa}_2} }{16    c_{\Delta}^{-2}  \sqrt{C_4\left \{ 1 + (\overline{\alpha}_\ma/\underline{\alpha}_\ma)^2  (r+2s) \right \}  }}, \; \sqrt{\frac{c_{\HW} (r+2s) \widetilde{\kappa}_2 /  \widetilde{\kappa}_1 }{16 C_4 C_{\Delta}^{-2} c }},\; c_{\Delta}\underline{\alpha}_\ma c_{\bm{\omega}}\right ]
\]
in the above inequality, then for all $\bm{\Delta}\in\bm{\Upsilon}\cap\pazocal{S}(\delta)$ it holds uniformly that
%\begin{equation}\label{eq:rsc1}
%	\frac{1}{\sqrt{T}}\|\bm{\Delta}\bm{X}\|_{\Fr} \geq \frac{  2+ \overline{\alpha}_\ma^{-1}\sqrt{2l_\phi^2/u_\phi^4} }{16 C_{\Delta}} \sqrt{\widetilde{\kappa}_1} \cdot  \|\bm{\Delta}\|_{\Fr} -  \sqrt{ \frac{ (r+2s)^2 \widetilde{\kappa}_2^2 \log(Nd)}{ c \widetilde{\kappa}_1 T} } \cdot \|\bm{d}\|_1.
%\end{equation}
\begin{equation}\label{eq:rsc1}
	\frac{1}{\sqrt{T}}\|\bm{\Delta}\bm{X}\|_{\Fr} \geq \frac{  3 \sqrt{\widetilde{\kappa}_1} }{16 C_{\Delta}}  \cdot  \|\bm{\Delta}\|_{\Fr} -  \sqrt{ \frac{ (r+2s)^2 \widetilde{\kappa}_2^2 \log(Nd)}{ c \widetilde{\kappa}_1 T} } \cdot \|\bm{d}\|_1.
\end{equation}
As mentioned earlier, for any $\alpha\neq 0$, we have $\alpha\bm{\Delta}(\bm{\phi},\bm{d})=\bm{\Delta}(\bm{\phi},\alpha\bm{d})$ and hence
\[
\frac{1}{\sqrt{T}}\|(\alpha\bm{\Delta})\bm{X}\|_{\Fr} \geq \frac{  3 \sqrt{\widetilde{\kappa}_1}  }{16 C_{\Delta}} \cdot  \|\alpha \bm{\Delta}\|_{\Fr} - \sqrt{ \frac{ (r+2s)^2 \widetilde{\kappa}_2^2 \log(Nd)}{ c \widetilde{\kappa}_1 T} } \cdot \| \alpha \bm{d}\|_1. 
\]
This shows that \eqref{eq:rsc1}  will remain true uniformly for all $\bm{\Delta}\in\bm{\Upsilon}\cap\pazocal{S}(\alpha \delta)$ with any $\alpha\neq 0$, and hence \eqref{eq:rsc1} holds for all $\bm{\Delta}\in\bm{\Upsilon}$.

Note that for any $x,y,z\geq0$, if $x\geq y-z$, then $y^2\leq (x+z)^2\leq 2(x^2+z^2)$ and hence $x^2\geq y^2/2-z^2$. As a result,
\eqref{eq:rsc1} implies that
%\[
%\frac{1}{T}\|\bm{\Delta}\bm{X}\|_{\Fr}^2 \geq \frac{1}{2}\left (\frac{  2+ \overline{\alpha}_\ma^{-1}\sqrt{2l_\phi^2/u_\phi^4} }{16 C_{\Delta}} \right )^2 \widetilde{\kappa}_1 \cdot  \|\bm{\Delta}\|_{\Fr}^2  - \frac{  (r+2s)^2 \widetilde{\kappa}_2^2 \log(Nd)}{ c \widetilde{\kappa}_1 T}  \cdot \|\bm{d}\|_1^2.
%\]
\[
\frac{1}{T}\sum_{t=1}^{T}\|\bm{\Delta}\bm{x}_t\|_{2}^2 =\frac{1}{T}\|\bm{\Delta}\bm{X}\|_{\Fr}^2 \geq C \left \{\widetilde{\kappa}_1 \|\bm{\Delta}\|_{\Fr}^2  - \frac{  (r+2s)^2 \widetilde{\kappa}_2^2 \log \{N(p\vee1)\}}{\widetilde{\kappa}_1 T}  \|\bm{d}\|_1^2 \right\}.
\]
Finally, note that $\widetilde{\kappa}_i \asymp \kappa_i$ for $i=1,2$, and $r+2s\lesssim1$. Combining all tails probabilities and conditions on $T$ from claims (i)--(iv), we accomplish the proof of this lemma.

\bigskip
Below we give the proofs of claims (i)--(iv).

\bigskip
\noindent\textbf{Proof of (i):} Note that 
\[
\frac{1}{T}\|\bm{D} \bm{Z}\|_{\Fr}^2=\frac{1}{T}\trace( \bm{Z}^\top\bm{D}^\top\bm{D} \bm{Z})= \trace\left (\bm{D} \widehat{\bm{\Sigma}}_z \bm{D}^\top\right )=\vect(\bm{D}^\top)^\top  (\bm{I}_N \otimes \widehat{\bm{\Sigma}}_z ) \vect(\bm{D}^\top),
\]
where $\widehat{\bm{\Sigma}}_z=\bm{Z}\bm{Z}^\top/T=T^{-1}\sum_{t=1}^{T}\bm{z}_t \bm{z}_t^\top$. 
Then, the result of this lemma can be rewritten as 
\begin{equation}\label{eq:re1}
	|\bm{u}^\top  (\bm{I}_N\otimes \widehat{\bm{\Sigma}}_z) \bm{u}|^{1/2} \geq \frac{\sqrt{\widetilde{\kappa}_1}}{2} \|\bm{u}\|_2-  \sqrt{\frac{ (r+2s)^2\kappa_2^2 \log(Nd)}{c_1 \widetilde{\kappa}_1  T} }
	\|\bm{u}\|_1, \quad \forall \bm{u}\in\mathbb{R}^{N^2d},
\end{equation}	
with probability at least $1-2e^{-0.5c_1\widetilde{\kappa}_1^2  T/\{(r+2s)^2\kappa_2^2\}}$.

Let $\bm{\Sigma}_z=\mathbb{E}(\bm{z}_t \bm{z}_t^\top)$.	In addition, let $\underline{\bm{z}}_T=(\bm{z}_T^\top, \dots, \bm{z}_1^\top)^\top$, and denote its covariance matrix by
\[
\underline{\bm{\Sigma}}_z=\mathbb{E}(\underline{\bm{z}}_T\underline{\bm{z}}_T^\top)=\left (\bm{\Sigma}_z(j-i)\right )_{1\leq i,j\leq T},
\]
where  $\bm{\Sigma}_z(\ell) = \mathbb{E}(\bm{z}_t\bm{z}_{t-\ell}^\top)$ is the lag-$\ell$ autocovariance matrix of $\bm{z}_t$ for $\ell\in\mathbb{Z}$, and $\bm{\Sigma}_z(0)=\bm{\Sigma}_z$. 
We will first prove the following intermediate result:
\begin{equation}\label{eq:re2}
	\left |\bm{u}^\top \{\bm{I}_N\otimes (\widehat{\bm{\Sigma}}_z - \bm{\Sigma}_z)\}\bm{u} \right | \leq\frac{\widetilde{\kappa}_1}{4} \|\bm{u}\|_2^2 + \frac{(r+2s)^2 \kappa_2^2 \log(Nd)}{c_1 \widetilde{\kappa}_1  T}
	\|\bm{u}\|_1^2, \quad \forall \bm{u}\in\mathbb{R}^{N^2d}, 
\end{equation}	
with probability at least $1-2e^{-0.5c_1\widetilde{\kappa}_1^2  T/\{(r+2s)^2\kappa_2^2\}}$.

Denote $\bm{U}=\bm{L}^\top(\bm{\omega}^*)\otimes \bm{I}_N$, and let $\bm{\ell}_h(\bm{\omega}^*)$ be the $h$th row of $\bm{L}(\bm{\omega}^*)$ for $h\geq1$. Then 
$\bm{z}_t =\bm{U}\bm{x}_{t}=\sum_{h=1}^\infty \bm{U}_h\bm{y}_{t-h}$ and $\bm{U}=(\bm{U}_1, \bm{U}_2, \dots)$,
where $\bm{U}_h=\bm{\ell}_h(\bm{\omega}^*)\otimes \bm{I}_N$  for $h\geq1$. By the definition of $\bm{L}(\bm{\omega}^*)$, we have
$\|\bm{\ell}_h(\bm{\omega}^*)\|_{2}=1$ for $1\leq h\leq p$ and $\|\bm{\ell}_h(\bm{\omega}^*)\|_{2}\leq \sqrt{r+2s}\bar{\rho}^h$ for $h\geq p+1$,
which implies
\[
\sum_{h=1}^{\infty}\|\bm{U}_h\|_{\op}=\sum_{h=1}^{\infty}\|\bm{\ell}_h(\bm{\omega}^*)\|_{2}\leq \sqrt{r+2s}\bar{\rho}(1-\bar{\rho})^{-1}.
\]
In addition,  we have
\[ 
\sigma_{\min}(\bm{U})\geq \sigma_{\min, L}. 
\]
Consequently, applying Lemma \ref{lemma:Wcov}(ii) with $\bm{w}_t=\bm{z}_t$,  we can show that
\begin{equation}\label{eq:sigmaz0}
	\lambda_{\min}(\bm{\Sigma}_z) \geq \kappa_1\sigma_{\min}^2(\bm{U}) \geq \widetilde{\kappa}_1 
\end{equation}
and 
\begin{equation}\label{eq:sigmaz}
	\lambda_{\max}(\underline{\bm{\Sigma}}_z)\leq (r+2s) \bar{\rho}^2(1-\bar{\rho})^{-2} \kappa_2.
\end{equation}

Note that $T^{-1}\sum_{t=1}^{T}\|\bm{u}^\top \bm{z}_t\|_2^2 =\bm{u}^\top  \widehat{\bm{\Sigma}}_z\bm{u}$ and $\mathbb{E}(\|\bm{u}^\top \bm{z}_t\|_2^2)= \bm{u}^\top\bm{\Sigma}_z\bm{u}$.
Furthermore, since $\bm{z}_t=\mathcal{W}(B)\bm{y}_t=\mathcal{W}(B)\bm{\Psi}_*(B)\bm{\varepsilon}_{t}$ is a zero-mean and stationary time series, where $\mathcal{W}(B) = \sum_{i=1}^{\infty}\bm{W}_i B^i$, we can  apply Lemma \ref{lemma:hansonw} with $T_0=0$, $T_1=T$, $\bm{w}_t=\bm{z}_t$, $\bm{M}=\bm{u}^\top$, and  $\eta=\widetilde{\kappa}_1/\{108\sigma^2(r+2s) \bar{\rho}^2(1-\bar{\rho})^{-2} \kappa_2\}$, in conjunction with \eqref{eq:sigmaz}, 
to obtain  the following pointwise bound: for any $\bm{u}\in\mathbb{R}^{Nd}$ with $\|\bm{u}\|_2\leq 1$,
\begin{equation}\label{eq:hansonw}
	\mathbb{P}\left \{ \bm{u}^\top  (\widehat{\bm{\Sigma}}_z - \bm{\Sigma}_z)\bm{u} \geq \widetilde{\kappa}_1/108\right \} \leq 2\exp\left [-c_1\widetilde{\kappa}_1^2 T/\{(r+2s)^2  \kappa_2^2\} \right ],
\end{equation}
where $c_1=c_{\HW}\min[\{108\sigma^{2}\bar{\rho}^{2}(1-\bar{\rho})^{-2}\}^{-1}, \{108\sigma^{2}\bar{\rho}^{2}(1-\bar{\rho})^{-2}\}^{-2}]$. 

Let $\pazocal{K}(2K)=\{\bm{u}\in\mathbb{R}^{Nd}: \|\bm{u}\|_2\leq 1, \|\bm{u}\|_0\leq 2 K\}$ be a set of sparse vectors, where $K\geq1$ is an integer to be specified later. Then, by arguments similar to the proof of Lemma F.2 in \cite{basu2015regularized}, we can strengthen \eqref{eq:hansonw} to the union bound that holds for all  $\bm{u}\in\pazocal{K}(2K)$ as follows:
\begin{equation*}%\label{eq:hansonw1}
	\mathbb{P}\left \{ \sup_{\bm{u}\in \pazocal{K}(2K)}\bm{u}^\top  (\widehat{\bm{\Sigma}}_z - \bm{\Sigma}_z)\bm{u} \geq  \widetilde{\kappa}_1/108\right \} \leq 2\exp\left[ -c_1\widetilde{\kappa}_1^2 T/\{(r+2s)^2  \kappa_2^2\}+2K\log(Nd)\right ],
\end{equation*}
Now we choose $K=\lceil 0.25c_1\widetilde{\kappa}_1^2 T/\{(r+2s)^2 \kappa_2^2\log(Nd)\} \rceil\geq 1$. Thus, applying Supplementary Lemma 12 in \cite{LW12}, we have
\begin{align*}%\label{eq:reclaim1}
	&	\mathbb{P}\left \{ \forall \bm{u}\in\mathbb{R}^{Nd}: |\bm{u}^\top ( \widehat{\bm{\Sigma}}_z - \bm{\Sigma}_z) \bm{u}| \leq \frac{\widetilde{\kappa}_1}{4} \|\bm{u}\|_2^2 + \frac{(r+2s)^2  \kappa_2^2 \log(Nd)}{c_1 \widetilde{\kappa}_1  T}
	\|\bm{u}\|_1^2 \right \}\\
	&\hspace{5mm} \geq 1-2\exp\left [-0.5c_1\widetilde{\kappa}_1^2 T/\{(r+2s)^2 \kappa_2^2\} \right ],
\end{align*}	
and hence \eqref{eq:re2}.
Furthermore, by \eqref{eq:sigmaz0} and the inequality $|x+y|^{1/2}\leq |x|^{1/2}+|y|^{1/2}$, for all $\bm{u}\in\mathbb{R}^{N^2d}$, we have
\begin{align*}
	\sqrt{\widetilde{\kappa}_1} \|\bm{u}\|_2 \leq \lambda_{\min}^{1/2}(\bm{\Sigma}_z)\|\bm{u}\|_2  & \leq |\bm{u}^\top (\bm{I}_N\otimes\bm{\Sigma}_z) \bm{u}|^{1/2}\\
	&
	\leq |\bm{u}^\top  (\bm{I}_N\otimes \widehat{\bm{\Sigma}}_z) \bm{u}|^{1/2}+ |\bm{u}^\top \{ \bm{I}_N\otimes(\widehat{\bm{\Sigma}}_z-\bm{\Sigma}_z) \}\bm{u}|^{1/2}.
\end{align*}	
Finally,   combining this with \eqref{eq:re2} and the inequality $\sqrt{x^2+y^2}\leq |x|+|y|$, we have \eqref{eq:re1}. This completes the proof of (i).

\smallskip
\noindent\textbf{Proof of (ii):}
It is worth noting that $\bm{M}(\bm{\phi})$ is linear in $\bm{\phi}$, which implies that
\begin{equation}\label{eq:Xi1}
	\frac{\bm{M}(\bm{\phi})}{\|\bm{M}(\bm{\phi}) \|_{\Fr}}\in
	\bm{\Xi}_1=\{\bm{M} \in \bm{\Xi}\mid \|\bm{M}\|_{\Fr} = 1\}, \quad\forall \bm{\phi}\in \bm{\Phi},
\end{equation}
where $\bm{\Xi} = \left\{ \bm{M}(\bm{\phi}) \in  \mathbb{R}^{N \times N  (r+2s)} \mid \bm{\phi}\in \bm{\Phi} \right\}$. To prove the result of this lemma, we begin by establishing the following intermediate result: 
\begin{equation}\label{eq:reclaim2}
	\mathbb{P}\left ( \frac{\widetilde{\kappa}_1 l_\phi^2}{8 u_\phi^2} \leq  \inf_{\bm{M}\in\bm{\Xi}_1} \frac{1}{T} \|\bm{M} \bm{V}\|_{\Fr}^2 \leq \sup_{\bm{M}\in\bm{\Xi}_1} \frac{1}{T} \|\bm{M} \bm{V}\|_{\Fr}^2 \leq  \frac{6\widetilde{\kappa}_2 u_\phi^2}{l_\phi^2} \right ) \geq  1- 2e^{-0.5 c_2\widetilde{\kappa}_1^2 T/\{(r+2s)^2\kappa_2^2\}}.
\end{equation}

Similar to the proof of claim (i), let $\bm{\Sigma}_v=\mathbb{E}(\bm{v}_t \bm{v}_t^\top)$. In addition,  let $\underline{\bm{v}}_T=(\bm{v}_T^\top, \dots, \bm{v}_1^\top)^\top$, and denote its covariance matrix by
\[
\underline{\bm{\Sigma}}_v=\mathbb{E}(\underline{\bm{v}}_T\underline{\bm{v}}_T^\top)=\left (\bm{\Sigma}_v(j-i)\right )_{1\leq i,j\leq T},
\]
where  $\bm{\Sigma}_v(\ell) = \mathbb{E}(\bm{v}_t\bm{v}_{t-\ell}^\top)$ is the lag-$\ell$ autocovariance matrix of $\bm{v}_t$ for $\ell\in\mathbb{Z}$, and $\bm{\Sigma}_v(0)=\bm{\Sigma}_v$. 

Denote $\bm{U}=\bm{P}^\top(\bm{\omega}^*)\otimes \bm{I}_N$ and let $\bm{p}_h(\bm{\omega}^*)$ be the $h$th row of $\bm{P}(\bm{\omega}^*)$ for $h\geq1$. Then 
$\bm{v}_t =\bm{U}\bm{x}_{t}=\sum_{h=1}^\infty \bm{U}_h\bm{y}_{t-h}$ and $\bm{U}=(\bm{U}_1, \bm{U}_2, \dots)$,
where $\bm{U}_h=\bm{p}_h(\bm{\omega}^*)\otimes \bm{I}_N$  for $h\geq1$. By the definition of $\bm{P}(\bm{\omega}^*)$, we have $\|\bm{p}_h(\bm{\omega}^*)\|_{2}\leq \sqrt{r+2s}C_{\ell}\bar{\rho}^h$ for $h\geq 1$,
which implies
\[
\sum_{h=1}^{\infty}\|\bm{U}_h\|_{\op}=\sum_{h=1}^{\infty}\|\bm{p}_h(\bm{\omega}^*)\|_{2}\leq \sqrt{r+2s}C_{\ell}\bar{\rho}(1-\bar{\rho})^{-1}.
\]
In addition,  we have
\[ 
\sigma_{\min, L} \leq \sigma_{\min}(\bm{U})\leq \sigma_{\max}(\bm{U})\leq \sigma_{\max, L}.
\]
%\[ 
%\sigma_{\min}(\bm{U})\geq \sigma_{\min, L}  \geq \min\{1, c_{\bar{\rho}}\}. 
%\]
Consequently, applying Lemma \ref{lemma:Wcov}(ii) with $\bm{w}_t=\bm{v}_t$,  we can show that
\begin{equation}\label{eq:sigmav0}
	\widetilde{\kappa}_1 \leq \kappa_1\sigma_{\min}^2(\bm{U}) \leq \lambda_{\min}(\bm{\Sigma}_v) \leq  \lambda_{\max}(\bm{\Sigma}_v)  \leq  \kappa_2\sigma_{\max}^2(\bm{U}) \leq  \widetilde{\kappa}_2
\end{equation}
%By applying Lemma \ref{lemma:Wcov}(ii) with $\bm{w}_t=\bm{v}_t$ and a method similar to that for \eqref{eq:sigmaz0} and \eqref{eq:sigmaz}, it can be verified that
%	\begin{equation}\label{eq:sigmav0}
	%		\widetilde{\kappa}_1 \leq \lambda_{\min}(\bm{\Sigma}_v) \leq  \lambda_{\max}(\bm{\Sigma}_v)  \leq  \widetilde{\kappa}_2
	%	\end{equation}
%\begin{equation}\label{eq:sigmav0}
%\lambda_{\min}(\bm{\Sigma}_v) \geq \widetilde{\kappa}_1
%\end{equation}
and 
\begin{equation}\label{eq:sigmav}
	\lambda_{\max}(\underline{\bm{\Sigma}}_v)\leq(r+2s) C_{\ell}^2\bar{\rho}^2(1-\bar{\rho})^{-2}  \kappa_2,
\end{equation}

Note that $T^{-1}\|\bm{M} \bm{V}\|_{\Fr}^2=T^{-1}\sum_{t=1}^{T}\|\bm{M}\bm{v}_t\|_2^2 =\trace( \bm{M}\widehat{\bm{\Sigma}}_v\bm{M}^\top)$, where $\widehat{\bm{\Sigma}}_v=\bm{V}\bm{V}^\top/T=T^{-1}\sum_{t=1}^{T}\bm{v}_t \bm{v}_t^\top$,  and $\mathbb{E}(\|\bm{M}\bm{v}_t\|_2^2)=\trace( \bm{M}\bm{\Sigma}_v\bm{M}^\top)$. By \eqref{eq:sigmav0}, for any  $\bm{M} \in  \mathbb{R}^{N \times N (r+2s)}$, we have
\[
\widetilde{\kappa}_1 \|\bm{M}\|_{\Fr}^2 \leq \lambda_{\min}(\bm{\Sigma}_v)\|\bm{M}\|_{\Fr}^2\leq
\mathbb{E}\left (\|\bm{M}\bm{v}_t\|_2^2 \right )
\leq\lambda_{\max}(\bm{\Sigma}_v) \|\bm{M}\|_{\Fr}^2\leq \widetilde{\kappa}_2 \|\bm{M}\|_{\Fr}^2.
\]
Moreover, by Lemma \ref{lemma:hansonw} with $T_0=0$, $T_1=T$, $\bm{w}_t=\bm{v}_t$, and  $\eta= \widetilde{\kappa}_1/\{2\sigma^2 (r+2s) C_{\ell}^2\bar{\rho}^2 (1-\bar{\rho})^{-2} \kappa_2  \}$, in conjunction with \eqref{eq:sigmav}, we can show that for any  $\bm{M} \in  \mathbb{R}^{N \times N (r+2s)}$,
\begin{align*}
	%\mathbb{P}\left [\left |\trace\{ \bm{M}(\widehat{\bm{\Sigma}}_v-\bm{\Sigma}_v)\bm{M}^\top\}\right | \geq \widetilde{\kappa}_1/2\right] &=	
	\mathbb{P}\left \{\left |\frac{1}{T}\sum_{t=1}^{T}\|\bm{M}\bm{v}_t\|_2^2 - \mathbb{E}\left (\|\bm{M}\bm{v}_t\|_2^2 \right )\right | \geq \frac{\widetilde{\kappa}_1}{2} \|\bm{M}\|_{\Fr}^2 \right \} 
	&\leq 2\exp\left [-c_2\widetilde{\kappa}_1^2 T/\{(r+2s)^2\kappa_2^2\}  \right ].
\end{align*}
where $c_2=c_{\HW}\min[\{2\sigma^{2}C_{\ell}^2\bar{\rho}^2 (1-\bar{\rho})^{-2}\}^{-1}, \{2\sigma^{2}C_{\ell}^2\bar{\rho}^2 (1-\bar{\rho})^{-2}\}^{-2}]$. As a result, we have the following pointwise bound: for any  $\bm{M} \in  \mathbb{R}^{N \times N  (r+2s)}$,
\begin{equation}\label{eq:hansonw1}
	\mathbb{P}\left(  \frac{\widetilde{\kappa}_1}{2} \|\bm{M}\|_{\Fr}^2\leq \frac{1}{T} \|\bm{M} \bm{V}\|_{\Fr}^2 \leq  \frac{3\widetilde{\kappa}_2}{2} \|\bm{M}\|_{\Fr}^2  \right)  
	\geq 1-2\exp\left [-c_2\widetilde{\kappa}_1^2 T/\{(r+2s)^2\kappa_2^2\}  \right ].
\end{equation}

Next we strengthen the above pointwise bound to a union bound  that holds for all $\bm{M}\in\bm{\Xi}_1$. Let $\bm{\bar{\Xi}}(\epsilon_0)$ be a minimal generalized $\epsilon_0$-net of $\bm{\Xi}_1$ in the Frobenius norm, where $0<\epsilon_0<1$ will be chosen later. 
By  Lemma \ref{lemma:epsilon-net}(ii), any $\bm{M}\in\bm{\bar{\Xi}}(\epsilon_0)$ satisfies $l_\phi/u_\phi\leq \|\bm{M}\|_{\Fr}\leq u_\phi/l_\phi$. Define the event 
\[
\mathcal{E}(\epsilon_0)=\left \{\forall \cm{M}\in\bm{\bar{\Xi}}(\epsilon_0): \sqrt{\frac{ \widetilde{\kappa}_1 l_\phi^2}{2 u_\phi^2}} < \frac{1}{\sqrt{T}}\|\bm{M} \bm{V}\|_{\Fr} <\sqrt{\frac{ 3\widetilde{\kappa}_2 u_\phi^2}{2 l_\phi^2}} \right \}.
\]
Then, by the pointwise bounds in \eqref{eq:hansonw1} and the covering number in Lemma \ref{lemma:epsilon-net}(i), 
we have
\begin{align}\label{eq:event}
	\mathbb{P}\{\stcomp{\mathcal{E}}(\epsilon_0)\}&\leq e^{(r+2s) \log\{3/(c_{\bm{M}}\epsilon_0)\}} \max_{\bm{M}\in\bm{\bar{\Xi}}(\epsilon_0)} \mathbb{P} \left [\stcomp{\left \{ \frac{\widetilde{\kappa}_1 l_\phi^2}{2 u_\phi^2} \leq \frac{1}{T}\|\bm{M}\bm{V}\|_{\Fr}^2 \leq \frac{ 3\widetilde{\kappa}_2 u_\phi^2}{2 l_\phi^2} \right \} }\right ] \notag\\
	&\leq  2\exp\left[- c_2\widetilde{\kappa}_1^2 T/\{(r+2s)^2\kappa_2^2\}  +(r+2s) \log\{3u_\phi/(l_\phi\epsilon_0)\}\right].
\end{align}
By  Lemma \ref{lemma:epsilon-net}(iii), it holds
\begin{equation}\label{eq:upper}
	\mathcal{E}(\epsilon_0)\subset  \left \{ \max_{\bm{M}\in\bm{\bar{\Xi}}(\epsilon_0)}\frac{1}{\sqrt{T}}\|\bm{M} \bm{V}\|_{\Fr}\leq \sqrt{\frac{ 3 \widetilde{\kappa}_2 u_\phi^2}{2 l_\phi^2}}\right \}\subset \left \{ \sup_{\bm{M}\in\bm{\Xi}_1}\frac{1}{\sqrt{T}}\|\bm{M} \bm{V}\|_{\Fr}\leq \frac{\sqrt{ 3\widetilde{\kappa}_2u_\phi^2/(2 l_\phi^2 )} }{1-\epsilon_0}\right \}.
\end{equation}
Moreover,  by a method similar to that for the proof of Lemma \ref{lemma:epsilon-net}(iii), for any $\bm{M}\in\bm{\Xi}_1$ and its corresponding $\bar{\bm{M}}\in\bm{\bar{\Xi}}(\epsilon_0)$ defined therein, we can show that
\begin{align*}
	\frac{1}{\sqrt{T}}	\| \bm{M} \bm{V}\|_{\Fr}
	&	\geq \frac{1}{\sqrt{T}}\| \bar{\bm{M}}_{(1)} \bm{V}\|_{\Fr} - \frac{1}{\sqrt{T}}\| (\bm{M}-\bar{\bm{M}})_{(1)}\bm{V}\|_{\Fr} \\
	&\geq \min_{\bar{\bm{M}} \in \bm{\bar{\Xi}}(\epsilon)} \frac{1}{\sqrt{T}}\| \bar{\bm{M}}_{(1)} \bm{V}\|_{\Fr}  -  \epsilon_0	\sup_{\bm{M} \in \bm{\Xi}_1} \frac{1}{\sqrt{T}} \| \bm{M} \bm{V}\|_{\Fr}.
\end{align*}
Taking the infimum over all  $\bm{M}\in\bm{\Xi}_1$ and combining the result with \eqref{eq:upper},  we can show that on the event $\mathcal{E}(\epsilon_0)$, it holds
\begin{align*}
	\inf_{\bm{M}\in\bm{\Xi}_1}\frac{1}{\sqrt{T}}	\| \bm{M} \bm{V}\|_{\Fr} &\geq \sqrt{\frac{\widetilde{\kappa}_1 l_\phi^2}{2 u_\phi^2}} - \epsilon_0\cdot \frac{\sqrt{3\widetilde{\kappa}_2 u_\phi^2/(2 l_\phi^2)}}{1-\epsilon_0} 
	\geq \sqrt{\frac{\widetilde{\kappa}_1 l_\phi^2}{2 u_\phi^2}}-2 \epsilon_0 \sqrt{\frac{3\widetilde{\kappa}_2 u_\phi^2}{2 l_\phi^2}}
\end{align*}
if $0<\epsilon_0\leq1/2$. Thus, by setting 
\[
\epsilon_0= \min\left\{\frac{l_\phi^2}{4 u_\phi^2}\sqrt{\frac{\widetilde{\kappa}_1}{3 \widetilde{\kappa}_2}}, \frac{1}{2} \right \},
\] 
we have
\begin{equation}\label{eq:lower}
	\mathcal{E}(\epsilon_0)\subset \left \{ \inf_{\bm{M}\in\bm{\Xi}_1}\frac{1}{\sqrt{T}}\|\bm{M} \bm{V}\|_{\Fr}\geq \frac{\sqrt{\widetilde{\kappa}_1 l_\phi^2/(2 u_\phi^2 )}}{2}\right \}.
\end{equation} 
Consequently, with the above choice of $\epsilon_0$, we have
\[
\mathcal{E}(\epsilon_0) \subset \left \{ \frac{\widetilde{\kappa}_1 l_\phi^2}{8 u_\phi^2} \leq \inf_{\bm{M}\in\bm{\Xi}_1} \frac{1}{T} \|\bm{M} \bm{V}\|_{\Fr}^2   \leq 
\sup_{\bm{M}\in\bm{\Xi}_1} \frac{1}{T} \|\bm{M} \bm{V}\|_{\Fr}^2  \leq \frac{6\widetilde{\kappa}_2 u_\phi^2}{l_\phi^2}  \right \},
\]
which, together with \eqref{eq:event}, implies that 
\[
\mathbb{P}\left (  \frac{\widetilde{\kappa}_1 l_\phi^2}{8 u_\phi^2} \leq \inf_{\bm{M}\in\bm{\Xi}_1} \frac{1}{T} \|\bm{M} \bm{V}\|_{\Fr}^2   \leq 
\sup_{\bm{M}\in\bm{\Xi}_1}\frac{1}{T} \|\bm{M} \bm{V}\|_{\Fr}^2  \leq \frac{6\widetilde{\kappa}_2 u_\phi^2}{l_\phi^2} \right ) \geq 1- 2e^{-0.5 c_2\widetilde{\kappa}_1^2 T/\{(r+2s)^2\kappa_2^2\}}
\]
under the condition on $T$ stated in (ii). Then \eqref{eq:reclaim2} follows immediately. By combining 
\eqref{eq:Xi1}, \eqref{eq:reclaim2},  and the bounds in \eqref{eq:DFr}, we accomplish the proof of (ii).

\smallskip
\noindent\textbf{Proof of (iii):}
Similar to the proof of claim (i), we can show that
\[
\frac{1}{T}\|\bm{D}_{\ma} \bm{H}(\bm{\phi})\|_{\Fr}^2= \trace\left \{ \bm{D}_{\ma} \widehat{\bm{\Sigma}}_H(\bm{\phi}) \bm{D}_\ma^\top\right \}
=\vect(\bm{D}_\ma^\top)^\top  \{ \bm{I}_N \otimes \widehat{\bm{\Sigma}}_H(\bm{\phi})\} \vect(\bm{D}_\ma^\top),
\]
where $\widehat{\bm{\Sigma}}_H(\bm{\phi})=\bm{H}(\bm{\phi})\bm{H}^\top(\bm{\phi})/T=T^{-1}\sum_{t=1}^{T}\bm{h}_t(\bm{\phi}) \bm{h}_t^\top(\bm{\phi})$. Then, the high probability event stated in this lemma is equivalent to 
\[
\sup_{\bm{\phi}\in \bm{\Phi}_1} \frac{|\bm{u}^\top  \{\bm{I}_N \otimes \widehat{\bm{\Sigma}}_H(\bm{\phi}) \} \bm{u}|}{\|\bm{\phi}\|_2^2} \leq C_4 (r+2s) \widetilde{\kappa}_2\left [\|\bm{u}\|_2^2+\frac{4 \log\{N(r+2s)\}}{c_{\HW}  T}  \|\bm{u}\|_1^2\right], \quad \forall \bm{u}\in\mathbb{R}^{N^2(r+2s)}.
\]
Thus, similar to the proof of \eqref{eq:re2}, it suffices to show that  with probability at least $1-4e^ {-0.5c_{\HW} T}$,
\begin{equation}\label{eq:reclaim3}
	\sup_{\bm{\phi}\in \bm{\Phi}_1} \frac{|\bm{u}^\top  \widehat{\bm{\Sigma}}_H(\bm{\phi}) \bm{u} |}{\| \bm{\phi}\|_2^2}   \leq C_4 (r+2s) \widetilde{\kappa}_2\left [\|\bm{u}\|_2^2+\frac{4 \log\{N(r+2s)\}}{c_{\HW}  T}  \|\bm{u}\|_1^2\right], \quad \forall \bm{u}\in\mathbb{R}^{N(r+2s)}.
\end{equation}

To prove \eqref{eq:reclaim3}, we first aim to establish an upper bound of $|\bm{u}^\top  \widehat{\bm{\Sigma}}_H(\bm{\phi}) \bm{u} |$ for a fixed $\bm{u}=(\bm{u}_1^\top, \dots, \bm{u}_{r+2s}^\top)^\top\in\mathbb{R}^{N(r+2s)}$, where $\bm{u}_k\in\mathbb{R}^N$ for $1\leq k\leq r+2s$. 
Note that $\bm{h}_t(\bm{\phi})=\sum_{h=1}^{\infty}\{\bm{q}_h(\bm{\phi})\otimes \bm{I}_N\}\bm{y}_{t-p-h}$, where  $\bm{q}_h(\bm{\phi})=(q_{h,1}(\bm{\phi}), q_{h,2}(\bm{\phi}), \dots)^\top$ is the transpose of the $h$th row of $\bm{Q}(\bm{\phi})$.  
Then
\begin{align*}
	|\bm{u}^\top  \widehat{\bm{\Sigma}}_H(\bm{\phi}) \bm{u} |&=\left |\frac{1}{T}\sum_{t=1}^T \bm{u}^\top \bm{h}_t(\bm{\phi}) \bm{h}_t^\top(\bm{\phi})\bm{u} \right |\\
	&\leq \frac{1}{T}\sum_{i=1}^{\infty}\sum_{h=1}^{\infty} \left |\sum_{t=1}^{T} \bm{u}^\top\{\bm{q}_i(\bm{\phi})\otimes \bm{I}_N\}\bm{y}_{t-p-i} \bm{y}_{t-p-h}^\top\{\bm{q}_h^\top(\bm{\phi})\otimes \bm{I}_N\} \bm{u} \right | \\
	&\leq \frac{1}{T}\sum_{i=1}^{\infty} \left ( \sum_{t=1}^T  [\bm{u}^\top\{\bm{q}_i(\bm{\phi})\otimes \bm{I}_N\}\bm{y}_{t-p-i}]^2 \right )^{1/2} 
	\sum_{h=1}^{\infty} \left (\sum_{t=1}^T  [\bm{u}^\top\{\bm{q}_h(\bm{\phi})\otimes \bm{I}_N\}\bm{y}_{t-p-h}]^2 \right )^{1/2}\\
	&= \left \{\sum_{h=1}^{\infty} \left ( \frac{1}{T}  \sum_{t=1}^T  [\bm{u}^\top\{\bm{q}_h(\bm{\phi})\otimes \bm{I}_N\}\bm{y}_{t-p-h}]^2 \right )^{1/2} \right \}^2.
\end{align*}
In addition,
\begin{align*}
	\frac{1}{T}  \sum_{t=1}^T  [\bm{u}^\top\{\bm{q}_h(\bm{\phi})\otimes \bm{I}_N\}\bm{y}_{t-p-h}]^2 &
	= \frac{1}{T}  \sum_{t=1}^T \left \{ \sum_{h=1}^{\infty}\sum_{k=1}^{r+2s}   q_{h,k}(\bm{\phi})  \bm{u}_k^\top\bm{y}_{t-p-h} \right \}^2 \\
	& \leq \frac{1}{T}  \sum_{t=1}^T \sum_{h=1}^{\infty}\sum_{k=1}^{r+2s}   q_{h,k}^2(\bm{\phi})  \sum_{k=1}^{r+2s} (\bm{u}_k^\top\bm{y}_{t-p-h})^2\\
	& = \|\bm{q}_{h}(\bm{\phi})\|_2^2 
	\sum_{k=1}^{r+2s} \frac{1}{T} \sum_{t=1}^T  (\bm{u}_k^\top\bm{y}_{t-p-h})^2\\
	&\leq \|\bm{q}_{h}(\bm{\phi})\|_2^2 \left \{ 
	\sum_{k=1}^{r+2s} \sqrt{\frac{1}{T} \sum_{t=1}^T  (\bm{u}_k^\top\bm{y}_{t-p-h})^2 } \right \}^2.
\end{align*}
Furthermore, by Lemma \ref{cor1} and a method similar to that for upper bounding $\| \bm{R}_{1h}\|_{\Fr}$ and $\| \bm{R}_{2h}\|_{\Fr}$ in the proof of Proposition \ref{prop:perturb}, we can show that  %$\|\bm{\phi}\|_2\leq c_{\bm{\omega}}$
\begin{align*}%\label{eq:R1}
	\| \bm{q}_{h}(\bm{\phi})\|_{2}
	&\leq \sqrt{2}C_{\ell}\bar{\rho}^h \|\bm{\phi}\|_2 + \frac{\sqrt{2}}{2} C_{\ell}\bar{\rho}^h \|\bm{\phi}\|_2^2 \leq 2\sqrt{2}C_{\ell}\bar{\rho}^h \|\bm{\phi}\|_2, \quad\forall \bm{\phi}\in \bm{\Phi}_1.
\end{align*}
Combining the above results, we have
\begin{equation*}
	|\bm{u}^\top  \widehat{\bm{\Sigma}}_H(\bm{\phi}) \bm{u} | 
	\leq  \left \{ 2\sqrt{2}C_{\ell} \|\bm{\phi}\|_2 \sum_{k=1}^{r+2s} \sum_{h=1}^{\infty} \bar{\rho}^h  \sqrt{\frac{1}{T} \sum_{t=1}^T  (\bm{u}_k^\top\bm{y}_{t-p-h})^2 }\right \}^2, \quad\forall \bm{\phi}\in \bm{\Phi}_1.
\end{equation*}
Hence, by Lemma \ref{alem:rsc3aux}, if $T\geq c_{\HW}^{-1}\log 2$, for any fixed $\bm{u}\in\mathbb{R}^{N(r+2s)}$,  it holds with probability at least $1-4e^{-c_{\HW} T}$ that
\begin{align}\label{eq:alem1}
	\sup_{\bm{\phi}\in \bm{\Phi}_1} \frac{|\bm{u}^\top  \widehat{\bm{\Sigma}}_H(\bm{\phi}) \bm{u} |}{\| \bm{\phi}\|_2^2} &\leq 8 C_{\ell}^2  \left \{ \sum_{k=1}^{r+2s} \sum_{h=1}^{\infty} \bar{\rho}^h  \sqrt{\frac{1}{T} \sum_{t=1}^T  (\bm{u}_k^\top\bm{y}_{t-p-h})^2 }\right \}^2 \notag\\
	& \leq 8 C_{\ell}^2 (r+2s) \sum_{k=1}^{r+2s} \sum_{h=1}^{\infty} \bar{\rho}^{2h} \frac{1}{T} \sum_{t=1}^T  (\bm{u}_k^\top\bm{y}_{t-p-h})^2 \notag\\
	&\leq 8 C_{\ell}^2 (r+2s) \sum_{k=1}^{r+2s} \sum_{h=1}^{\infty} \bar{\rho}^{2h} \lambda_{\max}(\bm{\Sigma}_\varepsilon) \mu_{\max}(\bm{\Psi}_*) (h \sigma^2+1) \|\bm{u}_k\|_2^2 \notag\\
	& = C_4 (r+2s) \kappa_2 \|\bm{u}\|_2^2 \leq C_4 (r+2s) \widetilde{\kappa}_2 \|\bm{u}\|_2^2,
\end{align}
where  
$0<C_4=8 C_{\ell}^2 \sum_{h=1}^{\infty} \bar{\rho}^{2h} (h \sigma^2+1) <\infty$ is an absolute constant.

Next we strengthen the above bound to \eqref{eq:reclaim3} by a method similar to that for \eqref{eq:re2} in the proof of claim (i).
Let $\pazocal{K}(2K)=\{\bm{u}\in\mathbb{R}^{N(r+2s)}: \|\bm{u}\|_2\leq 1, \|\bm{u}\|_0\leq 2 K\}$ be a set of sparse vectors, where $K\geq1$ is an integer to be specified later. Then, by arguments similar to the proof of  Lemma F.2 in \cite{basu2015regularized}, we have the union bound:
\begin{equation*}%\label{eq:hansonw1}
	\mathbb{P}\left \{ \sup_{\bm{u}\in \pazocal{K}(2K)}\sup_{\bm{\phi}\in \bm{\Phi}_1} \frac{|\bm{u}^\top  \widehat{\bm{\Sigma}}_H(\bm{\phi}) \bm{u} |}{\| \bm{\phi}\|_2^2}  \geq  C_4 (r+2s)\widetilde{\kappa}_2 \|\bm{u}\|_2^2 \right \} \leq 4e^{-c_{\HW} T+2K\log\{N(r+2s)\}},
\end{equation*}
By choosing $K=\lceil 0.25c_{\HW} T/\log\{N(r+2s)\} \rceil \geq 1$ and using  Supplementary Lemma 12 in \cite{LW12}, we can readily verify \eqref{eq:reclaim3} and thus accomplish the proof of (iii).

\smallskip
\noindent\textbf{Proof of (iv):}
Similar to the proof of claim (iii), we have
\begin{equation}\label{eq:alem4}
	\frac{1}{T}\|\bm{G}_{\ma}^{*} \bm{B}(\bm{\phi})\|_{\Fr}^2= \trace\left \{ \bm{G}_{\ma}^{*} \widehat{\bm{\Sigma}}_b(\bm{\phi}) \bm{G}_{\ma}^{* \top}\right \}
	=\vect(\bm{G}_{\ma}^{* \top})^\top  \{\bm{I}_N\otimes \widehat{\bm{\Sigma}}_b(\bm{\phi}) \} \vect(\bm{G}_{\ma}^{* \top}),
\end{equation}
where $\widehat{\bm{\Sigma}}_b(\bm{\phi})=\bm{B}(\bm{\phi})\bm{B}^\top(\bm{\phi})/T=T^{-1}\sum_{t=1}^{T}\bm{b}_t(\bm{\phi}) \bm{b}_t^\top(\bm{\phi})$.  
Moreover, we can establish an upper bound of $|\bm{u}^\top  \widehat{\bm{\Sigma}}_b(\bm{\phi}) \bm{u} |$ for any fixed $\bm{u}\in\mathbb{R}^{N(r+2s)}$. Note that $\bm{b}_t(\bm{\phi})=\sum_{h=1}^{\infty}\{\bm{s}_h(\bm{\phi})\otimes \bm{I}_N\}\bm{y}_{t-p-h}$, where  $\bm{s}_h(\bm{\phi})=(s_{h,1}(\bm{\phi}), s_{h,2}(\bm{\phi}), \dots)^\top$ is the transpose of the $h$th row of $\bm{S}(\bm{\phi})$. In addition,  by Lemma \ref{cor1} and a method similar to that for upper bounding $\| \bm{R}_{3h}\|_{\Fr}$  in the proof of Proposition \ref{prop:perturb}, we can show that  %$\|\bm{\phi}\|_2\leq c_{\bm{\omega}}$
\begin{equation*}%\label{eq:R1}
	\| \bm{s}_{h}(\bm{\phi})\|_{2}
	\leq \frac{\sqrt{2}}{2}C_{\ell}\bar{\rho}^h \|\bm{\phi}\|_2^4,  \quad\forall \bm{\phi}\in \bm{\Phi}_1.
\end{equation*}
Then  by Lemma \ref{alem:rsc3aux}, along the lines of \eqref{eq:alem1} it can be readily proved that  if $T\geq c_{\HW}^{-1}\log 2$, for any fixed $\bm{u}\in\mathbb{R}^{N(r+2s)}$,  with probability at least $1-4e^{-c_{\HW} T}$,
\begin{align*}
	\sup_{\bm{\phi}\in \bm{\Phi}_1} \frac{|\bm{u}^\top  \widehat{\bm{\Sigma}}_b(\bm{\phi}) \bm{u} |}{\| \bm{\phi}\|_2^4} 
	& \leq C_4 (r+2s) \kappa_2 \|\bm{u}\|_2^2 \leq C_4 (r+2s)\widetilde{\kappa}_2 \|\bm{u}\|_2^2,
\end{align*}
where $C_4>0$ is the absolute constant defined as in \eqref{eq:alem1}. For simplicity, denote $\vect(\bm{G}_{\ma}^{* \top})=(\bm{u}_1^\top, \dots, \bm{u}_{N}^\top)^\top\in\mathbb{R}^{N^2(r+2s)}$, where $\bm{u}_i\in\mathbb{R}^{N(r+2s)}$ for $1\leq i\leq N$.
Then 
\begin{align*}
	&	\mathbb{P}\left \{\sup_{\bm{\phi}\in \bm{\Phi}_1} \frac{|\bm{u}^\top  \{\bm{I}_N \otimes \widehat{\bm{\Sigma}}_b(\bm{\phi}) \} \bm{u}|}{\|\bm{\phi}\|_2^4} \geq  C_4 (r+2s) \widetilde{\kappa}_2 \|\bm{u}\|_2^2 \right \}\\
	&\hspace{5mm}\leq \mathbb{P}\left \{\sum_{i=1}^{N}\sup_{\bm{\phi}\in \bm{\Phi}_1} \frac{|\bm{u}_i^\top   \widehat{\bm{\Sigma}}_b(\bm{\phi}) \bm{u}_i|}{\|\bm{\phi}\|_2^4} \geq  C_4 (r+2s) \widetilde{\kappa}_2 \sum_{i=1}^{N}\|\bm{u}_i\|_2^2 \right \}\\
	&\hspace{5mm}\leq \sum_{i=1}^{N} \mathbb{P}\left \{\sup_{\bm{\phi}\in \bm{\Phi}_1} \frac{|\bm{u}_i^\top   \widehat{\bm{\Sigma}}_b(\bm{\phi}) \bm{u}_i|}{\|\bm{\phi}\|_2^4} \geq  C_4 (r+2s)\widetilde{\kappa}_2 \|\bm{u}_i\|_2^2 \right \}\\
	&\hspace{5mm}\leq 4 e^{-c_{\HW} T+\log N} \leq 4 e^{-c_{\HW} T/2},
\end{align*}
if $T\geq 2c_{\HW}^{-1}\log N$.
Note that $\|\bm{G}_{\ma}^{*}\|_{\Fr}^2 \leq (r+2s)\overline{\alpha}_\ma^2$.
Combining these results with \eqref{eq:alem4}, we accomplish the proof of (iv).

%%%%%%%%%%%%%%%%%%%%%%%%%%%%%%%%%%%%%%%%%%%%%%%%%%%%%%%%%%%%%%%%%%%%%%%%%%%%%%%%%%
\subsection{Proof of Lemma \ref{lemma:init1} (Effect of initial values I))}

Note that
\begin{equation}\label{eq:S10}
	S_1(\bm{\widehat{\Delta}}) = \frac{2}{T}\sum_{t=1}^{T}\langle \bm{\varepsilon}_t, \sum_{h=t}^{\infty}\bm{\widehat{\Delta}}_h \bm{y}_{t-h} \rangle = \frac{2}{T}\sum_{i=1}^{3}S_{1i}(\bm{\widehat{\Delta}}),
\end{equation}
where
\begin{align*}
	&S_{11}(\bm{\widehat{\Delta}}) = \sum_{t=1}^{p}\langle \bm{\varepsilon}_t, \sum_{h=t}^{p}\bm{\widehat{\Delta}}_h \bm{y}_{t-h} \rangle = \sum_{t=1}^{p}\langle \bm{\varepsilon}_t, \sum_{h=t}^{p} \bm{\widehat{D}}_h \bm{y}_{t-h} \rangle,\\
	&S_{12}(\bm{\widehat{\Delta}}) = \sum_{t=1}^{p}\langle \bm{\varepsilon}_t, \sum_{h=p+1}^{\infty}\bm{\widehat{\Delta}}_h \bm{y}_{t-h} \rangle, \quad\text{and}\quad
	S_{13}(\bm{\widehat{\Delta}})  = \sum_{t=p+1}^{T}\langle \bm{\varepsilon}_t, \sum_{h=t}^{\infty}\bm{\widehat{\Delta}}_h \bm{y}_{t-h} \rangle,
\end{align*}
with $\bm{\widehat{D}}_h=\bm{\widehat{G}}_h-\bm{G}_h^*=\bm{\widehat{\Delta}}_h$ for $1\leq h\leq p$.  Without loss of generality, we assume that $p\geq 1$; otherwise, $S_{11}(\bm{\widehat{\Delta}})$ will simply disappear.

Note that
\begin{align*}
	|S_{11}(\bm{\widehat{\Delta}})| = \left | \sum_{h=1}^{p} \sum_{t=1}^{h} \langle \bm{\varepsilon}_t, \bm{\widehat{D}}_h \bm{y}_{t-h} \rangle \right |
	= \left | \sum_{h=1}^{p} \langle \sum_{t=1}^{h}  \bm{\varepsilon}_t \bm{y}_{t-h}^\top, \bm{\widehat{D}}_h \rangle \right |
	&\leq \sum_{h=1}^{p} \|\vect(\bm{\widehat{D}}_h)\|_1 \left \| \sum_{t=1}^{h}  \bm{\varepsilon}_t \bm{y}_{t-h}^\top \right \|_{\max}\\
	%\leq  \left \{\sum_{h=1}^{p} \|\vect(\bm{\widehat{D}}_h)\|_1^2 \right \}^{1/2} \left \{ \sum_{h=1}^p \left \| \sum_{t=1}^{h}  \bm{\varepsilon}_t \bm{y}_{t-h}^\top \right \|_{\max}^2 \right \}^{1/2}\\
	&\leq \|\bm{\widehat{d}}_{\ar}\|_1  \max_{1\leq h\leq p} \left \| \sum_{t=1}^{h}  \bm{\varepsilon}_t \bm{y}_{t-h}^\top \right \|_{\max}.
\end{align*}
For any fixed $1\leq h\leq p$, by a method similar to that for claim (i) in the proof of Lemma \ref{lemma:devb}, we can show that
\[
\mathbb{P}\left \{\left \| \sum_{t=1}^{h}\bm{\varepsilon}_t \bm{y}_{t-h}^{\top} \right \|_{\max} \leq C_1 \sqrt{h\kappa_2\lambda_{\max}(\bm{\Sigma}_{\varepsilon})\log N} \right \} \geq 1-4 e^{-2\log N}.
\]
As a result, with probability at least $1-4p e^{-2\log N}$, we have 
\begin{equation}\label{eq:S11}
	|S_{11}(\bm{\widehat{\Delta}})| \leq  C_1 \|\bm{\widehat{d}}_{\ar}\|_1\sqrt{p\kappa_2\lambda_{\max}(\bm{\Sigma}_{\varepsilon})\log N }.
\end{equation}

For $S_{12}(\bm{\widehat{\Delta}})$, similar to \eqref{eq:devb3}, we have
\begin{align*}
	|S_{12}(\bm{\widehat{\Delta}})| &\leq  \|\bm{\widehat{d}}_{\ma}\|_1 \sup_{\bm{\omega}\in\bm{\Omega}} \left \| \sum_{t=1}^{p}\bm{\varepsilon}_t \bm{x}_{t-p}^\top \{ \bm{L}^{\ma}(\bm{\omega}) \otimes \bm{I}_N\}\right \|_{\max} \notag\\
	&\hspace{5mm} +  \|\bm{g}_{\ma}^{*}\|_1 \sup_{\bm{\phi}\in \bm{\Phi}_1} \left \| \sum_{t=1}^{p} \bm{\varepsilon}_t \bm{x}_{t-p}^\top  \left [\left \{\bm{L}^{\ma}(\bm{\omega}^*+\bm{\phi}) -\bm{L}^{\ma}(\bm{\omega}^*) \right \}\otimes\bm{I}_N\right ] \right \|_{\max}
\end{align*}
By a method similar to that for claim (ii) in the proof of Lemma \ref{lemma:devb}, we can show that with probability at least  $1- 4  e^{-4 \log N}$,
\begin{align*}
	&\sup_{\bm{\omega}\in\bm{\Omega}} \left \| \sum_{t=1}^{p}\bm{\varepsilon}_t \bm{x}_{t-p}^\top \{ \bm{L}^{\ma}(\bm{\omega}) \otimes \bm{I}_N\}\right \|_{\max} \leq C_2 \sqrt{p\kappa_2\lambda_{\max}(\bm{\Sigma}_{\varepsilon})\log N},\\
	&\sup_{\bm{\phi}\in\bm{\Phi}_1} \frac{ \left \| \sum_{t=1}^{p} \bm{\varepsilon}_t \bm{x}_{t-p}^\top  \left [\left \{\bm{L}^{\ma}(\bm{\omega}^*+\bm{\phi}) -\bm{L}^{\ma}(\bm{\omega}^*) \right \}\otimes\bm{I}_N\right ] \right \|_{\max} }{\|\bm{\phi}\|_2 }  \leq C_3 \sqrt{p\kappa_2\lambda_{\max}(\bm{\Sigma}_{\varepsilon})\log N}.
\end{align*}
Therefore, with probability at least  $1- 5 e^{-4 \log N}$,
\begin{equation}\label{eq:S12}
	|S_{12}(\bm{\widehat{\Delta}})| \leq \sqrt{p}(C_2+C_3) ( \|\bm{\widehat{d}}_{\ma}\|_1  + \|\bm{g}_{\ma}^{*}\|_1 \|\bm{\widehat{\phi}}\|_2 ) \sqrt{\kappa_2\lambda_{\max}(\bm{\Sigma}_{\varepsilon})\log N},
\end{equation}

Now we handle $S_{13}(\bm{\widehat{\Delta}})$.
For any $t\geq p+1$, let $\bm{\widehat{\Delta}}_{[t]}=(\bm{\widehat{\Delta}}_t, \bm{\widehat{\Delta}}_{t+1}, \dots)$ be the horizontal concatenation of $\{\bm{\widehat{\Delta}}_h\}_{h\geq t}$. For any $h\geq1$, let $\bm{L}^{\ma}_{[h]}(\bm{\omega})$ be the matrix obtained by removing the first $h-1$ rows of $\bm{L}^{\ma}(\bm{\omega})$. For any $t\geq p+1$, we have
\begin{align*}
	\sum_{h=t}^\infty \widehat{\bm{\Delta}}_h \bm{y}_{t-h} &= \bm{\widehat{\Delta}}_{[t]} \bm{x}_{1} 
	=  \left [ \widehat{\bm{G}}_{\ma} \{\bm{L}^{\ma}_{[t-p]}(\widehat{\bm{\omega}})\otimes \bm{I}_N\}^\top - \bm{G}_{\ma}^{*} \{\bm{L}^{\ma}_{[t-p]}(\bm{\omega}^*) \otimes \bm{I}_N\}^\top \right ]\bm{x}_{1}\\
	&=  \widehat{\bm{D}}_{\ma} \{ \bm{L}^{\ma}_{[t-p]}(\widehat{\bm{\omega}}) \otimes \bm{I}_N\}^\top \bm{x}_{1} + \bm{G}_{\ma}^{*} \left [\left \{\bm{L}^{\ma}_{[t-p]}(\widehat{\bm{\omega}}) -\bm{L}^{\ma}_{[t-p]}(\bm{\omega}^*) \right \}\otimes\bm{I}_N\right ]^\top \bm{x}_{1}.
\end{align*}
Thus, we can apply arguments similar to those for claim (ii) in the proof of Lemma \ref{lemma:devb} to handle  $S_{13}(\bm{\widehat{\Delta}})$. First, similar to \eqref{eq:devb3}, we can show that
\begin{align*}%\label{eq:S33a}
	|S_{13}(\bm{\widehat{\Delta}})| & \leq   \|\bm{\widehat{d}}_{\ma}\|_1 \sup_{\bm{\omega}\in\bm{\Omega}} \left \|\sum_{t=p+1}^{T}\bm{\varepsilon}_t \bm{x}_{1}^\top \{ \bm{L}^{\ma}_{[t-p]}(\bm{\omega}) \otimes \bm{I}_N\}\right \|_{\max} \notag\\
	&\hspace{5mm} +  \|\bm{g}_{\ma}^{*}\|_1 \sup_{\bm{\phi}\in \bm{\Phi}_1} \left \| \sum_{t=p+1}^{T} \bm{\varepsilon}_t \bm{x}_{1}^\top  \left [\left \{\bm{L}^{\ma}_{[t-p]}(\bm{\omega}^*+\bm{\phi}) -\bm{L}^{\ma}_{[t-p]}(\bm{\omega}^*) \right \}\otimes\bm{I}_N\right ] \right \|_{\max}. 
\end{align*}
%Consider the first term on the right-hand side of \eqref{eq:S33a}. 
Similar to \eqref{eq:dvb1}, we can show that
\begin{align}\label{eq:init1eq1}
	& \sup_{\bm{\omega}\in\bm{\Omega}} \left \|\sum_{t=p+1}^{T}\bm{\varepsilon}_t \bm{x}_{1}^\top \{ \bm{L}^{\ma}_{[t-p]}(\bm{\omega}) \otimes \bm{I}_N\}\right \|_{\max} = \sup_{\bm{\omega}\in\bm{\Omega}}\max_{1\leq i, j\leq N,p+1\leq k\leq d}\left |  \sum_{t=p+1}^{T}\varepsilon_{i,t} \sum_{h=t}^\infty\ell_{h,k}(\bm{\omega})y_{j,t-h} \right |\notag \\
	&\hspace{5mm}\leq  \sum_{h=p+1}^\infty \sup_{\bm{\omega}\in\bm{\Omega}}\max_{p+1\leq k\leq d}|\ell_{h,k}(\bm{\omega})| \max_{1\leq i, j\leq N}  \left | \sum_{t=p+1}^{h \wedge T}\varepsilon_{i,t} y_{j,t-h} \right |\notag \\
	&\hspace{5mm}\leq 
	\sum_{h=p+1}^\infty \bar{\rho}^{h-p} \max_{1\leq i, j\leq N}  \left | \sum_{t=p+1}^{h \wedge T}\varepsilon_{i,t} y_{j,t-h} \right |,
\end{align}
and, similar to \eqref{eq:dvb4}, it can be verified that
\begin{align*}
	&\mathbb{P}\left \{\forall h\geq p+1:  \max_{1\leq i, j\leq N} \left |  \sum_{t=p+1}^{h\wedge T}\varepsilon_{i,t} y_{j,t-h} \right |  \geq \{2(h-p)\sigma^2+1\} \sqrt{8(h-p)\kappa_2\lambda_{\max}(\bm{\Sigma}_{\varepsilon})\log N}  \right \}\notag\\
	&\hspace{5mm}  \leq 5 e^{-4 \log N}.
\end{align*}
As a result, with probability at least $1-5 e^{-4 \log N}$, we have
\begin{align*}
	\sup_{\bm{\omega}\in\bm{\Omega}} \left \|\sum_{t=p+1}^{T}\bm{\varepsilon}_t \bm{x}_{1}^\top \{ \bm{L}^{\ma}_{[t-p]}(\bm{\omega}) \otimes \bm{I}_N\}\right \|_{\max} 
	&\leq \sum_{h=p+1}^{\infty}  \bar{\rho}^{h-p} \{2(h-p)\sigma^2+1\} \sqrt{8(h-p)\kappa_2\lambda_{\max}(\bm{\Sigma}_{\varepsilon})\log N} \\
	& \lesssim \sqrt{\kappa_2\lambda_{\max}(\bm{\Sigma}_{\varepsilon})\log N}.
\end{align*}
Furthermore, along the lines of  \eqref{eq:lemdvb2}, we can simultaneously derive the upper bound:
\[
\sup_{\bm{\phi}\in \bm{\Phi}_1} \frac{\left \| \sum_{t=p+1}^{T} \bm{\varepsilon}_t \bm{x}_{1}^\top  \left [\left \{\bm{L}^{\ma}_{[t-p]}(\bm{\omega}^*+\bm{\phi}) -\bm{L}^{\ma}_{[t-p]}(\bm{\omega}^*) \right \}\otimes\bm{I}_N\right ] \right \|_{\max} }{\|\bm{\phi}\|_2} \lesssim \sqrt{\kappa_2\lambda_{\max}(\bm{\Sigma}_{\varepsilon})\log N}.
\] 
In view of the above results,  with probability at least $1-5 e^{-4 \log N}$, we have
\begin{equation}\label{eq:S13}
	|S_{13}(\bm{\widehat{\Delta}})| \leq  C_5 ( \|\bm{\widehat{d}}_{\ma}\|_1  + \|\bm{g}_{\ma}^{*}\|_1 \|\bm{\widehat{\phi}}\|_2 ) \sqrt{\kappa_2\lambda_{\max}(\bm{\Sigma}_{\varepsilon})\log N},
\end{equation}
for some absolute constant $C_5>0$.

Let $C_{\init1}=2(C_1+C_2+C_3+C_5)$, where $C_i$'s are from \eqref{eq:S11}--\eqref{eq:S13}. By \eqref{eq:S10}--\eqref{eq:S13} and the fact that $\|\widehat{\bm{d}}_{\ar}\|_1+\|\bm{\widehat{d}}_{\ma}\|_1=\|\widehat{\bm{d}}\|_1$, we accomplish the proof of this lemma.

%%%%%%%%%%%%%%%%%%%%%%%%%%%%%%%%%%%%%%%%%%%%%%%%%%%%%%%%%%%%%%%%%%%%%%%%%%%%%%%%%%
\subsection{Proof of Lemma \ref{lemma:init2} (Effect of initial values II)}

Similar to the proof of Lemma \ref{lemma:init1}, consider the partition
\begin{equation}\label{eq:S20}
	S_2(\bm{\widehat{\Delta}}) = \frac{2}{T}\sum_{t=2}^{T}\langle \sum_{h=t}^{\infty}\bm{A}_h^* \bm{y}_{t-h}, \sum_{k=1}^{t-1}\bm{\widehat{\Delta}}_k \bm{y}_{t-k} \rangle = \frac{2}{T}\sum_{i=1}^{3}S_{2i}(\bm{\widehat{\Delta}}),
\end{equation}
where
\begin{align*}
	&S_{21}(\bm{\widehat{\Delta}}) = \sum_{t=2}^{p+1}\langle \sum_{h=t}^{\infty}\bm{A}_h^* \bm{y}_{t-h}, \sum_{k=1}^{t-1}\bm{\widehat{\Delta}}_k \bm{y}_{t-k} \rangle = \sum_{t=2}^{p+1}\langle \sum_{h=t}^{\infty}\bm{A}_h^* \bm{y}_{t-h}, \sum_{k=1}^{t-1}\bm{\widehat{D}}_k \bm{y}_{t-k} \rangle,\\
	&S_{22}(\bm{\widehat{\Delta}}) = \sum_{t=p+2}^{T}\langle \sum_{h=t}^{\infty}\bm{A}_h^* \bm{y}_{t-h}, \sum_{k=1}^{p}\bm{\widehat{\Delta}}_k \bm{y}_{t-k} \rangle = \sum_{t=p+2}^{T}\langle \sum_{h=t}^{\infty}\bm{A}_h^* \bm{y}_{t-h}, \sum_{k=1}^{p}\bm{\widehat{D}}_k \bm{y}_{t-k} \rangle, \\
	&S_{23}(\bm{\widehat{\Delta}})  = \sum_{t=p+2}^{T}\langle \sum_{h=t}^{\infty}\bm{A}_h^* \bm{y}_{t-h}, \sum_{k=p+1}^{t-1}\bm{\widehat{\Delta}}_k \bm{y}_{t-k} \rangle,
\end{align*}
with $\bm{\widehat{D}}_h=\bm{\widehat{G}}_h-\bm{G}_h^*=\bm{\widehat{\Delta}}_h$ for $1\leq h\leq p$. Without loss of generality, we assume that $p\geq 1$; otherwise, $S_{21}(\bm{\widehat{\Delta}})$ will simply disappear. The above partition allows us to upper bound  $|S_{2i}(\bm{\widehat{\Delta}})|$ by arguments similar to that for $S_{1i}(\bm{\widehat{\Delta}})$ in the proof of Lemma \ref{lemma:init1}, for each $1\leq i\leq 3$.  

Specifically, we begin by considering  $S_{21}(\bm{\widehat{\Delta}})$. Note that
\begin{align}\label{eq:S21a}
	|S_{21}(\bm{\widehat{\Delta}})| &= \left | \sum_{k=1}^{p} \langle \sum_{t=k+1}^{p+1}  \sum_{h=t}^{\infty}  \bm{A}_h^* \bm{y}_{t-h}, \bm{\widehat{D}}_k \bm{y}_{t-k} \rangle \right |
	= \left | \sum_{k=1}^{p} \langle  \sum_{h=k+1}^{\infty}  \sum_{t=k+1}^{h\wedge(p+1)} \bm{A}_h^* \bm{y}_{t-h} \bm{y}_{t-k}^\top, \bm{\widehat{D}}_k \rangle \right | \notag\\ 
	&\leq \sum_{k=1}^{p} \|\vect(\bm{\widehat{D}}_k)\|_1 \left \|  \sum_{h=k+1}^{\infty}  \bm{A}_h^* \sum_{t=k+1}^{h\wedge(p+1)}  \bm{y}_{t-h} \bm{y}_{t-k}^\top  \right \|_{\max} \notag\\
	%	&\leq  \left \{\sum_{k=1}^{p} \|\vect(\bm{\widehat{D}}_k)\|_1^2 \right \}^{1/2} \left \{ \sum_{k=1}^{p} \left \|  \sum_{h=k+1}^{\infty}  \bm{A}_h^* \sum_{t=k+1}^{h\wedge(p+1)}  \bm{y}_{t-h} \bm{y}_{t-k}^\top  \right \|_{\max}^2 \right \}^{1/2} \notag\\
	&\leq \|\bm{\widehat{d}}_{\ar}\|_1 \cdot \max_{1\leq k\leq p} \left \| \sum_{h=k+1}^{\infty}  \bm{A}_h^* \sum_{t=k+1}^{h\wedge(p+1)}  \bm{y}_{t-h} \bm{y}_{t-k}^\top \right \|_{\max}.
\end{align}
Let $\bm{a}_{i,h}^* \in\mathbb{R}^N$ denote the $i$th row vector of $\bm{A}_h^*$,  for $1\leq i\leq N$ and $h\geq 1$. We can show that
\begin{align} \label{eq:S21b}
	\max_{1\leq k\leq p} \left \| \sum_{h=k+1}^{\infty}  \bm{A}_h^* \sum_{t=k+1}^{h\wedge(p+1)}  \bm{y}_{t-h} \bm{y}_{t-k}^\top \right \|_{\max} &= \max_{1\leq k\leq p}  \max_{1\leq i, j\leq N} \left | \sum_{h=k+1}^{\infty} \sum_{t=k+1}^{h\wedge(p+1)}  y_{i,t-k} \bm{y}_{t-h}^\top \bm{a}_{j,h}^* \right | \notag\\
	&\leq \max_{1\leq k\leq p}  \sum_{h=k+1}^{\infty}  \max_{1\leq i, j\leq N}  \left | \sum_{t=k+1}^{h\wedge(p+1)}  y_{i,t-k} \bm{y}_{t-h}^\top \bm{a}_{j,h}^* \right | \notag\\
	&=  \max_{1\leq k\leq p}  \sum_{h=1}^{\infty} \max_{1\leq i, j\leq N} \left | \sum_{t=1}^{h\wedge(p+1-k)}  y_{i,t} \bm{y}_{t-h}^\top \bm{a}_{j,h+k}^* \right | \notag\\
	&\leq  \sum_{h=1}^{\infty} \max_{1\leq k\leq p}  \max_{1\leq i, j\leq N} \left | \sum_{t=1}^{h\wedge(p+1-k)}  y_{i,t} \bm{y}_{t-h}^\top \bm{a}_{j,h+k}^* \right |,
\end{align}
where the second last equality follows from a change of variables. For any fixed  $(i,h,k,j)$ with $1\leq i,j\leq N$, $1\leq k\leq p$ and $h\geq1$, note that $h\wedge(p+1-k)\leq p$.  

We first focus on the case where $h\geq p+1$.
Similar to \eqref{eq:reunivarj}, we can show that
\begin{equation*}%\label{eq:reunivarh1}
	\mathbb{P}\left \{\frac{1}{p}\left | \sum_{t=1}^{h\wedge(p+1-k)}  y_{i,t} \bm{y}_{t-h}^\top \bm{a}_{j,h+k}^* \right | \geq  \{(h-p)\sigma^2+1\}\kappa_2 \|\bm{a}_{j,h+k}^*\|_2 \right \} \leq 2e^{-c(h-p) T}.
\end{equation*}
By Lemma \ref{cor1}, for any $1\leq j\leq N$ and $h\geq p+1$ we have 
\begin{equation}\label{eq:adecay}
	\|\bm{a}_{j,h+k}^*\|_2 \leq C\bar{\rho}^{h-p},
\end{equation}
for some absolute constant $C>0$.
As a result, if $T\geq 4c^{-1} \log (N^2p)$, then
\begin{align*}
	&	\mathbb{P}\left \{ \max_{1\leq k\leq p}  \max_{1\leq i, j\leq N}  \left | \sum_{t=1}^{h\wedge(p+1-k)}  y_{i,t} \bm{y}_{t-h}^\top \bm{a}_{j,h+k}^* \right | \geq C\kappa_2  \{(h-p)\sigma^2+1\} p \bar{\rho}^{h-p} \right \} \\ 
	&\hspace{5mm} \leq 2N^2pe^{-c(h-p) T} \leq 2 e^{-4(h-p) \log(Np)},
\end{align*}
which  can be further strengthened to a union bound for all $h\geq p+1$ as follows:
\begin{align*}%\label{eq:S21dvb}
	&\mathbb{P}\Bigg\{\forall h\geq p+1:   \max_{1\leq k\leq p}   \max_{1\leq i, j\leq N} \left | \sum_{t=1}^{h\wedge(p+1-k)}  y_{i,t} \bm{y}_{t-h}^\top \bm{a}_{j,h+k}^* \right |  \geq C\kappa_2  \{(h-p)\sigma^2+1\}  p \bar{\rho}^{h-p} \Bigg\}\notag\\
	&\hspace{5mm}  \leq \sum_{h=p+1}^{\infty} 2 e^{-4(h-p)\log (Np)} \leq 3 e^{-4\log (Np)},
\end{align*}
where the last inequality holds as long as $N\geq 2$. In addition, for each $1\leq h\leq p$, by a similar method, we can show that
\begin{align*}
	&\mathbb{P}\left \{  \max_{1\leq k\leq p}  \max_{1\leq i, j\leq N} \left | \sum_{t=1}^{h\wedge(p+1-k)}  y_{i,t} \bm{y}_{t-h}^\top \bm{a}_{j,h+k}^* \right |   \geq C\kappa_2  (2\sigma^2+1)p \right \}\\
	&\hspace{5mm} \leq  2 e^{-4\log (Np)}, 
\end{align*}
Combining the above results with \eqref{eq:S21a} and \eqref{eq:S21b}, we have   with probability at least $1-(3+4p) e^{-4\log (Np)}$,
\begin{equation}\label{eq:S21}
	|S_{21}(\bm{\widehat{\Delta}})| \lesssim   p \|\bm{\widehat{d}}_{\ar}\|_1 \kappa_2.
\end{equation}

Next, for $i=2$ and 3, the upper bound for $|S_{2i}(\bm{\widehat{\Delta}})|$ can be readily  established by combining techniques we have used above for $|S_{21}(\bm{\widehat{\Delta}})|$ and methods similar to those for $|S_{1i}(\bm{\widehat{\Delta}})|$ in the proof of Lemma \ref{lemma:init1}. That is, for each $i=2$ and 3, we can  show that
with probability at least $1- Cp e^{-c\log (Np)}$,
\begin{equation}\label{eq:S2i}
	|S_{2i}(\bm{\widehat{\Delta}})| \lesssim  p ( \|\bm{\widehat{d}}_{\ma}\|_1  + \|\bm{g}_{\ma}^{*}\|_1 \|\bm{\widehat{\phi}}\|_2 ) \kappa_2.
\end{equation}
Since the proof of this result follows closely the lines of \eqref{eq:S12} and \eqref{eq:S13} in the proof of Lemma \ref{lemma:init1} (with only slight modifications to exploit the decay property similar to  \eqref{eq:adecay}),  but will be rather tedious, we omit the details here. 

Combining \eqref{eq:S20}, \eqref{eq:S21}, \eqref{eq:S2i}, and the fact that $\|\widehat{\bm{d}}_{\ar}\|_1+\|\bm{\widehat{d}}_{\ma}\|_1=\|\widehat{\bm{d}}\|_1$, we accomplish the proof of this lemma.

%%%%%%%%%%%%%%%%%%%%%%%%%%%%%%%%%%%%%%%%%%%%%%%%%%%%%%%%%%%%%%%%%%%%%%%%%%%%%%%%%%
\subsection{Proof of Lemma \ref{lemma:init3} (Effect of initial values III)}
For any $t\geq p+1$, let $\bm{\Delta}_{[t]}=(\bm{\Delta}_t, \bm{\Delta}_{t+1}, \dots)$ be the horizontal concatenation of $\{\bm{\Delta}_h\}_{h\geq t}$. Note that
\begin{align}\label{eq:S30}
	|S_3(\bm{\Delta})|=\frac{3}{T}\sum_{t=1}^{T}\Big \|\sum_{k=t}^{\infty}\bm{\Delta}_k\bm{y}_{t-k} \Big \|_2^2 &=\frac{3}{T} \Bigg \{\sum_{t=1}^{p} \Big \|\sum_{k=t}^{\infty}\bm{\Delta}_k\bm{y}_{t-k} \Big \|_2^2 + \underbrace{\sum_{t=p+1}^{T} \Big \|\sum_{k=t}^{\infty}\bm{\Delta}_k\bm{y}_{t-k} \Big \|_2^2 }_{ S_{33}(\bm{\Delta}) } \Bigg \} \notag \\&\leq \frac{3}{T} \left \{2\sum_{i=1}^{2}S_{3i}(\bm{\Delta}) + S_{33}(\bm{\Delta}) \right \},
\end{align}
where 
\begin{align*}
	S_{31}(\bm{\Delta}) & = \sum_{t=1}^{p} \Big \|\sum_{k=t}^{p}\bm{\Delta}_k\bm{y}_{t-k} \Big \|_2^2 = \sum_{t=1}^{p} \Big \|\sum_{k=t}^{p}\bm{D}_k\bm{y}_{t-k} \Big \|_2^2,\\
	S_{32}(\bm{\Delta}) & =\sum_{t=1}^{p} \Big \|\sum_{k=p+1}^{\infty}\bm{\Delta}_k\bm{y}_{t-k} \Big \|_2^2 = \sum_{t=1}^{p}\|\bm{\Delta}_{[p+1]} \bm{x}_{t-p}\|_2^2, \\
	S_{33}(\bm{\Delta}) & =\sum_{t=p+1}^{T} \Big \|\sum_{k=t}^{\infty}\bm{\Delta}_k\bm{y}_{t-k} \Big \|_2^2 = \sum_{t=p+1}^{T} \|\bm{\Delta}_{[t]}\bm{x}_1\|_2^2,
\end{align*}
with $\bm{D}_h=\bm{G}_h-\bm{G}_h^*=\bm{\Delta}_h$ for $1\leq h\leq p$. Without loss of generality, we assume that $p\geq 1$; otherwise, both $S_{31}(\bm{\Delta})$ and $S_{32}(\bm{\Delta})$ will simply disappear.

We first consider $S_{31}(\bm{\Delta})$. For any $k\geq 1$, denote $\bm{X}_0^k=(\bm{y}_0, \dots, \bm{y}_{1-k})$. It can be verified that 
\begin{align}\label{eq:S31a}
	S_{31}(\bm{\Delta}) & =\sum_{t=1}^{p} \langle \sum_{k=t}^{p}\bm{D}_k\bm{y}_{t-k}, \sum_{j=t}^{p}\bm{D}_j\bm{y}_{t-j} \rangle = \sum_{k=1}^{p}\sum_{j=1}^{p} \sum_{t=1}^{k\wedge j}\langle \bm{D}_k\bm{y}_{t-k}, \bm{D}_j\bm{y}_{t-j} \rangle \notag\\
	&\leq \sum_{k=1}^{p}\sum_{j=1}^{p} \left (\sum_{t=1}^{k\wedge j} \|\bm{D}_k\bm{y}_{t-k}\|_2^2\right )^{1/2}
	\left (\sum_{t=1}^{k\wedge j} \|\bm{D}_j\bm{y}_{t-j}\|_2^2\right )^{1/2} \notag\\
	&\leq \left \{\sum_{k=1}^{p} \left (\sum_{t=1}^{k} \|\bm{D}_k\bm{y}_{t-k}\|_2^2\right )^{1/2}\right \}^2 
	= \left (\sum_{k=1}^{p} \|\bm{D}_k\bm{X}_0^k\|_{\Fr}\right)^2. 
\end{align}

For each fixed $1\leq k\leq p$, we can apply techniques similar to those for the proof of claim (i) in Section  \ref{asec:rsc} to upper bound $\|\bm{D}_k\bm{X}_0^k\|_{\Fr}$.
Specifically, note that
\begin{equation}\label{eq:S31b}
	\frac{1}{k}\|\bm{D}_k\bm{X}_0^k\|_{\Fr}^2=\frac{1}{k}\trace( \bm{X}_0^{k\top}\bm{D}_k^\top\bm{D}_{k} \bm{X}_0^k)= \trace\left (\bm{D}_k \widehat{\bm{\Sigma}}_y^k \bm{D}_{k}^\top\right )=\vect(\bm{D}_{k}^\top)^\top  (\bm{I}_N \otimes \widehat{\bm{\Sigma}}_y^k ) \vect(\bm{D}_{k}^\top),
\end{equation}
where $\widehat{\bm{\Sigma}}_y^k=\bm{X}_0^k\bm{X}_0^{k\top}/k=k^{-1}\sum_{t=1}^{k}\bm{y}_{t-k} \bm{y}_{t-k}^\top$. Similar to \eqref{eq:hansonw}, by applying Lemmas \ref{lemma:hansonw}(ii) and \ref{lemma:Wcov}, where we take $T_0=0$, $T_1=k$, $\bm{w}_t=\bm{y}_{t-k}$, $\bm{M}=\bm{u}^\top$, and  $\eta=\log N/(108\sigma^2)$, we can derive 
the following pointwise bound: for any $\bm{u}\in\mathbb{R}^{N}$ with $\|\bm{u}\|_2\leq 1$,
\begin{equation*}
	\mathbb{P}\left \{ \bm{u}^\top  (\widehat{\bm{\Sigma}}_y^k - \bm{\Sigma}_y)\bm{u} \geq \kappa_2 \log N/108\right \} \leq 2e^{-c k \log N},
\end{equation*}
where $c=c_{\HW}\min\{(108\sigma^{2})^{-1}, (108\sigma^{2})^{-2}\}$. Let $\pazocal{K}(2K)=\{\bm{u}\in\mathbb{R}^{N}: \|\bm{u}\|_2\leq 1, \|\bm{u}\|_0\leq 2 K\}$ be a set of sparse vectors, where $K\geq1$ is an integer to be specified later. Then, by arguments similar to the proof of Lemma F.2 in \cite{basu2015regularized}, we can strengthen the above pointwise bound to the union bound as follows:
\begin{equation*}%\label{eq:hansonw1}
	\mathbb{P}\left \{ \sup_{\bm{u}\in \pazocal{K}(2K)}\bm{u}^\top  (\widehat{\bm{\Sigma}}_y^k - \bm{\Sigma}_y)\bm{u} \geq  \kappa_2 \log N/108\right \} \leq 2e^{-c k \log N +2K\log N},
\end{equation*}
Now we choose $K=\lceil 0.25c k \log N \rceil$. Consequently, by Supplementary Lemma 12 in \cite{LW12}, we  have
\begin{equation*}%\label{eq:reclaim1}
	\mathbb{P}\left \{ \forall \bm{u}\in\mathbb{R}^{N}: |\bm{u}^\top ( \widehat{\bm{\Sigma}}_y^k - \bm{\Sigma}_y) \bm{u}| \leq \frac{\kappa_2 \log N}{4} \|\bm{u}\|_2^2 + \frac{\kappa_2}{c k}
	\|\bm{u}\|_1^2 \right \} \geq 1-2e^ {-0.5c k \log N}.
\end{equation*}	
This further implies that 
\begin{equation*}%\label{eq:reclaim1}
	\mathbb{P}\left \{ \forall \bm{u}\in\mathbb{R}^{N^2}: |\bm{u}^\top \{\bm{I}_N\otimes ( \widehat{\bm{\Sigma}}_y^k - \bm{\Sigma}_y)\} \bm{u}| \leq \frac{\kappa_2 \log N}{4} \|\bm{u}\|_2^2 + \frac{\kappa_2 }{c k}
	\|\bm{u}\|_1^2 \right \} \geq 1-2e^ {-0.5c k \log N}.
\end{equation*}	
Furthermore, by Lemma \ref{lemma:Wcov}, we have $|\bm{u}^\top (\bm{I}_N\otimes\bm{\Sigma}_y) \bm{u}| \leq \kappa_2 \|\bm{u}\|_2^2 \leq 2\kappa_2 \log N \|\bm{u}\|_2^2$ if $N\geq 2$. As a result, for any $1\leq k\leq p$, we have
\begin{align*}
	|\bm{u}^\top  (\bm{I}_N\otimes \widehat{\bm{\Sigma}}_y^k) \bm{u}| &\leq |\bm{u}^\top (\bm{I}_N\otimes\bm{\Sigma}_y) \bm{u}|+ |\bm{u}^\top \{\bm{I}_N\otimes ( \widehat{\bm{\Sigma}}_y^k - \bm{\Sigma}_y)\} \bm{u}|\\
	&\leq \frac{9\kappa_2 \log N }{4} \|\bm{u}\|_2^2 + \frac{\kappa_2}{c} \|\bm{u}\|_1^2, \quad \forall  \bm{u}\in\mathbb{R}^{N^2},
\end{align*}
with probability at least $1-2e^ {-0.5c \log N}$. Then,  applying the inequality $|x+y|^{1/2}\leq |x|^{1/2}+|y|^{1/2}$,   from the above result we further have 
\begin{equation*}%\label{eq:reclaim1}
	\mathbb{P}\left \{ \forall \bm{u}\in\mathbb{R}^{N^2}:|\bm{u}^\top  (\bm{I}_N\otimes \widehat{\bm{\Sigma}}_y^k) \bm{u}|^{1/2} \leq \sqrt{\frac{9\kappa_2 \log N }{4}} \|\bm{u}\|_2 + \sqrt{\frac{\kappa_2}{c} }\|\bm{u}\|_1 \right \} \geq 1-2e^ {-0.5c \log N}.
\end{equation*}	
Thus, in view of \eqref{eq:S31b}, for any $1\leq k\leq p$, letting $\bm{u}=\vect(\bm{D}_{k}^\top)^\top$, we have
\begin{equation}\label{eq:S31c}
	\frac{\|\bm{D}_k\bm{X}_0^k\|_{\Fr}}{\sqrt{k}} \leq  \sqrt{\frac{9\kappa_2 \log N}{4} }\|\bm{D}_k\|_{\Fr} + \sqrt{\frac{ \kappa_2}{c} } \|\bm{D}_k\|_1, \quad \forall \bm{D}_k\in \mathbb{R}^{N\times N},
\end{equation}
with probability at least $1-2e^ {-0.5c \log N}$. This, together with \eqref{eq:S31a}, implies that
\begin{align}\label{eq:S31}
	S_{31}(\bm{\Delta}) &\leq  p\left ( \sqrt{\frac{9\kappa_2 \log N}{4} } \sum_{k=1}^{p} \|\bm{D}_k\|_{\Fr} + \sqrt{\frac{ \kappa_2}{c} } \sum_{k=1}^{p} \|\bm{D}_k\|_1  \right)^2 \notag\\
	& \lesssim (\kappa_2  p\log N) \|\bm{d}_{\ar}\|_{2}^2 + \kappa_2 p\|\bm{d}_{\ar}\|_{1}^2, \quad \forall \bm{\Delta}\in\bm{\Upsilon},
\end{align}
with probability at least $1-2e^ {-0.5c \log N}$.

Next we consider $S_{32}(\bm{\Delta})$. The method will be similar to that for Lemma \ref{lemma:rsclasso}. Specifically, by \eqref{eq:stackH} and \eqref{eq:Delta}, we can show that
\[
\bm{\Delta}_{[p+1]}=\bm{D}_{\ma}\{\bm{L}^\ma(\bm{\omega}^*) \otimes\bm{I}_N\}^\top+\bm{M}(\bm{\phi}) \{\bm{P}(\bm{\omega}^*)\otimes\bm{I}_N\}^\top+\bm{D}_{\ma}\{ \bm{Q}(\bm{\phi})\otimes\bm{I}_N\}^\top+\bm{G}_{\ma}^{*}\{ \bm{S}(\bm{\phi})\otimes\bm{I}_N\}^\top,
\]  
where  $\bm{P}(\bm{\omega}^*), \bm{Q}(\bm{\phi}), \bm{S}(\bm{\phi}) \in\mathbb{R}^{\infty \times (r+2s)}$ and 
$\bm{M}(\bm{\phi})\in\mathbb{R}^{N\times N(r+2s)}$ are defined as in the proof of Lemma \ref{lemma:rsclasso}. For simplicity, with a slight modification to the notation in \eqref{eq:ts-z}, we define  
\begin{align}\label{eq:S32notation}
	\begin{split}
		\bm{Z}_{-p}&=(\bm{z}_{1-p},\dots, \bm{z}_{0}), \quad \bm{z}_t =\left \{\bm{L}^\ma(\bm{\omega}^*)\otimes \bm{I}_N \right \}^\top\bm{x}_{t},\\
		\bm{V}_{-p}&=(\bm{v}_{1-p},\dots, \bm{v}_{0}), \quad \bm{v}_t =\left \{\bm{P}(\bm{\omega}^*)\otimes \bm{I}_N \right \}^\top\bm{x}_{t},\\
		\bm{H}_{-p}(\bm{\phi})&=(\bm{h}_{1-p}(\bm{\phi}),\dots, \bm{h}_{0}(\bm{\phi})), \quad \bm{h}_t(\bm{\phi})=\left \{\bm{Q}(\bm{\phi})\otimes \bm{I}_N \right \}^\top\bm{x}_{t},\\ 
		\bm{B}_{-p}(\bm{\phi})&=(\bm{b}_{1-p}(\bm{\phi}),\dots, \bm{b}_{0}(\bm{\phi})), \quad \bm{b}_{t}(\bm{\phi})=\left \{\bm{S}(\bm{\phi})\otimes \bm{I}_N \right \}^\top\bm{x}_{t},
	\end{split}
\end{align}
and $\bm{X}_{-p}=(\bm{x}_{1-p}, \dots, \bm{x}_{0})$.
Consequently,
\begin{equation*}%\label{eq:split}
	\bm{\Delta}_{[p+1]}\bm{x}_{t}
	=\bm{D}_{\ma} \bm{z}_{t} + \bm{M}(\bm{\phi}) \bm{v}_{t} +\bm{D}_{\ma}\bm{h}_{t}(\bm{\phi}) + \bm{G}_{\ma}^{*}\bm{b}_{t}(\bm{\phi}),
\end{equation*}
and then
\[
\bm{\Delta}_{[p+1]}\bm{X}_{-p}=\bm{D}_{\ma} \bm{Z}_{-p} + \bm{M}(\bm{\phi}) \bm{V}_{-p} +\bm{D}_{\ma}\bm{H}_{-p}(\bm{\phi})+\bm{G}_{\ma}^{*}\bm{B}_{-p}(\bm{\phi}).
\]
Moreover, by the triangle inequality,
\begin{align}\label{eq:S32triangle}
	S_{32}^{1/2}(\bm{\Delta}) &=\left \{\sum_{t=1}^{p}\|\bm{\Delta}_{[p+1]}\bm{x}_{t-p}\|_{2}^2 \right \}^{1/2}= \|\bm{\Delta}_{[p+1]}\bm{X}_{-p}\|_{\Fr} \notag\\
	&\leq \|\bm{D}_{\ma} \bm{Z}_{-p}\|_{\Fr}+\|\bm{M}(\bm{\phi}) \bm{V}_{-p}\|_{\Fr}+\|\bm{D}_{\ma}\bm{H}_{-p}(\bm{\phi})\|_{\Fr}+\|\bm{G}_{\ma}^{*}\bm{B}_{-p}(\bm{\phi})\|_{\Fr}.
\end{align}
Now our task is to upper bound  each of the four terms on the right-hand side of \eqref{eq:S32triangle}. It is worth noting the resemblance of the above terms to those in \eqref{eq:triangle}. In fact, although claim (i) in the proof of Lemma \ref{lemma:rsclasso} focuses on the lower bound, similar techniques can be used to derive an upper bound for $\|\bm{D}_{\ma} \bm{Z}_{-p}\|_{\Fr}$; see also the arguments that lead to \eqref{eq:S31c} above. Specifically, we can show that 
\begin{equation}\label{eq:S32a}
	\frac{\|\bm{D}_{\ma} \bm{Z}_{-p}\|_{\Fr}}{\sqrt{p}} \leq  \sqrt{\frac{9(r+2s)\kappa_2 \log (Np)}{4} }\|\bm{d}_{\ma}\|_{2} + \sqrt{\frac{(r+2s)\kappa_2}{c} } \|\bm{d}_{\ma}\|_1, \quad \forall \bm{d}_{\ma}\in \mathbb{R}^{N^2 p},
\end{equation}
with probability at least $1-2e^ {-0.5c \log (Np)}$.  

Furthermore, by arguments similar to those for \eqref{eq:hansonw1}, we have for any  $\bm{M} \in  \mathbb{R}^{N \times N(r+2s)}$ the pointwise bound:
\begin{equation*}
	\mathbb{P}\left( \frac{\|\bm{M} \bm{V}_{-p}\|_{\Fr} }{\sqrt{p}} \leq \sqrt{ \widetilde{\kappa}_2 \{1+ \log (Np)\}} \|\bm{M}\|_{\Fr}  \right)  
	\geq 1- 2e^{-2c\widetilde{\kappa}_1^2 p\log(Np)/\{(r+2s)\kappa_2\}^2}.
\end{equation*}
To strengthen it to a union bound that holds for all $\bm{M}\in\bm{\Xi}_1$, consider a minimal generalized $1/2$-net $\bm{\bar{\Xi}}(1/2)$ of $\bm{\Xi}_1$ in the Frobenius norm. 
By  Lemma \ref{lemma:epsilon-net}(ii), any $\bm{M}\in\bm{\bar{\Xi}}(1/2)$ satisfies $\|\bm{M}\|_{\Fr}\leq u_\phi / l_\phi$. 
Then, by the discretization and covering number in Lemma \ref{lemma:epsilon-net}, we can show that
\begin{align*}
	&\mathbb{P}\left [ \sup_{\bm{M}\in\bm{\Xi}_1}\frac{\|\bm{M} \bm{V}_{-p}\|_{\Fr}}{\sqrt{p}}\geq 2 (u_\phi / l_\phi)\sqrt{ \widetilde{\kappa}_2 \{1+ \log (Np)\} } \right ]\\
	&\leq \mathbb{P}\left[ \max_{ \bm{M}\in \bm{\bar{\Xi}}(1/2)} \frac{\|\bm{M} \bm{V}_{-p}\|_{\Fr}}{\sqrt{p}}\geq  (u_\phi / l_\phi) \sqrt{ \widetilde{\kappa}_2\{1+ \log (Np)\}} \right]\\
	&\leq e^{(r+2s) \log(6/c_{\bm{M}})} \max_{\bm{M}\in\bm{\bar{\Xi}}(1/2)} \mathbb{P} \left [\frac{\|\bm{M} \bm{V}_{-p}\|_{\Fr}}{\sqrt{p}} \geq  (u_\phi / l_\phi) \sqrt{ \widetilde{\kappa}_2 \{1+ \log (Np)\}}  \right ] \notag\\
	&\leq  2\exp\left[- 2c\widetilde{\kappa}_1^2 p\log(Np)/\{(r+2s)\kappa_2\}^2 +(r+2s) \log(6u_\phi/l_\phi) \right].
\end{align*}
Combining this  with 
\eqref{eq:Xi1} and the upper bound in \eqref{eq:DFr},
under the condition that $\log(Np)\geq  c^{-1} (r+2s)^2 (\kappa_2/\widetilde{\kappa}_1)^2
\log (6 u_\phi/l_\phi)$, we have
\begin{equation}\label{eq:S32b}
	\sup_{\bm{\phi}\in\bm{\Phi}} \frac{\|\bm{M}(\bm{\phi}) \bm{V}_{-p}\|_{\Fr}}{\sqrt{p}\| \bm{\phi}\|_2} 
	\leq
	(u_\phi / l_\phi^2) \sqrt{\widetilde{\kappa}_2 \{1+ \log (Np)\} }, 
\end{equation}
with probability at least $1- 2e^{-c \widetilde{\kappa}_1^2 p\log(Np)/\{(r+2s)\kappa_2\}^2 }$.  

We can also derive upper bounds for the third and last terms in \eqref{eq:S32triangle} by slightly modifying the proofs of claims (iii) and (iv) in the proof of Lemma \ref{lemma:rsclasso}, respectively. Denote $\widehat{\bm{\Sigma}}_h^{p}(\bm{\phi})=\bm{H}_{-p}(\bm{\phi})\bm{H}_{-p}^\top(\bm{\phi})/p=p^{-1}\sum_{t=1}^{p}\bm{h}_{t-p}(\bm{\phi}) \bm{h}_{t-p}^\top(\bm{\phi})$. Along the lines of \eqref{eq:alem1} we can show that for any fixed $\bm{u}\in\mathbb{R}^{N(r+2s)}$,     if $p\log \{N(r+2s)\}\geq \max\{1, c_{\HW}^{-1}\log 2\}$,  then with probability at least $1-4e^{-c_{\HW} p\log\{N(r+2s)\}}$,
\begin{equation*}
	\sup_{\bm{\phi}\in \bm{\Phi}_1} \frac{|\bm{u}^\top  \widehat{\bm{\Sigma}}_h^p(\bm{\phi}) \bm{u} |}{\| \bm{\phi}\|_2^2} \leq   C_4\widetilde{\kappa}_2 \|\bm{u}\|_2^2 \log\{N(r+2s)\},
\end{equation*}
where $C_4>0$ is the absolute constant defined as in the proof of Lemma \ref{lemma:rsclasso}.  Note that, however, a bit different from \eqref{eq:alem1},  the above result is obtained by taking $\eta=\log \{N(r+2s)\}$ when applying Lemma \ref{alem:rsc3aux}. Then, by a method similar to that for  \eqref{eq:reclaim3} but taking the sparsity level $K=\lceil 0.25c_{\HW} p\log \{N(r+2s)\}\rceil$, we can show that  with probability at least $1-4e^{-0.5c_{\HW} p\log\{N(r+2s)\}}$,
\begin{equation*}
	\sup_{\bm{\phi}\in \bm{\Phi}_1} \frac{|\bm{u}^\top  \widehat{\bm{\Sigma}}_h^p(\bm{\phi}) \bm{u} |}{\| \bm{\phi}\|_2^2}   \leq C_4\widetilde{\kappa}_2\left [\log\{N(r+2s)\}\|\bm{u}\|_2^2+\frac{4}{c_{\HW} p}  \|\bm{u}\|_1^2\right], \quad \forall \bm{u}\in\mathbb{R}^{N(r+2s)}.
\end{equation*}
Thus, analogous to the result of claim (iii) in the proof of Lemma \ref{lemma:rsclasso}, it then follows that
\begin{equation}\label{eq:S32c}
	\sup_{\bm{\phi}\in \bm{\Phi}_1} \frac{ \|\bm{D}_{\ma}\bm{H}_{-p}(\bm{\phi})\|_{\Fr}^2 }{p \|\bm{\phi}\|_2^2} \leq C_4 \widetilde{\kappa}_2 \left [ \log\{N(r+2s)\} \|\bm{d}_{\ma}\|_{2}^2+ \frac{ 4}{c_{\HW} p} \|\bm{d}_{\ma}\|_1^2 \right ], \quad \forall \bm{d}_{\ma}\in \mathbb{R}^{N^2 (r+2s)},
\end{equation}
with probability at least $1-4e^{-0.5c_{\HW} p\log\{N(r+2s)\}}$. In addition, we can  derive an upper bound for the last term in \eqref{eq:S32triangle} by a slight modification to the proof of claim (iv) in Section \ref{asec:rsc} in the same spirit as above.  The key is to apply Lemma \ref{alem:rsc3aux} with $\eta=(2\log N)/(c_{\HW} p)$. It can be readily verified that if $2\log N\geq c_{\HW} p$, then
\begin{equation}\label{eq:S32d}
	\sup_{\bm{\phi}\in \bm{\Phi}_1}	\frac{\|\bm{G}_{\ma}^{*} \bm{B}_{-p}(\bm{\phi})\|_{\Fr}^2}{p \|\bm{\phi}\|_2^4} \leq  C_4 \overline{\alpha}_\ma^2  (r+2s) \widetilde{\kappa}_2 \cdot \frac{2\log N}{c_{\HW} p},
\end{equation}
with probability at least  $1-4 e^{-\log N}$. Therefore, in view of \eqref{eq:S32triangle}--\eqref{eq:S32d}, by a method similar to that for the proof of Lemma \ref{lemma:rsclasso}, we can show that 
\begin{equation}\label{eq:S32}
	S_{32}(\bm{\Delta})\lesssim \{\widetilde{\kappa}_2 (r+2s) p \log \{N(p\vee1)\}\} \|\bm{\Delta}\|_{\Fr}^2  + \widetilde{\kappa}_2 p \|\bm{d}_{\ma}\|_1^2, \quad \forall \bm{\Delta}\in\bm{\Upsilon},
\end{equation}
with probability at least $1-2e^ {-0.5c \log (Np)}- 2e^{-c(\widetilde{\kappa}_1/\widetilde{\kappa}_2)^2 p\log(Np) }- 4e^{-0.5c_{\HW} p\log\{N(r+2s)\}}-4 e^{-\log N}=1-Ce^ {-c (\widetilde{\kappa}_1/\widetilde{\kappa}_2)^2 p \log \{N(p\vee1)\}}$.

Lastly, we derive an upper bound for  $S_{33}(\bm{\Delta})$. In fact, the method will be very similar to that for $S_{32}(\bm{\Delta})$.  For any $h\geq1$, let $\bm{L}^{\ma}_{[h]}(\bm{\omega})$ be the matrix obtained by removing the first $h-1$ rows of $\bm{L}^{\ma}(\bm{\omega})$. Similarly, let  
$\bm{P}_{[h]}(\bm{\omega}^*), \bm{Q}_{[h]}(\bm{\phi})$, and $\bm{S}_{[h]}(\bm{\phi})$  be the matrices obtained by removing the first $h-1$ rows of 
$\bm{P}(\bm{\omega}^*), \bm{Q}(\bm{\phi})$, and $\bm{S}(\bm{\phi})$, respectively. Then for any $t\geq p+1$, we have
\begin{align*}
	\bm{\Delta}_{[t]}&=\bm{D}_{\ma}\{ \bm{L}^\ma_{[t-p]}(\bm{\omega}^*)\otimes\bm{I}_N\}^\top +\bm{M}(\bm{\phi}) \{\bm{P}_{[t-p]}(\bm{\omega}^*)\otimes\bm{I}_N\}^\top\\
	&\hspace{5mm}+\bm{D}_{\ma}\{ \bm{Q}_{[t-p]}(\bm{\phi})\otimes\bm{I}_N\}^\top+\bm{G}_{\ma}^{*}\{ \bm{S}_{[t-p]}(\bm{\phi})\otimes\bm{I}_N\}^\top.
\end{align*}
As a result, we can show that
\begin{equation*}%\label{eq:split}
	\bm{\Delta}_{[t]}\bm{x}_{1}
	=\bm{D}_{\ma} \bm{\widetilde{z}}_{t} + \bm{M}(\bm{\phi}) \bm{\widetilde{v}}_{t} +\bm{D}_{\ma}\bm{\widetilde{h}}_{t}(\bm{\phi}) + \bm{G}_{\ma}^{*}\bm{\widetilde{b}}_{t}(\bm{\phi}),
\end{equation*}
and further
\begin{align}\label{eq:S33triangle}
	S_{33}(\bm{\Delta}) &=\sum_{t=p+1}^{T}\|\bm{\Delta}_{[t]}\bm{x}_{1}\|_{2}^2 \notag\\
	&\leq 4\sum_{t=p+1}^{T} \left \{\|\bm{D}_\ma \bm{\widetilde{z}}_{t}\|_{2}^2 + \|\bm{M}(\bm{\phi}) \bm{\widetilde{v}}_{t}\|_2^2 +\|\bm{D}_{\ma}\bm{\widetilde{h}}_{t}(\bm{\phi})\|_2^2 + \|\bm{G}_{\ma}^{*}\bm{\widetilde{b}}_{t}(\bm{\phi})\|_2^2 \right \} \notag\\
	&= 4 \left \{\|\bm{D}_{\ma} \bm{\widetilde{Z}}\|_{\Fr}^2+\|\bm{M}(\bm{\phi}) \bm{\widetilde{V}}\|_{\Fr}^2+\|\bm{D}_{\ma}\bm{\widetilde{H}}(\bm{\phi})\|_{\Fr}^2+\|\bm{G}_{\ma}^{*}\bm{\widetilde{B}}(\bm{\phi})\|_{\Fr}^2 \right \},
\end{align}
where
\begin{align*}%\label{eq:S32notation}
	\begin{split}
		\bm{\widetilde{Z}}&=(\bm{\widetilde{z}}_{p+1},\dots, \bm{\widetilde{z}}_{T}), \quad \bm{\widetilde{z}}_t =\left \{\bm{L}^\ma_{[t-p]}(\bm{\omega}^*)\otimes \bm{I}_N \right \}^\top\bm{x}_{1},\\
		\bm{\widetilde{V}}&=(\bm{\widetilde{v}}_{p+1},\dots, \bm{\widetilde{v}}_{T}), \quad \bm{\widetilde{v}}_t =\left \{\bm{P}_{[t-p]}(\bm{\omega}^*)\otimes \bm{I}_N \right \}^\top\bm{x}_{1},\\
		\bm{\widetilde{H}}(\bm{\phi})&=(\bm{\widetilde{h}}_{p+1}(\bm{\phi}),\dots, \bm{\widetilde{h}}_{T}(\bm{\phi})), \quad \bm{\widetilde{h}}_t(\bm{\phi})=\left \{\bm{Q}_{[t-p]}(\bm{\phi})\otimes \bm{I}_N \right \}^\top\bm{x}_{1},\\ 
		\bm{\widetilde{B}}(\bm{\phi})&=(\bm{\widetilde{b}}_{p+1}(\bm{\phi}),\dots, \bm{\widetilde{b}}_{T}(\bm{\phi})), \quad \bm{\widetilde{b}}_{t}(\bm{\phi})=\left \{\bm{S}_{[t-p]}(\bm{\phi})\otimes \bm{I}_N \right \}^\top\bm{x}_{1}.
	\end{split}
\end{align*}
It then remains to derive upper bounds for each of the four summands in \eqref{eq:S33triangle}. Despite the resemblance of the above to \eqref{eq:S32triangle}, it is important to recognize that $\bm{\widetilde{z}}_t, \bm{\widetilde{v}}_t, \bm{\widetilde{h}}_t(\bm{\phi})$ and $\bm{\widetilde{b}}_t(\bm{\phi})$ are not stationary, unlike $\bm{z}_t, \bm{v}_t, \bm{h}_t(\bm{\phi})$ and $\bm{b}_t(\bm{\phi})$. Indeed, the key to establishing upper bounds for the terms in \eqref{eq:S33triangle} is to exploit the property that the magnitude of these variables diminishes exponentially fast as $t$ increases.   For succinctness, we will demonstrate the key trick using $\|\bm{D}_{\ma} \bm{\widetilde{Z}}\|_{\Fr}^2$ as an example. The other three summands in \eqref{eq:S33triangle} can be handled by using the same trick in conjunction with methods for upper bounding the analogous terms in \eqref{eq:S32triangle}. 

Note that  by the Cauchy-Schwarz inequality,
\begin{align*}
	\left \{\sum_{t=p+1}^{h} \Big\| \sum_{k=p+1}^{d} \ell_{h,k}(\bm{\omega}^*)\bm{D}_k\bm{y}_{t-h} \Big\|_2^2\right \}^{1/2} & \leq \sqrt{r+2s} \left \{\sum_{t=p+1}^{h} \sum_{k=p+1}^{d} |\ell_{h,k}(\bm{\omega}^*)|^2 \|\bm{D}_k\bm{y}_{t-h} \|_2^2\right \}^{1/2}\\
	& \leq \bar{\rho}^{h-p} \sqrt{r+2s}  \left \{  \sum_{k=p+1}^{d} \sum_{t=p+1}^{h} \|\bm{D}_k\bm{y}_{t-h} \|_2^2\right \}^{1/2}\\
	&\leq \bar{\rho}^{h-p} \sqrt{r+2s}  \sum_{k=p+1}^{d} \left \{   \sum_{t=p+1}^{h} \|\bm{D}_k\bm{y}_{t-h} \|_2^2\right \}^{1/2}\\
	&=\bar{\rho}^{h-p} \sqrt{r+2s} \sum_{k=p+1}^{d} \| \bm{D}_k \bm{X}_0^{h-p}\|_{\Fr},
\end{align*}
where $\bm{X}_0^{h-p}=(\bm{y}_{p+1-h}, \dots, \bm{y}_{0})$. This leads to 
\begin{align*}
	\|\bm{D}_{\ma} \bm{\widetilde{Z}}\|_{\Fr}^2 &
	= \sum_{t=p+1}^{T} \Big\|\sum_{k=p+1}^{d} \bm{D}_k\sum_{h=t}^{\infty}\ell_{h,k}(\bm{\omega}^*)\bm{y}_{t-h} \Big \|_2^2\\
	&=\sum_{t=p+1}^{T} \Big\langle \sum_{h=t}^{\infty} \sum_{k=p+1}^{d} \ell_{h,k}(\bm{\omega}^*)\bm{D}_k\bm{y}_{t-h}, \sum_{h=t}^{\infty}\sum_{i=p+1}^{d} \ell_{h,i}(\bm{\omega}^*)\bm{D}_i\bm{y}_{t-i} \Big \rangle\\
	&= \sum_{h=p+1}^{\infty} \sum_{h=p+1}^{\infty} \sum_{t=p+1}^{h\wedge h\wedge T} \Big\langle \sum_{k=p+1}^{d} \ell_{h,k}(\bm{\omega}^*)\bm{D}_k\bm{y}_{t-h}, \sum_{i=p+1}^{d} \ell_{h,i}(\bm{\omega}^*)\bm{D}_i\bm{y}_{t-i} \Big \rangle\\
	&\leq \left [\sum_{h=p+1}^{\infty} \left \{\sum_{t=p+1}^{h\wedge T} \Big\| \sum_{k=p+1}^{d} \ell_{h,k}(\bm{\omega}^*)\bm{D}_k\bm{y}_{t-h} \Big\|_2^2\right \}^{1/2} \right ]^2\\
	&\leq (r+2s) \left [ \sum_{k=p+1}^{d} \sum_{h=p+1}^{\infty} \bar{\rho}^{h-p} \| \bm{D}_k \bm{X}_0^{h-p}\|_{\Fr}  \right ]^2.
\end{align*}
By Lemma \ref{alem:rsc3aux} and a method similar to that for \eqref{eq:S31c},  for any fixed $p+1\leq k\leq d$, we can show that  
\begin{equation*}
	\frac{\|\bm{D}_k\bm{X}_0^{h-p}\|_{\Fr}}{\sqrt{h-p}} \leq  \sqrt{\frac{9\kappa_2 (h-p)\log N}{4} }\|\bm{D}_k\|_{\Fr} + \sqrt{\frac{ \kappa_2 (h-p)}{c} } \|\bm{D}_k\|_1, \quad \forall \bm{D}_k\in \mathbb{R}^{N\times N}, \forall h\geq p+1
\end{equation*}
with probability at least $1-4e^{-0.5c_{\HW} \log N}$. As a result, we have
\[
\|\bm{D}_{\ma} \bm{\widetilde{Z}}\|_{\Fr}^2  \lesssim (r+2s)\left \{(\kappa_2  \log N) \|\bm{d}_{\ma}\|_{2}^2 + \kappa_2 \|\bm{d}_{\ma}\|_{1}^2 \right \}, \quad \forall \bm{d}_{\ma}\in \mathbb{R}^{N^2 (r+2s)},
\]
with probability at least $1-4(r+2s)e^{-0.5c_{\HW} \log N}$.  Along the same lines, we can establish upper bounds for  the other three summands in \eqref{eq:S33triangle} and obtain
\begin{equation}\label{eq:S33}
	S_{33}(\bm{\Delta})\lesssim (r+2s)\left \{(\kappa_2 \log N) \|\bm{\Delta}\|_{\Fr}^2  + \kappa_2 \|\bm{d}_{\ma}\|_1^2\right \}, \quad \forall \bm{\Delta}\in\bm{\Upsilon},
\end{equation}
with probability at least $1-C (r+2s) e^ {-c (\widetilde{\kappa}_1/\widetilde{\kappa}_2)^2 p \log \{N(p\vee1)\}}$. 

Finally, note that $\widetilde{\kappa}_i \asymp \kappa_i$ for $i=1,2$. Thus, combining  \eqref{eq:S30}, \eqref{eq:S31}, \eqref{eq:S32} and \eqref{eq:S33}, we have
\begin{equation*}%\label{eq:scs}
	|S_{3}(\bm{\Delta})| \leq  \frac{ C_{\init3}  \kappa_2 (r+2s) }{T} \left(  \|\bm{\Delta}\|_{\Fr}^2 \log N +\|\bm{d}\|_1^2 \right), \quad \forall \bm{\Delta}\in\bm{\Upsilon},
\end{equation*}
with probability at least $1-C (r+2s) e^ {-c (\kappa_1/\kappa_2)^2 p \log \{N(p\vee1)\}}$. Since $\widehat{\bm{\Delta}} \in\bm{\Upsilon}$, the proof is complete.

%%%%%%%%%%%%%%%%%%%%%%%%%%%%%%%%%%%%%%%%%%%%%%%%%%%%%%%%%%%%%%%%%%%%%%%%%%%%%%%%%%
%%%%%%%%%%%%%%%%%%%%%%%%%%%%%%%%%%%%%%%%%%%%%%%%%%%%%%%%%%%%%%%%%%%%%%%%%%%%%%%%%%
%%%%%%%%%%%%%%%%%%%%%%%%%%%%%%%%%%%%%%%%%%%%%%%%%%%%%%%%%%%%%%%%%%%%%%%%%%%%%%%%%%
%%%%%%%%%%%%%%%%%%%%%%%%%%%%%%%%%%%%%%%%%%%%%%%%%%%%%%%%%%%%%%%%%%%%%%%%%%%%%%%%%%
\subsection{Additional lemmas for proofs of Lemmas \ref{lemma:devb}--\ref{lemma:init3}} \label{asec:aux2}
This section contains several lemmas used to establish Lemmas \ref{lemma:devb}--\ref{lemma:init3}. Their proofs are given in Section \ref{asec:auxproof}.

Firstly, in Lemmas \ref{lemma:hansonw}--\ref{lemma:Wcov} below, we adopt the following notations.  Let $\{\bm{w}_t\}$ be a generic time series taking values in $\mathbb{R}^M$, 
where $M$ is an arbitrary positive integer. If $\{\bm{w}_t\}$ is stationary with mean zero, then we denote the  covariance matrix of $\bm{w}_t$ by
$\bm{\Sigma}_w = \mathbb{E}(\bm{w}_t\bm{w}_t^\top)$.
In addition,  let $\underline{\bm{w}}_T=(\bm{w}_T^\top, \dots, \bm{w}_1^\top)^\top$, and denote its covariance matrix by
\[
\underline{\bm{\Sigma}}_w=\mathbb{E}(\underline{\bm{w}}_T\underline{\bm{w}}_T^\top)=\left (\bm{\Sigma}_w(j-i)\right )_{1\leq i,j\leq T},
\]
where  $\bm{\Sigma}_w(\ell) = \mathbb{E}(\bm{w}_t\bm{w}_{t-\ell}^\top)$ is the lag-$\ell$ autocovariance matrix of $\bm{w}_t$ for $\ell\in\mathbb{Z}$, and $\bm{\Sigma}_w(0)=\bm{\Sigma}_w$. 
For a particular time series $\{\bm{y}_t\}$, accordingly we define $\bm{\Sigma}_y=\mathbb{E}(\bm{y}_t\bm{y}_t^\top)$ and  $\underline{\bm{\Sigma}}_y=\mathbb{E}(\underline{\bm{y}}_T\underline{\bm{y}}_T^\top)=\left (\bm{\Sigma}_y(j-i)\right )_{1\leq i,j\leq T}$, where  $\underline{\bm{y}}_T=(\bm{y}_{T}^\top, \dots, \bm{y}_{1}^\top)^\top$,  $\bm{\Sigma}_y (\ell) = \mathbb{E}(\bm{y}_t\bm{y}_{t-\ell}^\top)$  is the lag-$\ell$ covariance matrix of $\bm{y}_t$ for $\ell\in \mathbb{Z}$, and $\bm{\Sigma}_y=\bm{\Sigma}_y(0)$.

%%%%%%%%%%%%%%%%%%%%%%%%%%%%%%%%%%%%%%%%%%%%%%%%%%%%%%%%%%%%%%%%%%%%%%%%%%%%%%%%%%
\begin{lemma}[Hanson-Wright inequalities for stationary time series]\label{lemma:hansonw}
	Suppose that Assumption \ref{assum:error}  holds for $\{\bm{\varepsilon}_t\}$, and $\{\bm{w}_t\}$ is a time series with the VMA($\infty$) representation,
	\[
	\bm{w}_t= \sum_{j=1}^\infty \bm{\Psi}_j^w  \bm{\varepsilon}_{t-j},
	\]
	where 
	$\bm{\Psi}_j^w\in\mathbb{R}^{M\times N}$ for all $j$, and $\sum_{j=1}^{\infty}\|\bm{\Psi}_j^w \|_{\op}<\infty$. Let  $T_0$ be a fixed integer, and let $T_1$  be a fixed positive integer.
	Then, for any $\bm{M}\in\mathbb{R}^{Q\times M}$ with $Q\geq 1$ and any $\eta>0$, it holds
	\begin{equation*}
		\mathbb{P}\left \{\left |\frac{1}{T_1}\sum_{t=T_0+1}^{T_0+T_1}\|\bm{M}\bm{w}_t\|_2^2 - \mathbb{E}\left (\|\bm{M}\bm{w}_t\|_2^2\right )\right | \geq \eta \sigma^2 \lambda_{\max}(\underline{\bm{\Sigma}}_w)\|\bm{M}\|_{\Fr}^2\right \} \leq 2e^{-c_{\HW}\min(\eta, \eta^2)T_1}.
	\end{equation*}
	%	\end{itemize}
	\end{lemma}
	
	%%%%%%%%%%%%%%%%%%%%%%%%%%%%%%%%%%%%%%%%%%%%%%%%%%%%%%%%%%%%%%%%%%%%%%%%%%%%%%%%%%
	\begin{lemma}[Martingale concentration inequality]\label{lemma:martgl} 
Suppose that Assumption \ref{assum:error} holds for $\{\bm{\varepsilon}_t\}$. Let 
$\mathcal{F}_t=\sigma\{\bm{\varepsilon}_t, \bm{\varepsilon}_{t-1}, \dots\}$ for $t\in\mathbb{Z}$ be a filtration.  Let $\{\bm{y}_t\}$ be a zero-mean  time series, where  $\bm{y}_t=(y_{1, t},\dots, y_{N,t})^\top \in\mathbb{R}^N$ is $\mathcal{F}_{t-1}$-measurable.
Let $T_0$ be a fixed integer,  and let $T_1$ be a fixed positive integer.  
Fix $1\leq i,j\leq N$ and $k\geq 1$. 	For any $a,b>0$, we have 
\begin{equation*}
	\mathbb{P}\left \{ \left |\sum_{t=T_0+1}^{T_0+T_1}\varepsilon_{i,t} y_{j,t-k} \right |\geq a, \; \sum_{t=T_0+1}^{T_0+T_1} y_{j,t-k}^2 \leq b\right \} \leq  2\exp\left \{-\frac{a^2}{2\sigma^2 \lambda_{\max}(\bm{\Sigma}_{\varepsilon})b}\right \}.
\end{equation*}
\end{lemma}

%%%%%%%%%%%%%%%%%%%%%%%%%%%%%%%%%%%%%%%%%%%%%%%%%%%%%%%%%%%%%%%%%%%%%%%%%%%%%%%%%%
\begin{lemma}[Bounds for covariance matrices of stationary time series]\label{lemma:Wcov}
Suppose that Assumption \ref{assum:error} holds for $\{\bm{\varepsilon}_t\}$, and  
%	the vector MA($\infty$) representation of $\{\bm{y}_t\}$  is $\bm{y}_t =\bm{\Psi}_*(B)\bm{\varepsilon}_{t}$, where $\bm{\Psi}_*(B) = \bm{I}_N+\sum_{j=1}^{\infty}\bm{\Psi}_j^* B^j$, and $B$ is the backshift operator. 
$\{\bm{y}_t\}$ has the VMA($\infty$) representation,  $\bm{y}_t =\bm{\Psi}_*(B)\bm{\varepsilon}_{t}$, where $\bm{\Psi}_*(B) = \sum_{j=0}^{\infty}\bm{\Psi}_j^* B^j$,  $B$ is the backshift operator, $\bm{\Psi}_0^*=\bm{I}_N$, and  $\sum_{j=0}^{\infty}\|\bm{\Psi}_j^* \|_{\op}<\infty$.
Let 
\[
\kappa_1=	\lambda_{\min}(\bm{\Sigma}_\varepsilon)\mu_{\min}(\bm{\Psi}_*) \quad\text{and}\quad \kappa_2=\lambda_{\max}(\bm{\Sigma}_\varepsilon)\mu_{\max}(\bm{\Psi}_*),
\] 
where
$\mu_{\min}(\bm{\Psi}_*) = \min_{|z|=1}\lambda_{\min}(\bm{\Psi}_*(z)\bm{\Psi}_*^{\HH}(z))$, $\mu_{\max}(\bm{\Psi}_*) = \max_{|z|=1}\lambda_{\max}(\bm{\Psi}_*(z)\bm{\Psi}_*^{\HH}(z))$, and $\bm{\Psi}_*^{\HH}(z)$ is the conjugate transpose of $\bm{\Psi}_*(z)$.
\begin{itemize}
	\item [(i)] It holds
	\begin{equation*}
		\kappa_1 \leq \lambda_{\min}(\underline{\bm{\Sigma}}_y)\leq \lambda_{\max}(\underline{\bm{\Sigma}}_y) \leq \kappa_2
		\quad \text{and}\quad
		\kappa_1 \leq \lambda_{\min}(\bm{\Sigma}_y)\leq \lambda_{\max}(\bm{\Sigma}_y) \leq  \kappa_2.
	\end{equation*}
	\item [(ii)]  Define  the time series $\{\bm{w}_t\}$ by
	%\[
	$\bm{w}_t=\bm{U}\bm{x}_t=\sum_{i=1}^\infty \bm{U}_i\bm{y}_{t-i}$,
	%\] 
	where $\bm{x}_t=(\bm{y}_{t-1}^\top, \bm{y}_{t-2}^\top, \dots)^\top$, $\bm{U}=(\bm{U}_1, \bm{U}_2, \dots)\in\mathbb{R}^{M\times \infty}$, and $\bm{U}_i$'s are $M\times N$ blocks such that $\sum_{i=1}^{\infty}\|\bm{U}_i\|_{\op}<\infty$. Then, $\{\bm{w}_t\}$ is a zero-mean stationary time series. Moreover,
	\begin{equation}\label{eq:sigmaw0}
		\kappa_1 \sigma_{\min}^2(\bm{U})   \leq \lambda_{\min}(\bm{\Sigma}_w) \leq  \lambda_{\max}(\bm{\Sigma}_w)  \leq \kappa_2 \sigma_{\max}^2(\bm{U})
	\end{equation}
	and 
	\begin{equation}\label{eq:sigmaw}
		\lambda_{\max}(\underline{\bm{\Sigma}}_w)\leq \kappa_2 \left (\sum_{i=1}^\infty\|\bm{U}_i\|_{\op}\right )^2.
	\end{equation}
\end{itemize}
\end{lemma}

%%%%%%%%%%%%%%%%%%%%%%%%%%%%%%%%%%%%%%%%%%%%%%%%%%%%%%%%%%%%%%%%%%%%%%%%%%%%%%%%%%
\begin{lemma}\label{alem:rsc3aux}
Suppose that the conditions in Lemma \ref{lemma:Wcov} hold,  $T_0$ is a fixed integer, and  $T_1$  is a fixed positive integer.
For any $\bm{u}\in\mathbb{R}^N$ and  $\eta\geq 1$, if $\eta T_1 \geq c_{\HW}^{-1}\log 2$, then
\begin{equation*}
	\mathbb{P}\left \{\forall j\geq 1: \frac{1}{T_1} \sum_{t=T_0+1}^{T_0+T_1}  (\bm{u}^\top\bm{y}_{t-j})^2   \leq  \kappa_2 (\eta j \sigma^2+1) \|\bm{u}\|_2^2 \right \} \geq 1-4e^{-c_{\HW}\eta T_1},
\end{equation*}
where $c_{\HW}>0$ is the absolute constant in Lemma  \ref{lemma:hansonw}.
\end{lemma}
%%%%%%%%%%%%%%%%%%%%%%%%%%%%%%%%%%%%%%%%%%%%%%%%%%%%%%%%%%%%%%%%%%%%%%%%%%%%%%%%%%

Lastly, the proof of Lemma \ref{lemma:rsclasso} also relies on Lemma \ref{lemma:epsilon-net} below.
Let
\begin{equation*}%\label{eq:Xi}
\bm{\Xi} = \left\{ \bm{M}(\bm{\phi}) \in  \mathbb{R}^{N \times N (r+2s)} \mid \bm{\phi} \in \bm{\Phi} \right\}\quad \text{and}\quad \bm{\Xi}_1=\{\bm{M} \in \bm{\Xi}\mid \|\bm{M}  \|_{\Fr} = 1\},
\end{equation*}
where $\bm{M}(\bm{\phi})$ is defined as in Section \ref{subsec:notations}. The following definition is used in Lemma \ref{lemma:epsilon-net}.

\begin{definition}[Generalized $\epsilon$-net of $\bm{\Xi}_1$]\label{def}
For any $\epsilon>0$, we say that
$\bm{\bar{\Xi}}(\epsilon)$ is a generalized $\epsilon$-net of $\bm{\Xi}_1$ if $\bm{\bar{\Xi}}(\epsilon)\subset\bm{\Xi}$, and for any $\bm{M}(\bm{\phi}) \in \bm{\Xi}_1$, there exists $\bm{M}(\bm{\bar{\phi}})\in \bm{\bar{\Xi}}(\epsilon)$ such that $\| \bm{M}(\bm{\phi})- \bm{M}(\bm{\bar{\phi}})\|_{\Fr} \leq \epsilon$. However,  $\bm{\bar{\Xi}}(\epsilon)$ is not required to be a subset of $\bm{\Xi}_1$; that is, $\bm{\bar{\Xi}}(\epsilon)$ may not be an $\epsilon$-net of $\bm{\Xi}_1$.
\end{definition}

\begin{lemma}[Covering number and discretization for $\bm{\Xi}_1$]\label{lemma:epsilon-net}
For any $0<\epsilon<1$, let $\bm{\bar{\Xi}}(\epsilon)$ be a minimal generalized $\epsilon$-net of $\bm{\Xi}_1$ in the Frobenius norm. 
\begin{itemize}
	\item [(i)] The cardinality of $\bm{\bar{\Xi}}(\epsilon)$ satisfies
	\[
	\log |\bm{\bar{\Xi}}(\epsilon)| \leq (r+2s)\log\{3u_\phi/(l_\phi \epsilon)\}, 
	\]
	where $l_\phi=(\sqrt{2}\overline{\alpha}_\ma)^{-1} \min_{1\leq k\leq s}\gamma_{k}^*$ and $u_\phi=\underline{\alpha}_\ma^{-1}$.
	
	\item[(ii)] For any $\bm{M}\in \bm{\bar{\Xi}}(\epsilon)$, it holds $l_\phi/u_\phi\leq \|\bm{M}\|_{\Fr}\leq u_\phi/l_\phi$.
	
	\item[(iii)] For any matrix $\bm{V} \in \mathbb{R}^{N(r+2s) \times T}$, it holds 
	\begin{align*} %\label{eq:disc2}
		\sup_{ \bm{M}\in \bm{\Xi}_1} \|\bm{M} \bm{V}\|_{\Fr} \leq (1- \epsilon)^{-1} \max_{ \bm{M} \in \bm{\bar{\Xi}}(\epsilon)}\|\bm{M} \bm{V}\|_{\Fr}.
	\end{align*}
\end{itemize} 
\end{lemma}

%%%%%%%%%%%%%%%%%%%%%%%%%%%%%%%%%%%%%%%%%%%%%%%%%%%%%%%%%%%%%%%%%%%%%%%%%%%%%%%%%%
\subsection{Proofs of Lemmas \ref{lemma:hansonw}--\ref{lemma:epsilon-net}}\label{asec:auxproof}

\begin{proof}[Proof of Lemma \ref{lemma:hansonw}]
First it is obvious that $\{\bm{w}_t\}$  is a zero-mean stationary time series. Without loss of generality, we let $T_0=0$ and $T_1=T$ in what follows. 

Under Assumption \ref{assum:error}, $\bm{\varepsilon}_{t}=	\bm{\Sigma}_\varepsilon^{1/2}\bm{\xi}_t$, and  all coordinates of the vector $\bm{\xi} = (\bm{\xi}_{T-1}^\top, \bm{\xi}_{T-2}^\top, \dots)^\top$ are independent and $\sigma^2$-sub-Gaussian with mean zero and variance one.
In addition, by the vector MA($\infty$) representation of $\bm{w}_t$, we have  $\underline{\bm{w}}_T = \underline{\bm{\Psi}}^w\bm{\xi}$, 
where
\begin{align*}% \label{eq:HD}
	\underset{TM \times \infty }{\underline{\bm{\Psi}}^w}= \left( \begin{matrix}
		\bm{\Psi}_1^w \bm{\Sigma}_\varepsilon^{1/2}&\bm{\Psi}_2^w \bm{\Sigma}_\varepsilon^{1/2} &\bm{\Psi}_3^w \bm{\Sigma}_\varepsilon^{1/2}&\cdots&\bm{\Psi}_{T}^w \bm{\Sigma}_\varepsilon^{1/2}&\cdots\\
		&\bm{\Psi}_1^w \bm{\Sigma}_\varepsilon^{1/2}&\bm{\Psi}_2^w \bm{\Sigma}_\varepsilon^{1/2}&\cdots&\bm{\Psi}_{T-1}^w \bm{\Sigma}_\varepsilon^{1/2}&\cdots\\
		&&\ddots&&\vdots&&\\
		&&&&\bm{\Psi}_1^w \bm{\Sigma}_\varepsilon^{1/2} &\cdots
	\end{matrix}\right).
\end{align*}
Then, it holds
\begin{equation}\label{eq:varw}
	\underline{\bm{\Sigma}}_w=\mathbb{E}(\underline{\bm{w}}_T\underline{\bm{w}}_T^\top)=\underline{\bm{\Psi}}^w(\underline{\bm{\Psi}}^{w})^\top.
\end{equation}

Define the vector $\underline{\bm{m}}_T=((\bm{M}\bm{w}_T)^\top, \dots,  (\bm{M}\bm{w}_1)^\top)^\top=(\bm{I}_T\otimes \bm{M})\underline{\bm{w}}_T$.  Then $\underline{\bm{m}}_T=\bm{P}\bm{\xi}$, where   $\bm{P}=(\bm{I}_T\otimes \bm{M})  \underline{\bm{\Psi}}^w$. As a result,
$\sum_{t=1}^{T}\|\bm{M}\bm{w}_t\|_2^2=\underline{\bm{m}}_T^\top \underline{\bm{m}}_T =\bm{\xi}^\top\bm{P}^\top\bm{P}\bm{\xi}$. Similar to \eqref{eq:varw}, it follows from the Hanson-Wright inequality that
for any $\iota>0$,
\begin{equation}\label{eq:hansonw2}
	\mathbb{P}\left ( \left |\sum_{t=1}^{T} \|\bm{M} \bm{w}_t\|_2^2 - T\mathbb{E}\left (\|\bm{M}\bm{w}_t\|_2^2\right )\right | \geq \iota\right ) \leq 2\exp\left\{ - c_{\HW} \min\left( \frac{\iota}{\sigma^2\|\bm{P}^\top\bm{P}\|_{\op}}, \frac{\iota^2}{\sigma^4 \|\bm{P}^\top\bm{P}\|_{\Fr}^2}\right)\right\}.
\end{equation}
By \eqref{eq:varw}, we have
$\|\bm{P}^\top\bm{P}\|_{\op}=\|\bm{P}\bm{P}^\top\|_{\op}\leq 
\|\bm{M}\bm{M}^\top\|_{\op} \| \underline{\bm{\Psi}}^w(\underline{\bm{\Psi}}^{w})^\top\|_{\op}
\leq \lambda_{\max}(\underline{\bm{\Sigma}}_w)\|\bm{M}\|_{\Fr}^2$. Moreover, 
\begin{align*}
	\trace(\bm{P}^\top\bm{P})=\trace(\bm{P}\bm{P}^\top)&=\trace\{(\bm{I}_T\otimes \bm{M})\underline{\bm{\Sigma}}_w (\bm{I}_T\otimes \bm{M}^\top)\}\\
	&= \textrm{vec}(\bm{I}_T\otimes \bm{M})^\top(\underline{\bm{\Sigma}}_w\otimes\bm{I}_{TQ})\textrm{vec}(\bm{I}_T\otimes \bm{M})
	\leq T\lambda_{\max}(\underline{\bm{\Sigma}}_w) \|\bm{M}\|_{\Fr}^2,
\end{align*}
where the second equality follows from \eqref{eq:varw}. As a result,
\[
\|\bm{P}^\top\bm{P}\|_{\Fr}\leq \sqrt{\|\bm{P}^\top\bm{P}\|_{\op}\trace(\bm{P}^\top\bm{P})}
\leq 
\sqrt{\|\bm{P}\bm{P}^\top\|_{\op}\trace(\bm{P}\bm{P}^\top)} \leq  \sqrt{T} \lambda_{\max}(\underline{\bm{\Sigma}}_w) \|\bm{M}\|_{\Fr}^2.
\]
Taking $\iota=\eta\sigma^2T\lambda_{\max}(\underline{\bm{\Sigma}}_w) \|\bm{M}\|_{\Fr}^2$ in \eqref{eq:hansonw2}, the proof of this lemma is complete.
\end{proof}

%%%%%%%%%%%%%%%%%%%%%%%%%%%%%%%%%%%%%%%%%%%%%%%%%%%%%%%%%%%%%%%%%%%%%%%%%%%%%

\begin{proof}[Proof of Lemma \ref{lemma:martgl}]
By Assumption \ref{assum:error}, $\varepsilon_{i,t}$ is $\sigma^2 \lambda_{\max}(\bm{\Sigma}_{\varepsilon})$-sub-Gaussian. Then, the result follows from  Lemma 4.2 in \cite{simchowitz2018learning}.
\end{proof}

%%%%%%%%%%%%%%%%%%%%%%%%%%%%%%%%%%%%%%%%%%%%%%%%%%%%%%%%%%%%%%%%%%%%%%%%%%%%%

\begin{proof}[Proof of Lemma \ref{lemma:Wcov}]
\noindent\textbf{Proof of (i):}  Consider the spectral density of $\{\bm{y}_t\}$,
\begin{equation*}
	\bm{f}_y(\theta) = (2\pi)^{-1} \bm{\Psi}_*(e^{-i\theta})\bm{\Sigma}_{\varepsilon}\bm{\Psi}_*^{\HH}(e^{-i\theta}), \hspace{5mm}\theta\in [-\pi, \pi].
\end{equation*}
Let 
\[
\mathpzc{M}(\bm{f}_y) = \max_{\theta\in[-\pi,\pi]} \lambda_{\max}(\bm{f}_y(\theta))
\quad\text{and}\quad
\mathpzc{m}(\bm{f}_y) = \min_{\theta\in[-\pi,\pi]} \lambda_{\min}(\bm{f}_y(\theta))
\]
Along the lines of \cite{basu2015regularized},  it holds 
\[
2\pi \mathpzc{m}(\bm{f}_y) \leq \lambda_{\min}(\underline{\bm{\Sigma}}_y)\leq \lambda_{\max}(\underline{\bm{\Sigma}}_y) \leq  2\pi \mathpzc{M}(\bm{f}_y),
\]
\[
2\pi \mathpzc{m}(\bm{f}_y) \leq \lambda_{\min}(\bm{\Sigma}_y)\leq \lambda_{\max}(\bm{\Sigma}_y) \leq  2\pi \mathpzc{M}(\bm{f}_y),
\]
and
\begin{equation}\label{eq:fy}
	\lambda_{\min}(\bm{\Sigma}_\varepsilon)\mu_{\min}(\bm{\Psi}_*) \leq 2\pi \mathpzc{m}(\bm{f}_y)\leq 2\pi \mathpzc{M}(\bm{f}_y) \leq \lambda_{\max}(\bm{\Sigma}_\varepsilon)\mu_{\max}(\bm{\Psi}_*);
\end{equation}
see Proposition 2.3 therein.
Thus, (i) is proved.

\smallskip
\noindent\textbf{Proof of (ii):} Since $\sum_{i=1}^{\infty}\|\bm{U}_i\|_{\op}<\infty$ and $\{\bm{y}_t\}$ is   stationary with mean zero, the time series $\bm{w}_t=\mathcal{W}(B)\bm{y}_t=\mathcal{W}(B)\bm{\Psi}_*(B)\bm{\varepsilon}_{t}$ is also zero-mean and stationary, where $\mathcal{W}(B) = \sum_{i=1}^{\infty}\bm{U}_i B^i$.

For any $\ell\in \mathbb{Z}$, denote by $\bm{\Sigma}_y (\ell) = \mathbb{E}(\bm{y}_t\bm{y}_{t-\ell}^\top)$ the lag-$\ell$ covariance matrix of $\bm{y}_t$, and then $\bm{\Sigma}_y(\ell) = \int_{-\pi}^{\pi} \bm{f}_y(\theta) e^{i \ell \theta} d\theta$. 
For any fixed $\bm{u} \in \mathbb{R}^{N}$ with $\|\bm{u}\|_2 = 1$, 
\begin{align}\label{eq:eigenW}
	\bm{u}^\top \bm{\Sigma}_w \bm{u} &= \bm{u}^\top \mathbb{E}\left (\sum_{j=1}^{\infty} \bm{U}_{j}\bm{y}_{t-j} \sum_{k=1}^{\infty} \bm{U}^\top_{k}\bm{y}_{t-k} \right )\bm{u}  \notag\\
	& =\bm{u}^\top \sum_{j=1}^{\infty} \sum_{k=1}^{\infty} \bm{U}_{j}\bm{\Sigma}_y(k-j) \bm{U}^\top_{k} \bm{u} \notag\\
	& = \int_{-\pi}^{\pi} \sum_{j=1}^{\infty} \sum_{k=1}^{\infty} \bm{u}^\top  \bm{U}_j\bm{f}_y(\theta)  e^{-i(j-k)\theta} \bm{U}_k^\top \bm{u} \, d\theta \notag\\
	& = \int_{-\pi}^{\pi} \bm{u}^\top \mathcal{W}(e^{-i\theta}) \bm{f}_y(\theta)  \mathcal{W}^{\HH}(e^{-i\theta}) \bm{u} \, d\theta,
\end{align}
where  $\mathcal{W}(z) = \sum_{j=1}^{\infty}\bm{U}_jz^j$ for $z\in \mathbb{C}$, and $\mathcal{W}^{\HH}(e^{-i\theta})=\big\{ \mathcal{W}(e^{i\theta}) \big\}^\top$ is the  conjugate transpose of $\mathcal{W}(e^{-i\theta})$.
Since $\bm{f}_y(\theta)$ is Hermitian, $\bm{u}^\top \mathcal{W}(e^{-i\theta}) \bm{f}_y(\theta)  \mathcal{W}^{\HH}(e^{-i\theta})\bm{u}$ is  real for all $\theta \in [-\pi, \pi]$. Then it is easy to see that
\[
\mathpzc{m}(\bm{f}_y) \cdot  \bm{u}^\top \mathcal{W}(e^{-i\theta}) \mathcal{W}^{\HH}(e^{-i\theta}) \bm{u}  \leq
\bm{u}^\top \mathcal{W}(e^{-i\theta}) \bm{f}_y(\theta) \mathcal{W}^{\HH}(e^{-i\theta}) \bm{u}
\leq \mathpzc{M}(\bm{f}_y) \cdot \bm{u}^\top \mathcal{W}(e^{-i\theta}) \mathcal{W}^{\HH}(e^{-i\theta}) \bm{u}.
\]
Moreover, since $\int_{-\pi}^{\pi} e^{i\ell\theta} d \theta = 0$ for any $\ell\neq0$, we can show that
\begin{align*}
	\int_{-\pi}^{\pi}\bm{u}^\top \mathcal{W}(e^{-i\theta}) \mathcal{W}^{\HH}(e^{-i\theta}) \bm{u} \, d\theta 
	&	=\int_{-\pi}^{\pi} \sum_{j=1}^{\infty} \sum_{k=1}^{\infty} \bm{u}^\top  \bm{U}_j e^{-i(j-k)\theta} \bm{U}_k^\top \bm{u} \, d\theta \\
	&	= 2\pi \bm{u}^\top  \bm{U} \bm{U}^\top \bm{u}.
\end{align*}
which, together with the fact of $\|\bm{u}\|_2 = 1$, implies that
\begin{align} \label{eq:WW}
	2\pi \sigma_{\min}^2(\bm{U}) \leq \int_{-\pi}^{\pi}\bm{u}^\top \mathcal{W}(e^{-i\theta}) \mathcal{W}^{\HH}(e^{-i\theta})  \bm{u} \, d\theta  \leq 2\pi \sigma_{\max}^2(\bm{U}).
\end{align}
In view of \eqref{eq:fy}--\eqref{eq:WW}, we accomplish the proof of \eqref{eq:sigmaw0}.

To verify \eqref{eq:sigmaw}, note that   the spectral density of $\{\bm{w}_t\}$ is  
\[
\bm{f}_w(\theta) = \mathcal{W}(e^{-i\theta})\bm{f}_y(\theta)\mathcal{W}^{\HH}(e^{-i\theta}), \hspace{5mm}\theta\in [-\pi, \pi];\]  
see Section 9.2 of \cite{Priestley81}. Then 
\begin{align*}
	\mathpzc{M}(\bm{f}_w)=\max_{\theta\in[-\pi,\pi]} \lambda_{\max}(\bm{f}_w(\theta)) & \leq \mathpzc{M}(\bm{f}_y) \max_{\theta\in[-\pi,\pi]} \lambda_{\max}\{\mathcal{W}(e^{-i\theta})\mathcal{W}^{\HH}(e^{-i\theta})\}\\
	&= \mathpzc{M}(\bm{f}_y) \max_{\theta\in[-\pi,\pi]} \left \| \sum_{j=1}^{\infty}  \bm{U}_j e^{-ij\theta} \right \|_{\op}^2\\
	&\leq \mathpzc{M}(\bm{f}_y) \left (\sum_{j=1}^{\infty}  \|\bm{U}_j\|_{\op} \right )^2 
\end{align*}
In addition, by a method similar to the proof of Proposition 2.3 in \cite{basu2015regularized}, we can show that 
\[
\lambda_{\max}(\underline{\bm{\Sigma}}_w)\leq 2\pi \mathpzc{M}(\bm{f}_w).
\]
Combining the above results with \eqref{eq:fy}, the proof of \eqref{eq:sigmaw} is complete.
\end{proof}

%The next two lemmas are auxiliary results for establishing the restricted strong convexity. Used in the proof of Lemma \ref{lemma:rsclasso}, Lemma \ref{lemma:epsilon-net}  is for the proof of claim (ii), and Lemma \ref{alem:rsc3aux} for claims (iii) and (iv). 

\begin{proof}[Proof of Lemma \ref{alem:rsc3aux}]
We first fix $j\geq 1$. Applying Lemma \ref{lemma:hansonw}(ii) with $\bm{M}=\bm{u}^\top$ and $\bm{w}_t=\bm{y}_{t-j}$, together with the result
\[
\lambda_{\max}(\underline{\bm{\Sigma}}_w)=\lambda_{\max}(\underline{\bm{\Sigma}}_y)\leq  \kappa_2
%\quad\text{and}\quad
\]
as implied by Lemma \ref{lemma:Wcov}(i), we can show that 
\begin{equation*}
	\mathbb{P}\left \{\left |\frac{1}{T_1} \sum_{t=T_0+1}^{T_0+T_1} (\bm{u}^\top\bm{y}_{t-j})^2 -\mathbb{E}\{(\bm{u}^\top\bm{y}_{t-j})^2\}\right | \geq  \eta j \sigma^2\kappa_2  \|\bm{u}\|_2^2 \right \} \leq 2e^{-c_{\HW}\min(\eta j, \eta^2 j^2)T_1}= 2e^{-c_{\HW} j\eta  T_1}.
\end{equation*}
holds for any $\eta>0$. In addition, by Lemma \ref{lemma:Wcov}(i),
\[
\mathbb{E}\{(\bm{u}^\top\bm{y}_{t-j})^2\}\leq \lambda_{\max}(\bm{\Sigma}_y)  \|\bm{u}\|_2^2 \leq  \kappa_2 \|\bm{u}\|_2^2.
\]
Thus, we further have
\begin{equation*}
	\mathbb{P}\left \{\frac{1}{T_1} \sum_{t=T_0+1}^{T_0+T_1} (\bm{u}^\top\bm{y}_{t-j})^2 \geq  \kappa_2 (\eta j \sigma^2+1)\right \} \leq 2e^{-cj \eta  T_1}.
\end{equation*}
By considering the union bound over all $j\geq1$, we have
\[
\mathbb{P}\left \{\exists j\geq 1: \frac{1}{T_1} \sum_{t=T_0+1}^{T_0+T_1} (\bm{u}^\top\bm{y}_{t-j})^2  \geq  \kappa_2 (\eta j \sigma^2+1) \right \} \leq \sum_{j=1}^{\infty}2e^{-cj\eta T_1}\leq 4e^{-c_{\HW}\eta T_1},
\]
if $\eta T_1\geq c_{\HW}^{-1}\log 2$.  The proof is complete.
\end{proof}
%%%%%%%%%%%%%%%%%%%%%%%%%%%%%%%%%%%%%%%%%%%%%%%%%%%%%%%%%%%%%%%%%%%%%%%%%%%%%%%%%%
\begin{proof}[Proof of Lemma \ref{lemma:epsilon-net}]
\textbf{Proof of (i):}	
Note that if $\|\bm{M}(\bm{\phi})\|_{\Fr}=1$, it follows from \eqref{eq:DFr} that
$l_\phi \leq \|\bm{\phi}\|_2\leq u_\phi$. This implies
$\bm{\Xi}_1\subset \{\bm{M}(\bm{\phi})\mid \bm{\phi} \in \bm{\Pi}\}$,
where
\[
\bm{\Pi}=\{\bm{\phi}\in\mathbb{R}^{r+2s}\mid l_\phi \leq \|\bm{\phi}\|_2\leq u_\phi\}.
\]
Hence, the problem of covering $\bm{\Xi}_1$ can be converted into that of covering $\bm{\Pi}$.

For any fixed $\epsilon >0 $,  let $\bm{\bar{\Pi}}(\epsilon)$ be a minimal $(l_\phi \epsilon)$-net for $\bm{\Pi}$ in the Euclidean norm. Denote
\[
\bm{\bar{\Xi}}(\epsilon) = \left\{ \bm{M}(\bm{\phi}) \in  \mathbb{R}^{N \times N  (r+2s)} \mid  \bm{\phi}  \in \bm{\bar{\Pi}}(\epsilon) \right\}.
\]
Thus,  for every $\bm{M}(\bm{\phi}) \in \bm{\Xi}_1$, there exists $\bm{M}( \bm{\bar{\phi}})\in \bm{\bar{\Xi}}(\epsilon)$ with    $\bm{\bar{\phi}} \in \bm{\bar{\Pi}}(\epsilon)$ such that $\|\bm{\phi} - \bm{\bar{\phi}}\|_2\leq l_{\phi}\epsilon$. By \eqref{eq:DFr}, we further have
\begin{align*}
	\|\bm{M}(\bm{\phi})- \bm{M}( \bm{\bar{\phi}})\|_{\Fr} = \|\bm{M}(\bm{\phi} - \bm{\bar{\phi}}) \|_{\Fr} \leq \epsilon.
\end{align*}
In addition, note that $\bm{\bar{\Xi}}(\epsilon)\subset \bm{\Xi}$. Therefore, $\bm{\bar{\Xi}}(\epsilon)$ is a generalized $\epsilon$-net of $\bm{\Xi}_1$.
Moreover, by a standard volumetric argument (see also Corollary 4.2.13 in \cite{Vershynin2018} for details), the cardinality of $\bm{\bar{\Pi}}(\epsilon)$ satisfy
\begin{align*}%\label{eq:Pi1net}
	\log|\bm{\bar{\Pi}}(\epsilon)| \leq (r+2s)\log\{3u_\phi/(l_\phi \epsilon)\}.
\end{align*}
Noting that  $|\bm{\bar{\Xi}}(\epsilon)| \leq  |\bm{\bar{\Pi}}(\epsilon)|$, the proof of (i) is complete.

\smallskip\noindent
\textbf{Proof of (ii):} Since 
$\bm{\bar{\Pi}}(\epsilon)\subset\bm{\Pi}$, we have
\[\bm{\bar{\Xi}}(\epsilon)\subset \left\{ \bm{M}(\bm{\phi}) \in  \mathbb{R}^{N \times N (r+2s)} \mid  \bm{\phi}  \in \bm{\Pi}\right\}.\] 
Then by \eqref{eq:DFr}, for any $\bm{M}\in \bm{\bar{\Xi}}(\epsilon)$, it holds
\[
l_\phi/u_\phi=\underline{\alpha}_\ma l_\phi  \leq \|\bm{M}(\bm{\phi})\|_{\Fr} \leq \frac{\sqrt{2}\overline{\alpha}_\ma}{\min_{1\leq k\leq s}\gamma_{k}^*} u_\phi =u_\phi/l_\phi.
\]
Thus, (ii) is proved.

\smallskip\noindent
\textbf{Proof of (iii):} From the proof of (i), for every $\bm{M}:=\bm{M}(\bm{\phi}) \in \bm{\Xi}_1$, there exists $\bar{\bm{M}}:=\bm{M}( \bm{\bar{\phi}})\in \bm{\bar{\Xi}}(\epsilon)$ with    $\bm{\bar{\phi}} \in \bm{\bar{\Pi}}(\epsilon)$ such that $\|\bm{M}- \bar{\bm{M}}\|_{\Fr} = \|\bm{M}(\bm{\phi} - \bm{\bar{\phi}}) \|_{\Fr} \leq \epsilon$. In addition, since $\bm{M}(\bm{\phi})$ is linear in $\bm{\phi}$, we have $(\bm{M}- \bar{\bm{M}})/\|\bm{M}- \bar{\bm{M}}\|_{\Fr}=\bm{M}(\bm{\phi} - \bm{\bar{\phi}})/\|\bm{M}(\bm{\phi} - \bm{\bar{\phi}}) \|_{\Fr} \in\bm{\Xi}_1$. Then for any $\bm{M}\in\bm{\Xi}_1$, we can show that
\[
\| \bm{M}_{(1)} \bm{V}\|_{\Fr}  
\leq \| \bar{\bm{M}} \bm{V}\|_{\Fr} + \| (\bm{M}-\bar{\bm{M}}) \bm{V}\|_{\Fr} 
\leq \max_{\bar{\bm{M}} \in \bm{\bar{\Xi}}(\epsilon)} \| \bar{\bm{M}} \bm{V}\|_{\Fr}  + \epsilon\sup_{\bm{M} \in \bm{\Xi}_1} \| \bm{M} \bm{V}\|_{\Fr}.
\]
Taking supremum over all  $\bm{M} \in \bm{\Xi}_1$ on both sides, we accomplish the proof of Lemma \ref{lemma:epsilon-net}.
\end{proof}

\putbib[SparseSARMA]
\end{bibunit}

\end{document}